\documentclass[pdflatex, sn-mathphys-num]{sn-jnl}

\usepackage[utf8]{inputenc}

\usepackage{lineno}

\usepackage{graphicx}%
\usepackage{multirow}%
\usepackage{amsmath,amssymb,amsfonts}%
\usepackage{amsthm}%
\usepackage{mathrsfs}%
\usepackage[title]{appendix}%
\usepackage{textcomp}%
\usepackage{manyfoot}%
\usepackage{booktabs}%
\usepackage{algorithm}%
\usepackage{algorithmicx}%
\usepackage{algpseudocode}%
\usepackage{listings}%

\usepackage[dvipsnames]{xcolor}

\usepackage{mathtools}

\usepackage{xspace}

\usepackage{nicefrac}
\usepackage{stmaryrd} %
\let\merge\undefined
\usepackage[misc]{ifsym} %
\usepackage{cancel}

\usepackage{tikz}
\usetikzlibrary{matrix}
\usetikzlibrary{calc}
\usetikzlibrary{positioning}
\usetikzlibrary{scopes}
\usetikzlibrary{chains}
\usetikzlibrary{matrix}
\usetikzlibrary{decorations.pathmorphing, decorations.pathreplacing, decorations.shapes} 
\usetikzlibrary{fit}
\usetikzlibrary{backgrounds}
\usetikzlibrary{shapes.geometric, shapes.multipart}
\usetikzlibrary{intersections}

\usepackage{graphicx}
\usepackage{caption}
\usepackage{subcaption}
\usepackage{wrapfig}

\usepackage[inline]{enumitem}

\usepackage{prettyref}
\newcommand{\rref}[2][]{\prettyref{#2}}
\newrefformat{sec}{Section\,\ref{#1}}
\newrefformat{def}{Def.\,\ref{#1}}

\newrefformat{thm}{Theorem\,\ref{#1}}
\newrefformat{prop}{Proposition\,\ref{#1}}
\newrefformat{lem}{Lemma\,\ref{#1}}
\newcommand{\rreflems}[2]{Lemmas\,\ref{#1} and~\ref{#2}}
\newrefformat{cor}{Corollary\,\ref{#1}}
\newrefformat{ex}{Example\,\ref{#1}}
\newcommand{\rrefexamples}[2]{Examples\,\ref{#1} and~\ref{#2}}
\newrefformat{tab}{Table\,\ref{#1}}
\newrefformat{fig}{Fig.\,\ref{#1}}
\newcommand{\rreffigs}[2]{Figs.\,\ref{#1} and~\ref{#2}}
\newrefformat{app}{Appendix\,\ref{#1}}
\newrefformat{elec}{Electronic Appendix\,\ref{#1}}
\newrefformat{itm}{item \ref{#1}\@\xspace}
\newrefformat{eq}{equation~(\ref{#1})}
\newrefformat{rem}{Remark\,\ref{#1}}
\newrefformat{ft}{Footnote\,\ref{#1}}
\newrefformat{notation}{Notation\,\ref{#1}}

\newrefformat{pre}{(\ref{#1})}

\makeatletter
\newcommand\footnoteref[1]{\protected@xdef\@thefnmark{\ref*{#1}}\@footnotemark}
\makeatother

\xspaceaddexceptions{\footnoteref}

\usepackage{mbmaths}
\usepackage[
    greekpnames, 
    smallsig
]{mblogic}
\usepackage[usnames]{mbmisc}
\usepackage{xproofs} %

\tikzstyle{disc}=[draw, shape=circle, inner sep=.05cm]
\tikzstyle{ghost}=[shape=circle, inner sep=0.05cm]
\tikzstyle{continuous}=[decorate, decoration={coil, aspect=0, amplitude=.2mm, segment length=4pt}]
\tikzstyle{badtrace}=[dashed, BrickRed]

\tikzstyle{comarrow}=[dashed, -stealth, MidnightBlue]
\tikzstyle{noisy}=[comarrow, solid, decorate, decoration={zigzag, segment length=2mm, amplitude=0.4mm, post length= 1mm}, -stealth]

\newcommand{\notcheckmark}{\checkmark\!\!\!\raisebox{2.2pt}{\textbf{\tiny\textbackslash}}}

\newcommand{\distpicto}{\text{\scriptsize$|\kern-4pt\leftrightarrow\kern-4pt|$}}

\newcommand{\success}{\textcolor{OliveGreen}{$\checkmark$}}
\newcommand{\deny}{\textcolor{BrickRed}{$\notcheckmark$}}
\newcommand{\termination}{\textcolor{BrickRed}{$\times$}}
\newcommand{\measurment}{\textcolor{MidnightBlue}{$\distpicto$}}

\newsavebox\disc
\sbox\disc{\tikz[baseline=-0.5ex]{\node[disc] {};}}

\newsavebox\continuous
\sbox\continuous{\tikz[baseline=-0.5ex]{\draw[continuous] (0,0) -- (.5,0);}}

\newsavebox\communication
\sbox\communication{\tikz[baseline=-0.5ex]{\draw[comarrow] (0,0) -- (.5,0);}}

\newsavebox\lossy
\sbox\lossy{\tikz[baseline=-0.5ex]{\draw[comarrow] (0,0) -- (.5,0) node [right, inner sep = 0] {\termination};}}

\newsavebox\noisy
\sbox\noisy{\tikz[baseline=-0.5ex]{\draw[noisy] (0,0) -- (.5,0);}}

\newsavebox\crash
\sbox\crash{\tikz[baseline=-0.5ex]{\draw[badtrace] (0,0) -- (.5,0);}}

\theoremstyle{thmstyleone}%
\newtheorem{theorem}{Theorem}%
\newtheorem{proposition}[theorem]{Proposition}%
\newtheorem{lemma}[theorem]{Lemma}%
\newtheorem{corollary}[theorem]{Corollary}%

\theoremstyle{thmstyletwo}%
\newtheorem{example}[theorem]{Example}%
\newtheorem{remark}[theorem]{Remark}%

\theoremstyle{thmstylethree}%
\newtheorem{definition}[theorem]{Definition}%

\begin{document}

\title{Complete Dynamic Logic of Communicating Hybrid Programs}

\author*[3,1]{\fnm{Marvin} \sur{Brieger}}\email{brieger@kit.edu}

\author[2]{\fnm{Stefan} \sur{Mitsch}}\email{smitsch@depaul.edu}

\author[3]{\fnm{Andr\'e} \sur{Platzer}}\email{platzer@kit.edu}

\affil*[1]{\orgname{LMU Munich}, \orgaddress{\country{Germany}}}

\affil[2]{\orgname{DePaul University}, \orgaddress{\state{Illinois}, \country{USA}}}

\affil[3]{\orgname{Karlsruhe Institute of Technology}, \orgaddress{\country{Germany}}}

\newcommand{\ctrlNotify}{\progtt{comm}}
\newcommand{\ctrlUpdate}{\progtt{upd}}
\newcommand{\ctrlVelocity}{\progtt{velo}}
\newcommand{\ctrlDistance}{\progtt{dist}}

\newcommand{\tarvelo}{v_{tar}}
\newcommand{\maxvelo}{V}
\newcommand{\mespos}{m}
\newcommand{\waitvar}{w}

\newcommand{\periodicity}{\epsilon}

\newcommand{\safevelo}[1]{\nicefrac{#1}{\periodicity}}
\newcommand{\maxdist}{\periodicity \maxvelo}

\newcommand{\parameters}[1]{(#1)}
\newcommand{\invariant}[1]{\textnormal{\uppercase{#1}}}

\newcommand{\fAssume}{\A}
\newcommand{\lCommit}{\Commit}

\newcommand{\since}[2]{\traceOp[#1]{since}(#2)}
\newcommand{\till}[2]{\traceOp[#1]{till}(#2)}

\newif\iflongversion\longversiontrue

\abstract{
    This article presents a relatively complete proof calculus for the dynamic logic of communicating hybrid programs \dLCHP.
	Beyond hybrid systems, 
	communicating hybrid programs not only feature mixed discrete and continuous dynamics 
	but also their parallel interactions in parallel hybrid systems.
	This not only combines the subtleties of hybrid and discrete parallel systems,
	but parallel hybrid dynamics necessitates that all parallel subsystems synchronize in time and evolve truly simultaneously.
	To enable compositional reasoning 
	nevertheless,
	\dLCHP combines differential dynamic logic \dL with mutual abstraction of subsystems by assumption-commitment (ac) reasoning.
	The resulting proof calculus preserves the essence of dynamic logic axiomatizations,
	while revealing---and being driven by---a new modal logic view onto ac-reasoning.
	
	The \dLCHP proof calculus is shown to be complete relative to \comFOD, 
	the first-order logic of differential equation properties FOD augmented with communication traces.
	This confirms that the calculus covers \emph{all} aspects of parallel hybrid systems,
	because it lacks no axioms 
	to reduce \emph{all} their dynamical effects to the assertion logic.
	Additional axioms for encoding communication traces enable a provably correct equitranslation between \comFOD and FOD,
	which reveals the possibility of representational succinctness in parallel hybrid systems proofs.
	Transitively, this establishes a full proof-theoretical alignment of \dLCHP and \dL,
	and shows that reasoning about parallel hybrid systems is exactly as hard as reasoning about hybrid systems, continuous systems, or discrete systems.
}

\keywords{Parallel hybrid systems, Parallel programs, Hybrid systems, Differential dynamic logic, Assumption-commitment reasoning, CSP, Completeness}

\maketitle

\section{Introduction} \label{sec:intro}

This article studies parallel interactions of hybrid systems,
which model combined discrete, continuous, \emph{and} parallel dynamics.
The study of parallel hybrid systems safety is important, because safety-critical cyber-physical systems naturally feature separate subsystems that evolve in parallel,
\eg multiple trains in train control systems~\cite{DBLP:conf/icfem/PlatzerQ09},
multiple planes in aircraft collision avoidance~\cite{DBLP:conf/fm/PlatzerC09},
and multiple cars in adaptive cruise control~\cite{DBLP:conf/fm/LoosPN11}.
Despite this prevalence in applications
and special component-based approaches to mitigate the parallel state space explosion~\cite{DBLP:journals/sttt/MullerMRSP18,Lunel2019,Kamburjan2020,Kamburjan2022},
the \emph{parallelism} of hybrid systems itself is only partially understood.
That is why models and safety proofs for hybrid systems are still limited to cumbersome, ad-hoc workarounds for the absence of parallelism.

This article studies the \emph{dynamic logic of communicating hybrid programs \dLCHP}~\cite{Brieger2023} for modeling and verification of parallel hybrid systems.
The logic \dLCHP studies parallel hybrid systems from the perspective of logics of multi-dynamical systems \cite{DBLP:conf/cade/Platzer16,Platzer18} in order to identify their fundamental building blocks
and reasoning principles.
Beyond previous work \cite{Brieger2023},
we show relative completeness of \dLCHP's proof calculus.
This yields a \emph{comprehensive} characterization of parallel hybrid systems,
because the \dLCHP calculus identifies \emph{all} essential logical principles behind the mixture of dynamics in parallel hybrid systems.
Since the \dLCHP proof calculus only features axioms for modular structural decomposition, this relative completeness result also justifies why compositional verification is possible for parallel hybrid systems.

Hybrid systems combine discrete and continuous dynamics following discrete jumps and differential equations.
They are fundamentally challenging, as they are not semidecidable \cite{Henzinger1995}.
In contrast, parallel systems run multiple discrete subsystems simultaneously that are tied together by communication.
They are subject to the state space explosion problem~\cite{Clarke2001},
and it took considerable effort \cite{Roever1997,deRoever2001} to turn 
the early noncompositional proof systems 
\cite{OwickiGries1976,LevinGries1981}
into compositional methods for discrete parallelism \cite{Misra1981,Xu1997},
which mitigate the state space explosion.
As a result, 
each---the verification of hybrid \cite{Alur2011,Alur1993,Henzinger1996,DBLP:journals/jar/Platzer08} and parallel systems \cite{AptdeBoerOlderog10,LevinGries1981,OwickiGries1976,deRoever2001}---already is a significant separate challenge.
The combination of hybrid and parallel dynamics poses genuinely new additional challenges:
While different variants of parallelism are considered for discrete systems~\cite{Brookes1996,Brookes2002}, 
parallel hybrid systems require true simultaneous parallel composition,
where the subsystems always agree on the duration of their continuous dynamics.%
\footnote{
	By contrast, a semantics in which the parallel hybrid dynamics could disagree on their duration%
	---even by as little as a second---%
	would be counterfactual for the correct passage of time in classical mechanics.
}
Consequently, even compositional proof techniques must maintain enough insight into the global flow of time when decomposing parallel hybrid systems.

The \emph{dynamic logic of communicating hybrid programs \dLCHP}~\cite{Brieger2023} extends differential dynamic logic \dL for hybrid systems~\cite{DBLP:journals/jar/Platzer08, DBLP:journals/jar/Platzer17, Platzer18} 
with support to model and verify parallelism.
\emph{Communicating hybrid programs} (CHPs) are a compositional model for parallel interactions of hybrid systems.
Given CHPs~$\alpha, \beta$,
\eg modeling cars or robots, the parallel composition $\alpha \parOp \beta$ models their simultaneous evolution during which~$\alpha$ and~$\beta$ may communicate synchronously via channels (loss and delay can be modeled).
For compositional verification,
\dLCHP blends the dynamic logic \cite{Pratt1976, Harel1979,Harel_et_al_2000} setup of \dL with assumption-commitment~(ac) reasoning~\cite{Misra1981,AcHoare_Zwiers,deRoever2001}
to enable mutual abstraction of parallel program effects.
The ac-box $[ \alpha ] \ac \psi$ introduced for this purpose complements the safety modality $[ \alpha ] \psi$,
which states that the postcondition~$\psi$ holds in all worlds reachable by program $\alpha$,
where the assumption $\A$ limits any possible incoming communication while the commitment $\Commit$ is a promise about all outgoing communication.

The main subject of this article is a modular, compositional, and sound Hilbert-type proof calculus for~\dLCHP,
which supports verification of parallel hybrid systems by structural top-down decomposition.
The calculus is \emph{truly compositional}
because it verifies a parallel composition from local specifications of the subsystems that are based only on their observable behavior \cite{deRoever2001}.
This is to be contrasted with Hybrid Hoare-logics (HHLs)~\cite{Liu2010,Guelev2017,Wang2012},
which are noncompositional by design~\cite{Liu2010}
or whose calculi explode the combinatorial product of all parallel interactions~\cite{Wang2012, Guelev2017}.
This exhaustive unfolding leaves no room for local abstractions based on the relevant program behavior.
In contrast, the \dLCHP calculus supports \emph{complexity-on-demand reasoning},
which admits coarse local abstractions of the parallel program effects,
yet always supports sufficiently eloquent abstractions for completeness.

The \dLCHP calculus is \emph{modular} and develops a new axiomatic foundation for parallel systems.
Its highlight is the parallel injection axiom 
\([ \alpha ] \psi \rightarrow [ \alpha \parOp \beta ] \psi\)
as the \emph{only} reasoning principle for safety of parallel hybrid systems,
which is sound if $\beta$ has no influence on the truth of $\psi$ \cite{Brieger2023}.
Despite its asymmetry, 
safety reasoning for parallel hybrid systems can be completely based on parallel injection
once combined with elementary modal logic principles that enable the suitable combination of the insights from successive injections of parallel subsystems.
This also reveals that parallel systems do not need the classical but complex and highly composite proof rules in Hoare-style ac-reasoning~\cite{AcHoare_Zwiers,AcSemantics_Zwiers}.
The modularity of \dLCHP is grounded in a
novel modal view of ac-reasoning,
which enables its graceful blending into \dLCHP.
Graceful means that \dLCHP generalizes the Pratt-Segerberg axiom system \cite{Pratt1976,Segerberg1982, Harel1979} for dynamic logic whenever possible.
For completeness,
this article adds axioms for the previously \cite{Brieger2023} omitted ac-diamond $\langle \alpha \rangle \ac \psi$, 
the modal dual of the ac-box $[ \alpha ] \ac \psi$.

The main contribution of this article is a relative completeness proof for the \dLCHP calculus.
This leads to the fundamental insight that parallel hybrid systems in~\dLCHP and hybrid systems in \dL 
are \emph{proof-theoretically equivalent}:
\begin{equation*}
	\text{parallel hybrid} = \text{hybrid} \qquad\text{(proof-theoretically)}
\end{equation*}
Formally, \dLCHP is proven relatively complete for \dL's  oracle logic FOD,
the first-order logic of differential equation properties~\cite{DBLP:journals/jar/Platzer08}.
As central milestone,
\dLCHP is proven complete relative to \comFOD,
which extends FOD with communication traces.
This already confirms that \dLCHP's calculus identifies \emph{all} essential dynamics constituting parallel hybrid systems,
because all dynamical effects reduce to the assertion logic \comFOD.
The subsequent reduction from \comFOD to FOD addresses the encoding of communication traces,
which was postponed for modularity.
In summary, properties of parallel hybrid systems can be proven in \dLCHP to the same extent as properties of hybrid systems in \dL,
because both align with the provability of properties of continuous systems in FOD:
\begin{equation*}%
	\dLCHP \overset{\text{new}}{=} \comFOD \overset{\text{new}}{=} \text{FOD} \overset{\text{\cite{DBLP:journals/jar/Platzer08}}}{=} \dL \qquad\text{(proof-theoretically)}
\end{equation*}
This does not mean
that parallel hybrid systems are best understood as hybrid systems
just like discrete parallel system are not best understood by their monolithic parallel product.
Instead, \dLCHP marks and solves the specific challenges of the parallel interplay of hybrid dynamics,
and \dLCHP's axioms identify the exact rules thereof.

Completeness of \dLCHP is related to a tradition of seminal completeness results for discrete, hybrid, and parallel systems:
Cook expresses sufficient loop invariants for Hoare-logic~\cite{Cook1978},
Harel adds variants for loop termination in dynamic logic~\cite{Harel1977},
Zwiers generalizes strongest postconditions to the environment of parallel programs~\cite{AcSemantics_Zwiers},
\dL lifts invariants and variants to the real domain of hybrid systems \cite{DBLP:journals/jar/Platzer08},
and differential game logic \dGL uses a finely-branched induction order to account for 
the adversarial dynamics in games~\cite{DBLP:journals/tocl/Platzer15}.
In \dLCHP, discrete, continuous, and parallel dynamics culminate all together,
so that expressiveness results, parallel environments, and inductive reduction span the whole mixture of dynamics.
Our completeness proof succeeds because \dLCHP's modular calculus in turn enables a modular completeness argument that disentangles the mixed dynamics.
For the reduction from \comFOD to FOD, 
we identify an extension of the calculus that lifts their semantic equiexpressiveness to syntactic completeness
following the idea of provably correct equitranslations \cite{AbouElWafa2024}.

\paragraph{Summary}

The article presents \dLCHP, a \emph{dynamic logic} for reasoning about parallel interactions of communicating hybrid systems. 
The core contribution is \dLCHP's \emph{compositional, sound, and complete} proof calculus
based on a graceful embedding of ac-reasoning into 
Pratt-Segerberg's well-established proof system for 
dynamic logic.
This development enables purely specification-based compositional reasoning
but also intensifies the question whether the calculus is strong enough to prove all properties of parallel hybrid systems.
To the best of our knowledge,
\dLCHP is the first logic for parallel hybrid systems that is complete \emph{and} compositional,
and the first complete dynamic logic
for ac-reasoning.
The article makes the following individual contributions:

\begin{enumerate}[resume, label=(\roman*)]
	\item
	\label{itm:sound}
	We refine \dLCHP's 
	proof calculus \cite{Brieger2023} to achieve completeness
	and prove its soundness.
	This includes a more liberal side condition for the parallel injection axiom,
	and new complete axioms 
	for the ac-diamond.
	The resulting calculus is modular
	and truly achieves only-by-specification compositional reasoning.

	\item 
	\label{itm:comfod_compl}
	\dLCHP is proven complete relative to \comFOD,  
	the first-order logic of communication traces and differential equations.
	This yields a comprehensive characterization of all dynamical effects of parallel hybrid systems.
	
	\item%
	\label{itm:fod_compl}
	By a provable equitranslation from \comFOD to FOD,
	\dLCHP becomes relatively complete for FOD.
	This proof-theoretically fully aligns \dLCHP and \dL, 
	and shows that reasoning about parallel hybrid systems is possible to the same extent as reasoning about hybrid systems or differential equation properties.

\end{enumerate}

\paragraph{Outline}

The article is structured as follows:
\rref{sec:dLCHP} recaps the dynamic logic of communicating hybrid programs \dLCHP, in particular, its syntax, semantics, static semantics, and substitution properties.
\rref{sec:calculus} presents a Hilbert-style proof calculus for \dLCHP and proves its soundness.
\rref{sec:completeness} contains the main contribution of this article and gives two complementary completeness results for the calculus in \rref{sec:calculus}.
\rref{sec:related} discusses related work,
and \rref{sec:conclusion} draws conclusions.

\section{Dynamic Logic of Communicating Hybrid Programs} 
\label{sec:dLCHP}

Cyber-physical systems commonly become safety-critical
when they operate in close proximity to other CPSs and obstacles in their environment.
For example,
autonomous vehicles must avoid collisions with other vehicles,
humans, and stationary obstacles.
The dynamic logic of communicating hybrid programs \dLCHP~\cite{Brieger2023}
supports the modeling and verification of parallel interactions among multiple CPSs
as parallel hybrid systems.
Thereby, \dLCHP underpins a broad variety of existing approaches,
including component-based hybrid systems in \dL \cite{DBLP:journals/sttt/MullerMRSP18,Lunel2019}
and parallel products of hybrid automata \cite{Lynch2003, Frehse2004, Henzinger2001, Benvenuti2014},
providing them with solid \emph{logical foundations}.

This section presents the syntax, semantics, and static semantics of \dLCHP.
The design of \dLCHP is centered around the principles of modularity and compositionality,
which enables modeling of complex parallel hybrid systems from simpler subsystems
and their verification by hierarchical (de)composition.
A convoy of two cars safely adjusting their speed despite lossy communication~\cite{BriegerMP2023} serves as running example.

\emph{Communicating hybrid programs} (CHPs) are a compositional programming language for modeling parallel interactions of hybrid systems.
They extend hybrid programs for modeling hybrid systems \cite{DBLP:journals/jar/Platzer08} 
with communication and parallelism in the style of communicating sequential processes (CSP) \cite{Hoare1978}.
Hybrid programs model mixed discrete and continuous dynamics,
where continuous evolutions along ordinary differential equations (ODEs)
describe physical processes.
Parallel compositions $\alpha\parOp\beta$ of hybrid systems
therefore require that the continuous evolutions of the subprograms~$\alpha$ and~$\beta$
execute \emph{truly simultaneously} for the same duration,
ensuring time-consistent evolution of the physical processes across all subprograms.
The subprograms $\alpha$ and $\beta$ of $\alpha\parOp\beta$
can communicate but \emph{may not} share state.
As in CSP,
the send and receive primitives of CHPs perform synchronous communication along channels,
\iest communication takes place  
only when both sender and receiver are ready to communicate at the current time.
Based on these primitives, 
CHPs modularly support a wide range of interaction patterns,
including synchronous, delayed, noisy, and lossy communication.
This modular design makes CHPs an expressive modeling language
for parallel hybrid dynamics
while keeping the semantics and axiomatization simple.

For studying safety of parallel hybrid systems,
\dLCHP features a compositional proof calculus (\rref{sec:calculus}) based on dynamic logic \cite{Harel_et_al_2000,Platzer18}.
Since parallel hybrid systems synchronize in time and interact via communication during their execution,
the observable behavior of a CHP~$\alpha$ \emph{cannot} be adequately described solely
by an initial-final state contract $\varphi \rightarrow [ \alpha ] \psi$,
which expresses that the postcondition $\psi$ holds in all final states reachable by $\alpha$ from states satisfying the precondition $\varphi$.
For compositional reasoning about parallelism,
\dLCHP thus embeds assumption-commitment (ac) reasoning \cite{Misra1981,AcHoare_Zwiers,AcSemantics_Zwiers,deRoever2001}
into the dynamic logic setup. 
The contract $\varphi \rightarrow [ \alpha ] \ac \psi$ complements the pre- and postcondition
with an \emph{assumption} $\A$ about any incoming communication and a \emph{commitment} $\Commit$ about all outgoing communication,
specifying the interaction of the program $\alpha$ with its environment 
during the run of $\alpha$.
Crucially, the ac-contract $(\A, \Commit)$ refers \emph{only} to the behavior of programs observable by the environment,
which enables truly compositional top-down decomposition into isolated subsystems
and modular bottom-up reuse of components \cite{deRoever2001,Zwiers1996}.

\subsection{Syntax} \label{sec:syntax}

The syntax builds on sets of real variables $\RVar$,
trace variables $\TVar$,
and channel names $\Chan = \naturals$,
and $\V = \RVar \cup \TVar$ are all variables.
By convention, $x,y \in \RVar$, and $\tvar \in \TVar$, 
and $\ch{}, \ch{dh} \in \Chan$,
and $\avar \in \V$.
Channel sets $\cset, \cset[alt] \subseteq \Chan$ are always assumed to be (co)finite,
because they either denote the finitely many channel names of a program 
or the complement of such a set.
Overlined expressions $\exprvec$ denote vectors.

The terms of \dLCHP (\rref{def:syntax_terms}) combine
real arithmetic as in \dL with communication traces,
which are adapted from ac-reasoning \cite{AcSemantics_Zwiers,deRoever2001} to hybrid systems.
Integer terms as in previous work \cite{Brieger2023} 
are not necessary because they are definable in \dL \cite{DBLP:journals/jar/Platzer17}.
Since CHPs model hybrid systems,
they operate over the real-valued state,
\iest CHPs can only read and write real variables.
Trace terms represent observations about the 
communications between parallel programs,
which can be reasoned about in \dLCHP.

\newcommand{\purecom}[1]{\restrict{#1}{\Chan}}

\begin{definition}[Terms] \label{def:syntax_terms}
	The terms $\Trm$ of \dLCHP are real terms $\Rtrm$ and trace terms $\Ttrm$ as defined by the following grammar, where $\ratconst \in \rationals$,
	and $\ch{} \in \Chan$, 
	and $\cset\subseteq\Chan$ is (co)finite,
	and $\rp \in \polynoms{\rationals}{\RVar} \subset \Rtrm$ is a polynomial in~$\RVar$ over rational coefficients:
	\begin{align*}
		\Rtrm: &&\re_1, \re_2 & \cceq \underbrace{
			x \mid \ratconst \mid \re_1 + \re_2 \mid \re_1 \cdot \re_2 
		}_{\polynoms{\rationals}{\RVar}}
        \mid \chan{\te} \mid \val{\te} \mid \len{\te} \\
		\Ttrm: &&\te_1, \te_2 & \cceq \tvar \mid \epsilon \mid \comItem{\ch{}, \rp} \mid \te_1 \cdot \te_2 \mid \te \downarrow \cset \mid \at{\te}{\re}
	\end{align*}
\end{definition}

Real terms $\Rtrm$ are polynomials $\polynoms{\rationals}{\RVar}$
in the real variables $\RVar$
plus the selectors $\chan{\te}$ and $\val{\te}$, 
returning the channel name and value,
respectively, of the last communication in the trace~$\te$,
and $\len{\te}$ denotes the length of $\te$.
If $\te$ is empty, 
$\val{\te}$ and $\chan{\te}$ default to $0$.
Differential forms $(\rp)'$ can be added as first-class real terms to \dLCHP~\cite{Brieger2023} to support \dL's axiomatic reasoning about differential equation invariants~\cite{DBLP:journals/jar/Platzer17},
but are omitted in this article for simplicity.
By convention $\rp \in \polynoms{\rationals}{\RVar} \subset \Rtrm$,
and $\re \in \Rtrm$ denotes an arbitrary real term.
In programs, only polynomials $\rp \in \polynoms{\rationals}{\RVar}$ 
occur as the program state is real-valued.

Trace terms $\Ttrm$ are variables $\tvar$, the empty trace $\epsilon$, communication items $\comItem{\ch{}, \rp}$, concatenation $\te_1 \cdot \te_2$, projection $\te \downarrow \cset$ onto channel set $\cset$, and access~$\at{\te}{\re}$ to the \mbox{$\lfloor \re \rfloor$-th} communication item in trace $\te$,
where $\lfloor \cdot \rfloor$ is rounding.
The item $\comItem{\ch{}, \rp}$ represents a communication event with value $\rp$ on channel $\ch{}$,
where $\rp \in \polynoms{\rationals}{\RVar}$ since the value comes from a program.
The projection $\te \downarrow \cset$ removes all items from $\te$ whose channel name is not in $\cset$.
If access $\at{\te}{\re}$ is out-of-bounds,
it yields the empty trace $\epsilon$.
For example, $\val{\at{(\tvar \downarrow \ch{})}{\len{\tvar \downarrow \ch{}}-2}}$ 
asks for the value of the 
second
last communication along the channel $\ch{}$ that is in $\tvar$.

The behavior of a CHP 
that is observable by the environment 
can be described by a trace that consists of communication events 
and the durations for which continuous behavior is executed.
During the run of a program,
this trace is collected in the reserved \emph{recorder} variable $\recvar\in\TVar$.
The central idea for compositional reasoning is to use the recorder $\recvar$ as the only interface through which specifications describe the 
interactions between parallel programs~\cite{AcHoare_Zwiers,Hooman1992}.
This supports truly compositional top-down reasoning
where properties of compound systems can be derived purely from observational specifications of the subsystems.
Since the specifications do not rely on the internals of programs,
they allow modular reuse of components in bottom-up verification \cite{Zwiers1996}.

Unlike in previous work on \dLCHP \cite{Brieger2023},
CHPs do not require a special global time variable that silently evolves with every continuous evolution.
Instead, the synchronization of parallel continuous evolutions in time will be modeled purely in the semantics.
This approach conceptually aligns the synchronization of evolution durations with that of communication,
and simplifies the axiomatization and completeness proof.

Programs (\rref{def:syntax_chps}) and formulas (\rref{def:syntax_formulas}) have a mutually dependent syntax
as programs occur 
in formulas via modalities 
and formulas as tests in programs.
Their context-sensitive grammars presume notions of free variables $\SFV(\cdot)$ and bound variables $\SBV(\cdot)$,
and $\SV(\cdot) = \SFV(\cdot) \cup \SBV(\cdot)$,
as defined in the static semantics.
\rref{sec:static_semantics} provides a definition of the static semantics
via a standard computation 
from the syntax,
alongside a precise definition based on the denotational semantics.%
\footnote{This circularity between syntax and semantics is well-founded 
because for each operator only free and bound variables of its subexpressions are involved.}
Both definitions can be used interchangeably
as both satisfy the coincidence properties (\rref{sec:static_semantics}),
which are the only properties of free and bound variables relevant for soundness.

\begin{definition}[Programs] \label{def:syntax_chps}
	\emph{Communicating hybrid programs} 
	are defined by the grammar below,
	where $\rp \in \polynoms{\rationals}{\RVar}$ 
	and 
	$\chi \in \FolRA$ is a formula of first-order real arithmetic,
	\iest a first-order formula over $\polynoms{\rationals}{\RVar}$-terms.
	In $\alpha \parOp \beta$,
	the subprograms must not share state,
	\iest $\SBV(\alpha) \cap \SV(\beta) \subseteq \gvars$
	and $\SBV(\beta) \cap \SV(\alpha) \subseteq \gvars$.
	
	\begin{align*}
		\alpha, \beta \cceq \;
		& 
		\underbrace{x \ceq \rp \mid x \ceq * \mid \test{} \mid \evolution{}{} \mid \alpha \seq \beta \mid \alpha \cup \beta \mid \repetition{\alpha}}
			_\text{hybrid programs from \dL}
		\mid
		\underbrace{\send{}{} \mid \receive{}{} \mid \alpha \parOp \beta}
			_\text{CSP extension}
	\end{align*}
\end{definition}

CHPs combine hybrid programs from \dL \cite{DBLP:journals/jar/Platzer08} with CSP-style \cite{Hoare1978} communication primitives and a parallel operator.
Assignment $x \ceq \rp$ updates $x$ to $\rp$, 
nondeterministic assignment $x \ceq *$ sets $x$ to any value,
and the test $\test{\chi}$ has no effect on the state if~$\chi$ is satisfied and aborts execution otherwise.
Continuous evolution $\evolution{}{}$ follows the differential equation $\evolution*{}{non}$ for any duration but only as long as the domain constraint~$\chi$ is not violated.
For example, $\evolution{x' = v, v' = a}{v \ge 0}$ 
models the motion of a car with position $x$, velocity $v$, and acceleration $a$
that cannot drive backward,
because the domain $v \ge 0$ prohibits negative velocity.
The domain constraint $\chi$ is omitted when the evolution $\evolution{}{non}$ is not restricted.
Terms $\rp$ and tests $\chi$ in programs 
are limited to $\polynoms{\rationals}{\RVar}$-polynomials and first-order real arithmetic $\FolRA$, respectively,
as the program state is real-valued.
Sequential composition $\alpha\seq\beta$ first executes~$\alpha$ and then~$\beta$,
nondeterministic choice $\alpha\cup\beta$ either executes~$\alpha$ or~$\beta$,
and repetition~$\repetition{\alpha}$ repeats $\alpha$ for zero or more times.
Control structures can be modeled as usual,
\eg $\ifelse{\chi}{\alpha}{\beta} \equiv \test{\chi} \seq \alpha \cup \test{\neg\chi} \seq \beta$
and $\whileloop{\chi}{\alpha} \equiv \repetition{(\test{\chi} \seq \alpha)} \seq \test{\neg\chi}$.

Communication in \dLCHP is synchronous as in CSP \cite{Hoare1978},
\iest it takes place on a channel if all parallel programs that share this channel can agree on the same value at the same time.
The send statement $\send{}{}$ instantaneously communicates the value $\rp$ along the channel~$\ch{}$ 
if the environment can accept $\rp$ on~$\ch{}$ at the current time,
and does not change the local state.
The receive statement $\receive{}{}$ assigns any value to the variable~$x$ that the environment can currently communicate 
along $\ch{}$.
If no communication is possible, the execution aborts.
Asynchronous communication including delay and loss of messages can be modeled
by explicit buffer programs which receive, delay, and forward communication.
For example, $\receive{}{b} \seq t \ceq 0 \seq \evolution{t' = 1}{t \le \epsilon} \seq \send{\ch{dh}}{b}$
models a directed 
channel with input port $\ch{}$, output port $\ch{dh}$,
and delay of at most~$\epsilon$ time units.

Parallel composition $\alpha\parOp\beta$ executes $\alpha$ and $\beta$ truly simultaneously,
\iest there is a run of $\alpha\parOp\beta$ if there are runs of the subprograms,
which agree on the duration of continuous evolution
and on the value and time of every communication on shared channels.
That is, parallel CHPs 
ensure the time-consistent evolution of parallel physical processes
(with Newtonian dynamics).
CHPs model distributed hybrid systems,
\iest the subprograms $\alpha, \beta$ of the parallel composition $\alpha \parOp \beta$ can communicate
but \emph{may not} share state,
which is represented by the constraints $\SBV(\alpha) \cap \SV(\beta) \subseteq \gvars$ 
and $\SBV(\beta) \cap \SV(\alpha) \subseteq \gvars$.%
\footnote{The weaker constraint $\SBV(\alpha) \cap \SBV(\beta) \subseteq \{ \recvar \}$ in previous work \cite{Brieger2023} yields equivalent modeling capabilities,
but the separation of free variables in \rref{def:syntax_chps}
simplifies axioms and completeness proofs.}
The recorder variable \mbox{$\recvar\in\TVar$} is shared between parallel programs,
because it collects the global communication history of $\alpha\parOp\beta$.

\begin{remark}
	\label{rem:broadcasting}
	Channels are \emph{not} limited to unidirectional communication between exactly two processes,
	although this is a common use case.
	For example, the program $\send{}{\rp} \parOp \receive{}{x_1} \parOp \dots \parOp \receive{}{x_\semvar}$
	assigns~$\rp$ to every $x_i$.
	This broadcast communication has useful applications such as the simultaneous announcement of a speed limit to all vehicles in a convoy.
	If more than one program can send on a channel,
	all programs that share this channel must agree on the value sent,
	\eg $\send{}{\rp_1} \parOp \ldots \parOp \send{}{\rp_l}$
	only has a run if all $\rp_i$ have the same value.
\end{remark}

\begin{figure}[t]
	\centering
	\begin{small}
		$\begin{aligned}
			&\ctrlVelocity \equiv
				\receive{\ch{vel}}{\tarvelo} \seq 
				\ifstat{d {>} \periodicity \maxvelo}{v_f \ceq \tarvelo} \\
			&\ctrlDistance \equiv 
				\receive{\ch{pos}}{\mespos} \seq 
				d \ceq \mespos - x_f \seq \waitvar \ceq 0 \seq \\ 
				&\qquad \ifstat{d {\le} \periodicity \maxvelo}{%
					\{v_f \ceq * \seq \test{\orange{0}{v_f}{\safevelo{d}}}\}
				} \\
			&\Plant_f \equiv
				\evolution{x_f' = v_f, \waitvar' = 1}{\waitvar \le \periodicity} \\
			&\progtt{follower} \equiv
				\big( (\ctrlVelocity \cup \ctrlDistance) \seq \Plant_f \big)^*
		\end{aligned}$\hspace{.4cm}
		$\begin{aligned}
			& \ctrlNotify \equiv v_l \ceq * \seq \test{0 {\le} v_l {\le} \maxvelo} \seq (\send{\ch{vel}}{v_l} \cup \skipProg) \\
			& \ctrlUpdate \equiv \send{\ch{pos}}{x_l} \\ 
			& \Plant_l \equiv \evolution{x_l' = v_l}{non} \\
			& \progtt{leader} \equiv \big( (\ctrlNotify \cup \ctrlUpdate) \seq \Plant_l \big)^* \\
			& \ifstat{\chi}{\alpha} \equiv \test{\chi} \seq \alpha \cup \test{\neg\chi}
		\end{aligned}$
	\end{small}
	\vspace{.5cm}
	\caption{%
		Models of two moving cars ($\progtt{follower}$ and $\progtt{leader}$) 
		whose parallel composition $\progtt{follower} \parOp \progtt{leader}$ forms a convoy.
	}
	\label{fig:follower_leader}
\end{figure}

\begin{example}
	[Communicating cars \cite{BriegerMP2023}]
	\label{ex:follower_leader}
	Figure \ref{fig:follower_leader} models a convoy of two cars safely adjusting their speed. 
	From time to time, 
	the $\progtt{leader}$ changes its speed $v_l$ in the range $0$ to $\maxvelo$ and notifies this to the $\progtt{follower}$.
	This communication, however, is lossy ($\send{\ch{vel}}{v_l} \cup \skipProg$).
	As a safety mechanism, the $\progtt{follower}$ measures its distance $d$ to the $\progtt{leader}$ at least every $\periodicity$ time units.
	This is modeled by receiving the $\progtt{leader}$'s position on channel $\ch{pos}$ 
	in the controller $\ctrlDistance$.
	If the distance~$d$ fell below~$\maxdist$, 
	the $\progtt{follower}$ slows down in $\ctrlDistance$ to avoid collision before the next measurement.
	Regularly, the $\progtt{follower}$ adopts speed updates in its controller $\ctrlVelocity$, but crucially refuses 
	if the last known distance $d$ is unsafe ($\neg(d {>} \maxdist)$).
	Even though the speed update is perfectly fine at the moment,
	an unsafe distance can cause a future collision
	if a future slow down of the $\progtt{leader}$ gets lost (\rref{fig:convoy_plot}).
	Then only a 
	position measurement can reliably tell if it is safe to obey the $\progtt{leader}$'s speed again.
\end{example}
\vspace{-.5cm}

\begin{figure}
	\begin{minipage}{.55\textwidth}
      \captionof{figure}{
			\small Qualitative plot of example positions~$x_f$ and~$x_l$ of the cars over time (\rref{ex:follower_leader}).
			First, the speed update is accepted (\success).
			The next communication is lost (\usebox{\lossy}).
			After reliably measuring the position (\measurment),
			the $\progtt{follower}$ adjusts its speed.
			Crucially, it conservatively rejects the speed update~(\deny) 
			when a crash~($\lightning$) with a slowing $\progtt{leader}$ is possible since speed communication may fail (\usebox{\lossy}) until the next reliable position measurement is expected, see dashed trajectory~(\usebox{\crash}).
		}
		\label{fig:convoy_plot}
	\end{minipage}\hspace{.3cm}
	\begin{minipage}{.41\textwidth}
		\vspace*{-.5\baselineskip}
		\centering
		\begin{small}
			\begin{tikzpicture}[node distance= 2em]
    \coordinate (origin) at (0,0);
    \coordinate (coorigin) at ([shift={(-.3, -.4)}] origin);
    \coordinate (epsorigin) at ([shift={(0, -.4)}] origin);

    \draw[black, ->] (coorigin) --++ (4, 0) node[below right, xshift=-2mm] {\small time};
    \draw[black, ->] (coorigin) --++ (0, 3) node[above] {\small position};

    \node[disc] [above=1.25 of origin] (l-0) {};
    \node[disc] [above right=2 and 1 of origin] (l-0-1) {};
    \node[disc] [above right=2 and 1.66 of origin] (l-1) {};
    \node[disc] [above right=2 and 2 of origin] (l-2) {};
    \node[disc] [above right=2.6 and 2.5 of origin] (l-2-1) {};
    \node[ghost] [above right=2.75 and 3.66 of origin] (l-3) {};

    \node[ghost][above right=2 and 4 of origin] (l-g) {};

    \node[inner sep=0] (terminate-1) [below=.2cm of l-0-1] {\termination};
    \draw[comarrow] (l-0-1) -- (terminate-1);

    \node[inner sep=0] (terminate-2) [below=.2cm of l-2-1] {\termination};
    \draw[comarrow] (l-2-1) -- (terminate-2);

    \node[disc] (f-0) at (origin) {};
    \node[disc] [above right=1.25 and 1.66 of origin] (f-1) {};
    \node[disc] [above right=1.35 and 2 of origin] (f-2) {};
    \node[ghost] [above right=1.75 and 3.33 of origin] (f-3) {};

    \node[ghost] [above right=3.15 and 3.5 of origin] (f-g) {};

    \draw[black] (l-0) -- (l-0-1);
    \draw[black] (l-0-1) -- (l-1);
    \draw[black] (l-1) -- (l-2);
    \draw[black] (l-2) -- (l-2-1);
    \draw[black, name path=slow] (l-2-1) -- (l-3);

    \draw[black] (f-0) -- (f-1);
    \draw[black] (f-1) -- (f-2);
    \draw[black] (f-2) -- (f-3);

    \draw[dotted, gray] (l-2) -- (l-g);
    \draw[badtrace, name path=bad] (f-2) -- (f-g);

    \draw[comarrow] (l-0) -- (f-0);
    \draw[comarrow] (l-1) -- (f-1);
    \draw[comarrow] (l-2) -- (f-2);

    \path[name intersections={of=slow and bad,by=crash}];
    \node[draw, lightgray, circle] (collision) at (crash) {};
    \node [above left=-.2 and -.1 of collision] {$\lightning$};

    \node [below right=-.15 and -.1 of f-0] {{\scriptsize$\ctrlVelocity$} \!\success};
    \node [below=.1 of f-1] {\scriptsize$\ctrlDistance$ \!\measurment};
    \node [below right=-.15 and -.1 of f-2] {{\scriptsize$\ctrlVelocity$} \!\deny};

    \node [above right=0 and .1 of l-0, rotate=90] {\scriptsize$\ctrlNotify$};
    \node [above right=0 and .1 of l-0-1, rotate=90] {\scriptsize$\ctrlNotify$};
    \node [above right=0 and .15 of l-1, rotate=90] {\scriptsize$\ctrlUpdate$};
    \node [above right=0 and .1 of l-2, rotate=90] {\scriptsize$\ctrlNotify$};
    \node [above right=0 and .1 of l-2-1, rotate=90] {\scriptsize$\ctrlNotify$};

    \node [left=.3 of f-0] {$x_f$};
    \node [left=.3 of l-0] {$x_l$};

    \draw[decoration={brace, mirror, raise=3pt}, decorate] (epsorigin) -- node[below=.5em] {$\le\periodicity$} ++(1.71,0);

    \draw[decoration={brace, mirror, raise=3pt}, decorate] ([shift={(1.71, 0)}] epsorigin) -- node[below=.5em] {$\le\periodicity$} ++(1.66, 0);

    \draw[dotted, gray] ([shift={(1.71, -.3)}] origin) -- (f-1.south);
    \draw[dotted, gray] ([shift={(3.37, -.3)}] origin) -- ([shift={(-.33, 0)}] l-3.north);

    \coordinate (cont) at ([shift={(3.1, .7)}] coorigin);
    \node [right=0 of cont] {\scriptsize controller};
    \node[disc] [left=0 of cont] {};

    \coordinate (com) at ([shift={(0, -.3)}] cont);
    \node [right=0 of com] {\scriptsize comm.};
    \draw[comarrow] ([shift={(-.5, 0)}] com) -- (com);
\end{tikzpicture}
		\end{small}
		\vspace*{-1.5em}
	\end{minipage}
\end{figure}

\begin{definition}
	[Formulas] 
	\label{def:syntax_formulas}
	The formulas of \dLCHP are defined by the grammar below 
	for relations $\sim$,
	terms $\expr_1, \expr_2 \in \Trm$ of equal type,
	and $\avar \in \V$.
	The relations $\sim$ include equality~$=$ on all types, 
	greater-equals $\ge$ on real terms, and the prefix relation $\preceq$ on traces.
	The ac-formulas are unaffected by state change in $\alpha$
	\iest $\SFV(\A, \Commit) \cap \SBV(\alpha) \subseteq \{\recvar\}$.
	\begin{align*}
		\varphi, \psi, \A, \Commit \cceq \; 
			& \expr_1 \sim \expr_2 \mid
			\neg \varphi \mid \varphi \wedge \psi 
			\mid \fa{\avar} \varphi 
			\mid [ \alpha ] \psi \mid [\alpha] \ac \psi  
			\mid \langle \alpha \rangle \psi \mid \langle \alpha \rangle \ac \psi
	\end{align*}
\end{definition}

Other relations,
\eg $\neq$, $\le$, $<$, $>$, and strict prefixing $\prec$,	
first-order connectives,
\eg~$\vee$ and $\ex{\avar} \varphi$,
and truth $\true$ and falsity $\false$
are definable as usual.
Comparisons $\exprvec_1 \sim \exprvec_2$ of vectors $\expr_1, \expr_2$
are always understood to be element-wise.
Where useful, we write $\fa{\avar{:}\anysort}$ 
with explicit type $\anysort \in \{\reals, \traces\}$ for $\avar \in \V_\anysort$
instead of $\fa{\avar}$ for emphasis.
Further, $\fa{\avar{=}\expr} \varphi$ is short for $\fa{\avar}(\avar = \expr \rightarrow \varphi)$,
and $\ex{\avar{=}\expr} \varphi \equiv \ex{\avar} (\avar=\expr \wedge \varphi)$.
For any context~$\phi(\cdot)$,
the conditional $\phi(\ifelse{\varphi}{\expr_1}{\expr_2})$
abbreviates $(\varphi \rightarrow \phi(\expr_1)) \wedge (\neg\varphi \rightarrow \phi(\expr_2))$.

The formulas of \dLCHP combine first-order dynamic logic \cite{Harel1979} with ac-reasoning~\mbox{\cite{Misra1981,AcHoare_Zwiers}} by adding ac-modalities.
As usual, the box $[ \alpha ] \psi$ holds if the postcondition $\psi$ holds in all states reachable by program $\alpha$,
and the diamond $\langle \alpha \rangle \psi$ holds if $\psi$ holds in some state reachable by $\alpha$.
Let an $(\A, \alpha)$-run be an $\alpha$-run whose incoming communication satisfies the assumption~$\A$.
Then the ac-box $[ \alpha ] \ac \psi$ states that for all $(\A, \alpha)$-runs,
every outgoing communication fulfills the commitment $\Commit$
\emph{and} if a final state is reached, 
the postcondition $\psi$ holds there.
Dually, the ac-diamond $\langle \alpha \rangle \ac \psi$ 
holds if there exists an $(\A, \alpha)$-run
that satisfies $\Commit$
\emph{or}
reaches a final state where $\psi$ holds.
Although the dynamic modalities $[ \alpha ] \psi$ and $\langle \alpha \rangle \psi$ are special cases $[ \alpha ] \acpair{\true, \true} \psi$ and $\langle \alpha \rangle \acpair{\true, \false} \psi$ of the ac-modalities, respectively,
they are included explicitly as they specify closed systems,
where the environment has no influence,
and enable succinct and modular axioms for communication-free dynamics.

The proof calculus (\rref{sec:calculus}) modularly integrates ac-reasoning,
which has only been supported for Hoare-logic previously,
and the dynamic logic \dL.
The cornerstone of this development is a modal logic interpretation of ac-reasoning, 
based on the insight
that the assumption-program pair $(\A, \alpha)$ induces the reachability relation of the ac-box $[ \alpha ] \ac \psi$ 
while the commitment-postcondition pair $(\Commit, \psi)$ is evaluated in the reachable worlds.
For an ac-modality $\dbleft \alpha \dbright \ac \psi$,
call $(\A, \alpha)$ the \emph{modal action},
$(\A, \Commit)$ the \emph{ac-contract},
and $(\Commit, \psi)$ the \emph{promise}.
This complements the classical view~\cite{Misra1981,AcHoare_Zwiers} 
that the ac-contract specifies $\alpha$'s communication interface with a clear modal perspective.

For compositional reasoning about parallelism, 
it is essential that specifications refer only to program behavior that is observable
by the environment and do not reveal internals of the program~\cite{deRoever2001}.
This is why the ac-contract $(\A, \Commit)$ of $\dbleft \alpha \dbright \ac \psi$ 
must not refer to the state variables $\SBV(\alpha)$ of $\alpha$,
because parallel programs can interact only through communication.
Instead, the ac-contract accesses the observable behavior exclusively via the recorder variable $\recvar$,
which reflects the communication history of $\alpha$.

\begin{example} \label{ex:convoy_safety}
	The formula below specifies when to consider 
	the convoy (\rref{ex:follower_leader}) formed by the parallel composition of the $\progtt{follower}$ and the $\progtt{leader}$ car
	to be safe.
	If the cars start with a distance~${>}d$, 
	if the $\progtt{follower}$ has a speed $v_f{\le}\nicefrac{d}{\periodicity}$ that prevents it from colliding with the $\progtt{leader}$ within the first $\periodicity$ time units,
	and if the $\progtt{leader}$ does not drive backward initially ($v_l{\ge}0$),
	then the cars do never collide ($x_f < x_l$) when run in parallel.
	That is, all possible behaviors of $\progtt{follower} \parOp \progtt{leader}$
	only lead to states satisfying $x_f < x_l$
	when starting in a state satisfying $\varphi$.

	\begin{equation*}
		\underbrace{\periodicity \ge 0 \wedge
			\waitvar = 0 \wedge
			0 {\le} v_f {\le} \safevelo{d} \wedge
			v_f \le \maxvelo \wedge
			v_l \ge 0 \wedge
			x_f + d < x_l}_\varphi
		\rightarrow [ \progtt{follower} \parOp \progtt{leader} ]
		\, x_f < x_l 
	\end{equation*}
\end{example}

\subsection{Semantics} \label{sec:semantics}

\newcommand{\traceone}{\trace}
\newcommand{\tracetwo}{\rho}

\newcommand{\tprecision}{W}

\newcommand{\indexvar}{i}

\newcommand{\rectpair}[2]{(#1, #2)}

The denotational semantics of \dLCHP
assigns a value to every term 
and a reachability relation to every program,
and it defines the satisfaction relation for formulas.

The behavior observable by the environment during the run of a CHP consists of communication events
\emph{and} the durations for which continuous evolutions are executed. 
A \emph{(communication) event} $\comItem{\ch{}, \semConst} \in \Chan \times \reals$ occurs on a channel $\ch{}$ and carries a value~$\semConst$.
The \emph{observable duration} $\duration \ge 0$ of continuous evolution 
is modeled as a communication~$\comItem{\durchan, \duration}$ along the reserved channel $\durchan\in\Chan$.
A \emph{trace} is a finite sequence of events,
which does \emph{not} contain consecutive events on channel $\durchan$.
The set of traces is denoted by $\traces \subseteq (\Chan \times \reals)^*$,
and $\epsilon$ denotes the empty trace.
Consecutive events 
are excluded from traces,
because only the overall duration of evolution is intended to be observable by the environment,
\iest $\comItem{\durchan, \duration_1} \cdot \comItem{\durchan, \duration_2}$
and 
$\comItem{\durchan, \duration_1 + \duration_2}$ are observably equivalent.%
\footnote{
	Alternatively, traces could be defined as the quotient of an equivalence relation on $(\Chan\times\reals)^*$
	that identifies all traces with observably equivalent durations.
	Instead, we represent each equivalence class by a canonical representative,
	which admits a simple treatment of traces as linear sequences.
}
In particular, 
we identify the zero-duration trace $\comItem{\durchan, 0}$ with the empty trace~$\epsilon$.

The concatenation of traces merges consecutive observable durations,
since they are equivalent to their sum.
For traces $\traceone=(\traceone_1, \ldots, \traceone_n)$ and $\tracetwo=(\tracetwo_1, \ldots, \tracetwo_m)$,
where $\traceone_i, \tracetwo_j \in (\Chan\times\reals) \setminus \{\comItem{\durchan, 0}\}$,
the \emph{concatenation} $\traceone\cdot\tracetwo$ is defined as follows:
\begin{equation*}
	\traceone \cdot \tracetwo = 
	\begin{cases}
		(\trace_1, \ldots, \trace_{n-1}, \comItem{\durchan, \duration_1+\duration_2}, \rawtrace_2, \ldots, \rawtrace_m)
			& \text{if } \trace_n = \comItem{\durchan, \duration_1} \text{ and } \rawtrace_1 = \comItem{\durchan, \duration_2} \\
		(\traceone_1, \ldots, \traceone_n, \tracetwo_1, \ldots, \tracetwo_m)
			& \text{else}
	\end{cases}
\end{equation*}
For $\trace\in\traces$ and channels $\cset \subseteq \Chan$,
the \emph{projection} $\trace \downarrow \cset$ 
is inductively defined by
$\epsilon \downarrow \cset = \epsilon$,
and $(\comItem{\ch{}, \semConst} \cdot \rawtrace) \downarrow \cset = \comItem{\ch{}, \semConst} \cdot (\rawtrace \downarrow \cset)$
if $\ch{} \in \cset$,
and $(\comItem{\ch{}, \semConst} \cdot \rawtrace) \downarrow \cset = \rawtrace \downarrow \cset$
if $\ch{}\not\in\cset$.
That is, $\trace\downarrow\cset$ removes all events from $\trace$ whose channel is not in $\cset$
and merges any consecutive observable durations that may result from the removal of events,
\eg $(\comItem{\durchan, \duration_1} \cdot \comItem{\ch{}, 0} \cdot \comItem{\durchan, \duration_2}) \downarrow \{\durchan\}
= \comItem{\durchan, \duration_1} \cdot \comItem{\durchan, \duration_2}
= \comItem{\durchan, \duration_1+\duration_2}$.
Further, $\chan{\trace}$ and $\val{\trace}$ denote the channel and value,
respectively, of the last event in the trace $\trace$,
defaulting to~$0$ if $\trace$ is empty.
The length of $\trace$ is denoted by $\len{\trace}$.
For $\trace\in\traces$ and $\semConst \in \reals$,
access~$\at{\trace}{\semConst}$ returns the~$\lfloor\semConst\rfloor$-th item of~$\trace$
and~$\epsilon$ if~$\lfloor\semConst\rfloor$ is out-of-bounds ($\lfloor \cdot \rfloor$ is rounding).
The (strict) prefix relation on traces is written ($\prec$) $\preceq$.

A \emph{state} is a mapping $\pstate{v} : \V \rightarrow \reals \cup \traces$ from variables to values 
such that $\pstate{v}(\avar) \in \anysort$ for all $\avar\in\V_\anysort$ and $\anysort \in \{\reals,\traces\}$.
The set of states is denoted $\states$.
If $\avar \in \V_\anysort$ and $\semConst \in \anysort$,
the \emph{update} $\pstate{v} \subs{\avar}{\semConst}$ is defined by $\pstate{v} \subs{\avar}{\semConst} (\avar) = \semConst$ and $\pstate{v} \subs{\avar}{\semConst} = \pstate{v}$ on $\{ \avar \}^\complement$.

The semantics of terms (\rref{def:semantics_terms}) evaluates variables by the state and every operator by its semantic counterpart,
\eg $\sem{\te \downarrow \cset}{\pstate{v}} = \sem{\te}{\pstate{v}} \downarrow \cset$.

\begin{definition}
	[Term semantics]
	\label{def:semantics_terms}
	The \emph{valuation} $\sem{\expr}{\pstate{v}} \in \reals \cup \traces$ of the term~$\expr$ in the state $\pstate{v}$ is defined as follows,
	where $\fsymb[builtin] \in \{\usarg+\usarg, \usarg\downarrow\cset,\ldots\}$ is any built-in operator including constants:
	\begin{align*}
		\sem{\avar}{\pstate{v}} 
			& = \pstate{v}(\avar) \\
		\qquad \sem{\fsymb[builtin](\expr_1, \ldots, \expr_k)}{\pstate{v}} 
			& = \fsymb[builtin](\sem{\expr_1}{\pstate{v}}, \ldots, \sem{\expr_k}{\pstate{v}})
	\end{align*}
\end{definition}

Since the syntax of programs and formulas is mutually dependent,
\rref{def:semantics_programs} and \rref{def:semantics_formulas} define their semantics by a mutual recursion on their structure.
The denotational semantics of CHPs \cite{Brieger2023} embeds \dL's Kripke semantics \cite{DBLP:journals/jar/Platzer08} into a linear history semantics for communicating programs \cite{AcHoare_Zwiers}.
The resulting denotations combine a reachability relation between states with traces of communications and observable durations.
An intuitive operational view onto the denotational semantics of CHPs is given in \rref{fig:opsem}.

\newcommand{\inistates}{\mathcal{I}}
\newcommand{\finstates}{\mathcal{F}}

The denotation $\sem{\alpha}{} \in \pDomain$ of a program $\alpha$ is drawn from a domain $\pDomain \subseteq \powerset(\states {\times} \traces {\times} \botop{\states})$ with $\botop{\states} = \states \cup \{ \bot \}$,
where $\powerset(\cdot)$ is the powerset.
Each run $\run \in \sem{\alpha}{}$ starts in an \emph{initial state} $\pstate{v}$, 
emits communication $\trace$,
and either approaches a \emph{final state} $\pstate{w} \neq \bot$,
or if $\pstate{w} = \bot$, the run is an \emph{unfinished computation}.
If unfinished, the run either can be continued \emph{or} failed a test or domain constraint
such that execution aborts.

The linear order of communication events imposes two natural properties on every denotation:
\emph{Prefix-closedness} requires all prefixes $\trace[pre] \preceq \trace$ to be observable before a program can communicate $\trace$.
\emph{Totality} requires that computation can start from every state even if it aborts immediately.
To lift the prefix relation~$\preceq$ on traces to trace-state pairs,
define $(\trace[pre], \pstate[pre]{w}) \preceq (\trace, \pstate{w})$
if $(\trace[pre], \pstate[pre]{w}) = (\trace, \pstate{w})$, or $\trace[pre] \preceq \trace$ and $\pstate[pre]{w} = \bot$.
Then define a denotation $\denotation \in \pDomain$ to be prefix-closed 
if $\run \in \denotation$ and $\obs[pre] \preceq \obs$,
imply $\run[pre] \in \denotation$.
Further, $\denotation$ is total if for every state $\pstate{v}$, there is some $\run \in \denotation$.
Then the \emph{domain} $\pDomain$ of CHPs 
are the
prefix-closed and total subsets of $\states {\times} \traces {\times} \botop{\states}$.
On~$\pDomain$, the subset relation $\subseteq$ is a partial order with least element $\pLeast = \states \times \{ \epsilon \} \times \{\bot\}$,
\iest~every denotation contains $\pLeast$
because computation can start in every state.

The semantics of composite programs is given by operators on denotations:
For denotations $\denotation, \denotation[alt] \in \pDomain$,
define \emph{lowering} $\botop{\denotation} = \{ (\pstate{v}, \trace, \bot) \mid \run \in \denotation \}$
and \emph{continuation} $\denotation \continuation \denotation[alt]$ by $(\pstate{v}, \trace_1 \cdot \trace_2, \pstate{w}) \in \denotation \continuation \denotation[alt]$ if there are $(\pstate{v}, \trace_1, \pstate{u}) \in \denotation$ and $(\pstate{u}, \trace_2, \pstate{w}) \in \denotation[alt]$.
Further, define (prefix-closed) sequential composition as $\denotation \closedComposition \denotation[alt] = \botop{\denotation} \cup (\denotation \continuation \denotation[alt])$.
Let $\sIdentity = \states \times \{\epsilon\} \times \states$.
Then $\pIdentity = \pLeast \cup \sIdentity$ is the neutral element of 
$\closedComposition$.
Semantical iteration~$\denotation^\semvar$ is inductively defined by $\denotation^0 = \pIdentity$
and $\denotation^{\semvar+1} = \denotation \closedComposition \denotation^\semvar$.
Syntactically, define~$\alpha^\semvar$  by $\alpha^0 \equiv \test{\true}$ and $\alpha^{\semvar+1} = \alpha \seq \alpha^\semvar$,
and indeed, $\sem{\alpha}{}^\semvar = \sem{\alpha^\semvar}{}$ for each $\semvar$.
For programs $\alpha, \beta$ and states $\pstate{w}_\alpha, \pstate{w}_\beta \in \botop{\states}$,
the \emph{merged state} $\pstate{w}_\alpha \merge \pstate{w}_\beta$ is $\bot$ if at least one of $\pstate{w}_\alpha$ and $\pstate{w}_\beta$ is $\bot$.
Otherwise, define $\pstate{w}_\alpha \merge \pstate{w}_\beta = \pstate{w}_\alpha$ on $\SBV(\alpha)$,
and define $\pstate{w}_\alpha \merge \pstate{w}_\beta = \pstate{w}_\beta$ on $\SBV(\alpha)^\complement$.
For final states~$\pstate{w}_\alpha, \pstate{w}_\beta$ of a parallel composition $\alpha\parOp\beta$,
merging commutes, \iest $\pstate{w}_\alpha \merge \pstate{w}_\beta = \pstate{w}_\beta \merge \pstate{w}_\alpha$,
because parallel programs do not share bound variables (\rref{def:syntax_chps}).

Parallel CHPs synchronize on their shared channels.
For a trace $\trace\in\traces$, 
we write $\trace \downarrow \alpha$ for the projection $\trace \downarrow (\CN(\alpha) \cup \{\durchan\})$
of $\trace$ onto the channels on which~$\alpha$ synchronizes,
where $\CN(\alpha)\subseteq\Chan$ are all channels which occur in $\alpha$ syntactically (\rref{sec:static_semantics}).
Programs always synchronize on the reserved channel $\durchan$ for durations of continuous evolutions,
because the fact that purely discrete programs 
do not consume time,
\iest never emit observable durations $\comItem{\durchan,\duration}$ with $\duration>0$,
is intended to be observable.

\begin{definition}
	[Program semantics]
	\label{def:semantics_programs}
	The \emph{semantics} $\sem{\alpha}{} \in \pDomain$ of a CHP $\alpha$ is defined below,
	where $\vDash$ denotes the satisfaction relation (\rref{def:semantics_formulas}).
	\begin{align*}%
		&\sem{x \ceq \rp}{}
		= \pLeast \cup \{ (\pstate{v}, \epsilon, \pstate{w}) \mid \pstate{w} = \pstate{v} \subs{x}{\semConst} \text{ where } \semConst = \sem{\rp}{\pstate{v}} \} \\
		&\sem{x \ceq *}{}
		= \pLeast \cup \{ (\pstate{v}, \epsilon, \pstate{w}) \mid \pstate{w} = \pstate{v} \subs{x}{\semConst} \text{ where } \semConst \in \reals \} \\
		&\sem{\test{}}{}
		= \pLeast \cup \{ (\pstate{v}, \epsilon, \pstate{v}) \mid \pstate{v} \vDash \chi \} \\
		&\sem{\evolution*{}{}}{} 
		= \pLeast \cup \big\{ (\odeSolution(0), \trace, \pstate{w}) \mid
			(\trace, \pstate{w}) \preceq (\comItem{\durchan, \duration}, \odeSolution(\duration)) \text{ and } \bfforall \zeta \in [0, \duration] : \odeSolution(\zeta) \vDash \evolution*{}{non} \wedge \chi \\
			&\quad \text{ for a solution } \odeSolution : [0, \duration] \rightarrow \states \text{ with } \odeSolution(\zeta)(x') = \solutionDerivative{\odeSolution}{x}(\zeta) 
			\text{ and } \odeSolution(\zeta) = \odeSolution(0) \text{ on } \{ x \}^\complement \big\} \\
		&\sem{\send{}{}}{}
			= \big\{ (\pstate{v}, \trace, \pstate{w}) \mid \observable \preceq (\comItem{\ch{}, \semConst}, \pstate{v}) \text{ where } \semConst = \sem{\rp}{\pstate{v}} \big\} \\
		&\sem{\receive{}{}}{} 
			= \big\{ (\pstate{v}, \trace, \pstate{w}) \mid \observable \preceq (\comItem{\ch{}, \semConst}, \pstate{v} \subs{x}{\semConst}) \text{ where } \semConst \in \reals \big\} \\
		&\sem{\alpha \cup \beta}{}
		= \sem{\alpha}{} \cup \sem{\beta}{} \\
		&\sem{\alpha \seq \beta}{} 
		= \sem{\alpha}{} \closedComposition \sem{\beta}{} 
		\equalsdef \botop{\sem{\alpha}{}} \cup (\sem{\alpha}{} \continuation \sem{\beta}{})\\
		&\sem{\repetition{\alpha}}{} 
		= \bigcup_{n \in \naturals} \sem{\alpha}{}^\semvar
		= \bigcup_{n \in \naturals} \sem{\alpha^\semvar}{} 
		\sidecondition[color=black]{\quad where $\alpha^0 \equiv \test{\true}$ and $\alpha^{\semvar+1} = \alpha \seq \alpha^\semvar$} \\
		&\sem{\alpha \parOp \beta}{}
		= \big\{
			(\pstate{v}, \trace, \pstate{w}_\alpha \merge \pstate{w}_\beta)
			\mid \computation[proj=\gamma] \in \sem{\gamma}{} \text{ for } \gamma\in\{\alpha,\beta\} \text{, and } \trace \downarrow (\alpha\parOp\beta) = \trace
		\}
	\end{align*}
\end{definition}

\begin{figure}[t]
	\begin{tikzpicture}[node distance=.8cm]
		\tikzstyle{state}+=[circle, fill=blue!5, draw=blue, inner sep=0, minimum width=1.5em];
		\tikzstyle{proglabel}+=[fill=Gray!20, rounded corners=2mm, minimum height=1.6em, align=center];
		\tikzstyle{smalllabel}+=[fill=Gray!20, rounded corners=2mm, align=center];

		\node[state] (assign-v) {$\pstate{v}$};
		\node[proglabel, below=.5em of assign-v] {$x \ceq \rp$};
		\node[state] [above=of assign-v] (assign-w) {$\pstate{v} \subs{x}{\semConst}$};

		\draw[->] (assign-v) -- node[left] {$\epsilon$} (assign-w);
		\node[above=0 of assign-w] {$\semConst {=} \sem{\rp}{\pstate{v}}$};

		\node[state] [right=.9cm of assign-v] (nondet-v) {$\pstate{v}$};
		\node[proglabel, below=.5em of nondet-v] {$x \ceq *$};
		\node[state] [above=of nondet-v] (nondet-w) {$\pstate{v} \subs{x}{\semConst}$};

		\draw[->] (nondet-v) -- node[left] {$\epsilon$} (nondet-w);
		\node[above=0 of nondet-w] {$\semConst \in \reals\vphantom{\sem{\rp}{\pstate{v}}}$};

		\node[state] [right=of nondet-v] (test-v) {$\pstate{v}$};
		\node[proglabel, below=.5em of test-v] {$\test{}$};
		\draw[->] (test-v) edge [loop above] node[above] {$\epsilon$} node[above, yshift=.45cm] {if $\pstate{v} \vDash \chi$} (test-v);

		\node[state] [right=.6cm of test-v] (ode-v) {$\pstate{v}$};
		\node[state] [above right=1cm and .9cm of ode-v] (ode-w) {$\pstate{w}$};
		\node[above=2mm of ode-w, xshift=-5mm, align=left] 
			{\footnotesize if $\pstate{w}$ is reachable\\ 
			\footnotesize from $\pstate{v}$ when following \\
			\footnotesize $\evolution*{}{non}$ for duration $\duration$};
		\node[proglabel, below=.5em of ode-v, xshift=.8cm] {$\evolution*{}{non}$};
		\draw[->, decorate, decoration={coil, aspect=0, amplitude=.5mm, segment length=5mm}] (ode-v) -- node[above left] {$\comItem{\durchan, \duration}$} (ode-w);

		\node[state] [right=2cm of ode-v] (send-v) {$\pstate{v}$};
		\node[proglabel, below=.5em of send-v] {$\send{}{}$};
		\draw[->] (send-v) edge [loop above] node[above] {\small$\comItem{\ch{}, \rp}$} (send-v);

		\node[state] [right=of send-v] (receive-v) {$\pstate{v}$};
		\node[state] [above=of receive-v] (receive-w) {$\pstate{v} \subs{x}{\semConst}$};
		\node[proglabel, below=.5em of receive-v] {$\receive{}{}$};

		\draw[->] (receive-v) -- node[right] {\small$\comItem{\ch{}, \semConst}$} (receive-w);
		\node[above=0 of receive-w] {$\semConst \in \reals\vphantom{\sem{\rp}{\pstate{v}}}$};

		\node[state] [above right=.1cm and 1.6cm of receive-v] (choice-v) {$\pstate{v}$};
		\node[state] [above right=.45cm and .6cm of choice-v] (choice-w-1) {$\pstate{w}_1$};
		\node[state] [below right=.2cm and 1.3cm of choice-v] (choice-w-2) {$\pstate{w}_2$};
		\node[proglabel, below right=.7cm and -.2cm of choice-v] {$\alpha\cup\beta$};

		\draw[->] (choice-v) -- 
			node[smalllabel, inner sep=1mm] {\footnotesize$\alpha\vphantom{\beta}$}
			node[above left=.7mm and .7mm] {$\trace_1$} (choice-w-1);
		\draw[->] (choice-v) -- 
			node[smalllabel, inner sep=1mm] {\footnotesize$\beta$}
			node[below left=.7mm and -1mm] {$\trace_2$} (choice-w-2);

		\draw[->, thick] (choice-v) 
			edge[bend right=30] (choice-w-1);
		\draw[->, thick] (choice-v) 
			edge[bend left= 30] 
			node[above=.2cm, right] {$\trace_1 {+} \trace_2$} (choice-w-2);

		\node[state] [below left=2.4cm and 0cm of assign-v] (seq-v) {$\pstate{v}$};
		\node[state] [above right=.25cm and .6cm of seq-v] (seq-u) {$\pstate{u}$};
		\node[state] [below right=.6cm and .9cm of seq-u] (seq-w) {$\pstate{w}$};

		\draw[->] (seq-v) -- 
			node[smalllabel, inner sep=1mm] {\footnotesize$\alpha\vphantom{\beta}$} 
			node[above left=.7mm and .7mm] {$\trace_1$} (seq-u);
		\draw[->] (seq-u) -- 
			node[smalllabel, inner sep=1mm] {\footnotesize$\beta$} 
			node[above right=.7mm and .7mm] {$\trace_2$} (seq-w);
		\draw[->, thick] (seq-v) 
			edge[bend right=20] 
			node[smalllabel, inner sep=1mm] {\footnotesize$\alpha\seq\beta$}
			node[below=2mm] {$\trace_1\cdot\trace_2$} (seq-w);

		\node[state] [right=2.8cm of seq-v] (par-v) {$\pstate{v}$};

		\node[state, circle split, font=\footnotesize] [above right=.4cm and 1.5cm of par-v] (par-w-a) {};
		\node[yshift=3pt] at (par-w-a.center) {\footnotesize$\pstate{w}_\alpha$};
		\node[yshift=-3.5pt] at(par-w-a.center) {\footnotesize$\pstate{v}$};

		\node[state, circle split, font=\footnotesize] [below right=.4cm and 1.5cm of par-v] (par-w-b) {};
		\node[yshift=3pt] at (par-w-b.center) {\footnotesize$\pstate{v}$};
		\node[yshift=-3.5pt] at(par-w-b.center) {\footnotesize$\pstate{w}_\beta$};

		\node[state, circle split, font=\footnotesize] [inner sep = 0] at ($(par-w-a.center)!.5!(par-w-b.center)$) (par-w) {};
		\node[yshift=3pt] at (par-w.center) {\footnotesize$\pstate{w}_\alpha$};
		\node[yshift=-3.5pt] at(par-w.center) {\footnotesize$\pstate{w}_\beta$};
		
		\draw[->] (par-v) 
			edge [bend left=20] 
			node[smalllabel, inner sep=1mm] {\footnotesize$\alpha\vphantom{\beta}$}
			node[above=1.5mm] {$\trace \downarrow \alpha$} (par-w-a);
		\draw[->] (par-v) 
			edge [bend right=20] 
			node[smalllabel, inner sep=1mm] {\footnotesize$\beta$}
			node[below=1.5mm] {$\trace \downarrow \beta$} (par-w-b);

		\draw[->] (par-w-a) -- (par-w);
		\draw[->] (par-w-b) -- (par-w);
		\draw[->, thick] (par-v) -- 
			node[xshift=-2mm, smalllabel, inner sep=1mm] {\footnotesize$\alpha{\parOp}\beta$}
			node[above, xshift=3.5mm] {$\trace$} (par-w);

		\node[state] [right=2.4cm of par-v] (rep-v) {$\pstate{v}$};
		\node[state] [above right=.4cm and .5cm of rep-v] (rep-u1) {$\pstate{w}_1$};
		\node[state] [below right=0cm and .8cm of rep-u1] (rep-u2) {$\pstate{w}_2$};
		\node[state] [right=.6cm of rep-u2] (rep-w) {$\pstate{w}_\semvar$};

		\draw[->] (rep-v) -- 
			node[smalllabel, inner sep=1mm] {\footnotesize$\alpha\vphantom{\beta}$} 
			node[above left=.7mm and .7mm] {$\trace_1$} (rep-u1);
		\draw[->] (rep-u1) -- 
			node[smalllabel, inner sep=1mm] {\footnotesize$\alpha\vphantom{\beta}$} 
			node[above=1.5mm] {$\ldots$} (rep-u2);
		\draw[->] (rep-u2) -- 
			node[smalllabel, inner sep=1mm] {\footnotesize$\alpha\vphantom{\beta}$} 
			node[above=1.5mm] {$\trace_\semvar$} (rep-w);
		\draw[->, thick] (rep-v) 
			edge[bend right=30] 
			node[smalllabel, inner sep=1mm] {\footnotesize$\repetition{\alpha}$}
			node[below=1.5mm] {$\trace_1\ldots\trace_\semvar$} (rep-w.south west);

		\coordinate [below=1.4cm of seq-v] (axis-start);
		\draw[->, thick] ([xshift=-.6cm]axis-start) --+ (11.4, 0) node[right] {time};
		\draw[->, thick] ([xshift=-.6cm]axis-start) --+ (0, 7) node[above right=1mm and -2mm] {state};

		\begin{scope}[transform canvas={scale=0.75}]
			\coordinate [above right= 0cm and 14.3cm of axis-start]  (legend-v);

			\node[state] (bsp-v) at (legend-v) {$\pstate{v}$};
			\node[state, above=1cm of bsp-v] (bsp-w) {$\pstate{w}$};

			\node[right=1mm of bsp-v, yshift=-1mm] (state-label-v) {state};
			\node[right=1mm of bsp-w, yshift=1mm] {state};
			\draw[->] (bsp-v) -- node[right=.2em, name={trace}, inner sep=.2em] {$\trace$} (bsp-w);

			\node[right=0 of trace, align=left, font=\linespread{0.8}] (trace-label) {comm.\\trace};

			\node[draw, thick, dotted, inner sep=3mm, fit=(bsp-w) (trace-label) (state-label-v)] {};
		\end{scope}
	\end{tikzpicture}
	\caption{
		Illustration of the operational behavior of CHPs
	}
	\label{fig:opsem}
\end{figure}

The key insight enabling \dLCHP's compositional proof calculus is that the semantics itself is compositional, 
\iest for every statement, the semantics is a simple function of the semantics of its pieces.
In case $\sem{\evolution*{}{}}{}$,
the final state $\odeSolution(\duration)$ can be reached 
when following a solution $\odeSolution$ for the differential equation $\evolution*{}{non}$ 
with duration $\duration$
that satisfies~$\chi$ at all times~\cite{DBLP:journals/jar/Platzer08,DBLP:journals/jar/Platzer17}.
Accordingly,
the duration $\comItem{\durchan, \duration}$ can be observed from the continuous evolution.
The execution of discrete statements is instantaneous (delay can be modeled).
Every denotation $\sem{\alpha}{}$ is indeed prefix-closed and total.%
\footnote{
	Totality and prefix-closedness of noncommunicating atomic programs and of continuous evolution
	is ensured by $\pLeast$,
	while the communication primitives always have a terminating run
	such that $\preceq$ adds all prefixes.
	The only remarkable compound cases are sequential composition,
	which is prefix-closed thanks to the inclusion of $\botop{\sem{\alpha}{}}$,
	and parallel composition $\alpha\parOp\beta$,
	which is prefix-closed because projection is a congruence on the prefix relation such that $\trace[pre] \preceq \trace$ implies $\trace[pre] \downarrow \gamma \preceq \trace \downarrow \gamma$ for $\gamma \in \{ \alpha, \beta \}$.
}

The communication~$\trace$ of a parallel composition $\alpha\parOp\beta$ is implicitly characterized by projections onto the subprograms (\rref{fig:partition})
to avoid the exhaustive enumeration of all possible interleavings.
Since all programs sharing a channel need to agree on the communication along this channel,
a trace~$\trace$ can be observed from~$\alpha\parOp\beta$ if the subtraces $\trace\downarrow\gamma$,
\iest the events in~$\trace$ along the channels of $\gamma$,
can be observed from~$\gamma$.
The guard $\trace \downarrow (\alpha\parOp\beta) = \trace$ excludes noncausal communication 
on channels not belonging to either subprogram.
The implicit characterization of the parallel semantics by projections
is fundamental to \dLCHP's compositional axiomatization of parallel CHPs,
which supports reasoning by projection onto the relevant subtrace for a property.

The semantics of $\alpha\parOp\beta$ enforces that the subprograms evolve simultaneously in time,
because they synchronize on the reserved channel $\durchan$
by definition of $\downarrow\gamma$.
Consequently, the subprograms not only agree on the overall duration,
but always execute continuous evolutions synchronously.
For the intended semantics, 
it is important that concatenation and projection merge consecutive observable durations.
For example, 
$\repetition{(\evolution{}{non})}$ 
only emits a single event $\comItem{\durchan, \duration}$,
where~$\duration$ is the sum of the durations over all iterations of $\evolution{}{non}$.
This ensures that $\evolution{t' = 1}{non} \parOp \repetition{(\evolution{}{non})}$
is not unintentionally restricted to $\le 1$ iterations,
since $\evolution{t'=1}{non}$ emits a single event $\comItem{\durchan, \semConst}$,
and the identity $\comItem{\durchan, 0} = \epsilon$ ensures that $0$ iterations are possible.
Likewise, the projection $\trace\downarrow\gamma$ in $\sem{\alpha\parOp\beta}{}$
may remove 
communications
between durations in $\trace \in \sem{\alpha\parOp\beta}{}$ (\rref{fig:partition}).
Merging thus ensures that the semantics of $\evolution{t' = 1}{non} \parOp \beta$ is not empty,
where $\beta\equiv\evolution{x' = v}{non} \seq \receive{}{v} \seq \evolution{x' = v}{non}$
emits traces of form $\comItem{\durchan,\duration_1} \cdot \comItem{\ch{}, \semConst} \cdot \comItem{\durchan, \duration_2}$.

\begin{figure}[t]
	\centering
	\begin{tikzpicture}		
		\tikzstyle{duration}+=[fill=gray!20, rectangle, inner ysep=1pt, minimum width=2cm, font=\footnotesize]
		\tikzstyle{com}+=[draw, inner ysep=1pt, circle]

		\node[duration] (t0) {$\comItem{\durchan, \duration_1}$};
		\node[com] [right=0 of t0] (t1) {};
		\node[duration] [right=0 of t1, minimum width=1.4cm] (t2) {$\comItem{\durchan, \duration_2}$};
		\node[com] [right=0 of t2] (t3) {};
		\node[duration] [right=0 of t3, minimum width=2.2cm] (t4) {$\comItem{\durchan, \duration_3}$};
		\node[com] [right=0 of t4] (t5) {};
		\node[duration] [right=0 of t5, minimum width=1cm] (t6) {$\vphantom{\comItem{\durchan, \duration_3}}$};

		\node[duration] [above=2mm of t0] (ta0) {$\comItem{\durchan, \duration_1}$};
		\node[com] [right=0 of ta0] (ta1) {};
		\node[duration] [right=0 of ta1, minimum width=3.86cm] (ta3) {\hspace{1.3mm}$\comItem{\durchan, \duration_2+\duration_3}$};
		\node[com] [right=0 of ta3] (ta4) {};
		\node[duration] [right=0 of ta4, minimum width=1cm] (ta5) {$\vphantom{\comItem{\durchan, \duration_3}}$};

		\node[duration] [below=2mm of t0] (tb0) {$\comItem{\durchan, \duration_1}$};
		\node[com] [right=0 of tb0] (tb1) {};
		\node[duration] [right=0 of tb1, minimum width=1.4cm] (tb2) {$\comItem{\durchan, \duration_2}$};
		\node[com] [right=0 of tb2] (tb3) {};
		\node[duration] [right=0 of tb3, minimum width=3.46cm] (tb5) {$\comItem{\durchan, \duration_3+\dots}$};

		\node[text width=10mm, align=left] [left=4mm of t0] (t) {$\trace$};
		\node[text width=10mm, align=left] [left=4mm of ta0] (ta) {$\trace\downarrow\alpha$};
		\node[text width=10mm, align=left] [left=4mm of tb0] (tb) {$\trace\downarrow\beta$};

		\node [below=2mm of tb1, inner sep=0] (com1)  {\footnotesize$\comItem{\ch{}, \semConst_1}$};
		\node [below=2mm of tb3,  inner sep=0] (com2) {\footnotesize$\comItem{\ch{dh}, \semConst_2}$};
		\node [above=2mm of ta4,  inner sep=0] (com3) {\footnotesize$\comItem{\ch{gh}, \semConst_3}$};
		\draw[dashed] (ta1) -- (t1) -- (tb1) -- (com1);
		\draw[dashed] (t3) -- (tb3) -- (com2);
		\draw[dashed] (t5) -- (ta4) -- (com3);

		\node[right=.5mm of t6]  {$\ldots$}; 
		\node[right=.5mm of ta5] {$\ldots$};
		\node[right=.5mm of tb5] {$\ldots$};
	\end{tikzpicture}
	\caption{
		Trace $\trace$ of a parallel CHP $\alpha\parOp\beta$,
		where the channel $\ch{}$ is shared between $\alpha$ and $\beta$, 
		and $\ch{dh}$ and $\ch{gh}$ are exclusive channels of~$\beta$ and $\alpha$, respectively.
	 }
	\vspace{-.3cm}
	\label{fig:partition}
\end{figure}

The formula semantics (\rref{def:semantics_formulas}) of the first-order constructs is standard.
The ac-box adapts its semantics from Hoare-style ac-reasoning \cite{AcHoare_Zwiers, AcSemantics_Zwiers},
and the ac-diamond is defined as its modal dual.
The dynamic modalities are special cases of the ac-modalities,
generalizing their semantics in dynamic logic \cite{Harel1979} 
with communication.

The program semantics separates observable communication from local sequential state change,
since this separation-of-concerns is convenient for parallel composition.
These separate effects induce a classical reachability relation when merging them into reachable worlds.
To this end, define \emph{state-trace concatenation} $\pstate{v} \cdot \trace$
by $\pstate{v} \cdot \trace = \pstate{v} \subs{\recvar}{\pstate{v}(\recvar) \cdot \trace}$,
which appends a trace $\trace$ to the recorder~$\recvar$ in a state $\pstate{v}$.
The reachable worlds of a run $\run \in \sem{\alpha}{}$
are the \emph{intermediate state} $\pstate{v} \cdot \trace$
and, if $\pstate{w}\neq\bot$, the \emph{final state} $\pstate{w} \cdot \trace$.

\begin{definition}
	[Formula semantics]
	\label{def:semantics_formulas}
	The \emph{satisfaction} $\pstate{v} \vDash \phi$ of a \dLCHP formula~$\phi$ in
	state $\pstate{v}$ is inductively defined below,
	and $\sem{\phi}{} = \{ \pstate{v} \in \states \mid \pstate{v} \vDash \phi \}$ denotes all states satisfying $\phi$.
	For a set of states $U \subseteq \states$
	and a formula $\phi$, 
	write $U \vDash \phi$ if $\pstate{v} \vDash \phi$ for all $\pstate{v} \in U$. 
	Trivially, $\emptyset \vDash \phi$.
	A formula $\phi$ is \emph{valid} (written $\vDash \phi$)
	if $\pstate{v} \vDash \phi$ for all states $\pstate{v}$.
	\begin{enumerate}
		\item $\pstate{v} \vDash \expr_1 {\sim} \expr_2$ if $\sem{\expr_1}{\pstate{v}} \sim \sem{\expr_2}{\pstate{v}}$ \sidecondition[color=black]{\quad where $\sim$ is any relation symbol}
        \item $\pstate{v} \vDash \varphi \wedge \psi$ if $\pstate{v} \vDash \varphi$ and $\pstate{v} \vDash \psi$
		\item $\pstate{v} \vDash \neg \varphi$ if $\pstate{v} \nvDash \varphi$, \iest it is not the case that $\pstate{v} \vDash \varphi$
		\item $\pstate{v} \vDash \fa{\avar} \varphi$ if $\pstate{v} \subs{\avar}{d} \vDash \varphi$ for all $d \in \anysort$, 
		where $\anysort \in \{\reals,\traces\}$ is the type of $\avar \in \V_\anysort$
		\item \label{itm:dynBoxSem}
		$\pstate{v} \vDash [ \alpha ] \psi$ if $\stconcat{\pstate{w}}{\trace}\vDash \psi$ for all $\computation \in \sem{\alpha}{}$ with $\pstate{w} \neq \bot$ 
		\item \label{itm:acBoxSem}
		$\pstate{v} \vDash [ \alpha ] \ac \psi$ if for all $\computation \in \sem{\alpha}{}$ the following conditions both hold:
		\begin{align}
			\vspace{-.7em}
			& \assCommit{\pstate{v}}{\trace} \vDash \A \text{ implies } \stconcat{\pstate{v}}{\trace} \vDash \Commit \tag{commit} \label{eq:commit} \\
			&\big( \assPost{\pstate{v}}{\trace} \vDash \A \text{ and } \pstate{w} \neq \bot \big) \text{ implies } \stconcat{\pstate{w}}{\trace} \vDash \psi \tag{post} \label{eq:post}
		\end{align}%
		
		\item \label{itm:dynDiaSem}
		$\pstate{v} \vDash \langle \alpha \rangle \psi$ if $\stconcat{\pstate{w}}{\trace}\vDash \psi$ for some $\computation \in \sem{\alpha}{}$ with $\pstate{w} \neq \bot$

		\item \label{itm:acDiaSem}
		$\pstate{v} \vDash \langle \alpha \rangle \ac \psi$ 
		if for some $\run \in \sem{\alpha}{}$ at least one of the following conditions holds:
		\begin{align}
			\vspace{-.7em}
			& \assCommit{\pstate{v}}{\trace} \vDash \A \text{ and } \stconcat{\pstate{v}}{\trace} \vDash \Commit \tag{commit} \\
			&\big( \assPost{\pstate{v}}{\trace} \vDash \A \text{ and } \pstate{w} \neq \bot \big) \text{ and } \stconcat{\pstate{w}}{\trace} \vDash \psi \tag{post} 
		\end{align}%
	\end{enumerate}
\end{definition}

The ac-contract $(\A, \Commit)$ of the ac-box $[ \alpha ] \ac \psi$ 
specifies the possible communications (and observable durations) of a program,
where the assumption $\A$ limits the incoming communication
and the commitment provides a guarantee about the outgoing communication.
In \acCommit and \acPost,
the ac-contract $(\A, \Commit)$ refers to the intermediate state $\pstate{v} \cdot \trace$,
because the communication $\trace$ of a program run $\run \in \sem{\alpha}{}$ is observable by the environment,
but parallel programs cannot observe each other's state change since the initial state~$\pstate{v}$. 
The final state $\pstate{w}\cdot\trace$ combines communication \emph{and} state change,
since termination exposes the complete program behavior.
In the reachable states $\pstate{v}\cdot\trace$ and $\pstate{w}\cdot\trace$,
the recorder $\recvar$ reflects the communication $\trace$,
because state-trace concatenation appends $\trace$ to $\recvar$.
Since programs may communicate before they abort execution,
even aborting programs (\eg $\send{}{} \seq \test{\false}$) can have intermediate states,
reflecting the observable communication prior to abortion.

The condition \acCommit of $[ \alpha ] \ac \psi$ 
asserts the commitment $\Commit$ after every communication of $\alpha$,
because every prefix of an $\alpha$-run 
is a run as well by prefix-closedness (\rref{fig:modsem}).
Therefore, the ac-contract can be understood as an invariant 
of the communication trace maintained by every communication.
The assumptions~$\A$ available for proving the commitment~$\Commit$
in \acCommit
are limited to the strict prefix~($\prec$) of the current trace
to avoid circular proofs~\cite{AcHoare_Zwiers},
ensuring that mutual guarantees between parallel programs are well-founded.
Upon termination as in \acPost all~($\preceq$) assumptions are available,
since no further 
mutual commitments are possible.

\begin{figure}[t]
	\begin{center}
		\begin{tikzpicture}
			\tikzstyle{state}+=[circle, fill=blue!5, draw=blue, inner sep=0, minimum width=1.5em]
			\tikzstyle{concat}+=[circle, fill=gray!5, draw=gray, inner sep=.8mm]
			\tikzstyle{bot}+=[circle, fill=gray!20, draw=darkgray, inner sep=.15mm]

			\tikzstyle{post}+=[circle, fill=green!20, inner sep=.15mm, outer xsep=1mm]
			
			\tikzstyle{ass}=[above, circle, fill=yellow!50, text=black, inner sep=0mm]
			\tikzstyle{commit}=[ass, fill=green!20, inner sep=.3mm]

			\node[state] (v) {$\pstate{v}$};
			\node[concat] [right=0 of v] (c-0) {};

			\node[concat] [right=.6cm of c-0] (c-1) {};
			\node[concat] [right=1cm of c-1] (c-3) {};
			\node[concat] [right=.6cm of c-3] (c-4) {};

			\node[state] [right=0 of c-4] (w) {$\pstate{w}$};
			
			\node[post] [right=0 of w] {$\psi$};

			\draw[->, text=black] (c-0) -- node[above] {$\trace_1$} (c-1);
			\draw[->] (c-1) -- node[fill=white] {$\ldots$} (c-3);
			\draw[->, text=black] (c-3) -- node[above] {$\trace_\semvar$} (c-4);
			
			\draw[stealth-] ($(c-0.center)+(170:1.1mm)$) -- +(-.1, .4) node[commit, outer sep=1mm] {$\Commit$};
			\draw[stealth-] ($(c-0.center)+(10:1.1mm)$) -- +(.1, .4) node[ass, outer sep=1mm] {$?\!\A$};

			\draw[stealth-] ($(c-1.center)+(170:1.1mm)$) -- +(-.1, .4) node[commit, outer sep=1mm] {$\Commit$};
			\draw[stealth-] ($(c-1.center)+(10:1.1mm)$) -- +(.1, .4) node[ass, outer sep=1mm] {$?\!\A$};

			\draw[stealth-] ($(c-3.center)+(170:1.1mm)$) -- +(-.1, .4) node[commit, outer sep=1mm] {$\Commit$};
			\draw[stealth-] ($(c-3.center)+(10:1.1mm)$) -- +(.1, .4) node[ass, outer sep=1mm] {$?\!\A$};

			\draw[stealth-] ($(c-4.center)+(170:1.1mm)$) -- +(-.1, .4) node[commit, outer sep=1mm] {$\Commit$};
			\draw[stealth-] ($(c-4.center)+(10:1.1mm)$) -- +(.1, .4) node[ass, outer sep=1mm] {$?\!\A$};

			\node[ass] [below left=5mm and 3mm of c-1]		(l1) 	{$?\!\A$};
			\node[commit] [below right=5mm and 3mm of c-1] 	(l2)	{$\Commit$};
			\node[right=0 of l2] {no circularity};
			\draw[->] (l1) edge[bend left=25] node[red] {\textsf{\bfseries X}} (l2);
			\draw[->] (l2) edge[bend left=25] (l1);
			\draw[gray,dashed] (l1.north west) -- (c-1);
			\draw[gray,dashed] (l2.north east) -- (c-1);

			\node[state] [above right=1cm and 1.4cm of w] (v2) {$\pstate{v}$};

			\node[ass] [right=0 of v2] (c2-0) 		{$?\!\A$};
			\node[ass] [right=.6cm of c2-0] (c2-1) {$?\!\A$};
			\node[ass] [right=1cm of c2-1] (c2-3) {$?\!\A$};
			\node[ass] [right=.6cm of c2-3] (c2-4) {$?\!\A$};
			\node[state] [right=0 of c2-4] (w2) {$\pstate{w}$};
			
			\node[post] [right=0 of w2] (post) {$\psi$};
			\node [right=0 of post] {\acPost};

			\draw[->, text=black] (c2-0) -- node[above] {$\trace_1$} (c2-1);
			\draw[->] (c2-1) -- node[fill=white] {$\ldots$} (c2-3);
			\draw[->, text=black] (c2-3) -- node[above] {$\trace_\semvar$} (c2-4);

			\node[state] [below=.15cm of v2] (v3) {$\pstate{v}$};

			\node[ass] [right=0 of v3] (c3-0) 		{$?\!\A$};
			\node[ass] [right=.6cm of c3-0] 		(c3-1) 	{$?\!\A$};
			\node[ass] [right=1cm of c3-1] 	(c3-3) 	{$?\!\A$};
			\node[commit] [right=.6cm of c3-3] 	(c3-4) 	{$\Commit$};
			\node[bot] [right=0 of c3-4] 		 (c3-bot)  		{$\bot$};

			\draw[->, text=black] 	(c3-0) -- node[above] {$\trace_1$} 			(c3-1);
			\draw[->] 				(c3-1) -- node[fill=white] {$\ldots$} 		(c3-3);
			\draw[->, text=black] 	(c3-3) -- node[above] {$\trace_\semvar$} 	(c3-4);
			
			\node [below=-.15cm of v3] (dots) {\vdots};

			\node[state] [below=0.05cm of dots] (v4) {$\pstate{v}$};

			\node[ass] [right=0 of v4] (c4-0) 		{$?\!\A$};
			\node[commit] [right=.6cm of c4-0] 		(c4-1) 	{$\Commit$};
			\node[bot] [right=0 of c4-1] 		   		{$\bot$};

			\draw[->, text=black] 	(c4-0) -- node[above] {$\trace_1$} 			(c4-1);

			\node[state] [below=.15cm of v4] (v5) {$\pstate{v}$};

			\node[commit] [right=0 of v5] (c5-0) 		{$\Commit$};
			\node[bot] [right=0 of c5-0] 	(c5-bot)	   		{$\bot$};

			\draw[decorate, decoration={brace, amplitude=10pt}] ([xshift=-3mm]v5.south) -- ([xshift=-3mm]v2.north);

			\draw[decorate, decoration={brace, amplitude=10pt}] ([xshift=3mm]c3-bot.north) -- 
				node[right=3mm, align=left] {
					\acCommit for\\
					each prefix/\\
					run to an\\ 
					intermediate\\ 
					state
				} ([xshift=3mm]c3-bot.north |- c5-bot.south);
		\end{tikzpicture}
	\end{center}
	\caption{Semantics of the ac-box $[ \alpha ] \ac \psi$ for a run $\run \in \sem{\alpha}{}$ with $\trace = (\trace_1, \ldots, \trace_\semvar)$}
	\label{fig:modsem}
\end{figure}

\subsection{Static Semantics} \label{sec:static_semantics}

The previous section gave the dynamic semantics of \dLCHP,
which precisely captures the evaluation of terms and formulas, and transition behavior of CHPs.
This section introduces \dLCHP's static semantics,
which determines free names,
\iest the variables and channels expressions and programs depend on,
and bound names,
\iest the variables and channels written by programs.
The bound effect property (\rref{lem:bound_effect}) and the coincidence properties (\rreflems{lem:expr_coincidence}{lem:program_coincidence})
confirm soundness of the static semantics,
\iest that expressions only depend on their free names 
and that only bound names change during program runs,
and are an essential tool for soundness arguments. 

Definition \ref{def:static_semantics} defines the static semantics based on the dynamic semantics~\cite{Brieger2023}.
For soundness arguments, this approach is preferred 
because it precisely identifies the aspects of the static semantics that influence soundness.
Algorithms that compute
sets of free and bound names
from the syntactic structure of the expressions, as usual,
are in \rref{app:stat_sem}.
These computed sets can be soundly used in place 
of \rref{def:static_semantics},
because they also satisfy the bound effect property and the coincidence properties (\rref{lem:sound_over}).

Definition \ref{def:static_semantics} considers formulas to be truth-valued to treat terms and formulas uniformly,
\iest $\sem{\phi}{\pstate{v}} = \mathbf{tt}$ if $\pstate{v} \vDash \phi$ and $\sem{\phi}{\pstate{v}} = \mathbf{ff}$ if $\pstate{v} \nvDash \phi$.
Further, for a relation $R \subseteq A \times B$ and $X \subseteq A$,
let $\relimage{R}{X} = \{ b \mid \bfexists a \in X : (a,b) \in R \}$ be the image of $R$ under~$X$.
For $\cset\subseteq\TVar\times\Chan$,
the \emph{(state) projection} $\pstate{v} \downarrow \cset$ projects every trace variable~$\tvar$
onto the image $\relimage{\cset}{\{\tvar\}}$,
\iest $\pstate{v} \downarrow \cset = \pstate{v}$ on $\V_\reals$
and $(\pstate{v} \downarrow \cset)(\tvar) = \pstate{v}(\tvar) \downarrow \relimage{\cset}{\{\tvar\}}$ for $\tvar \in \V_\traces$.
For states $\pstate{w} \in \botop{\states}$ and $\pstate{v} \in \states$,
define $\pstate{w} \eitherstate \pstate{v} = \pstate{v}$ if $\pstate{w} = \bot$
and $\pstate{w} \eitherstate \pstate{v} = \pstate{w}$ if $\pstate{w} \neq \bot$.

\begin{definition}
	[Static semantics]
	\label{def:static_semantics}
	For term or formula $\expr$, and program $\alpha$,
	define \emph{free variables} $\SFV(\expr)$ and $\SFV(\alpha)$, \emph{bound variables} $\SBV(\alpha)$, \emph{accessed channels per trace variable} $\SCN(\expr)$, 
	and \emph{written channels} $\SCN(\alpha)$.
	Further, $\SFV(\expr_1, \ldots, \expr_\semvar) = \bigcup_{j=1}^\semvar \SFV(\expr_j)$ and similar for $\SBV(\cdot), \SCN(\cdot)$.
	\newcommand{\suchthat}{:}
	\begin{align*}
		\SFV(\expr) 
			& = \{ \avar \in \V \mid \bfexists  
			\pstate{v}, \pstate[alt]{v} \suchthat \pstate{v} = \pstate[alt]{v} \text{ on } \{ \avar \}^\complement \text{ and } \sem{\expr}{\pstate{v}} \neq \sem{\expr}{\pstate[alt]{v}} \}%
			\\
		\SCN(\expr) 
			& = \{ (\tvar, \ch{}) \in \V_\traces \times \Chan \mid \bfexists 
			\pstate{v}, \pstate[alt]{v} \suchthat \pstate{v} \downarrow \{ (\tvar, \ch{}) \}^\complement = \pstate[alt]{v} \downarrow \{ (\tvar, \ch{}) \}^\complement \text{ and }
			\sem{\expr}{\pstate{v}} \neq \sem{\expr}{\pstate[alt]{v}} \} \\
		\SFV(\alpha) 
			& = \{ \avar \in \V \mid \bfexists 
			\pstate{v}, \pstate[alt]{v}, \trace, \pstate{w} \suchthat \pstate{v} = \pstate[alt]{v} \text{ on } \{ \avar \}^\complement \text{ and } \computation \in \sem{\alpha}{} \text{,} \\[-.3em]
			& \qquad\qquad\qquad \text{ and not } \bfexists \pstate[alt]{w} \suchthat
			(\pstate[alt]{v}, \trace, \pstate[alt]{w}) \in \sem{\alpha}{} \text{ and } \pstate{w} = \pstate[alt]{w} \text{ on } \{ \avar \}^\complement \} \\
		\SBV(\alpha)
			& = \{ \avar \in \V \mid \bfexists
			\computation \in \sem{\alpha}{} \suchthat (\stconcat{(\pstate{w} \eitherstate \pstate{v})}{\trace})(\avar) \neq \pstate{v}(\avar) \} \\
		\SCN(\alpha)
			& = \{ \ch{} \in \Chan \mid \bfexists 
			\computation \in \sem{\alpha}{} \suchthat \trace \downarrow \{\ch{}\} \neq \epsilon \}
	\end{align*}
\end{definition}

The static semantics of \dLCHP refines the static semantics of \dL \cite{DBLP:journals/jar/Platzer08} by taking communication into account.
A variable $\avar$ is free in an expression~$\expr$,
\iest $\avar\in\SFV(\expr)$,
if its value affects the evaluation of $\expr$.
The coincidence property (\rref{lem:expr_coincidence}) exploits an even more precise analysis of trace variables
based on accessed channels $\SCN(\expr)$.
Formally, $(\tvar, \ch{}) \in \SCN(\expr)$ if variations of the communication along the channel~$\ch{}$ in the trace~$\tvar$
can change the value of~$\expr$.
Intuitively, a trace variable~$\tvar$ depends on communication on all channels,
but each projection in whose context~$\tvar$ occurs restricts the set of channels 
accessible to the surrounding expression.
For example, the truth of the formula $\tvar_0 = \tvar \downarrow \{\ch{}\}$ 
only depends on communication on the channels~$\{\ch{}\}$ of~$\tvar$ 
but on all channels $\Chan$ of~$\tvar_0$.
Hence, $\SCN(\tvar_0 = \tvar \downarrow \{\ch{}\}) = (\{\tvar_0\}\times\Chan) \cup \{(\tvar,\ch{}) \}$.
This precision 
is crucial for the soundness argument of the parallel injection axiom,
which embeds a subprogram into a parallel composition only if the surrounding formula does not depend on the channels of that subprogram.
The denotation of accessed channels \emph{per trace variable}
refines previous work \cite{Brieger2023}.
This will be used to refine the side condition of parallel injection (\rref{def:noninterference})
as necessary for completeness.
In $\SBV(\alpha)$, the join $\pstate{w} \eitherstate \pstate{v}$ is necessary, 
because the recorder can be even bound if $\alpha$ does not terminate.

Based on the static semantics, the bound effect property (\rref{lem:bound_effect}) and coincidence properties for terms and formulas (\rref{lem:expr_coincidence}) and programs (\rref{lem:program_coincidence}) are given in the following.
Proofs are in previous work \cite{Brieger2023}.

\begin{lemma}
	[Bound effect property] 
	\label{lem:bound_effect}
	The sets $\SBV(\alpha)$ and $\SCN(\alpha)$ are the smallest sets with the \emph{bound effect property for program} $\alpha$. 
	That is, $\pstate{v} = \pstate{w}$ on~$\TVar$ if $\pstate{w} \neq \bot$, 
	and $\pstate{v} = \stconcat{(\pstate{w} \eitherstate \pstate{v})}{\trace}$ on $\SBV(\alpha)^\complement$,
	and $\trace \downarrow \SCN(\alpha)^\complement = \epsilon$ for all $\computation \in \sem{\alpha}{}$.
\end{lemma}

\begin{lemma}[Coincidence for terms and formulas] 
	\label{lem:expr_coincidence}
	The sets $\SFV(\expr)$ and $\SCN(\expr)$ are the smallest sets with the \emph{coincidence property for the term or formula} $\expr$.
	That is, if $\pstate{v} \downarrow \SCN(\expr) = \pstate[alt]{v} \downarrow \SCN(\expr)$ on $\SFV(\expr)$,
	then $\sem{\expr}{\pstate{v}} = \sem{\expr}{\pstate[alt]{v}}$.
	In particular, for a formula $\phi$,
	this implies $\pstate{v} \vDash \phi$ iff $\pstate[alt]{v} \vDash \phi$.
\end{lemma}

Note that, by the projection onto $\SCN(\expr)$, \rref{lem:expr_coincidence} only requires coincidence of the states $\pstate{v},\pstate[alt]{v}$ on the channels $\relimage{\SCN(\expr)}{\{\tvar\}}$ 
for every trace variable $\tvar\in\SFV(\alpha)$.

\begin{lemma}[Coincidence for programs] \label{lem:program_coincidence}	
	The set $\SFV(\alpha)$ is the smallest set with the 
	\emph{coincidence property for the program}~$\alpha$.
	That is, if $\pstate{v} = \pstate[alt]{v}$ on $\varset \supseteq \SFV(\alpha)$
	and $\run \in \sem{\alpha}{}$,
	then a state~$\pstate[alt]{w}$ exists such that $(\pstate[alt]{v}, \trace, \pstate[alt]{w}) \in \sem{\alpha}{}$
	and $\pstate{w} = \pstate[alt]{w}$ on $\varset$, and ($\pstate{w} = \bot$ iff $\pstate[alt]{w} = \bot$).
\end{lemma}

The bound effect property (\rref{lem:bound_effect}) and the coincidence properties (\rreflems{lem:expr_coincidence}{lem:program_coincidence})
are the only aspects of the static semantics (\rref{def:static_semantics}) that are relevant for soundness,
because they capture the essential meaning of free and bound names.
In particular,
any superset of the static semantics can be used soundly as replacement,
since $\SFV(\cdot), \SCN(\cdot), \SBV(\cdot)$ are the smallest sets satisfying \rref{lem:bound_effect}--\ref{lem:program_coincidence}.
Importantly, all results in this article continue to apply 
under such a replacement.%
\footnote{%
	\label{ft:exchange_stat_sem}%
	The well-formed syntax (\rref{def:syntax_chps} and \rref{def:syntax_formulas}) may become more restrictive
	though,
	but this only omits cases without semantical effect.
	For example, $x \ceq x \parOp x \ceq 0$ is a well-formed program (\rref{def:syntax_chps})
	for the static semantics (\rref{def:static_semantics}),
	as $\SBV(x \ceq x) = \emptyset$.
	For $\BV(x \ceq x) = \{x\}$,
	however,
	the program is not well-formed.
}
Supersets $\FV(\cdot), \CN(\cdot), \BV(\cdot)$ 
that are computed from the syntax of expressions and programs, as usual,
are in \rref{app:stat_sem}.
\rref{lem:sound_over} shows that they are indeed supersets,
ensuring that the coincidence properties apply.
Proofs are in previous work~\cite{BriegerMP2023}.

\begin{lemma}
	[Sound overapproximation]
	\label{lem:sound_over}
	The computable static semantics $\FV(\cdot), \CN(\cdot), \BV(\cdot)$ soundly computes supersets
	such that \rref{lem:bound_effect}--\ref{lem:program_coincidence} apply to $\FV(\cdot), \CN(\cdot), \BV(\cdot)$:
	\begin{equation*}
		\SFV(\expr) \subseteq \FV(\expr)
		\quad
		\SCN(\expr) \subseteq \CN(\expr)
		\quad
		\SFV(\alpha) \subseteq \FV(\alpha)
		\quad
		\SBV(\alpha) \subseteq \BV(\alpha)
		\quad
		\SCN(\alpha) \subseteq \CN(\alpha)
	\end{equation*}
\end{lemma}

\subsection{Static Properties of Programs}
\label{sec:static_properties}

This section presents corollaries of the static semantics (\rref{sec:static_semantics}),
which reflect fundamental static properties of CHPs
that are frequently used in soundness arguments.

Programs do not depend on the communication history,
\iest $\SFV(\alpha) \cap \TVar = \emptyset$,
because all terms $\rp$ and formulas $\chi$ in CHPs only contain real variables.
Further, $\pstate{v} = \pstate{w}$ on~$\TVar$ for all $(\pstate{v}, \trace, \pstate{w}) \in \sem{\alpha}{}$ by the bound effect property (\rref{lem:bound_effect}).
This suggests \rref{cor:history_coincidence},
which is a simple consequence of \rreflems{lem:bound_effect}{lem:program_coincidence}:

\begin{corollary}
	[History independence]
	\label{cor:history_coincidence}
	For every trace variable $\tvar$ and 
	every trace $\rawtrace$,
	obtain $\run \in \sem{\alpha}{}$ iff $(\pstate{v} \subs{\tvar}{\rawtrace}, \trace, \pstate{w} \subs{\tvar}{\rawtrace}) \in \sem{\alpha}{}$.
	In particular,
	$\run \in \sem{\alpha}{}$ iff $(\pstate{v} \cdot \rawtrace, \trace, \pstate{w} \cdot \rawtrace) \in \sem{\alpha}{}$.
	Further, if $(\pstate{v} \cdot \rawtrace, \trace, \pstate{w}) \in \sem{\alpha}{}$,
	then $(\pstate{v}, \trace, \pstate[alt]{w}) \in \sem{\alpha}{}$ with $\pstate{w} = \pstate[alt]{w} \cdot \rawtrace$.
\end{corollary}

The environment can only observe the communication of a program,
including the durations of continuous evolutions.
This is why for a well-formed (\rref{def:syntax_formulas}) ac-modality $\dbleft \alpha \dbright \ac \psi$,
the communication contract $(\A, \Commit)$ 
must not refer to the hybrid dynamics of $\alpha$,
\iest $\SFV(\A, \Commit) \cap \SBV(\alpha) \subseteq \V_\traces$.
This suggests \rref{cor:communicative_coincidence},
which is a simple consequence of \rreflems{lem:bound_effect}{lem:expr_coincidence}:

\begin{corollary}
	[Communicative coincidence]
	\label{cor:communicative_coincidence}
	Let the formula $\chi$ be unaffected by the hybrid dynamics of $\alpha$,
	\iest $\SFV(\chi) \cap \SBV(\alpha) \subseteq \V_\traces$.
	Then for every $\run \in \sem{\alpha}{}$ with $\pstate{w} \neq \bot$,
	the states $\pstate{v}$ and $\pstate{w}$ coincide on $\chi$,
	\iest $\pstate{v} = \pstate{w}$ on $\SFV(\chi)$.
	In particular, $\pstate{v} \cdot \trace \vDash \chi$ iff $\pstate{w} \cdot \trace \vDash \chi$.
\end{corollary}

Validity of a safety contract $\varphi \rightarrow [ \alpha ] \psi$ is preserved
when composing 
a program~$\beta$ with~$\alpha$ in parallel,
provided that any possible behavior of $\beta$ on which~$\psi$ depends synchronizes with $\alpha$.
\emph{Noninterference} (\rref{def:noninterference}) identifies the exact condition under which 
this is sound,
through a constraint on the static semantics. 
Noninterference will serve as the soundness criterion of the parallel injection axiom \mbox{$[ \alpha ] \ac \psi \rightarrow [ \alpha\parOp\beta ] \ac \psi$},
the key principle for safety reasoning about parallel CHPs.

\begin{definition}	
	[Noninterference] 
	\label{def:noninterference}
	A program $\beta$ \emph{does not interfere} with a formula-program pair $\nointfpair{\alpha}{\chi}$
	if the conditions in \rref{eq:noninterference} hold.
	The program~$\beta$ \emph{does not interfere} with an ac-box $[ \alpha ] \ac \psi$ if $\beta$ does not interfere with $\nointfpair{\alpha}{\chi}$ for all $\chi \in \{\A, \Commit, \psi\}$.
	\begin{align}
		\SFV(\chi) \cap \SBV(\beta) \subseteq \gvars
		&&\relimage{\SCN(\chi)}{\{\recvar\}} \cap \SCN(\beta) \subseteq \SCN(\alpha) \cup \{\durchan\}
		\label{eq:noninterference}
	\end{align}%
\end{definition}

Noninterference guarantees soundness of parallel injection,
because it ensures that~$\beta$ does not write any variables and channels that $\A, \Commit, \psi$ depend on,
except for the names $\varset = \SCN(\alpha) \cup \gvars \cup \{\durchan\}$.
The names $\varset$ can be soundly written by $\beta$,
because any communication of~$\beta$ changing values of $\varset$ 
matches with an $\alpha$-communication
by synchronization on shared channels.
Hence, $\A, \Commit, \psi$ are known to hold for that $\beta$-communication by $[ \alpha ] \ac \psi$.
The channel $\durchan$ is always shared,
because parallel programs always synchronize in time.
On~$\alpha$'s state, $\beta$ has no influence due to \dLCHP's distributed systems semantics,
where programs do not share state (\rref{def:syntax_chps}).
Since communication is recorded by $\recvar$,
the program $\beta$ only influences~$\chi$ on
the channels $\relimage{\SCN(\chi)}{\{\recvar\}}$ 
accessed in~$\chi$ via the recorder $\recvar$.
This closes a subtle completeness gap of previous work \cite{Brieger2023},
where the respective condition on channels is $\relimage{\SCN(\chi)}{\V_\traces} \cap \SCN(\beta) \subseteq \SCN(\alpha) \cup \{\durchan\}$,
prohibiting
$\beta$ to write channels that are accessed in~$\chi$ via \emph{any} trace variable.
For example, $\phi \equiv [ \test{\true} ] \tvar_0 = \epsilon \rightarrow [ \test{\true} \parOp \send{}{} ] \tvar_0 = \epsilon$
is valid, 
but $\phi$ does not fulfill the previous side condition
as $\relimage{\SCN(\tvar_0 = \epsilon)}{\V_\traces} = \Chan$.

The noninterference property (\rref{lem:noninterference})
proves that a formula $\psi$ is equivalently satisfied over the reachable states of a program $\alpha$
and the parallel composition $\alpha\parOp\beta$
when the injected program $\beta$ does not interfere with $\nointfpair{\alpha}{\psi}$.
The proof is in \rref{app:soundness}.

\begin{lemma}
	[Noninterference] 
	\label{lem:noninterference}
	Let the CHP $\beta$ not interfere (\rref{def:noninterference}) with 
	$\nointfpair{\alpha}{\psi}$.
	Moreover, let $\run \in \sem{\alpha \parOp \beta}{}$
	and $(\pstate{v}, \trace \downarrow \alpha, \pstate{w}_\alpha) \in \sem{\alpha}{}$,
	so that $\pstate{w} = \pstate{w}_\alpha \merge \pstate{w}_\beta$
	for some $(\pstate{v}, \trace \downarrow \beta, \pstate{w}_\beta) \in \sem{\beta}{}$.
	Then the following holds:
	\begin{enumerate}
		\item \label{itm:invalidation1}
		$\pstate{v} \cdot (\trace \downarrow \alpha) \vDash \psi \text{ iff } \pstate{v} \cdot \trace \vDash \psi$
		
		\item \label{itm:invalidation2} 
		$\pstate{w} \neq \bot$ implies $\big( 
			\pstate{w}_\alpha \cdot (\trace \downarrow \alpha) \vDash \psi \text{ iff } \pstate{w} \cdot \trace \vDash \psi 
		\big)$
	\end{enumerate}
\end{lemma}

\subsection{Substitution and Recorder Renaming}
\label{sec:substitution}

The calculus (\rref{sec:calculus}) uses 
substitutions of terms for variables in formulas and programs.
For $\avar$ and $\expr$ with equal type, 
the substitution $\phi \subs{\avar}{\expr}$ replaces the variable~$\avar$ by the term~$\expr$ in~$\phi$.
By bound variable renaming ($\alpha$-conversion),
we assume every substitution~$\phi \subs{\avar}{\expr}$ is admissible,
\iest in $\phi\subs{\avar}{\expr}$ neither the variable $\avar$ nor any free variable of the replacement~$\expr$ 
occur in the scope of a quantifier or modality 
binding this variable.

Substitution $\phi \subs{\avar}{\expr}$ can be implemented by standard capture-avoidance \cite{Platzer10}.
However, if $\avar$ is the recorder $\recvar$ and occurs in the scope of communication statements,
capture-avoidance is subtle,
because $\recvar$ collects communication
such that the same occurrence of~$\recvar$ is free \emph{and} bound.
For example, 
the postcondition $\len{\recvar}{=}2$ of $\phi \equiv [ \send{}{0} ] \len{\recvar} {=} 2$ 
depends on the initial length~$\len{\recvar}$ and the send statement,
\iest the occurrence of $\recvar$ in $\len{\recvar}{=}2$ is free and bound in $\phi$.
In this case,
substitution 
can be defined by separating the initial value assignment and bound variable renaming \cite[Section 2.5.1]{Platzer10}.

However, renaming of the recorder $\recvar$ when bound by a program~$\alpha$
is not supported directly, 
because it is fixed.
This is not an issue in Hoare-style ac-reasoning \cite{AcHoare_Zwiers},
because Hoare-logic does not enjoy the flexibility of dynamic logic,
which allows nesting of modalities.
Consequently, the pre- and postconditions of Hoare-triples targeted by substitutions never contain programs.
To recover bound variable renaming as a standard feature of logic, 
it can be simulated for the recorder $\recvar$ using appropriate quantifiers.
For this purpose, 
the modality $[ \recchp{\tvar}{\alpha} ] \psi$ with recorder $\tvar$
can be defined as follows (other ac-modalities are analogous),
where $\tvar_0$ is a fresh variable:
\begin{equation*}
	[ \recchp{\tvar}{\alpha} ] \psi \equiv \fa{\tvar_0{=}\recvar} \fa{\recvar{=}\tvar} [ \alpha ] \fa{\tvar{=}\recvar} \fa{\recvar{=}\tvar_0} \psi
\end{equation*}
This makes $\tvar$ appear as if it records the communication of $\alpha$.
Since~$\tvar_0$ is fresh, 
the semantics of $[ \recchp{\tvar}{\alpha} ] \psi$ 
indeed equals that of $[ \alpha ] \psi$ 
when replacing the recorder $\recvar$ with~$\tvar$:
\begin{equation*}
	\pstate{v} \vDash [ \recchp{\tvar}{\alpha} ] \psi \quad\text{ iff }\quad \pstate{w} \subs{\tvar}{\pstate{w}(\tvar)\cdot\trace} \vDash \psi \text{ for all } \run \in \sem{\alpha}{} \text{ with } \pstate{w}\neq\bot
\end{equation*}

The substitution~$\phi\subs{\tvar}{\te}$ for trace variables is standard capture-avoid substitution for the first-order connectives.
The case $\phi \equiv [ \recchp{\tilde{\tvar}}{\alpha} ] \psi$ of a box with recorder $\tilde{\tvar}$
is defined in \rref{eq:subs_ac_box}
by separating the initial value assignment,
where the bound variables $\SBV(\recchp{\tvar}{\alpha})$ with recorder $\tvar$
are defined by $\SBV(\recchp{\tvar}{\alpha}) = \SBV(\alpha) \setminus \{\recvar\} \cup \{\tvar \mid \recvar \in \SBV(\alpha) \}$:
\begin{equation}
	\label{eq:subs_ac_box}
	([ \recchp{\tilde{\tvar}}{\alpha} ] \psi) \subs{\tvar}{\te} \equiv
	\begin{cases}
		[ \recchp{\tilde{\tvar}}{\alpha} ] \psi \subs{\tvar}{\te}
			& \text{if } \tvar \not\equiv \tilde{\tvar} \text{ and } \SFV(\te)\cap\SBV( \recchp{\tilde{\tvar}}{\alpha})=\emptyset \\
		\fa{\tvar_0{=}\te} [ \recchp{{\tilde{\tvar}}\subs{\tvar}{{\tvar_0}}}{\alpha} ] \psi \subs{\tvar}{{\tvar_0}}
			& \text{else, where } \tvar_0 \text{ is fresh}
	\end{cases} 
\end{equation}
The case $\phi \equiv [ \alpha ] \psi$ is covered by $[ \recchp{\recvar}{\alpha} ] \psi$,
and other (ac-)modalities are analogous.
If $\te\in\TVar$,
the second case in \rref{eq:subs_ac_box} simplifies to $[ \recchp{{\tilde{\tvar}}\subs{\tvar}{{\te}}}{\alpha} ] \psi \subs{\tvar}{{\te}}$.
A full definition of substitution can be found in
\iflongversion
	\rref{app:substitution}.
\else
 	the report \cite{Brieger2024}.
\fi
The resulting substitution property (\rref{lem:rec_substitution}) for \dLCHP is standard.
The proof is in 
\iflongversion
	\rref{app:substitution}.
\else
	the report \cite{Brieger2024}.
\fi

\begin{lemma}
	[Substitution]
	\label{lem:rec_substitution}
	Let $\avar \in V$ be a variable of any type and $\expr$ is a term of equal type.
	Then $\pstate{v} \vDash \phi \subs{\avar}{\expr}$ iff $\pstate{v} \subs{\avar}{\sem{\expr}{\pstate{v}}} \vDash \phi$.
\end{lemma}

\section{Axiomatization} 
\label{sec:calculus}

\newcommand{\sidecond}[1]{\;\;{\color{gray}(#1)}}

\newcommand{\hdia}[1]{\langle\!\langle #1 \rangle\!\rangle}
\newcommand{\syncdia}[2]{\hdia{#1} \acpair{#2}}
\newcommand{\obsdia}[1]{\hdia{#1}}

\newcommand{\nojunkQ}[2]{\mathcal{Q}^{#1}{#2} \,}

\newcommand{\obsarrow}[1]{\overset{#1}{\mathrel{\raisebox{-2pt}{$\rightsquigarrow$}}}}

\newcommand{\runcommit}[2]{\langle #1 \etOp #2 \rangle}
\newcommand{\runpost}[4]{\langle #1 \etOp #3 \rangle}

Figure \ref{fig:calculus} presents a Hilbert-style proof calculus for \dLCHP,
which is sound and complete.
Figure~\ref{fig:derived} presents derived axioms and rules.
In \dLCHP, hybrid systems and discrete parallelism culminate.
Therefore, the \dLCHP calculus generalizes \dL's proof calculus for hybrid systems \cite{DBLP:journals/jar/Platzer08,DBLP:conf/lics/Platzer12b} and embeds ac-reasoning \cite{Misra1981,AcHoare_Zwiers} 
to enable compositional verification of parallelism
by mutual abstraction of parallel program effects.
Since Hoare-style ac-reasoning is not based on dynamic logic like \dL,
the \dLCHP calculus puts value in reconciling these two bases in a graceful way:
This manifests itself in the clear modal logic interpretation that ac-reasoning receives through the calculus
while generalizing the Pratt-Segerberg \cite{Pratt1976, Segerberg1982} proof system for dynamic logic whenever possible.
As a result, \dLCHP constitutes a genuine dynamic logic,
enabling ac-reasoning to integrate seamlessly with standard  reasoning in dynamic logic.

\tikzstyle{equi}=[semithick, stealth-stealth]
\tikzstyle{label}=[fill=white, inner sep=0, outer sep =0]

\tikzstyle{related}=[
	equi, decorate, decoration={
		coil,
		segment length=5,
		amplitude=.7, 
		pre=lineto,
		pre length=4pt,
		post=lineto,
		post length=2pt
	}
]

\tikzstyle{opposite}=[
	equi, decorate, decoration={
		zigzag, 
		segment length=+3pt, 
		amplitude=+.75pt, 
		post length=+4pt,
		pre length=+4pt
	}
]

\newsavebox\axiomaticRela
\newsavebox\derivedRela
\newsavebox\oppositeRela

\sbox\axiomaticRela{\tikz[baseline=-0.5ex]{\draw[equi] (0,0) -- (.7, 0);}}
\sbox\derivedRela{\tikz[baseline=-0.5ex]{\draw[equi, dashed] (0,0) -- (.7, 0);}}
\sbox\oppositeRela{\tikz[baseline=-0.5ex]{\draw[opposite] (0,0) -- (.7, 0);}}

\begin{figure}[b]
	\begin{minipage}{.42\textwidth}
		\captionof{figure}{
			\small The four modalities are related by duality \RuleName{acdbDual, dbDual}, flattening \RuleName{acNoCom, acDiaNoCom}, and embedding \RuleName{boxesDual, diasDual}.
			The arrows are axiomatic (\usebox{\axiomaticRela})
			and derived (\usebox{\derivedRela}) equivalences,
			and logical opposite (\usebox{\oppositeRela}).
		}
		\label{fig:axioms_modal_dualities}
	\end{minipage}%
	\hspace{.02\textwidth}
	\begin{minipage}{.55\textwidth}
		\let\tmp\RuleName
		\def\RuleName#1{\scalebox{.8}{\tmp{#1}}}

		\vspace*{.2cm}
		\centering
			\begin{small}
				\begin{tikzpicture}[
				every node/.style={
					minimum height=.7cm,
					outer sep=0,
					inner sep=1mm	
				},
				node distance=.7cm and 1cm
			]

				\tikzstyle{formula}=[fill=gray!20, rounded corners=1mm]

				\node[formula] (box) {$[ \alpha ] \psi$};
				\node[formula] (acBox) [right=1.2cm of box]  {$[ \alpha ] \acpair{\A, \Commit} \psi$};
				\node[formula] (boxAc) [right=1.2cm of acBox] {$\Commit {\wedge} (\A {\rightarrow} [ \alpha ] \psi )$};
				\node[formula] (dia) [below=of box] {$\langle \alpha \rangle \neg\psi$};
				\node[formula] (acDia) [below=of acBox] {$\langle \alpha \rangle \acpair{\A, \neg\Commit} \neg\psi$};
				\node[formula] (diaAc) [below=of boxAc] {$\neg\Commit {\vee} \A {\wedge} \langle \alpha \rangle \neg\psi$};

				\draw[opposite] (acBox) -- node[label, fill=none, right=1mm] {\RuleName{acdbDual}} (acDia);
				\draw[opposite] (box) -- node[label, fill=none, left=1mm] {\RuleName{dbDual}} (dia);
				\draw[opposite] (boxAc) -- node[label, fill=none, right=1mm] {\RuleName{dbDual}} (diaAc);	

				\draw[equi] (acBox) -- node[label, fill=none, above=-1mm, name=tt-ax] {\RuleName{boxesDual}} (box);
				\draw[equi] (acBox) -- node[label, fill=none, above=-1mm, name=nocom-ax] {\RuleName{acNoCom}} (boxAc);
				
				\draw[equi, dashed] (acDia) -- node[label, fill=none, below=-1mm, name=tf-ax] {\RuleName{diasDual}} (dia);
				\draw[equi, dashed] (acDia) -- node[label, fill=none, below=-1mm, name=nodia-ax] {\RuleName{acDiaNoCom}} (diaAc);

				\node (tt-cond) 		[below=2.5mm of tt-ax] {{\footnotesize $\A\equiv\Commit\equiv\true$}};
				\node (nocom-cond)    	[below=2.5mm of nocom-ax, xshift=2mm] {{\footnotesize $\SCN(\alpha)=\emptyset$}};

				\draw[gray, shorten <=-3.5pt] (tt-cond) -- (tt-cond |- tt-ax.south);
				\draw[gray, shorten <=-3.5pt] (nocom-cond) -- (nocom-cond |- nocom-ax.south);
				\draw[gray, shorten <=-3.5pt] (tt-cond) -- (tt-cond |- tf-ax.north);
				\draw[gray, shorten <=-3.5pt] (nocom-cond) -- (nocom-cond |- nodia-ax.north);

			\end{tikzpicture}
		\end{small}

		\let\RuleName\tmp
	\end{minipage}
\end{figure}

\begin{figure}[t]
	\begin{small}
		\begin{minipage}{\textwidth}
			\begin{calculus}
				\startAxiom{assign}
					$\axkey{[ x \ceq \rp] \psi(x)} \leftrightarrow \psi(\rp)$
				\stopAxiom
				\startAxiom{nondetAssign}
					$\axkey{[ x \ceq *] \psi} \leftrightarrow \fa{x} \psi$
				\stopAxiom
				\startAxiom{test}
					$\axkey{[ \test{} ] \psi} \leftrightarrow (\chi \rightarrow \psi)$
				\stopAxiom
				\startAxiom{boxesDual}
					$\axkey{[ \alpha ] \psi} \leftrightarrow [ \alpha ] \acpair{\true, \true} \psi$
				\stopAxiom
				\startAxiom{dbDual}
					$\axkey{\langle \alpha \rangle \psi} \leftrightarrow \neg [ \alpha ] \neg\psi$
				\stopAxiom
			\end{calculus}%
			\hspace*{.2cm}%
			\begin{calculus}
				\startAxiom{acComposition}
					$\axkey{[\alpha \seq \beta] \ac \psi} \leftrightarrow [\alpha] \ac [\beta] \ac \psi$
				\stopAxiom
				\startAxiom{acChoice}
					$\axkey{[\alpha \cup \beta] \ac \psi} \leftrightarrow [\alpha] \ac \psi \wedge [\beta] \ac \psi$
				\stopAxiom
				\startAxiom{acIteration}
					$\axkey{[ \repetition{\alpha }] \ac \psi} \leftrightarrow [ \alpha^0 ] \ac \psi \wedge [\alpha] \ac [ \repetition{\alpha} ] \ac \psi$%
					\footnote{\label{ft:abbrevs}\small Remember that $\alpha^0 \equiv \test{\true}$.}
				\stopAxiom
				\startAxiom{acInduction}
					$\axkey{[ \repetition{\alpha} ] \ac \psi} \leftrightarrow [\alpha^0] \ac \psi \wedge [ \repetition{\alpha} ] \acpair{\A, \true} (\psi \rightarrow [\alpha] \ac \psi)$%
					\footnoteref{ft:abbrevs}%
				\stopAxiom
				\startAxiom{acdbDual}
					$\axkey{\langle \alpha \rangle \ac \psi} \leftrightarrow \neg [ \alpha ] \acpair{\A, \neg\Commit} \neg\psi$
				\stopAxiom
			\end{calculus}
			\vspace*{.2cm}

			\parbox{.73\textwidth}{
				\begin{calculus}
					\startAxiom{acNoCom}
						$\axkey{[ \alpha ] \ac \psi} \leftrightarrow \Commit \wedge (\A \rightarrow [ \alpha ] \psi)$%
						\sidecond{$\SCN(\alpha) = \emptyset$}%
						\footnote{\label{ft:well-formed}%
							\small Care must be taken to obtain well-formed ac-boxes (\rref{def:syntax_formulas}) in right-to-left application.
						}
					\stopAxiom
					\startAxiom{send}
						$\axkey{[ \send{}{} ] \psi(\recvar)} 
						\leftrightarrow \fa{\tvar_0} 
							\big(
								\tvar_0 = \recvar \cdot \comItem{\ch{}, \rp} \rightarrow \psi(\tvar_0)
							\big)$%
						\footnote{\label{ft:fresh}%
							\small The variables $\tvar_0$, $v$, $\auxhis$, and $t$ are assumed to be fresh.
						}%
					\stopAxiom
					\startAxiom{acCom}
						$\axkey{[ \send{}{} ] \ac \psi} \leftrightarrow [ \test{\true} ] \ac [ \send{}{} ] [ \test{\true} ] \ac \psi$
					\stopAxiom
					\startAxiom{comDual}
						$\axkey{[ \receive{}{} ] \ac \psi} \leftrightarrow [ x \ceq * ] [ \send{}{x} ] \ac \psi$
					\stopAxiom
				\end{calculus}
				\vspace*{.2cm}

				\begin{calculus}
					\startAxiom{Aweak}
						$[\alpha] \acpair{\true, \Commit \wedge \weakA \rightarrow \A} \true 
						\rightarrow \big(
							[\alpha] \ac \psi \rightarrow \axkey{[\alpha] \acpair{\weakA, \Commit} \psi}
						\big)$
					\stopAxiom
					\startAxiom{acModalMP}
						$[\alpha] \acpair{\A, \Commit_1 \rightarrow \Commit_2} (\psi_1 \rightarrow \psi_2)
						\rightarrow \big( 
							[\alpha] \acpair{\A, \Commit_1} \psi_1 \rightarrow \axkey{[\alpha] \acpair{\A, \Commit_2} \psi_2}	
						\big)$
					\stopAxiom
				\end{calculus}
			}%
			\parbox{.23\textwidth}{
				\vspace*{-1cm}				
				\begin{calculus}[r]
					\startRule{MP}
						\Axiom{$\varphi$}
						\Axiom{$\varphi\rightarrow\psi$}
						\BinaryInf{$\psi$}
					\stopRule
					\startRule{acG}
						\Axiom{$\Commit \wedge \psi$}
						\UnaryInf{$[ \alpha ] \ac \psi$}
					\stopRule
					\startRule{forall}
						\Axiom{$\psi$}
						\UnaryInf{$\fa{\avar} \psi$}
					\stopRule
				\end{calculus}
			}
			\vspace*{.2cm}

			\begin{calculus}
				\startAxiom{acConvergence}
					$\A \wedge [ \repetition{\alpha} ] \acpair{\A, \true} \fa{v{>}0} \big(
						\varphi(v) \rightarrow \langle \alpha \rangle \acpair{\A, \false} \varphi(v - 1)
					\big) 
					\rightarrow
					\fa{v} \big(
						\varphi(v) \rightarrow
						\axkey{\langle \repetition{\alpha} \rangle \acpair{\A, \false} \ex{v{\le}0} \varphi(v)}
					\big)$%
					\footnoteref{ft:fresh}
				\stopAxiom
				\startAxiom{acDropComp}
					$[ \alpha ] \ac \psi \rightarrow \axkey{[ \alpha \parOp \beta ] \ac \psi}$%
					\sidecond{$\beta$ does not interfere with $[ \alpha ] \ac \psi$ (\rref{def:noninterference})}
				\stopAxiom
				\startAxiom{acLiveParCommit}
					\PrintAxiomDef{acLiveParCommit}%
					\footnoteref{ft:fresh}%
					\footnote{%
						\label{ft:trace_constraint}%
						\small The trace constraint $\alpha \etOp \te$ restricts the runs of $\alpha$
						to those with communication trace~$\te$.
						In context of an (ac-)diamond, 
						it is defined by
						$\runcommit{\gamma}{\te} \equiv \fa{\tvar{=}\te} \fa{\recvar{=}\epsilon} \langle \gamma \rangle \acpair{\A, \recvar=\tvar} \false$
						and $\runpost{\gamma}{}{\te}{} \psi \equiv \fa{\tvar{=}\te} \fa{\recvar{=}\epsilon} \langle \gamma \rangle (\recvar = \te \wedge \psi)$,
						where $\tvar$ is fresh.
					}
				\stopAxiom
				\startAxiom{acLivePar}
					\PrintAxiomDef{acLivePar}%
					\footnoteref{ft:fresh}%
					\footnoteref{ft:trace_constraint}
				\stopAxiom
				\startAxiom{odeSync}
					$\axkey{[ \evolution{}{} ] \ac \psi} 
					\leftrightarrow \fa{\tvar_0{=}{\recvar}} \fa{t{=}0}
						[ \evolution{x' = \rp, t' = 1}{} ] 
						\fa{\recvar{=}\tvar_0} [ \send{\durchan}{t} ] \ac \psi$%
					\footnoteref{ft:fresh}%
				\stopAxiom
			\end{calculus}
			\vspace*{.2cm}

			\begin{calculus}
				\startAxiom{hExtension}
					$\fa{\tvar_0{=}\recvar} [ \alpha ] \acpair{\true, \recvar \succeq \tvar_0} \recvar \succeq \tvar_0$%
					\footnoteref{ft:fresh}
				\stopAxiom
			\end{calculus}%
			\hspace*{.2cm}%
			\begin{calculus}
				\startAxiom{Atransfer}
					$\fa{\tvar_0{=}\recvar} \big(
						[ \alpha ] \acpair{\true, \Abutlast{\A} \rightarrow \Commit} (\Aglobally{\A} \rightarrow \psi)	
						\leftrightarrow
						[ \alpha ] \acpair{\A, \Commit} \psi
					\big)$%
					\footnoteref{ft:fresh}%
					\footnote{
						\small The modality $\Atrace{\A} \equiv \AclDef{}{\recvar}$ quantifies over the communication trace of~$\alpha$,
						without the last event if $\sim\equiv\prec$ and including the last event if $\sim\equiv\preceq$.
					}%
				\stopAxiom
			\end{calculus}
		\end{minipage}
	\end{small}
	\vspace{-.2em}
	\caption{\dLCHP proof calculus}
	\label{fig:calculus}
	\vspace{-.5cm}
\end{figure}

\subsection{Proof Calculus}

The \dLCHP calculus (\rref{fig:calculus}) is a first-order Hilbert-system based on the proof rules modus ponens \RuleName{MP} and \RuleName{forall}-generalization (\ForallGen).
Additionally, we assume the calculus to contain a complete axiomatization of first-order logic.
In particular, all valid formulas of (decidable) first-order real arithmetic are assumed as axioms \cite{Tarski1951}.
The calculus is an instance of system K \cite{Fitting1999},
like every classical dynamic logic~\cite{Pratt1976},
as it includes ac-versions of modal modus ponens (axiom \RuleName{acModalMP}) and Gödel's generalization rule (rule~\RuleName{acG}).
A formula $\phi$ is \emph{provable} in the \dLCHP calculus
if it is an instance of an axiom or the conclusion of a rule whose premises are all provable inductively.
If a \dLCHP formula~$\phi$ can be proven in the calculus,
write $\vdash \phi$.

The \dLCHP calculus is modular and features compositional axioms, 
each targeting one specific dynamical aspect of parallel hybrid systems.
This development enables modular soundness arguments,
\emph{and} completeness confirms that combinations of the \dLCHP axioms always suffice 
to prove \emph{all} properties of parallel hybrid systems.
Compositionality enables to reason about subcomponents in isolation,
effectively reducing the verification of complex systems to simpler subsystems.

Predominantly, 
each program statement is axiomatized by only one of the four modality types.
Switching between dynamic and ac-reasoning (\rref{fig:axioms_modal_dualities})
fills the gap to the other modalities,
thus minimizes the need for axioms,
and enables modular separation-of-concerns between communication and other dynamics.
Non-communicating atomic programs are sufficiently captured in boxes as axiom~\RuleName{acNoCom} can flatten an ac-box with these programs.
Conversely, axiom~\RuleName{boxesDual} transfers any axiom on ac-boxes to boxes.
The axioms \RuleName{acdbDual} and \RuleName{dbDual} bridge safety and liveness modalities.
Only repetition and parallelism require separate axioms for safety and liveness.

The \dLCHP calculus (\rref{fig:calculus}) is sound (\rref{thm:soundness}).
That is, every formula proven by the \dLCHP calculus from valid premises is valid.
\rref{cor:derived} establishes soundness of additional axioms and proof rules (\rref{fig:derived}) by deriving them in the calculus.
The proofs of \rref{thm:soundness} and \rref{cor:derived} are in \rref{app:soundness}.
\begin{theorem}
	[Soundness]
	\label{thm:soundness}
	The \dLCHP calculus (\rref{fig:calculus}) is sound,
	\iest every axiom is a valid formula and for every proof rule the conclusion is valid if the premises are valid.
	Consequently, every formula that derives from the axioms and rules in the \dLCHP calculus is valid.
\end{theorem}
\begin{corollary}
	[Derived axioms and rules]
	\label{cor:derived}
	The axioms and rules in \rref{fig:derived} derive in \dLCHP's proof calculus,
	thus they are sound.
\end{corollary}

\paragraph{Parallel Injection}

Parallel injection $[ \alpha ] \ac \psi \rightarrow [ \alpha\parOp\beta ] \ac \psi$
by axiom \RuleName{acDropComp}
is the key principle for safety reasoning about parallel CHPs.
It injects an additional program $\beta$ into a safety contract $[ \alpha \parOp \_\, ] \ac \psi$ 
if program $\beta$ does not interfere with the contract (\rref{def:noninterference}).
The axiom can be read modulo commutativity of parallel composition,
\iest~$\beta$ can be injected right \emph{and} left of~$\alpha$.
Parallel injection refines the dynamics of~$\alpha$
to the behavior that synchronizes with the environment~$\beta$ (\rref{fig:parInj}).
Noninterference~(\rref{def:noninterference}) guarantees soundness,
because it ensures that the program~$\beta$,
when injected into \mbox{$[ \alpha \parOp \_\, ] \ac \psi$},
has no influence on the ac-contract $(\A, \Commit)$ and the postcondition $\psi$.

Despite its convincing simplicity,
parallel injection suffices to prove \emph{all} safety properties even of parallel \emph{hybrid} systems
once combined with elementary modal logic principles
to combine the insights from successive injections for all parallel subsystems.
The axiom is truly compositional~\cite{deRoever2001},
because it reduces a property of $\alpha\parOp\beta$
to its subsystem~$\alpha$
based on specifications $\A$, $\Commit$, $\psi$ that refer only 
to the observable dynamics of~$\alpha$ via the recorder $\recvar$,
which enables to verify subcomponents in isolation.
By choosing suitable specifications,
the axiom supports coarse abstractions that mitigate state space explosion.
Importantly, noninterference (\rref{def:noninterference}) is permissive enough 
to be satisfied by all local safety properties of~$\alpha$,
providing enough insights about each subprogram of a parallel composition for completeness.
Recall that soundness is preserved when exchanging the static semantics (\rref{def:static_semantics})
used in \rref{def:noninterference}
with any superset, 
including computable notions of free and bound names (\rref{sec:static_semantics}),
which enable checking the side condition 
algorithmically.
Completeness is preserved,
because the well-formed syntax is restricted accordingly (\rref{ft:exchange_stat_sem}).

\tikzstyle{stateeq}+=[double distance=1pt]
\tikzstyle{error}+=[fill=red!10, draw=red, text=red, dashed]

\newsavebox\equalityRela
\newsavebox\droppedState

\sbox\equalityRela{\tikz[baseline=-0.5ex]{\draw[stateeq] (0,0) -- (.5, 0);}}
\sbox\droppedState{\tikz[baseline=-0.5ex]{\node[circle, error] {};}}

\begin{figure}[t]
	\centering
	\begin{tikzpicture}[node distance=.8cm]
		\tikzstyle{state}+=[circle, fill=blue!5, draw=blue, inner sep=0, minimum width=1.5em];
		\tikzstyle{proglabel}+=[fill=Gray!20, rounded corners=2mm, minimum height=1.6em, align=center];
		\tikzstyle{smalllabel}+=[fill=Gray!20, rounded corners=2mm, align=center];
		\tikzstyle{post}+=[circle, fill=green!20, inner sep=.15mm, outer xsep=1mm]

		\node[state] (seq-v) {$\pstate{v}$};

		\node[state, circle split] [right=1cm of seq-v] (seq-w-a) {};
		\node[yshift=3pt] at (seq-w-a.center) {\footnotesize$\pstate{w}_\alpha$};
		\node[yshift=-3.5pt] at(seq-w-a.center) {\footnotesize$\pstate{v}$};
		
		\draw[->] (seq-v) 
			--
			node[smalllabel, inner sep=1mm] {\footnotesize$\alpha\vphantom{\beta}$}
			node[above=1.5mm] {$\recvar \downarrow \alpha$} (seq-w-a);

		\node[post] [right=0 of seq-w-a] (seq-p) {$\psi$};

		\node[state] [right=2.7cm of seq-w-a] (par-v) {$\pstate{v}$};

		\node[state, circle split, font=\footnotesize] [above right=.4cm and 2cm of par-v] (par-w-a) {};
		\node[yshift=3pt] at (par-w-a.center) {\footnotesize$\pstate{w}_\alpha$};
		\node[yshift=-3.5pt] at(par-w-a.center) {\footnotesize$\pstate{v}$};

		\node[state, circle split, font=\footnotesize, error] [below right=.4cm and 2cm of par-v] (par-w-b) {};
		\node[yshift=3pt] at (par-w-b.center) {\footnotesize$\pstate{v}$};
		\node[yshift=-3.5pt] at(par-w-b.center) {\footnotesize$\pstate{w}_\beta$};

		\node[state, circle split, font=\footnotesize, draw=none, fill=none] [inner sep = 0] at ($(par-w-a.center)!.5!(par-w-b.center)$) (par-w) {};
		\draw[state] (par-w.center) ++(-.75em,0) arc (180:0:.75em);
		\draw[error] (par-w.center) ++(-.75em,0) arc (180:360:.75em);
		\draw[state] (par-w.east) -- (par-w.west);
		\node[yshift=3pt] at (par-w.center) {\footnotesize$\pstate{w}_\alpha$};
		\node[yshift=-3.5pt] at(par-w.center) {\footnotesize$\pstate{w}_\beta$};

		\draw (par-w-a) edge[stateeq, bend left=60] node[post, name=par-p] [right=0 of par-w-a] {$\psi$} (par-w);
		\node[right=-1mm of par-p] {\shortstack[l]{
			\footnotesize$\SFV(\psi) \subseteq \SBV(\beta)^\complement \cup \gvars$\\
			\footnotesize$\relimage{\SCN(\psi)}{\{\recvar\}} \subseteq \SCN(\beta)^\complement \cup \SCN(\alpha) \cup \{\durchan\}$\\
			\footnotesize by noninterference (\rref{def:noninterference})
		}};

		\draw[->] (par-v) 
			edge [bend left=20] 
			node[smalllabel, inner sep=1mm, pos=.3] {\footnotesize$\alpha\vphantom{\beta}$}
			node[fill=white, pos=.7, inner sep=.5mm] {$\recvar \downarrow \alpha$} (par-w-a);
		\draw[->, error] (par-v) 
			edge [bend right=20] 
			node[smalllabel, inner sep=1mm, pos=.3, text=black] {\footnotesize$\beta$}
			node[fill=white, pos=.7, inner sep=.5mm, text=black] {$\recvar \downarrow \beta$} (par-w-b);

		\draw[->] (par-w-a) -- (par-w);
		\draw[->, error] (par-w-b) -- (par-w);
		\draw[->, thick] (par-v) -- 
			node[xshift=-3mm, smalllabel, inner sep=1mm] {\footnotesize$\alpha{\parOp}\beta$}
			node[fill=white, pos=.7, inner sep=.5mm] {$\recvar$} (par-w);
			
		\coordinate (impl-start) at ([xshift=5mm]seq-p);
		\coordinate (impl-end) at ([xshift=-5mm] par-v);

		\node at ($(impl-start)!.5!(impl-end)$) {$\underset{\text{\rref{lem:program_coincidence}}}{\overset{\text{by coincidence}}{\Longrightarrow}}$};
	\end{tikzpicture}
	\caption{
		Safety of parallel composition (\rref{fig:opsem}) by parallel injection \RuleName{acDropComp},
		where \usebox{\equalityRela} denotes coincidence of states (\rref{lem:expr_coincidence}) on the names indicated by the label.
		Semantically, parallel injection corresponds to adding runs of $\beta$ ending in the state \usebox{\droppedState}.
	}
	\label{fig:parInj}
\end{figure}

\paragraph{Parallel Sequentialization}

Non-communicating discrete programs $\alpha, \beta$,
admit sequentialization,
\iest $\langle \alpha \rangle \langle \beta \rangle \psi \rightarrow \langle \alpha\parOp\beta \rangle \psi$ is sound,
because $\alpha$ and $\beta$ 
do not share state (\rref{def:syntax_chps}).
Axiom \RuleName{acLivePar} generalizes this sequentialization principle to arbitrary CHPs (\rref{fig:livePar}).
This requires existence of runs for the subprograms~$\alpha$ and~$\beta$ 
which synchronize,
\iest the runs agree on the communication on shared channels and the duration of continuous evolution.
The axiom~\RuleName{acLivePar} shows such runs exist
by proving that there are runs of the subprograms 
that communicate their projection $\auxhis\downarrow\gamma$ of a global trace $\auxhis$ for $\alpha\parOp\beta$.%
\footnote{
	Intuitively, $\langle \alpha\parOp\beta \rangle$ asks for a strategy to resolve the choices in~$\alpha$ and~$\beta$ such that the subprograms synchronize.
	The history $\auxhis$ is a witness for this strategy,
	identifying value and time of joint communications.
}
This also ensures that the subprograms execute  simultaneously in time,
because the projection~$\downarrow\gamma$ synchronizes the channel $\durchan$ containing durations of continuous evolutions.
Analogous to \RuleName{acLivePar},
axiom~\RuleName{acLiveParCommit} checks for synchronization
via the history~$\auxhis$, 
but does not need to consider state change,
as the commitment only specifies the observable communication.

The diamond $\runpost{\alpha}{\re_0}{\te}{\re_1} \psi$
with trace constraint $\te$
states there is a run of $\alpha$ 
that has communication trace~$\te$ and reaches a final state satisfying~$\psi$,
and $\runcommit{\alpha}{\te}$
states there is a run of~$\alpha$ communicating~$\te$,
but does not require to reach a final state.
By commutativity of parallelism,
$\alpha$ and~$\beta$ can be swapped in the premise of~\RuleName{acLivePar}.

\begin{figure}[t]
	\centering
	\begin{tikzpicture}[node distance=.8cm]
		\tikzstyle{state}+=[circle, fill=blue!5, draw=blue, inner sep=0, minimum width=1.5em];
		\tikzstyle{proglabel}+=[fill=Gray!20, rounded corners=2mm, minimum height=1.6em, align=center];
		\tikzstyle{smalllabel}+=[fill=Gray!20, rounded corners=2mm, align=center];
		\tikzstyle{post}+=[circle, fill=green!20, inner sep=.15mm, outer xsep=1mm]

		\node[state] (seq-v) {$\pstate{v}$};

		\node[state, circle split] [right=1cm of seq-v] (seq-w-a) {};
		\node[yshift=3pt] at (seq-w-a.center) {\footnotesize$\pstate{w}_\alpha$};
		\node[yshift=-3.5pt] at(seq-w-a.center) {\footnotesize$\pstate{v}$};

		\node[state, circle split] [right=1cm of seq-w-a] (seq-w-b) {};
		\node[yshift=3pt] at (seq-w-b.center) {\footnotesize$\pstate{w}_\alpha$};
		\node[yshift=-3.5pt] at(seq-w-b.center) {\footnotesize$\pstate{w}_\beta$};
		
		\node[post] [right=0 of seq-w-b] (seq-p) {$\psi$};

		\draw[->] (seq-v) 
			--
			node[smalllabel, inner sep=1mm] {\footnotesize$\alpha\vphantom{\beta}$}
			node[above=1.5mm] {$\auxhis \downarrow \alpha$} (seq-w-a);
		\draw[->] (seq-w-a) 
			--
			node[smalllabel, inner sep=1mm] {\footnotesize$\beta$}
			node[above=1.5mm] {$\auxhis \downarrow \beta$} (seq-w-b);

		\draw[->, thick] (seq-v) 
			edge[bend right=30]
			node[pos=.3, smalllabel, inner sep=1mm] {\footnotesize$\alpha{\parOp}\beta$}
			node[fill=white, pos=.7, inner sep= 0] {$\auxhis$} (seq-w-b);

		\coordinate [above=.45cm of seq-v] (bep-1) {};
		\coordinate [above=.45cm of seq-w-a] (bep-2) {};

		\draw[stateeq] (seq-v) -- node[left] {\shortstack[l]{
			\footnotesize$\SFV(\beta)\subseteq\SBV(\alpha)^\complement\cup\gvars$\\ 
			\footnotesize by well-formedness\\ 
			\footnotesize (\rref{def:syntax_chps})
		}} (bep-1) -- (bep-2) -- (seq-w-a);

		\node[state] [right=2.7cm of seq-w-b] (par-v) {$\pstate{v}$};

		\node[state, circle split, font=\footnotesize] [above right=.4cm and 2cm of par-v] (par-w-a) {};
		\node[yshift=3pt] at (par-w-a.center) {\footnotesize$\pstate{w}_\alpha$};
		\node[yshift=-3.5pt] at(par-w-a.center) {\footnotesize$\pstate{v}$};

		\node[state, circle split, font=\footnotesize] [below right=.4cm and 2cm of par-v] (par-w-b) {};
		\node[yshift=3pt] at (par-w-b.center) {\footnotesize$\pstate{v}$};
		\node[yshift=-3.5pt] at(par-w-b.center) {\footnotesize$\pstate{w}_\beta$};

		\node[state, circle split, font=\footnotesize] [inner sep = 0] at ($(par-w-a.center)!.5!(par-w-b.center)$) (par-w) {};
		\node[yshift=3pt] at (par-w.center) {\footnotesize$\pstate{w}_\alpha$};
		\node[yshift=-3.5pt] at(par-w.center) {\footnotesize$\pstate{w}_\beta$};

		\node[post] [right=0 of par-w] {$\psi$};
		
		\draw[->] (par-v) 
			edge [bend left=20] 
			node[smalllabel, inner sep=1mm, pos=.3] {\footnotesize$\alpha\vphantom{\beta}$}
			node[fill=white, pos=.7, inner sep=.5mm] {$\recvar \downarrow \alpha$} (par-w-a);
		\draw[->] (par-v) 
			edge [bend right=20] 
			node[smalllabel, inner sep=1mm, pos=.3] {\footnotesize$\beta$}
			node[fill=white, pos=.7, inner sep=.5mm] {$\recvar \downarrow \beta$} (par-w-b);

		\draw[->] (par-w-a) -- (par-w);
		\draw[->] (par-w-b) -- (par-w);
		\draw[->, thick] (par-v) -- 
			node[xshift=-3mm, smalllabel, inner sep=1mm] {\footnotesize$\alpha{\parOp}\beta$}
			node[fill=white, pos=.7, inner sep=.5mm] {$\recvar$} (par-w);
			
		\coordinate (impl-start) at ([xshift=5mm]seq-p);
		\coordinate (impl-end) at ([xshift=-5mm] par-v);

		\node at ($(impl-start)!.5!(impl-end)$) {$\underset{\text{\rref{lem:program_coincidence}}}{\overset{\text{by coincidence}}{\Longrightarrow}}$};

	\end{tikzpicture}
	\caption{
		Liveness of parallel composition (\rref{fig:opsem}) by sequentialization \RuleName{acLivePar},
		where \usebox{\equalityRela} denotes coincidence of states (\rref{lem:expr_coincidence}) on the variables indicated by the label.
	}
	\label{fig:livePar}
\end{figure}

\paragraph{Hybrid Programs}

Axioms \RuleName{assign}, \RuleName{nondetAssign}, and \RuleName{test} are as in \dL.
For continuous evolution $\dbleft \evolution*{}{non} \dbright \psi$,
\dLCHP inherits \dL's complete axiomatization of differential equation properties \cite{DBLP:journals/jacm/PlatzerT20}.
Since \dL does not consider the effect of the durations of continuous evolution on the history~$\recvar$,
the \dL axioms are applicable only if $\recvar\notin\SFV(\psi)$,
ensuring their conservative interplay with \dLCHP.
Otherwise, the duration axiom~\RuleName{odeSync} introduces a clock $t' = 1$
to uncover the local effect of the global synchronization in time
as a communication along the reserved channel $\durchan$.
Axiom~\RuleName{acNoCom} expands the ac-contract $(\A, \Commit)$ for noncommunicating programs.
Since the ac-contract $(\A, \Commit)$ is an invariant of~$\alpha$'s communication history,
the unfolding $\Commit \wedge (\A \rightarrow [ \alpha ] \psi)$ by~\RuleName{acNoCom}
corresponds to the base case with empty history.

Ac-composition \RuleName{acComposition}, 
ac-choice \RuleName{acChoice},
and ac-iteration~\RuleName{acIteration} 
are straightforward generalizations from dynamic logic.%
\footnote{
	The prefix-closed program semantics enables reasoning about nonterminating reactive systems \cite{AcHoare_Zwiers},
	and further renders the axiom \RuleName{acComposition} an equivalence
	since proving the commitment of $[ \alpha ] \ac [ \beta ] \ac \psi$ from $[ \alpha \seq \beta ] \ac \psi$ needs that all $\alpha$-prefixes are in the semantics of $\alpha\seq\beta$.
}
The base case $[ \alpha^0 ] \ac \psi$ in \RuleName{acIteration},
where $\alpha^0 \equiv \test{\true}$ yields no communication,
is provably equivalent to $\Commit \wedge (\A \rightarrow \psi)$ by \RuleName{acNoCom} and~\RuleName{test}.
Ac-induction~\RuleName{acInduction} carefully generalizes the induction axiom of dynamic logic, 
considering that assumption-program pairs are the modal actions.
Consequently, the induction step $\psi \rightarrow [ \alpha ] \ac \psi$  needs a proof in all worlds reachable by $(\A, \repetition{\alpha})$-runs,
proving both promises $\Commit$ and $\psi$ inductively.
As a result, the ac-induction rule \RuleName{acInvariant} derives as usual using Gödel generalization~\RuleName{acG}.
The required initial commitment $\Commit$ in \RuleName{acInvariant} reflects the base case when proving the ac-contract $(\A, \Commit)$ inductively.
Conversely, the environment guarantees the assumption $\A$ in the final state, 
even after zero iterations~$\alpha^0$.
This assumption can be obtained 
by a combination of the axioms \RuleName{Atransfer} and \RuleName{hExtension}.
The rule~\RuleName{acLoop} is a standard loop induction rule,
which derives from \RuleName{acInvariant}.

Axiom \RuleName{acConvergence} lifts \dL's hybrid version \cite{DBLP:journals/jar/Platzer08, DBLP:conf/lics/Platzer12b} of Harel's convergence axiom \cite{Harel1979} %
to assumption-repetition pairs as modal action.
It proves existence of a run to a final state as the commitment is unsatisfiable ($\false$).
Ac-arrival \RuleName{acArrival} is the ac-version of the arrival axiom~\cite{DBLP:journals/tocl/Platzer15} and the derivable dual of ac-induction \RuleName{acInduction}.
Using~\RuleName{acArrival},
convergence also covers
$\langle \repetition{\alpha} \rangle \ac \false$
by proving either $\langle \alpha^0 \rangle \ac \false$ or $\langle \alpha \rangle \ac \psi$ after some number of iterations.
Dually to the rule~\RuleName{acInvariant},
where the commitment $\Commit$ must be proven in the base case $[ \alpha^0 ] \ac \psi$ while the assumption $\A$ is given,~$\A$ must be proven in
$\langle \alpha^0 \rangle \ac \psi$
when $\Commit$ is not satisfied,
as $\langle \alpha^0 \rangle \ac \psi \leftrightarrow \Commit \vee \A \wedge \psi$.
This explains the premise $\A$ in~\RuleName{acConvergence}.

\paragraph{Communication}

The send axiom \RuleName{send} appends the communication to the recorder $\recvar$ and distinguishes the new world by the fresh trace variable~$\tvar_0$.
Ac-unfolding \RuleName{acCom} expands the invariant of the communication history represented by the ac-contract $(\A, \Commit)$ into 
the base case $[ \test{\true} ] \ac$ 
before and after the communication event emitted by $\send{}{}$.
Receiving $\receive{}{}$ obtains some value and binds it to the variable $x$.
The receive axiom \RuleName{comDual} equates this with testing whether the environment can agree on a nondeterministically chosen value for $x$ by sending.
This is analogous to axiom \RuleName{odeSync},
which turns the effect of the duration of continuous evolution on the history $\recvar$
into sending.

\begin{remark}
	\label{rem:recorder_renaming}
	The substitution of $\tvar_0$ for $\recvar$ in axiom \RuleName{send}
	needs recorder renaming (\rref{sec:substitution}) if the postcondition contains communicating statements,
	\eg if $\psi(\recvar) \equiv [ \send{\ch{dh}}{} ] \varphi(\recvar)$,
	then $\psi(\tvar_0) \equiv [ \recchp{\tvar_0}{\send{\ch{dh}}{}} ] \varphi(\tvar_0)$.
	Instead of unfolding the box $[ \recchp{\tvar}{\send{}{}} ]$ with recorder $\tvar$ to continue reasoning,
	the following variant of \RuleName{send} derives by first-order reasoning:
	\begin{equation*}
		[ \recchp{\tvar}{\send{}{}} ] \psi(\tvar) \leftrightarrow \fa{\tvar_0} \big( \tvar_0 = \tvar \cdot \comItem{\ch{}, \rp}  \rightarrow \psi(\tvar_0) \big)
	\end{equation*}
	In fact, every axiom in \rref{fig:calculus} is compatible with recorder renaming.
	That is, for every axiom and fresh $\tvar\in\TVar$, 
	the formula resulting from replacing every $\dbleft \alpha \dbright \psi$
	with $\dbleft \recchp{\tvar}{\alpha} \dbright \psi\subs{\recvar}{\tvar}$
	and every free occurrence of $\recvar$ with $\tvar$ derives in the \dLCHP calculus by first-order reasoning.
	These derived axioms allow to view $[ \recchp{\tvar}{\alpha} ] \psi$ as an atomic top-level formula,
	which is useful,
	because 
	it enables to execute CHPs symbolically \cite{King1976}
	as usual in dynamic logic.
\end{remark}

\paragraph{Auxiliary Axioms}

History invariance \RuleName{hExtension} and assumption transfer \RuleName{Atransfer} 
internalize properties of the computational domain \cite{AcSemantics_Zwiers} and of the semantics of ac-modalities \cite{Pandya1991}, respectively.
For example,
the order that the assumption fixes for communication on channels that are not shared between 
parallel subprograms
is \emph{not} the sum of local (\rref{def:noninterference}) properties of the subprograms
but guaranteed by the environment.
Axiom \RuleName{Atransfer} internalizes the global restriction of the reachable states by the assumption,
generalizing the assumption closure rule in Hoare-style ac-reasoning \cite{AcSemantics_Zwiers}
to support dual reasoning for the ac-diamond.
The notation~$\Abutlast{\A}$ and $\Aglobally{\A}$ borrowed from temporal logic reminds that~$\A$ holds \emph{for all} (strict) prefixes of~$\alpha$'s communication trace.
History invariance~\RuleName{hExtension} 
proves that all programs strictly extend the previous history.%
\footnote{%
	\label{ft:global_his_props}%
	In fact, history invariance~\RuleName{hExtension} 
	is logically independent of the other axioms \cite{AcSemantics_Zwiers}.
	The axiom excludes unintentional computational models,
	which admit interleaving of the communication of a parallel subprogram with the previous communication of other subprograms
	on channels that are not shared.
}

\begin{figure}[t]
	\begin{small}
		\begin{minipage}{\textwidth}
			\begin{calculus}
				\startAxiom{diasDual}
					$\axkey{\langle \alpha \rangle \psi} \leftrightarrow \langle \alpha \rangle\acpair{\true, \false} \psi$
				\stopAxiom
				\startAxiom{acDiaNoCom}
					$\axkey{\langle \alpha \rangle \ac \psi} \leftrightarrow \Commit \vee \A \wedge \langle \alpha \rangle \psi$
				\stopAxiom
			\end{calculus}
			\hspace*{.1cm}
			\begin{calculus}
				\startAxiom{acBoxesDist}
					$\axkey{[ \alpha ] \acpair{\A, \Commit_1 \wedge \Commit_2} (\psi_1 \wedge \psi_2)}
					\leftrightarrow
					[ \alpha ] \acpair{\A, \Commit_1} \psi_1 \wedge [ \alpha ] \acpair{\A, \Commit_2} \psi_2$
				\stopAxiom
				\startAxiom{acSplitDia}
					$\axkey{\langle \alpha \rangle \ac \psi} \leftrightarrow \langle \alpha \rangle \ac \false \vee \langle \alpha \rangle \acpair{\A, \false} \psi$
				\stopAxiom
			\end{calculus}
			\vspace*{.2cm}

			\begin{calculus}
				\startAxiom{acArrival}
					$\axkey{\langle \repetition{\alpha} \rangle \ac \psi}
					\leftrightarrow
					\langle \alpha^0 \rangle \ac \psi \vee \langle \repetition{\alpha} \rangle \acpair{\A, \false} (\neg\psi \wedge \langle \alpha \rangle \ac \psi)$
				\stopAxiom
			\end{calculus}
			\vspace*{.2cm}

			\begin{calculus}[r]
				\startRule{acMono}
					\Axiom{$\A_2 \rightarrow \A_1$}
					\Axiom{$\Commit_1 \rightarrow \Commit_2$}
					\Axiom{$\psi_1 \rightarrow \psi_2$}
					\TrinaryInf{$[ \alpha ] \acpair{\A_1, \Commit_1} \psi_1 \rightarrow [ \alpha ] \acpair{\A_2, \Commit_2} \psi_2$}
				\stopRule
			\end{calculus}
			\begin{calculus}[r]
				\startRule{acDiaMono}
					\Axiom{$\A_1 \rightarrow \A_2$}
					\Axiom{$\Commit_1 \rightarrow \Commit_2$}
					\Axiom{$\psi_1 \rightarrow \psi_2$}
					\TrinaryInf{$\langle \alpha \rangle \acpair{\A_1, \Commit_1} \psi_1 \rightarrow \langle \alpha \rangle \acpair{\A_2, \Commit_2} \psi_2$}
				\stopRule
			\end{calculus}
			\vspace*{.2cm}

			\begin{calculus}
				\startRule{acInvariant}
					\Axiom{$\psi \rightarrow [ \alpha ] \ac \psi$}
					\UnaryInf{$\Commit \wedge \psi \rightarrow [ \repetition{\alpha} ] \ac \psi$}
				\stopRule
			\end{calculus}
			\begin{calculus}
			\startRule{acLoop}
				\Axiom{$\varphi \rightarrow \Commit \wedge I$}
				\Axiom{$I \rightarrow [ \alpha ] \ac I$}
				\Axiom{$I \rightarrow \psi$}
				\TrinaryInf{$\varphi \rightarrow [ \repetition{\alpha} ] \ac \psi$}
			\stopRule
			\end{calculus}
		\end{minipage}
	\end{small}
	\caption{
		Derived axioms and proof rules of \dLCHP
	}
	\label{fig:derived}
\end{figure}

\tikzstyle{axname}=[
	align=center, outer sep=.5mm, inner sep=.8mm, rounded corners=1mm, font=\small
]
\tikzstyle{oldax}=[
	draw, dotted, semithick%
]
\tikzstyle{newax}=[draw]
\tikzstyle{derived}=[fill=black!10]

\newsavebox\dLCHPaxiom
\newsavebox\derivedAxiom
\newsavebox\knownFromAcHoare

\sbox\dLCHPaxiom{\tikz[baseline=-0.5ex]{\node[axname, newax] {\RuleStyle{Ax}}}}
\sbox\derivedAxiom{\tikz[baseline=-0.5ex]{\node[axname, derived] {\RuleStyle{Ax}}}}
\sbox\derivedAxiom{\tikz[baseline=-0.5ex]{\node[axname, derived] {\RuleStyle{Ax}}}}
\sbox\knownFromAcHoare{\tikz[baseline=-0.5ex]{\node[axname, derived, oldax] {\RuleStyle{Ax}}}}

\paragraph{Modal Logic Principles} 

Axiom \RuleName{acModalMP} is the ac-version of modal modus ponens covering monotonicity of \emph{both} promises.
Assumptions are antitone because under a weaker assumption more worlds are reachable, 
and assumption weakening~\RuleName{Aweak} further supports weakening by the commitment, 
since it is already guaranteed.
The latter enables the mutual abstraction of parallel programs,
which is the core principle of ac-reasoning for state space reduction,
by proving the assumption of a subprogram from the commitment of the other subprograms.
By antitonicity, axiom \RuleName{Aweak} can also be understood as a refinement rule~\cite{DBLP:conf/lics/LoosP16} for the environment.
The ac-version~\RuleName{acG} of the Gödel rule proves an ac-box if \emph{both} promises hold in all states.
The ac-version \RuleName{acdbDual} of modal duality~\RuleName{dbDual} again affects \emph{both} promises.
Axiom~\RuleName{boxesDual} 
embeds the dynamic modalities into ac-reasoning.

Thanks to the careful embedding of ac-reasoning into dynamic logic in \dLCHP,
other principles of modal logic derive (\rref{fig:derived}) by standard arguments:
Ac-monotonicity~\RuleName{acMono} combines \RuleName{acModalMP} for monotonicity of the promises and \RuleName{Aweak} for antitonicity of assumptions,
and drops the box by \RuleName{acG}.
Ac-distribution \RuleName{acBoxesDist} derives from \RuleName{acModalMP}.
The disjunctive relation of commitment and postcondition in $\langle \alpha \rangle \ac \psi$ is most apparent in the derivable axiom \RuleName{acSplitDia}.
An axiom for weakening the assumptions of parallel programs by their mutual commitments \cite{Brieger2023} derives from \RuleName{Aweak}.
Ac-monotonicity~\RuleName{acDiaMono} for ac-diamonds derives.
Assumptions become monotone in~\RuleName{acDiaMono} 
just like refinements reverse when switching from safety to liveness~\cite{DBLP:conf/lics/LoosP16}.

\subsection{Examples}
\label{sec:examples}

To illustrate \dLCHP's proof calculus in action,
we revisit the convoy of cars example (\rref{ex:follower_leader}).
The focus is on how the \dLCHP axioms are meant to be combined
to decompose the dynamics of parallel hybrid systems
into simpler subproblems.

The example proofs are presented in sequent-style.
A sequent $\Gamma \vdash \Delta$ abbreviates the formula $\bigwedge_{\varphi\in\Gamma} \varphi \rightarrow \bigvee_{\psi\in\Delta} \psi$.
We use the following standard proof rules of sequent calculus \cite{Platzer18},
which can be easily derived from the Hilbert-calculus of \dLCHP.
\vspace{1em}

\begin{small}
	\begin{calculus}
		\startRule{andR}
				\Axiom{$\Gamma \vdash \varphi, \Delta$}

				\Axiom{$\Gamma \vdash \psi, \Delta$}

			\BinaryInf{$\Gamma \vdash \varphi \wedge \psi, \Delta$}
		\stopRule
	\end{calculus}
	\begin{calculus}
		\startRule{implR}	
			\Axiom{$\Gamma, \varphi \vdash \psi, \Delta$}

			\UnaryInf{$\Gamma \vdash \varphi \rightarrow \psi, \Delta$}
		\stopRule
	\end{calculus}
	\begin{calculus}
		\startRule{conditional}	
			\Axiom{$\Gamma, \chi \vdash [ \alpha ] \psi, \Delta$}

			\Axiom{$\Gamma, \neg\chi \vdash [ \beta ] \psi, \Delta$}

			\BinaryInf{$\Gamma \vdash [ \ifelse{\chi}{\alpha}{\beta} ] \psi, \Delta$}
		\stopRule
	\end{calculus}

	\begin{calculus}
		\startRule{iG}	
			\Axiom{$\Gamma, \avar_0 = \expr \vdash \Delta$}

			\RightLabel{($\avar_0$ fresh)}
			\UnaryInf{$\Gamma \vdash \Delta$}
		\stopRule
	\end{calculus}
	\begin{calculus}
		\startRule{subsR}	
			\Axiom{$\Gamma, \avar = \expr \vdash \varphi(\expr), \Delta$}

			\UnaryInf{$\Gamma, \avar = \expr \vdash \varphi(\avar), \Delta$}
		\stopRule
	\end{calculus}
	\begin{calculus}
		\startRule{forallR}	
			\Axiom{$\Gamma \vdash \psi \subs{\avar}{\avar_0}, \Delta$}

			\RightLabel{($\avar_0$ fresh)}
			\UnaryInf{$\Gamma \vdash \fa{\avar} \psi, \Delta$}
		\stopRule
	\end{calculus}
\end{small}

\CreateRuleRef{acReceive}
\CreateRuleRef{composition}
\CreateRuleRef{solution}
\CreateRuleRef{forallR}
\CreateRuleRef{andR}
\CreateRuleRef{andL}
\CreateRuleRef{implR}

\newcommand{\hisupdate}{\Lambda}

\begin{example}
	\label{ex:receive}
	The \dLCHP calculus is modular.
	Complex reasoning decomposes into combinations of simple axioms,
	each capturing one specific aspect of parallel hybrid systems.
	In \rref{fig:deriveAcReceiveRight},
	the rule \RuleName{acReceiveRight} for executing a receive statement and unfolding the temporal structure of its ac-contract is derived.
	The proof combines axiom \RuleName{comDual},
	which reduces receiving to sending,
	axiom \RuleName{acCom},
	which unfolds the ac-contract around the send statement, 
	and axiom~\RuleName{send}, 
	which executes sending.
	The advantage of this modular approach is that the semantic soundness argument for each individual axiom remains simple,
	while the soundness of complex rules can be proven within the calculus. 
	In the following, $\hisupdate$ abbreviates $\tvar_0 = \recvar \cdot \comItem{\ch{}, x}$,
	where $\tvar_0$ is a fresh variable,
	and $x$ is assumed not to occur free in~$\Gamma, \Delta$.

	\CreateRuleRef{acSendRight}
	\CreateRuleRef{acReceiveRight}

	\begin{small}
		\begin{prooftree}[HypSeparation=.1em]
			\Axiom{$\Gamma \vdash \Commit(\recvar), \Delta$}

			\Axiom{$\Gamma, \A(\recvar), \hisupdate \vdash \Commit(\tvar_0), \Delta$}

			\Axiom{$\Gamma, \A(\recvar), \hisupdate, \A(\tvar_0) \vdash \psi(\tvar_0), \Delta$}
			
			\RuleNameLeft{acReceiveRight}{}
			\TrinaryInf{$\Gamma \vdash [ \receive{}{x} ] \acpair{\A(\recvar), \Commit(\recvar)} \psi(\recvar), \Delta$}
		\end{prooftree}
	\end{small}	
	An analogous rule~\RuleName{acSendRight} for sending can be obtained from \RuleName{acReceiveRight}
	by replacing $\receive{}{}$ with $\send{}{}$
	and~$\hisupdate$ with~$\hisupdate\subs{x}{\rp}$,
	and a derivation can be obtained by truncating the proof in \rref{fig:deriveAcReceiveRight}.
\end{example}
\vspace{-.5cm}

\newsavebox{\ptree}
\begin{lrbox}{\ptree}
	\begin{minipage}{.8\textwidth}
		\small
		\begin{prooftree}[shape=justified, JustifiedScoresWidth=.87\textwidth]
			\Axiom{$\Gamma, \A(\recvar), \hisupdate \vdash \Commit(\tvar_0), \Delta \qquad\qquad \Gamma, \A(\recvar), \hisupdate, \A(\tvar_0)\vdash \psi(\tvar_0), \Delta$}

			\RuleNameRight{andR, implR}{}
			\UnaryInf{$\Gamma, \A(\recvar), \tvar_0 = \recvar \cdot \comItem{\ch{}, x} \vdash \Commit(\tvar_0) \wedge (\A(\tvar_0) \rightarrow \psi(\tvar_0) ), \Delta$}

			\RuleNameRight{forallR, implR}{}
			\UnaryInf{$\Gamma, \A(\recvar) \vdash \fa{\tvar_0} \big( \tvar_0 = \recvar \cdot \comItem{\ch{}, x} \rightarrow \Commit(\tvar_0) \wedge (\A(\tvar_0) \rightarrow \psi(\tvar_0) ) \big), \Delta$}

			\RuleNameRight{send}{}
			\UnaryInf{$\Gamma, \A(\recvar) \vdash [ \send{}{x} ] \big( \Commit(\recvar) \wedge (\A(\recvar) \rightarrow \psi(\recvar) ) \big), \Delta$}

			\RuleNameRight{implR}{}
			\UnaryInf{$\Gamma \vdash \A(\recvar) \rightarrow [ \send{}{x} ] \big( \Commit(\recvar) \wedge (\A(\recvar) \rightarrow \psi(\recvar) ) \big), \Delta$}
		\end{prooftree}
	\end{minipage}
\end{lrbox}

\begin{figure}
	\small
	\begin{prooftree}[shape=justified]
		\Axiom{\raisebox{-1.08cm}{$\Gamma \vdash \Commit(\recvar), \Delta$}\qquad\hbox{\usebox{\ptree}}}

		\RuleNameRight{andR}{}
		\UnaryInf{$\Gamma \vdash \Commit(\recvar) \wedge \big( \A(\recvar) \rightarrow [ \send{}{x} ] \big( \Commit(\recvar) \wedge (\A(\recvar) \rightarrow \psi(\recvar) ) \big) \big) , \Delta$}

		\RuleNameRight{test}{}
		\UnaryInf{$\Gamma \vdash \Commit(\recvar) \wedge \big( \A(\recvar) \rightarrow [\test{\true}] [ \send{}{x} ] \big( \Commit(\recvar) \wedge (\A(\recvar) \rightarrow [\test{\true}] \psi(\recvar) ) \big) \big), \Delta$}

		\RuleNameRight{acNoCom}{}
		\UnaryInf{$\Gamma \vdash [ \test{\true} ] \acpair{\A(\recvar), \Commit(\recvar)} [ \send{}{x} ] [ \test{\true} ] \acpair{\A(\recvar), \Commit(\recvar)} \psi(\recvar), \Delta$}

		\RuleNameRight{acCom}{}
		\UnaryInf{$\Gamma \vdash  [ \send{}{x} ] \acpair{\A(\recvar), \Commit(\recvar)} \psi(\recvar), \Delta$}

		\RuleNameRight{nondetAssign, forallR}{}
		\UnaryInf{$\Gamma \vdash [ x \ceq * ] [ \send{}{x} ] \acpair{\A(\recvar), \Commit(\recvar)} \psi(\recvar), \Delta$}

		\RuleNameRight{comDual}{}
		\UnaryInf{$\Gamma \vdash [ \receive{}{} ] \acpair{\A(\recvar), \Commit(\recvar)} \psi(\recvar), \Delta$}
	\end{prooftree}
	\caption{
		Derivation of the rule \RuleName{acReceiveRight}, where $\hisupdate \equiv \tvar_0 = \recvar \cdot \comItem{\ch{}, x}$	
	}
	\label{fig:deriveAcReceiveRight}
\end{figure}

\begin{example}
	\label{ex:acSolution}
	In parallel hybrid systems,
	the duration of continuous dynamics synchronizes across parallel programs,
	which is modeled as a communication of the duration of each continuous evolution along the reserved channel $\durchan$.
	This relation between the duration and synchronization
	becomes apparent in the derived rule \RuleName{acSolutionR},
	which replaces a continuous evolution with an assignment of the solution $y(t)$
	of the symbolic initial value problem $y'(t) = f(y), y(0) = x$ for any time $t\ge0$
	and communicates the time $t$ along $\durchan$.
	The derivation (\rref{fig:acSolutionR}) of rule~\RuleName{acSolutionR}
	combines \dLCHP's duration axiom \RuleName{odeSync},
	which uncovers the time synchronization,
	with \dL's solution axiom~\RuleName{solution}~\cite{Platzer18}, 
	which solves the differential equation.
	\begin{prooftree}
		\Axiom{$\Gamma, t \ge 0, \fa{0{\le}s{\le}t} q(y(t)) \vdash [ x \ceq y(t) ] [ \send{\durchan}{t} ] \ac \psi, \Delta$}

		\RuleNameLeft{acSolutionR}{}
		\RightLabel{{\color{gray}($y'(t) = f(y)$)}}
		\UnaryInf{$\Gamma \vdash [ \evolution*{x' = f(x)}{q(x)} ] \ac \psi, \Delta$}
	\end{prooftree}
	\CreateRuleRef{acSolutionR}%
\end{example}

\begin{figure}
	\small
	\begin{prooftree}[shape=justified, JustifiedScoresWidth=.9\textwidth]
		\Axiom{$\Gamma, t \ge 0, \fa{0{\le}s{\le}t} q(y(t)) \vdash [ x \ceq y(t) ] [ \send{\durchan}{t} ] \ac \psi$}

		\RuleNameRight{composition, assign}{}
		\UnaryInf{$\Gamma, t \ge 0, \fa{0{\le}s{\le}t} q(y(t)) \vdash [ x \ceq y(t) \seq g \ceq t ] [ \send{\durchan}{g} ] \ac \psi, \Delta$}
		
		\RuleNameRight{subsR}{}
		\UnaryInf{$\Gamma, \tvar_0 = \recvar, g = 0, t \ge 0, \fa{0{\le}s{\le}t} q(y(t)) \vdash [ x \ceq y(t) \seq g \ceq g + t ] \fa{\recvar{=}\tvar_0} [ \send{\durchan}{g} ] \ac \psi, \Delta$}

		\RuleNameRight{forallR, implR}{}
		\UnaryInf{$\ldots \vdash \fa{t{\ge}0} \big( \fa{0{\le}s{\le}t} q(y(t)) \rightarrow [ x \ceq y(t) \seq g \ceq g + t ] \fa{\recvar{=}\tvar_0} [ \send{\durchan}{g} ] \ac \psi \big), \Delta$}

		\RuleNameRight{solution}{}
		\UnaryInf{$\Gamma, \tvar_0 = \recvar, g = 0 \vdash [ \evolution*{x' = f(x), g' = 1}{q(x)} ] \fa{\recvar{=}\tvar_0} [ \send{\durchan}{g} ] \ac \psi, \Delta$}

		\RuleNameRight{forallR, implR}{}
		\UnaryInf{$\Gamma \vdash \fa{\tvar_0{=}\recvar} \fa{g{=}0} [ \evolution*{x' = f(x), g' = 1}{q(x)} ] \fa{\recvar{=}\tvar_0} [ \send{\durchan}{g} ] \ac \psi, \Delta$}

		\RuleNameRight{odeSync}{}
		\UnaryInf{$\Gamma \vdash [ \evolution*{x' = f(x)}{q(x)} ] \ac \psi, \Delta$}
	\end{prooftree}
	\caption{Derivation of the rule \RuleName{acSolutionR}}
	\label{fig:acSolutionR}
\end{figure}

\rref{ex:derivation} decomposes the safety contract for the convoy of cars (\rref{ex:convoy_safety}) in the calculus.
This proof follows an idiomatic pattern for decomposing a safety contract \mbox{$[ \alpha\parOp\beta ] \psi$} or \mbox{$[ \alpha\parOp\beta ] \ac \psi$} about a parallel CHP into safety contracts for the subprograms,
where the box specifies closed systems without further environment
and the ac-box occurs as specification 
of nested parallel compositions:
\begin{enumerate}
	\item Introduce specifications $\Commit_\alpha$, $\Commit_\beta$, $\psi_\alpha$, $\psi_\beta$ for the subprograms by \RuleName{boxesDual} and monotonicity \RuleName{acMono} 
	that are strong enough to entail $\Commit$ and $\psi$, 
	\iest $\psi_\alpha \wedge \psi_\beta \rightarrow \psi$ and $\Commit_\alpha \wedge \Commit_\beta \rightarrow \Commit$ derive,
	but independent (\rref{def:noninterference}) of the other subprogram.
	\label{itm:idiomatic_conditions}

	\item Strengthen the overall assumption $\A$ from the commitments of the other subprograms by axiom~\RuleName{Aweak}
	to obtain local assumptions $\A_\alpha$ and $\A_\beta$.
	\label{itm:weaken}
	
	\item Distribute the 
	specifications
	by axiom \RuleName{acBoxesDist},
	creating a subgoal $[ \alpha\parOp\beta ] \acpair{\A_\gamma, \Commit_\gamma} \psi_\gamma$ for each subprogram $\gamma$.
	
	\item For each subgoal,
	drop the subprogram not belonging to the 
	local specification
	by the parallel injection axiom~\RuleName{acDropComp}.
	This yields subgoals $[ \alpha ] \acpair{\A_\alpha, \Commit_\alpha} \psi_\alpha$
	and $[ \beta ] \acpair{\A_\beta, \Commit_\beta} \psi_\beta$
	for the subprograms that can be verified independently.
	\label{itm:inject}
\end{enumerate}

While this is a canonical use of the interplay 
of \RuleName{acMono}, \RuleName{Aweak}, \RuleName{acBoxesDist}, and \RuleName{acDropComp} to prove safety of parallel hybrid systems
their individual responsibilities increase modularity, simplify soundness arguments, and can be used in other combinations as well.
Unlike nonmodular calculi,
which internalize this strategy and its soundness proof monolithically,
the classical proof rule from Hoare-style ac-reasoning \cite{AcHoare_Zwiers} derives in \dLCHP 
without any further semantical soundness arguments by completeness.
Further, the axioms \RuleName{Atransfer} and \RuleName{hExtension} can be added in step~\ref{itm:idiomatic_conditions},
but this is only necessary when proving these global properties is required.

\begin{figure}[t]
	\begin{minipage}{.46\textwidth}
		\captionof{figure}{
			Specifications used in \rref{ex:derivation},
			where the selector $\val[\rp]{\tvar \downarrow \ch{}}$ with default $\rp$ is defined for every context formula $\phi$.
			Moreover, $\varphi$ is the precondition of the $\progtt{convoy}$,
			and~$\psi_f$ and $\psi_l$ are the local postconditions of the $\progtt{follower}$ and the $\progtt{leader}$, respectively.
			The $\progtt{follower}$ assumes~$\A$
			while the $\progtt{leader}$ guarantees the commitment $\Commit$.
		}
		\label{fig:postconditions}
	\end{minipage}\hspace{.45cm}
	\begin{minipage}{.48\textwidth}
		$\begin{aligned}
			& \phi(\val[\rp]{\tvar \downarrow \ch{}}) 
				\equiv
				\ifelse{\tvar\downarrow\ch{}=\tvar_0\downarrow\ch{}}{
					\\ &\qquad \phi(\rp) 
				}{
					\phi(\val{\tvar\downarrow\ch{}})
				} \\
			& \varphi \equiv
				\periodicity \ge 0 \wedge
				\waitvar = 0 \wedge
				0 {\le} v_f {\le} \safevelo{d} \\
				&\qquad \wedge
				v_f \le \maxvelo \wedge
				v_l \ge 0 \wedge
				x_f + d < x_l \\
			& \psi_f \equiv
				x_f < \val[x_0]{\recvar\downarrow\ch{pos}}\! \\
			& \psi_l \equiv
				\val[x_0]{\recvar\downarrow\ch{pos}} \le x_l \\
			& \fAssume \equiv \lCommit \equiv\;
				0 \le \val[0]{\recvar\downarrow\ch{vel}} \le \maxvelo 
		\end{aligned}$
	\end{minipage}
\end{figure}

\begin{example}
	\label{ex:derivation}
	The safety contract in \rref{ex:convoy_safety} about the $\progtt{convoy}$ of cars in \rref{ex:follower_leader}
	can be decomposed following the idiomatic strategy described above
	into contracts about the $\progtt{follower}$ and $\progtt{leader}$ car.
	Figure~\ref{fig:postconditions} contains 
	the specifications used.
	The selector $\val[\rp]{\tvar \downarrow \ch{}}$ defaults to~$\rp$ 
	if the current history $\tvar \downarrow \ch{}$ equals the initial history $\tvar_0 \downarrow \ch{}$,
	\iest the $\progtt{convoy}$ did not communicate along channel $\ch{}$ yet,
	and otherwise 
	returns the value of the last communication on the channel~$\ch{}$.
	The specification
	$\psi_f$ establishes that the $\progtt{follower}$ always stays behind the last known position $\val[x_0]{\ch{pos}}$ of the $\progtt{leader}$,
	while the $\progtt{leader}$ never falls behind this %
	position by~$\psi_l$,
	where $x_0$ is the initial position of the $\progtt{leader}$.

	The safety contract (\rref{ex:convoy_safety}) of the $\progtt{convoy} \equiv \progtt{follower} \parOp \progtt{leader}$ 
	derives as shown below,
	where $\triangleright_1 \equiv \Commit \rightarrow \true$ and $\triangleright_2 \equiv \psi_f \wedge \psi_l \rightarrow x_f < x_l$ derive by first-order reasoning and decidable real arithmetic.
	Further, $\Gamma \equiv \tvar_0 = \recvar, x_0 = x_l, \varphi$,
	where $\tvar_0, x_0$ are fresh variables. 

	\begin{small}
		\begin{prooftree}
				\Axiom{$*$}

				\UnaryInf{$\vdash (\Commit \wedge \true \rightarrow \A) \wedge \true$}

				\RuleNameLeft{acG}{}
				\UnaryInf{$\Gamma \vdash [ \progtt{convoy} ] \acpair{\true, \Commit \wedge \true \rightarrow \A} \true$}

					\Axiom{\rref{fig:follower_velo_d_safe}}

					\UnaryInf{$\Gamma \vdash [ \progtt{follower} ] \acpair{\A, \true} \psi_f$}

					\RuleNameLeft{acDropComp}{}
					\UnaryInf{$\Gamma \vdash [ \progtt{convoy} ] \acpair{\A, \true} \psi_f$}

					\iflongversion
						\Axiom{\rref{fig:leader}}
					\else
						\Axiom{see \cite[Fig. 23]{Brieger2024}}
					\fi

					\UnaryInf{$\Gamma \vdash [ \progtt{leader} ] \acpair{\true, \Commit} \psi_l$}

					\RuleNameRight{acDropComp}{}
					\UnaryInf{$\Gamma \vdash [ \progtt{convoy} ] \acpair{\true, \Commit} \psi_l$}

				\SetOption{HypSeparation}{.5em}

				\RuleNameRight{acBoxesDist}{}
				\BinaryInf{$\Gamma \vdash [ \progtt{convoy} ] \acpair{\A, \Commit} (\psi_f \wedge \psi_l)$}

			\RuleNameRight{andR}{}
			\BinaryInf{$\Gamma \vdash [ \progtt{convoy} ] \acpair{\true, \Commit \wedge \true \rightarrow \A} \true \wedge [ \progtt{convoy} ] \acpair{\A, \Commit} (\psi_f \wedge \psi_l)$}

			\RuleNameRight{Aweak}{}
			\UnaryInf{$\Gamma \vdash [ \progtt{convoy} ] \acpair{\true, \Commit} (\psi_f \wedge \psi_l)$}

			\RuleNameRight{acMono}{, $\triangleright_1$, $\triangleright_2$}
			\UnaryInf{$\Gamma \vdash [ \progtt{convoy} ] \acpair{\true, \true} x_f < x_l$}

			\RuleNameRight{iG}{}
			\UnaryInf{$\varphi \vdash [ \progtt{convoy} ] \acpair{\true, \true} x_f < x_l$}

			\RuleNameRight{implR, boxesDual}{}
			\UnaryInf{$\vdash \varphi \rightarrow [ \progtt{convoy} ] x_f < x_l$}
		\end{prooftree}
	\end{small}
\end{example}

\newcommand{\solutionDomain}[1]{
	\fa{0{\le}s{\le}t} #1
}

\begin{example}
	\label{ex:follower_proof}
	Figure~\ref{fig:follower_velo_d_safe} continues the reasoning for contract about the \progtt{follower} from \rref{ex:derivation},
	proving its local safety contract.
	In order to prove
	that the \progtt{follower} stays safely behind the \progtt{leader},
	the timing of control actions needs to be analyzed.
	Figure~\ref{fig:follower_invariant} defines the duration $\since{}{\tvar, \ch{}}$ of continuous evolution
	since the last communication along the channel~$\ch{}$ in the trace~$\tvar$,
	and $\till{}{\tvar, \ch{}}$ is the time till the next $\ch{}$-communication.
	The key property of the loop invariant~$F$
	is that the distance between the current position $x_f$ of the \progtt{follower}
	and the last measured position $\val[x_0]{\recvar\downarrow\ch{pos}}$ of the \progtt{leader}
	is at least $\till{}{\recvar, \ch{pos}} \safevelo{d}$,
	so that no collision happens till the \progtt{follower}'s controller \progtt{velo} gets the chance to slow down.

	In \rref{fig:follower_velo_d_safe}, 
	the rules \RuleName{prop} and \RuleName{fol} denote propositional and first-order reasoning, respectively.
	In combination, \RuleName{acMono}, \RuleName{boxesDual} drop the assumption $\A$
	(trivial premises of \RuleName{acMono} are omitted).
	Rule \RuleName{acSendRight} 
	renames the recorder $\recvar$ in $\Plant_f$ to $\tvar_0$ (\rref{sec:substitution}),
	and axiom~\RuleName{send} uses that all axioms are compatible with recorder renaming (\rref{rem:recorder_renaming}).
	Note that this enables a statement-by-statement symbolic execution.
	The proof closes ($*$) by decidable real arithmetic and evaluation of trace terms
	via axioms for trace simplification (\rref{sec:fod-completeness}).
\end{example}

\begin{figure}
	\small
	\begin{align*}
		& \phi(\since{}{\tvar, \ch{}}) \equiv 
		\ifelse{\ex{\tvar_0, d} d>0 \wedge \tvar\downarrow\{\ch{},\durchan\} = \tvar_0 \cdot \comItem{\durchan, d}}
			{\phi(\val{\tvar\downarrow\{\ch{},\durchan\}})}{\phi(0)} \\
		& \phi(\till{}{\tvar, \ch{}}) \equiv \phi(\periodicity - \since{}{\tvar, \ch{}}) \\
		& F \equiv 
			0 \le v_f \le \safevelo{d} \wedge v_f \le \maxvelo \wedge x_f + \till{}{\recvar, \ch{pos}} \safevelo{d} < \val[x_0]{\recvar\downarrow\ch{pos}}
			\wedge \waitvar = \since{}{\recvar, \ch{pos}} \le \periodicity
	\end{align*}
	\caption{
		Specifications used for \rref{ex:follower_proof}, 
		where $\phi(\cdot)$ is some context formula.
		The time $\since{}{\tvar,\ch{}}$ elapsed since the last communication along channel $\ch{}$
		stored in~$\tvar$
		is the value of the last event of $\tvar\downarrow \{\ch{},\durchan\}$
		whenever the duration since the last $\ch{}$-communication is positive, and $0$ otherwise.
	}
	\label{fig:follower_invariant}
\end{figure}

\begin{figure}
	\small
	\begin{prooftree}[shape=justified, JustifiedScoresWidth=.8\textwidth]
		\Axiom{$*$}
		
		\UnaryInf{$\ldots, \tvar_1 = \tvar_0 \cdot \comItem{\durchan, t} \vdash F(\tvar_1, \tarvelo, x_f + t \cdot \tarvelo, \waitvar + t)$}

		\RuleNameRight{forallR,implR}{}
		\UnaryInf{$\ldots \vdash \fa{\tvar_1} \big( \tvar_1 = \tvar_0 \cdot \comItem{\durchan, t} \rightarrow F(\tvar_1, \tarvelo, x_f + t \cdot \tarvelo, \waitvar + t) \big)$}

		\RuleNameRight{send}{}
		\UnaryInf{$\ldots \vdash [ \recchp{\tvar_0}{\send{\durchan}{t}} ] F(\tvar_0, \tarvelo, x_f + t \cdot \tarvelo, \waitvar + t)$}

		\RuleNameRight{composition, assign}{}
		\UnaryInf{$\begin{aligned}
			&F, \tvar_0 = \recvar \cdot \comItem{\ch{vel}, \tarvelo}, \A(\tvar_0), d > \periodicity \maxvelo, t \ge 0, \solutionDomain{\waitvar + s \le \periodicity} \\
			&\qquad\quad\quad \vdash [ x_f \ceq x_f + t \cdot \tarvelo \seq \waitvar \ceq \waitvar + t ] [ \recchp{\tvar_0}{\send{\durchan}{t}} ] F(\tvar_0, \tarvelo, x_f, \waitvar)
		\end{aligned}$}

		\RuleNameRight{acSolutionR}{}
		\UnaryInf{$F, \tvar_0 = \recvar \cdot \comItem{\ch{vel}, \tarvelo}, \A(\tvar_0), d > \periodicity \maxvelo \vdash [ \recchp{\tvar_0}{\Plant_f(\tarvelo)} ] F(\tvar_0, \tarvelo)$}

		\iflongversion
			\SideAx{$\triangleright$ \rref{fig:follower_velo_d_unsafe}}
		\else
			\SideAx{$\triangleright$ see \cite[Fig. 22a]{Brieger2024}}
		\fi

		\RuleNameRight{assign}{}
		\UnaryInf{$\ldots, d > \periodicity \maxvelo \vdash [ v_f \ceq \tarvelo ] [ \recchp{\tvar_0}{\Plant_f(v_f)} ] F(\tvar_0, v_f)$}

		\SideAx{$\triangleright$ \RuleName{prop}}

		\RuleNameRight{conditional}{}
		\UnaryInf{$\begin{aligned}
			& F, \tvar_0 = \recvar \cdot \comItem{\ch{vel}, \tarvelo}, \A(\tvar_0) \\
			&\qquad\qquad\quad \vdash [ \ifstat{d > \periodicity \maxvelo}{v_f \ceq \tarvelo} ] [ \recchp{\tvar_0}{\Plant_f} ] F(\tvar_0)
		\end{aligned}$}

		\RuleNameRight{acReceiveRight}{}
		\UnaryInf{$F \vdash [ \receive{\ch{vel}}{\tarvelo} ] \acpair{\A, \true} [
		\ifstat{d > \periodicity \maxvelo}{v_f \ceq \tarvelo} ] [ \Plant_f ] F(\recvar)$}

		\RuleNameRight{acMono, boxesDual}{}
		\UnaryInf{$F \vdash [ \receive{\ch{vel}}{\tarvelo} ] \acpair{\A, \true} [
		\ifstat{d > \periodicity \maxvelo}{v_f \ceq \tarvelo} ] \acpair{\A, \true} [ \Plant_f ] \acpair{\A, \true} F$}

		\iflongversion
			\SideAx{$\triangleright$ \rref{fig:follower_dist_d_unsafe}}
		\else
			\SideAx{$\triangleright$ see \cite[Fig. 22b]{Brieger2024}}
		\fi

		\RuleNameRight{acComposition}{}
		\UnaryInf{$F \vdash [ \ctrlVelocity ] \acpair{\A, \true} [ \Plant_f ] \acpair{\A, \true} F$}

		\RuleNameRight{andR}{}
		\UnaryInf{$F \vdash [ \ctrlVelocity ] \acpair{\A, \true} [ \Plant_f ] \acpair{\A, \true} F \wedge [ \ctrlDistance ] \acpair{\A, \true} [ \Plant_f ] \acpair{\A, \true} F$}

		\RuleNameRight{acChoice}{}
		\UnaryInf{$F \vdash [ (\ctrlVelocity \cup \ctrlDistance) ] \acpair{\A, \true} [ \Plant_f ] \acpair{\A, \true} F$}

		\SideAx{$\triangleright$ \RuleName{fol}}
		
		\RuleNameRight{acComposition}{}
		\UnaryInf{$F \vdash [ (\ctrlVelocity \cup \ctrlDistance) \seq \Plant_f ] \acpair{\A, \true} F$}

		\RuleNameRight{acLoop}{}
		\UnaryInf{$\Gamma \vdash [ \progtt{follower}^* ] \acpair{\A, \true} x_f < \val[x_0]{\recvar \downarrow \ch{pos}}$}
	\end{prooftree}
	\caption{
		Derivation of the contract about the \progtt{follower} (\rref{ex:follower_proof})
	}
	\label{fig:follower_velo_d_safe}
\end{figure}

The strategy taken for \rref{ex:convoy_safety} already hints an outline for the completeness proof,
except that completeness does not use the axiom \RuleName{Aweak},
which supports compositional 
mutual abstraction of parallel program effects.
Instead of abstraction, 
completeness uses specifications which conservatively enumerate all parallel interleavings,
because in extremal cases every interleaving leads to a different reachable state.
Although the axiom \RuleName{Aweak} derives from the base logic by completeness,~\RuleName{Aweak} is important in practice,
because it guarantees that mutual abstractions can be used schematically.
This is similar to the compositionality condition in classical proof rules~\cite{AcHoare_Zwiers},
which is also not necessary for completeness of discrete parallel systems \cite{deRoever2001}.

\section{Completeness}
\label{sec:completeness}

\rref{thm:soundness} shows that \dLCHP's proof calculus is sound,
\iest every provable formula is valid.
This section is concerned with the converse question
whether every valid \dLCHP formula is provable in the calculus.
Since Gödel's incompleteness theorem \cite{Goedel1931} applies to 
\dL \cite[Theorem 2]{DBLP:journals/jar/Platzer08},
there cannot be a complete and effective axiomatization for \dLCHP either.
The standard way to evaluate the deductive power of a proof calculus nevertheless is to prove completeness relative to an oracle logic~\cite{Cook1978, Harel1977}.
The central contribution of this article is a positive answer to the completeness question,
consisting of two complementary results based on progressively simpler oracle logics.

The fundamental result is \rref{thm:fod-completeness},
which shows that all properties of parallel hybrid systems in \dLCHP can be effectively reduced to properties of continuous systems.
This proof-theoretically fully aligns parallel hybrid systems with hybrid systems,
because the latter proof-theoretically also align with continuous systems \cite{DBLP:journals/jar/Platzer08}.
Formally, \rref{thm:fod-completeness} proves that \dLCHP is complete relative to the first-order logic of differential equation properties FOD,%
\footnote{
	The base logic FOD \cite{DBLP:journals/jar/Platzer08} combines first-order real arithmetic $\FolRA$ with safety and liveness $\dbleft \evolution*{}{non} \dbright \psi$ constraints $\psi \in \FolRA$ on differential equations.
}
just like \dL \mbox{\cite[Theorem 3]{DBLP:journals/jar/Platzer08}}.
Consequently, completeness relative to discrete systems and relative semidecidability results \cite{DBLP:conf/lics/Platzer12b} carry over to \dLCHP.
In summary, properties of parallel hybrid systems can be proven to exactly the same extent as properties of hybrid, continuous, or discrete systems.

Completeness is already quite challenging for hybrid systems~\cite{DBLP:journals/jar/Platzer08,DBLP:conf/lics/Platzer12b}.
The major additional challenge of parallel hybrid systems is at the tension between compositionality and completeness.
The calculus is intended to support as much compositional reduction as possible,
without compromising its ability to prove all properties of parallel hybrid systems. 
For this purpose, parallel injection \RuleName{acDropComp} exploits that interleavings often form equivalence classes,
\eg robot collision avoidance can often 
be concluded from the worst-case trajectories,
and it supports local abstractions contributing only the minimal necessary insight about each subsystem.
This promising development for compositionality raises the key question 
for completeness
whether parallel injection can always prove sufficient insights---in extremal cases, up to the full parallel product.

\rref{thm:com-fod-completeness}
gives a positive answer to this question,
and shows that \dLCHP's calculus can reduce all dynamical effects of parallel hybrid systems.
In particular, parallel injection~\RuleName{acDropComp} 
can always decompose safety of parallel hybrid systems into safety of its subsystems.
Formally, \rref{thm:com-fod-completeness} proves \dLCHP complete relative to $\comFOD$,
the extension of FOD with \dLCHP's first-order fragment.
This confirms that \dLCHP's calculus (\rref{fig:calculus}) captures all multi-dynamical aspects of parallel hybrid systems,
because it shows that \dLCHP includes all axioms required to disentangle the interwoven discrete, continuous, and communication dynamics of CHPs 
into the base logic \comFOD.

The proof is modular to separate its specific challenges into manageable pieces.
The major technical challenge solved by 
\rref{thm:com-fod-completeness} is the construction of invariants, termination conditions, and verification conditions for parallel composition,
which simultaneously span discrete, continuous, \emph{and} parallel dynamics,
as opposed to hybrid systems \cite{DBLP:journals/jar/Platzer08} and discrete parallelism \cite{AcSemantics_Zwiers}.
\rref{thm:fod-completeness} reduces the communication traces remaining in \comFOD to FOD by $\reals$-Gödel encoding~\cite{DBLP:journals/jar/Platzer08}.
For the latter, we identify an extension of \dLCHP's calculus that internalizes the relation between communication traces and $\reals$-Gödel encodings.

\subsection{Completeness Relative to \texorpdfstring{\comFOD}{Omega-FOD}}
\label{sec:com-fod-completeness}

This section shows that \dLCHP is complete relative to \comFOD (\rref{thm:com-fod-completeness}),
which is proven inductively by an equivalent reduction to the base logic \comFOD.
Due to the mixed dynamics and subtle dependencies within parallel hybrid systems,
the actual proof (\rref{sec:completeness_theorem}) only succeeds by a
combination and generalization of strategies from \dL \cite{DBLP:journals/jar/Platzer08,DBLP:conf/lics/Platzer12b,DBLP:journals/jar/Platzer17} 
for expressiveness results (\rref{sec:expressiveness_of_comFOD}) to obtain sufficient invariants and termination conditions,
ac-reasoning \cite{AcSemantics_Zwiers,deRoever2001} to obtain verification conditions for parallel composition (\rref{sec:par-safety-lemmas}),
and \dGL \cite{DBLP:journals/tocl/Platzer15} to obtain an induction order.
Therefore, a proof outline is presented prior to the actual proof.

Analogous to \dL \cite{DBLP:journals/jar/Platzer08}, \rref{sec:expressiveness_of_comFOD} proves \comFOD expressive for the transition relation of CHPs (\rref{lem:rendition}),
and from this, proves \comFOD expressive for \dLCHP (\rref{lem:com-fod-expressiveness}).
This renders \comFOD expressive enough to state sufficient loop invariants and variants.
However, the combination of hybrid dynamics and 
communication in \dLCHP---%
including 
synchronization,
multi-typed states,
and prefix-closedness---%
is significantly more subtle than~\dL's simple reachability relation.
Further, for the decomposition of parallel CHPs 
the exact transition relation of \rref{lem:rendition} is too rigid 
as it entails absence of environmental computation,
which is in conflict with parallel injection \RuleName{acDropComp} embedding properties into environments.
As solution, \rref{sec:par-safety-lemmas} generalizes a notion of strongest promises from Hoare-style ac-reasoning~\cite{AcSemantics_Zwiers, deRoever2001},
which is receptive for environmental computation,
to the hybrid setup and dynamic logic.
\rref{sec:completeness_theorem} contains the actual proof of \rref{thm:com-fod-completeness} and discusses its insights.

At the core of our proof of \rref{thm:com-fod-completeness} is an effective 
reduction of any valid \dLCHP formula in \dLCHP's calculus (\rref{fig:calculus}) to \comFOD tautologies. 
Unlike \dL's original completeness proof \cite[Theorem 3]{DBLP:journals/jar/Platzer08},
we do
not stick to the classical structure of Harel's completeness for dynamic logic \cite[Theorem 3.1]{Harel1979},
which handles safety $\varphi \rightarrow [ \alpha ] \psi$ and liveness $\varphi \rightarrow \langle \alpha \rangle \psi$ separately.
Harel's approach is not well-behaved \wrt liveness of parallel CHPs,
as their liveness 
does not follow from independent liveness of the subprograms but additionally needs matching communication and duration.
In the axioms \RuleName{acLiveParCommit} and \RuleName{acLivePar},
this is apparent in the $\exists$-quantification that does not distribute over 
the $\langle\rangle$-modalities on the subprograms.
Instead, we embark on a strategy successfully applied for \dGL \cite{DBLP:journals/tocl/Platzer15} and~\dL's uniform substitution calculus~\cite{DBLP:journals/jar/Platzer17} that uses a well-founded order on all formulas.
This order gives precedence to program decomposition---as in~\RuleName{acLiveParCommit} and \RuleName{acLivePar}---over the usual structural complexity.

\newcommand{\invac}{\acpair{\true, \Commit}}

\newcommand{\strongestPost}[2]{\spSymb_{#1, #2}}

\begin{theorem}
	[\comFOD completeness]
	\label{thm:com-fod-completeness}
	The \dLCHP calculus (\rref{fig:calculus}) is complete relative to \comFOD,
	\iest every valid \dLCHP formula $\phi$ can be proven 
	in the
	calculus from \comFOD tautologies.
\end{theorem}
\vspace{-.7cm}
\begin{proof}[Proof outline]
	The proof is by induction along a well-founded partial order on \dLCHP formulas
	induced by the overall structural complexity of programs in $\phi$.
	By propositional recombination, 
	decompose $\phi$ into safety $\varphi \rightarrow [ \alpha ] \ac \psi$ and liveness $\varphi \rightarrow \langle \alpha \rangle \ac \psi$ fragments.
	Using axiom \RuleName{Atransfer} and its dual by \RuleName{acdbDual},
	the assumption $\A$ can be moved into the promises
	such that it suffices to consider 
	$\varphi \rightarrow [ \alpha ] \invac \psi$ and 
	$\varphi \rightarrow \langle \alpha \rangle \invac \psi$
	for any $\Commit$ and $\psi$.

	Except for $\repetition{\alpha}$ and $\alpha\parOp\beta$,
	safety directly reduces to simpler formulas by the corresponding axioms.
	Liveness is analogous by duality~\RuleName{dbDual},~\RuleName{acdbDual}, 
	as all involved axioms are equivalences.
	The induction hypothesis (IH) is applicable as needed,
	because all axioms strictly reduce the program complexity.
	For~$\repetition{\alpha}$,
	standard arguments~\cite{Harel1979, DBLP:conf/lics/Platzer12b} are generalized to ac-modalities.
	The key observation is that sufficient invariants for induction~\RuleName{acInduction} 
	and termination conditions for convergence \RuleName{acConvergence}
	are always expressible in the base logic \comFOD (\rref{sec:expressiveness_of_comFOD}).
	
	Safety $\varphi \rightarrow [ \alpha\parOp\beta ] \psi$ (ac-boxes are analogous)
	generalizes an argument for Hoare-style ac-reasoning \cite{AcSemantics_Zwiers,deRoever2001} to hybrid systems and dynamic logic.
	The idea is to decompose $\psi$ into the strongest postconditions for the subprograms (\rref{sec:par-safety-lemmas})
	as follows,
	because safety of each subprogram for its strongest postcondition derives by IH:
	\begin{equation}
		\label{eq:strg_composition}
		\strongestPost{\beta, \varphi}{\alpha} \wedge \strongestPost{\alpha, \varphi}{\beta} \wedge \recvar \succeq \tvar_0 \rightarrow \psi
	\end{equation}
	The strongest postcondition $\strongestPost{\ogamma, \varphi}{\gamma}$ of a program $\gamma$ \wrt the precondition~$\varphi$ and environment~$\ogamma$
	holds in exactly those states reachable by~$\gamma$ from a state satisfying $\varphi$
	when arbitrary $\ogamma$-communication may interleave,
	where $\otherprog{\alpha} \equiv \beta$ and $\otherprog{\beta} \equiv \alpha$.
	Equation (\ref{eq:strg_composition}) is valid,
	thus derives by IH,
	because the intersection $\strongestPost{\beta, \varphi}{\alpha} \wedge \strongestPost{\alpha, \varphi}{\beta} \wedge \recvar \succeq \tvar_0$ 
	exactly denotes the states reachable by $\alpha$ and $\beta$ in parallel
	when the other program may interleave communication (\rref{fig:strongest_post_decompose}),
	which entails $\psi$ as $\varphi \rightarrow [ \alpha\parOp\beta ] \psi$ is valid.
	The constraint $\recvar \succeq \tvar_0$ sorts out states with a nonlinear history (\rref{ft:global_his_props}).
	Since $\strongestPost{\ogamma, \varphi}{\gamma}$ is a strongest postcondition,
	$\varphi \rightarrow [ \gamma ] \strongestPost{\ogamma, \varphi}{\gamma}$ is valid,
	thus derives by IH.
	Then parallel injection~\RuleName{acDropComp} 
	proves \mbox{$\varphi \rightarrow [ \alpha\parOp\beta ] \strongestPost{\ogamma, \varphi}{\gamma}$},
	because~$\ogamma$ does not interfere with $\strongestPost{\ogamma, \varphi}{\gamma}$, 
	since $\strongestPost{\ogamma, \varphi}{\gamma}$ admits $\ogamma$-interleaving.
	Further, history invariance~\RuleName{hExtension} proves $\varphi \rightarrow [ \alpha\parOp\beta ] \recvar \succeq \tvar_0$,
	assuming~$\varphi$ contains $\tvar_0 = \recvar$.
	Finally, $\varphi \rightarrow [ \alpha \parOp \beta ] \psi$ derives by monotonicity~\RuleName{acMono}
	using \rref{eq:strg_composition}.

	By~\RuleName{acSplitDia},
	liveness $\langle \alpha\parOp\beta \rangle \acpair{\true, \Commit} \psi$ is split into 
	$\langle \alpha\parOp\beta \rangle \acpair{\true, \Commit} \false$ and $\langle \alpha\parOp\beta \rangle \acpair{\true, \false} \psi$,
	which derive by \RuleName{acLiveParCommit} and \RuleName{acLivePar}, respectively.
	The premises of \RuleName{acLiveParCommit} and \RuleName{acLivePar} derive by IH,
	because they are valid,	
	as they equivalently express 
	$\langle \alpha\parOp\beta \rangle \acpair{\true, \Commit} \psi$.
	The premises are smaller in the induction order,
	because decomposition of the parallel composition reduces the overall structural complexity of programs,
	even though the formula itself increased in complexity.
\end{proof}

\newsavebox\stateSet
\newsavebox\stateSetNeighbors
\newsavebox\denoRela
\newsavebox\envDenoRela

\sbox\stateSet{\tikz[baseline=-0.5ex]{\node[circle, fill=blue!5, draw=blue, inner sep=.8mm] {};}}
\sbox\stateSetNeighbors{\tikz[baseline=-0.5ex]{\node[circle, fill=blue!5, draw=blue, inner sep=.8mm, densely dashed] {};}}
\sbox\denoRela{\tikz[baseline=-0.5ex]{\draw[->] (0,0) to (.6,0);}}
\sbox\envDenoRela{\tikz[baseline=-0.5ex]{\draw[->, densely dashed] (0,0) to (.6,0);}}

\begin{figure}[t]
	\begin{minipage}{.38\textwidth}
		\captionof{figure}{
			Decomposition of strongest postconditions for parallel CHPs.
			The nodes \usebox{\stateSet} and \usebox{\stateSetNeighbors} denote \emph{sets} of states,
			the arrows \usebox{\denoRela} denote reachability relations, 
			and \usebox{\envDenoRela} denotes reachability when the environment may interleave. 
		}
		\label{fig:strongest_post_decompose}
	\end{minipage}\hspace{.4cm}
	\begin{minipage}{.57\textwidth}
		\begin{tikzpicture}
			\tikzstyle{state}+=[
				fill=blue!5, draw=blue
			];
			\tikzstyle{opaque}+=[
				fill opacity=.5, text opacity=1
			];	
			
			\node[state, circle] (va) {\parbox[t][.4cm][t]{.5cm}{\centering$\varphi_\alpha$}};
			\node[state, circle, opaque] (vb) at ([yshift=-2mm] va.south) {\parbox[b][.4cm][b]{.5cm}{\centering$\varphi_{\smash{\beta}\vphantom{\alpha}}$}};

			\draw[state, dashed, rotate=-20, opaque, name path=env-a] 
				(va.east)++(3,1.25) coordinate (env-a-center) ellipse (1.7cm and 1cm);
			\draw[state, dashed, rotate=20, opaque, name path=env-b] 
				(vb.east)++(3,-1.25) coordinate (env-b-center) ellipse (1.7cm and 1cm);
			
			\coordinate (reach-a) at ([xshift=2.7cm, yshift=.65cm]va.east);
			\coordinate (reach-b) at ([xshift=2.7cm, yshift=-.65cm]vb.east);
			\coordinate (reach-ab) at ([xshift=3.5cm, yshift=-.35cm]va.east);

			\node[state, ellipse, inner xsep=-1mm, inner ysep=1mm] at (reach-a) (sp-a) {$\strongestPost{\emptyset, \varphi_\alpha}{\alpha}$};

			\node[state, ellipse, inner xsep=-1mm, inner ysep=1mm] at (reach-b) (sp-b) {$\strongestPost{\emptyset, \varphi_{\smash{\beta}\vphantom{\alpha}}}{\beta}$};

			\node[state, ellipse, inner xsep=-1mm, inner ysep=1mm, fill=none] at (reach-ab) (sp-ab) {$\strongestPost{\emptyset, \varphi}{\alpha\parOp\beta}$};

			\draw[->] (va) edge[bend left=30] node[fill=white] {$\alpha$} (sp-a);
			\draw[->] (vb) edge[bend right=30] node[fill=white] {$\beta$} (sp-b);
			\draw[->] ($(va.center)!0.5!(vb.center)$) to node[fill=white, pos=.6] {$\alpha\parOp\beta$} (sp-ab);

			\path [name path=a-env-run] (va) -- (sp-a.south west);
			\path [name intersections={of=env-a and a-env-run, by={a-env-head}}];
			\path [name path=b-env-run] (vb) -- (sp-b.north west);
			\path [name intersections={of=env-b and b-env-run, by={b-env-head}}];

			\draw[->, densely dashed] (va) edge[bend left=10] node[fill=white, inner sep=.5mm] {$\alpha\parOp\_$} (a-env-head);
			\draw[->, densely dashed] (vb) edge[bend right=10] node[fill=white, inner sep=.5mm] {$\_\parOp\beta$} (b-env-head);

			\node [above right=-.2cm and 1.6cm of sp-a] (env-a-label) {$\strongestPost{\smash{\beta}, \varphi_\alpha}{\alpha}$};
			\node [below right=-.2cm and 1.6cm of sp-b] (env-b-label) {$\strongestPost{\alpha, \varphi_{\smash{\beta}\vphantom{\alpha}}}{\beta}$};

			\path [name path=env-a-path] (env-a-center) -- (env-a-label);
			\path [name intersections={of=env-a and env-a-path, by={edge-a}}];
			\path [name path=env-b-path] (env-b-center) -- (env-b-label);
			\path [name intersections={of=env-b and env-b-path, by={edge-b}}];

			\draw[->, semithick, dotted] (env-a-label.west) -- (edge-a);
			\draw[->, semithick, dotted] (env-b-label.west) -- (edge-b);

			\node [right=.7cm of sp-ab.east] (not-h-ext) {$\neg \recvar \succeq \tvar_0$};
			\draw[->, semithick, dotted] (not-h-ext) -- ([xshift=2mm]sp-ab.east);
		\end{tikzpicture}
	\end{minipage}
\end{figure}

\subsubsection{Expressiveness of \texorpdfstring{\comFOD}{Omega} for \dLCHP}
\label{sec:expressiveness_of_comFOD}

This section shows that the base logic \comFOD is expressive for \dLCHP (\rref{lem:com-fod-expressiveness}),
generalizing results from \dL \cite{DBLP:journals/jar/Platzer08}.
This guarantees that sufficient loop invariants and termination conditions are expressible in the base logic \comFOD.
Central to the expressiveness result is a characterization of the transition semantics of CHPs in \comFOD.
\rref{lem:rendition} effectively maps each CHP $\alpha$ to an \comFOD formula $\rendition{\alpha}{}{}$ that holds in a state if and only if there is an $\alpha$-run 
from an initial state characterized by $\alpha$'s variables~$\varvec$ 
to a state characterized by the variables $\varvec[alt]$.
The predicate symbol~$\finmarker{}$ tells whether~$\varvec[alt]$ is an intermediate ($\finmarker{}=\false$) or final ($\finmarker{}=\true$) state.
For final states,
the rendition $\rendition{}{}{\renparams{}{}{\true}}$ is equivalent to $\langle \alpha \rangle \varvec[alt] = \varvec$.
\rref{lem:rendition} generalizes this
to cover the intermediate states reachable by CHPs.
All proofs for this section are in \rref{app:com-fod-completeness}.

In contrast to \dL,
the rendition does \emph{not} use $\reals$-Gödel encodings to capture the unboundedly many intermediate states of repetitions.
Instead, the sequence of states is stored in trace variables,
whose $\reals$-Gödel encoding is later covered by completeness relative to FOD (\rref{sec:fod-completeness}).
The base logic \comFOD enables a rendition that closely mirrors the program semantics,
in particular for parallel composition.
Yet intermediate states and prefix-closedness of the program semantics increase the technicality of the rendition of CHPs 
compared to hybrid programs in~\dL.

\newcommand{\tvarof}[1]{\tvar_{\vecof{#1}}}

\newcommand{\mergeren}[1]{(\varvec[alt]_\alpha {\merge} \varvec[alt]_\beta)_{#1}}

\begin{lemma}
	[Rendition of programs]
	\label{lem:rendition}
	Let $\alpha$ be a CHP and let $\varvec \equiv (\recvar, \avar_1, ..., \avar_\semvar)$ be all variables of~$\alpha$.
	Moreover, let $\varvec[alt] \equiv (\tvarof{\avar[alt]}, \avar[alt]_1, ..., \avar[alt]_\semvar)$ 
	and $\varvec[third] \equiv (\tvarof{\avar[third]}, \avar[third]_1, ..., \avar[third]_\semvar)$ be fresh,
	and let $\finmarker{}$ be a predicate symbol.
	Then there is an \comFOD formula $\rendition{\alpha}{}{}$ such that the following formula is valid,
	where the fresh variables~$\varvec[third]$ 
	ensure that no variables except for the recorder $\recvar$ change in runs to intermediate states
	and guarantee well-formedness (\rref{def:syntax_formulas}) of the ac-diamond:
	\begin{equation*}
		\rendition{\alpha}{}{}
		\leftrightarrow
		\fa{\varvec[third]{=}\varvec}\langle \alpha \rangle \acpair{\true, \neg\finmarker{} \wedge \varvec[alt] = \varvec[third] \subs{\tvarof{\avar[third]}}{\recvar}}
		(\finmarker{} \wedge \varvec[alt] = \varvec) 
	\end{equation*}
\end{lemma}

\begin{lemma}
	[Expressiveness of \comFOD]
	\label{lem:com-fod-expressiveness}
	The logic \dLCHP is expressible in \comFOD.
	That is, for every \dLCHP formula $\phi$, there is an \comFOD formula $\toComFOD{\phi}$ over the same free variables such that $\vDash \phi \leftrightarrow \toComFOD{\phi}$.
\end{lemma}

\subsubsection{Verification Conditions for Parallelism}
\label{sec:par-safety-lemmas}

\newcommand{\opencom}[3]{\text{Com}(#1, #2, #3)}

A compositional proof of safety $\varphi \rightarrow [\alpha\parOp\beta] \psi$ naturally asks for splitting~$\psi$ 
such that $\psi_\alpha \wedge \psi_\beta \rightarrow \psi$,
and $\varphi \rightarrow  [ \alpha ] \psi_\alpha$,
and $\varphi \rightarrow  [ \beta ] \psi_\beta$ 
derive,
where $\psi_\gamma$ specifies the local behavior of $\gamma \in \{\alpha,\beta\}$.
If the subprograms do not interfere (\rref{def:noninterference}) with each other's postcondition,
parallel injection \RuleName{acDropComp} 
embeds them into the parallel composition,
\iest \mbox{$\varphi \rightarrow [ \alpha\parOp\beta ] \psi_\gamma$} derives 
for each $\psi_\gamma$,
so that $\varphi \rightarrow [ \alpha \parOp \beta ] \psi$ derives by modal logic reasoning.
The challenge for completeness is to find $\psi_\alpha$ and~$\psi_\beta$,
which capture sufficiently much of~$\alpha$'s and $\beta$'s behavior to entail~$\psi$ 
but also satisfy noninterference.
A natural choice for $\psi_\gamma$ seems to be the strongest postcondition~$\strongestPost{\varphi}{\gamma}$ of $\gamma$ \wrt the precondition~$\varphi$,
because it exactly demarcates $\gamma$'s behavior in terms of its reachable states.
Strict reachability, however, 
requires absence of environmental communication.%
\footnote{%
	Let $\strongestPost{\varphi}{\alpha}$ be the strongest postcondition of $\alpha$ \wrt to the precondition $\varphi$,
	\iest $\strongestPost{\varphi}{\alpha}$ exactly denotes all states reachable by an $\alpha$-run from some state satisfying $\varphi$.
	Then $\varphi \rightarrow [ \send{}{0} ] \strongestPost{\varphi}{\send{}{}}$ is valid,
	where $\varphi \equiv \tvar \downarrow \ch{dh} = \epsilon$,
	but $\varphi \rightarrow [ \send{}{0} \parOp \receive{\ch{dh}}{} ] \strongestPost{\varphi}{\send{}{0}}$ is not valid because $\strongestPost{\varphi}{\send{}{0}}$ requires that there was no communication on $\ch{dh}$ previously,
	\iest in the initial state.
	In fact, $\receive{\ch{dh}}{x}$ interferes (\rref{def:noninterference}) with $(\send{}{0}, \strongestPost{\varphi}{\send{}{0}})$. 
}

As solution,
we adapt an approach from Hoare-style ac-reasoning \cite{AcSemantics_Zwiers} to hybrid systems and dynamic logic.
The idea is to extend the strongest postcondition~$\strongestPost{\varphi}{\gamma}$ with all variations of the original states,
which cover some interleaving of communication potentially stemming from another program $\ogamma$.
This defines the strongest postcondition~$\strongestPost{\ogamma, \varphi}{\gamma}$ \wrt an environment~$\ogamma$.
Since $\strongestPost{\beta, \varphi}{\alpha}$ and $\strongestPost{\alpha, \varphi}{\beta}$ cover each other's final states,
their intersection covers the final states of~$\alpha\parOp\beta$,
as opposed to classical strongest postconditions.
Definition \ref{def:state_variations} characterizes the strongest promises semantically,
\rref{lem:state_variations_comfod} expresses them in the base logic \comFOD,
and \rreflems{lem:reachable_states_correct}{lem:ac_biggest_promise_split}
prove that the strongest promises completely demarcate the reachable states of 
parallel CHPs.
All proofs for this section are in \rref{app:com-fod-completeness}.

\newcommand{\comvarrun}[2]{(\pstate{v}, \trace \downarrow (\comvariation{#1}{#2}), \pstate{w})}

\begin{definition}
	[Environmental state variations]
	\label{def:state_variations}
	For a CHP $\alpha$,
	define
	\emph{intermediate state variations} 
	$\reachableInter{\cset, \varphi}{\alpha}$ and \emph{final state variations} $\reachableFin{\cset, \varphi}{\alpha}$
	\wrt the precondition $\varphi$ and channels $\cset \subseteq \Chan$,
	where $\sem{\varphi}{} \circ \denotation = \{\run \in \denotation \mid \pstate{v} \vDash \varphi \}$:
	\begin{align*}
		\reachableInter{\cset, \varphi}{\alpha} 
		& = \big\{ 
			\pstate{v} \cdot \trace \mid 
				\bfexists \pstate{w} : (\pstate{v}, \trace \downarrow (\alpha \cup \cset^\complement), \pstate{w}) \in \sem{\varphi}{} \circ \sem{\alpha}{}
		\big\} \\
		\reachableFin{\cset, \varphi}{\alpha}
		& = \big\{ 
			\pstate{w} \cdot \trace \mid 
				\bfexists \pstate{v} : (\pstate{v}, \trace \downarrow (\alpha \cup \cset^\complement), \pstate{w}) \in \sem{\varphi}{} \circ \sem{\alpha}{}
		\big\}
	\end{align*} 
\end{definition}

Environmental state variations (\rref{def:state_variations}) 
denote the reachable states 
in presence of environmental computation.
A variation results from an $\alpha$-run by interleaving some communication on the \mbox{non-$\alpha$} channels $\alpha^\complement \cap \cset$.
Final state variations implicitly cover environmental effects on the state as well,
because programs do not share state (\rref{def:syntax_chps}) and the initial state is  $\bfexists$-quantified.
Different from Hoare-style ac-reasoning~\cite{AcSemantics_Zwiers},
variation is not characterized within the transition relation (\rref{lem:rendition}),
which simplifies their encoding.
Instead, \rref{lem:state_variations_comfod} modularly
combines reachability~$\langle \alpha \rangle$ 
with projections in \comFOD to express variations.

\begin{lemma}
	[Strongest promises]
	\label{lem:state_variations_comfod}
	For any (co)finite $\cset\subseteq\Chan$,
	there are formulas $\strgCommit{\cset}{\varphi}{}{\alpha}$ and $\strgPost{\cset}{\varphi}{}{\alpha}$,
	where $\strgCommit{\cset}{\varphi}{}{}, \strgPost{\cset}{\varphi}{}{} \in \comFOD$,
	called the \emph{strongest commitment}~and \emph{strongest postcondition}, respectively,
	of the CHP $\alpha$ \wrt the precondition $\varphi$ and environmental communication on channels~$\cset$,
	which characterize the
	state variations (\rref{def:state_variations}):
	\begin{equation*}
		\reachableInter{\cset, \varphi}{\alpha} = \sem{\strgCommit{\cset}{\varphi}{}{\alpha}}{}
		\qquad\qquad
		\reachableFin{\cset, \varphi}{\alpha} = \sem{\strgPost{\cset}{\varphi}{}{\alpha}}{}
	\end{equation*}
	For $\scspSymb \in \{\scSymb,\spSymb\}$ and program~$\beta$,
	define $\strgPromise{\beta}{\varphi}{}{\alpha} \equiv \strgPromise{\cset}{\varphi}{}{\alpha}$ with $\cset = \SCN(\beta)$.
	For every well-formed (\rref{def:syntax_chps}) $\alpha\parOp\beta$, 
	if $\beta$ does not interfere (\rref{def:noninterference}) with $\nointfpair{\alpha}{\varphi}$,
	then $\beta$ does not interfere with $\nointfpair{\alpha}{\strgPromise{\beta}{\varphi}{}{\alpha}}$.
\end{lemma}

\rref{lem:reachable_states_correct} proves that the strongest promises are indeed strong enough to entail any valid promise (\rref{itm:reachable_states_entailment})
but not so strong as to cease being valid promises themselves (\rref{itm:reachable_states_correct}).
For $\cset=\emptyset$,
\iest the environment may not interleave,
$\strgPost{\cset}{\varphi}{}{\alpha}$ coincides with classical strongest postconditions.
If $\cset\not\subseteq\SCN(\alpha)$,
then $\strgPost{\cset}{\varphi}{}{\alpha}$ 
does not entail~$\psi$ (\rref{itm:reachable_states_entailment}),
because $\strgPost{\cset}{\varphi}{}{\alpha}$ contains states with environmental communication along non-$\alpha$ channels $\alpha^\complement \cap \cset$.
Still $\strgPost{\cset}{\varphi}{}{\alpha}$ is a valid promise (\rref{itm:reachable_states_correct}) because those states are just not reachable by~$\alpha$.
Since $\strgPost{\beta}{\varphi}{}{\alpha}$ covers all possible interleavings of~$\beta$'s communication,
it stays valid in all final states of $\alpha\parOp\beta$,
and
parallel injection \RuleName{acDropComp} proves $[ \alpha ] \strgPost{\beta}{\varphi}{}{\alpha} \rightarrow [ \alpha\parOp\beta ] \strgPost{\beta}{\varphi}{}{\alpha}$
as~$\beta$ does not interfere (\rref{def:noninterference}). 

\begin{lemma}
	[Correctness of the strongest promises]
	\label{lem:reachable_states_correct}
	The strongest promises (\rref{lem:state_variations_comfod}) satisfy the following properties.
	Hiding $\fa{\rvarvec{=}\rvarvec[alt]}$ of the variables $\rvarvec \supseteq \SBV(\alpha)$
	in the commitment ensures well-formedness of the ac-box (\rref{def:syntax_formulas}),
	where $\rvarvec[alt]$ is fresh: 
	\begin{enumerate}
		\item 
		\label{itm:reachable_states_entailment}
		If $\vDash \varphi \rightarrow [ \alpha ] \acpair{\true, \Commit} \psi$, then
		\begin{enumerate*}[label=(\roman*)]
			\item $\vDash \strgCommit{\emptyset}{\varphi}{}{\alpha} \rightarrow \Commit$ and\label{itm:strongest_com_implies_commit}
			
			\item $\vDash \strgPost{\emptyset}{\varphi}{}{\alpha} \rightarrow \psi$\label{itm:strongest_reach_implies_post}
		\end{enumerate*}

		\item 
		\label{itm:reachable_states_correct}
		$\vDash \rvarvec[alt] = \rvarvec \wedge \varphi \rightarrow [ \alpha ] \acpair{\true, \scSymb} \strgPost{\cset}{\rvarvec[alt] = \rvarvec \wedge \varphi}{}{\alpha}$,
		where $\scSymb \equiv \fa{\rvarvec{=}\rvarvec[alt]} \strgCommit{\cset}{\rvarvec[alt] = \rvarvec \wedge \varphi}{}{\alpha}$
	\end{enumerate}
\end{lemma}

\rref{lem:ac_biggest_promise_split} 
splits
the strongest promises for $\alpha\parOp\beta$ into strongest promises for the subprograms
when they admit interleaving of each other's communication.
The preconditions $\varphi_\alpha, \varphi_\beta$ characterize $\alpha$'s and $\beta$'s local share of the initial state.
Further, $F_\alpha, F_\beta$ additionally align the duration and previous history of the subprograms
on shared channels by the constraint $\tvar_0 \downarrow \cset_\gamma = \recvar \downarrow \cset_\gamma$.
History invariance $\recvar \succeq \tvar_0$ 
rejects runs,
where~$\alpha$ or~$\beta$ interleave with each other's previous history (\rref{ft:global_his_props}).

\begin{lemma}
	[Promise splitting]
	\label{lem:ac_biggest_promise_split}
	For each $(\gamma, \ogamma) \in \{(\alpha,\beta), (\beta,\alpha)\}$,
	let $\varphi_\gamma$ be a formula such that~$\ogamma$ does not interfere (\rref{def:noninterference}) with $(\gamma, \varphi_\gamma)$,
	and let $\cset_\gamma \supseteq \SCN(\gamma) \cup \{\durchan\} \cup \SCN(\ogamma)^\complement$.
	Then the following formula is valid
	for each strongest promise $\scspSymb \in \{\scSymb,\spSymb\}$ (\rref{lem:state_variations_comfod}),
	where $F_\gamma \equiv \varphi_\gamma \wedge \tvar_0 \downarrow \cset_\gamma = \recvar \downarrow \cset_\gamma$, 
	and $F \equiv \varphi_\alpha\wedge\varphi_\beta \wedge \tvar_0 = \recvar$,
	and $\tvar_0$ are fresh:
	\begin{align*}
		\strgPromise{\beta}{F_\alpha}{}{\alpha} \wedge
		\strgPromise{\alpha}{F_\beta}{}{\beta}
		\wedge 
		\recvar \succeq \tvar_0
		\rightarrow \strgPromise{\emptyset}{F}{}{\alpha\parOp\beta}
	\end{align*}
\end{lemma}

\subsubsection{Proof of Completeness Relative to \texorpdfstring{\comFOD}{Omega-FOD}}
\label{sec:completeness_theorem}

This section proves \dLCHP complete relative to \comFOD (\rref{thm:com-fod-completeness})
by an effective 
reduction of any valid \dLCHP formula to \comFOD tautologies in \dLCHP's proof calculus~(\rref{fig:calculus}).
A proof outline is at the beginning of \rref{sec:com-fod-completeness}.

\CreateRuleRef{prop}%
In \rref{sec:calculus}, we assumed that \dLCHP's proof calculus contains a complete axiomatization of first-order logic.
To make 
this precise,
\rref{thm:com-fod-completeness} uses the axioms~\RuleName{uniInstance} for universal instantiation, \RuleName{faDist} for distributivity,~\RuleName{vacuousQuanti} for vacuous quantification,
and \RuleName{subs} for substitution.
Axiom \RuleName{iG} for introducing a ghost variable derives.%
\footnote{By equality in first-order logic, obtain $\vdash \expr = \expr$,
so $\vdash (\expr = \expr \rightarrow \psi) \rightarrow \psi$ propositionally.
Then $\vdash \fa{\avar} (\avar = \expr \rightarrow \psi) \rightarrow \psi$ by \RuleName{uniInstance} as $\avar$ is fresh.}
\vspace{1em}

\begin{small}
	\begin{calculus}
		\startAxiom{uniInstance}
			$\fa{\avar} \psi(\avar) \rightarrow \psi(\expr)$
		\stopAxiom
		\startAxiom{faDist}
			$\fa{\avar} (\varphi \rightarrow \psi) \rightarrow (\fa{\avar} \varphi \rightarrow \fa{\avar} \psi)$
		\stopAxiom
		\startAxiom{vacuousQuanti}
			$\psi \rightarrow \fa{\avar} \psi 
			\quad\sidecondition{($\avar\not\in\SFV(\psi)$)}$
		\stopAxiom
	\end{calculus}\hspace{.4cm}
	\begin{calculus}
		\startAxiom{subs}
			$\avar_0 = \avar \rightarrow (\psi(\avar_0) \rightarrow \psi(\avar))$
		\stopAxiom
		\startAxiom{iG}
			$\fa{\avar} (\avar = \expr \rightarrow \psi) \rightarrow \psi
			\quad\sidecondition{($\avar$ fresh)}$
		\stopAxiom
	\end{calculus}
\end{small}
	
\begin{proof}
	[Proof of \rref{thm:com-fod-completeness}.]
	Write $\comFODderives \phi$ when the formula $\phi$ derives in \dLCHP's calculus (\rref{fig:calculus}) from \comFOD tautologies.
	Hence, for every \dLCHP formula $\phi$, 
	it is to be proven that $\vDash \phi$ implies~$\comFODderives \phi$.
	The formula $\phi$ is assumed to contain only ac-modalities using the equivalences~\RuleName{boxesDual} and~\RuleName{diasDual}.
	Further, $\phi$ is assumed to be in conjunctive normal form
	with negations pushed inside over modalities and quantifiers using the %
	equivalences $\neg [ \alpha ] \ac \psi \leftrightarrow \langle \alpha \rangle \acpair{\A, \neg\Commit} \neg\psi$ and $\neg \langle \alpha \rangle \ac \psi \leftrightarrow [ \alpha ] \acpair{\A, \neg\Commit} \neg\psi$ (by~\RuleName{acdbDual}),
	and $\neg \fa{\avar} \psi \leftrightarrow \ex{\avar} \neg\psi$ and $\neg \ex{\avar} \psi \leftrightarrow \fa{\avar} \neg\psi$.
	Unlike \dL's completeness proof \cite{DBLP:journals/jar/Platzer08}, 
	this proof explicitly handles quantifiers,
	because $\fa{\avar}\!$ and $\ex{\avar}\!$ have no simple differential equation encoding if $\avar$ is a trace variable.

	The proof is by induction along a well-founded partial order $\indOrder$ on \dLCHP formulas similar to an order used for \dGL~\cite{DBLP:journals/tocl/Platzer15}.
	The order $\indOrder$ lexicographically combines the ordering $\progorder$ of formulas by the overall structural complexity of the programs they contain
	and the ordering~$\fmlorder$ of formulas by the number of logical operators as usual,
	and both orders~$\progorder$ and~$\fmlorder$ put the base logic $\comFOD$ at their bottom,
	because every valid \comFOD formula derives in $\comFODderives$.
	Hence, if $\varphi \indOrder \psi$, 
	the formula~$\varphi$ might even have a more complex logical structure (\eg more quantifiers) than $\psi$ as long as some program got simpler and none got worse.
	A formula becomes smaller in $\progorder$ if some program is removed or decomposed.
	Consequently, the order~$\indOrder$ is well-founded,
	because the overall structural complexity of programs can only decrease finitely often
	such that every descending chain in $\indOrder$ eventually removed all programs
	and reaches the base logic $\comFOD$.
	In fact, $\indOrder$ is well-founded as lexicographic combination of well-founded orders.
	Formally,~$\progorder$ and~$\fmlorder$ can be defined from rank functions 
	\iflongversion
		(\rref{app:order})
	\else
		\cite{Brieger2024}.
	\fi
	
	Now, let $\vDash \phi$.
	Then $\comFODderives \phi$ is proven by well-founded induction on the structure of $\phi$ along the order $\indOrder$.
	Throughout the proof IH is short for induction hypothesis.
	\begin{enumerate}[leftmargin=*]
		\item 
		\label{itm:compl-case-com-fod}
		If $\phi$ contains no program,
		then $\phi$ is an \comFOD formula,
		thus $\comFODderives \phi$.

		\item $\phi \equiv \neg \psi$, then $\phi$
		is covered by case \ref{itm:compl-case-com-fod},
		because negations are assumed to be pushed inside over modalities
		such that $\psi$ cannot contain any program.

		\item $\phi \equiv \phi_1 \wedge \phi_2$,
		then $\vDash \phi_j$ for $j \in \{1,2\}$.
		Since $\phi_j \indOrder \phi$,
		as $\phi_j$ is structurally simpler than~$\phi$,
		obtain $\comFODderives \phi_j$ by IH.
		Then $\comFODderives \phi_1$ and $\comFODderives \phi_2$ combine to $\comFODderives \phi_1 \wedge \phi_2$ by propositional reasoning.
		
		\item $\phi \equiv \fa{\avar} \psi$,
		or $\phi \equiv \ex{\avar} \psi$,
		or $\phi \equiv \dbleft \alpha \dbright \ac \psi$,
		where $\dbleft \alpha \dbright$ is a unifying notation for $[ \alpha ]$ and $\langle \alpha \rangle$,
		then obtain $\comFODderives \false \vee \phi$ via case \ref{itm:compl-case-fa},
		\ref{itm:compl-case-ex}, or \ref{itm:compl-case-box},
		respectively,
		which yields $\comFODderives \phi$ propositionally.
	\end{enumerate}

	In case $\phi \equiv \phi_1 \vee \phi_2$, 
	\wlossg assume 
	$\phi_2 \equiv \fa{\avar} G$, 
	or $\phi_2 \equiv \ex{\avar} G$,
	or $\phi_2 \equiv \dbleft \alpha \dbright \ac G$ by derivable associativity and commutativity,
	and that $\phi_2 \not\in \comFOD$ (\eg $\fa{\avar} G$ would be an \comFOD formula if~$G$ is).
	In the remainder, 
	abbreviate $\neg\phi_1$ as $F$,
	so $\vDash \phi$ implies $\vDash F \rightarrow \phi_2$,
	then show $\comFODderives F \rightarrow \phi_2$,
	which yields $\comFODderives \phi_1 \vee \phi_2$ by propositional reasoning.
	Without further notice, the proof uses that $(F \rightarrow \lambda) \indOrder (F \rightarrow \chi)$ if $\lambda \indOrder \chi$, for any formulas $\lambda, \chi$.
	\begin{enumerate}[leftmargin=*, labelwidth=!, resume]
		\item 
		\label{itm:compl-case-fa}
		$\phi \equiv F \rightarrow \fa{\avar} G$,
		then assume $\avar \not\in F$ by bound variable renaming.
		Hence, $\vDash F \rightarrow G$.
		Since $G \indOrder \fa{\avar} G$,
		because $G$ has less quantifiers than $\fa{\avar} G$,
		obtain $\comFODderives F \rightarrow G$ by IH. 
		Then $\comFODderives \fa{\avar} (F \rightarrow G)$ by \ForallGen.
		Hence, $\comFODderives \fa{\avar} F \rightarrow \fa{\avar} G$ by \RuleName{faDist},
		so $\comFODderives F \rightarrow \fa{\avar} G$ by~\RuleName{vacuousQuanti}.

		If $\avar$ is the recorder $\recvar$ and $G$ contains a program binding $\recvar$,
		\eg $G \equiv [ \send{}{} ] \psi$,
		bound renaming of $\avar$ is not possible,
		since the fixed recorder $\recvar$ cannot be uniformly renamed in~$G$ as necessary.
		Instead, $G$ can be encoded as in case~\ref{itm:compl-case-ex}.
		Alternatively, recorder renaming can be simulated via modalities $\dbleft \recchp{\tvar}{\alpha} \dbright \psi$ with recorder~$\tvar$ (\rref{sec:substitution}),
		which can be treated as atomic formulas
		for the purposes of this proof (\rref{rem:recorder_renaming}).

		\item 
		\label{itm:compl-case-ex}
		$\phi \equiv F \rightarrow \ex{\avar} G$,
		then there is an \comFOD formula $\toComFOD{G}$ by \rref{lem:com-fod-expressiveness} 
		such that $\vDash G \leftrightarrow \toComFOD{G}$.
		Since $\ex{\avar} G \not\in \comFOD$
		but $\ex{\avar} \toComFOD{G} \in \comFOD$,
		obtain $\ex{\avar} \toComFOD{G} \indOrder \ex{\avar} G$,
		so $\comFODderives F \rightarrow \ex{\avar} \toComFOD{G}$ by IH.
		Further, $(\toComFOD{G} \rightarrow G) \indOrder \phi$,
		as~$\toComFOD{G} \in \comFOD$, and $G$ has less quantifiers than $\ex{\avar} G$.
		Hence, $\comFODderives \toComFOD{G} \rightarrow G$ by IH.
		By \ForallGen, $\comFODderives \fa{\avar} (\toComFOD{G} \rightarrow G)$.
		Then $\comFODderives \ex{\avar} \toComFOD{G} \rightarrow \ex{\avar} G$ by the derivable dual of \RuleName{faDist}.
		This combines with $\comFODderives F \rightarrow \ex{\avar} \toComFOD{G}$ to $\comFODderives F \rightarrow \ex{\avar} G$ using~\RuleName{MP}. 

		\item 
		\label{itm:compl-case-box}
		$\phi \equiv F \rightarrow \dbleft \alpha \dbright \ac G$,
		then the proof is by the following case analysis (\ref{itm:compl-case-assign}--\ref{itm:par-liveness-special}) of the structure of $\dbleft \alpha \dbright$.
		Missing ac-diamond cases derive analogously to their ac-box counterpart since the axioms used are equivalences such that 
		dual axioms derive by~\RuleName{dbDual} and~\RuleName{acdbDual}.
		
		Further, it suffices to consider $\A \equiv \true$,
		because assumptions $\A \not\equiv \true$ can be moved to the promises
		using \RuleName{Atransfer} as follows,
		where the premise is valid since $\tvar_0$ is fresh:
		\begin{prooftree}
			\Axiom{$F \wedge \tvar_0 = \recvar \rightarrow [ \alpha ] \acpair{\true, \Abutlast{\A} \rightarrow \Commit} (\Aglobally{\A} \rightarrow G)$}

			\RuleNameRight{forall}{}
			\UnaryInf{$F \rightarrow \fa{\tvar_0{=}\recvar} [ \alpha ] \acpair{\true, \Abutlast{\A} \rightarrow \Commit} (\Aglobally{\A} \rightarrow G)$}

			\RuleNameRight{Atransfer}{}
			\UnaryInf{$F \rightarrow [ \alpha ] \ac G$}
		\end{prooftree}

		The case $F \rightarrow \langle \alpha \rangle \acpair{\A, \Commit} G$
		is analogous using the dual of~\RuleName{Atransfer} by \RuleName{acdbDual}.
		
		If $\SCN(\alpha)=\emptyset$,
		as in the cases \ref{itm:compl-case-assign}--\ref{itm:compl-test} below,
		it suffices to prove that $\vDash F \rightarrow \dbleft \alpha \dbright G$ implies $\comFODderives F \rightarrow \dbleft \alpha \dbright G$
		for any $F$
		as follows:
		Let $\vDash F \rightarrow [ \alpha ] \ac G$.
		Then $\vDash F \rightarrow \Commit$ and $\vDash F \wedge \A \rightarrow [ \alpha ] G$
		by \RuleName{acNoCom}, because $\SCN(\alpha) = \emptyset$.
		Hence, $\comFODderives F \rightarrow \Commit$ by IH,
		because~$\Commit$ has less modalities than $[ \alpha ] \ac G$, so $(F \rightarrow \Commit) \indOrder \phi$.
		Further, $\comFODderives F \wedge \A \rightarrow [ \alpha ] G$
		by the assumption.
		The latter and $\comFODderives F \rightarrow \Commit$
		combine to $\comFODderives F \rightarrow \Commit \wedge (\A \rightarrow [ \alpha ] G)$
		by~\RuleName{MP} and propositional reasoning.
		Finally, $\comFODderives F \rightarrow [ \alpha ] \ac G$ by \RuleName{acNoCom}. 
		The case $F \rightarrow \langle \alpha \rangle G$ is analogous,
		using that the dual \RuleName{acDiaNoCom} of \RuleName{acNoCom} derives.

		\begin{enumerate}[leftmargin=0pt, label={\arabic{enumi}\alph*.}]
			\item 
			\label{itm:compl-case-assign}
			$\vDash F \rightarrow [ x \ceq \rp ] G$, 
			then $\vDash F \rightarrow G \subs{x}{\rp}$ by \RuleName{assign},
			where $G \subs{x}{\rp}$ is the capture-avoid substitution of~$\rp$ for $x$ in $G$,
			so that no free variable of $\rp$ gets bound in $G \subs{x}{\rp}$.
			Since the number of programs decreased,%
			\footnote{%
				Capture-avoidance can be defined such that no new program is introduced,
				\eg $G \subs{x}{\rp} \equiv \fa{y} (y = \rp \rightarrow  G \subs{x}{y})$ for a fresh variable $y$,
				where $G \subs{x}{y}$ needs no further capture-avoidance as $y$ is fresh \cite{DBLP:journals/tocl/Platzer15}.
			}
			obtain $G \subs{x}{\rp} \indOrder [ x \ceq \rp ] G$.
			Hence, $\comFODderives F \rightarrow G \subs{x}{\rp}$ by IH.
			Finally, $\comFODderives F \rightarrow [ x \ceq \rp ] G$ by~\RuleName{assign}.
			
			\item 
			\label{itm:compl-case-nondet}
			$\vDash F \rightarrow [ x \ceq * ] G$, then $\vDash F \rightarrow \fa{x} G$ by \RuleName{nondetAssign}.
			Since the number of programs decreased,
			obtain $\fa{x} G \indOrder [ x \ceq \rp ] G$.
			Hence, $\comFODderives F \rightarrow \fa{x} G$ by IH.
			Finally, $\comFODderives F \rightarrow [ x \ceq * ] G$ by \RuleName{nondetAssign}.

			\item 
			\label{itm:compl-test}
			$\vDash F \rightarrow [ \test{\chi} ] G$, then $\vDash F \rightarrow (\chi \rightarrow G)$ by \RuleName{test}.
			Since $\chi \rightarrow G$ has less programs, 
			obtain $(\chi \rightarrow G) \indOrder [ \test{\chi} ] G$,
			so $\comFODderives F \rightarrow (\chi \rightarrow G)$ by IH.
			Hence, $\comFODderives F \rightarrow [ \test{\chi} ] G$ by \RuleName{test}.

			\item 
			\label{itm:relativeSafety-send}
			$\vDash F \rightarrow [ \send{}{} ] \invac G$,
			then $\vDash F \rightarrow [ \test{\true} ] \invac [ \send{}{} ] [ \test{\true} ] \invac G$ by \RuleName{acCom}.
			Further, by \RuleName{send},
			the following formula is valid:
			\begin{equation*}
				\phi_0 \equiv F \rightarrow [ \test{\true} ] \invac \fa{\tvar_0} \big( 
					\tvar_0 = \recvar \cdot \comItem{\ch{}, \rp} \rightarrow ([ \test{\true} ] \invac G) \subs{\recvar}{\tvar_0} 
				\big) 
			\end{equation*}
			By \RuleName{acNoCom} and \RuleName{test}, 
			$[ \test{\true} ] \invac \lambda$ is provably equivalent to 
			$\Commit \wedge \lambda$
			for any formula $\lambda$.
			Hence, \wlossg~$\phi_0$ can be considered to contain less programs than $\phi$ such that $\phi_0 \indOrder \phi$.
			Therefore, $\comFODderives \phi_0$ by IH.
			Then $\comFODderives F \rightarrow [ \test{\true} ] \invac [ \send{}{} ] [ \test{\true} ] \invac G$ by \RuleName{send},
			and finally, $\comFODderives F \rightarrow [ \send{}{} ] \invac G$ by \RuleName{acCom}.

			\item $\vDash F \rightarrow [ \receive{}{} ] \invac G$,
			then assume $x \not\in F$ by bound variable renaming.
			Hence, $\vDash F \rightarrow \fa{x} [ \send{}{x} ] \invac G$ 
			by \RuleName{comDual} and \RuleName{nondetAssign}.
			Since $x \not\in F$,
			obtain $\vDash F \rightarrow [ \send{}{x} ] \invac G$.
			Then $\comFODderives F \rightarrow [ \send{}{x} ] \invac G$ 
			by \rref{itm:relativeSafety-send},
			which yields $\comFODderives \fa{x} (F \rightarrow  [ \send{}{x} ] \invac G)$ by \ForallGen.
			Hence, $\comFODderives \fa{x} F \rightarrow \fa{x} [ \send{}{x} ] \invac G$ by \RuleName{faDist},
			so $\comFODderives F \rightarrow \fa{x} [ \send{}{x} ] \invac G$ by \RuleName{vacuousQuanti}.
			Finally, $\comFODderives F \rightarrow [ \receive{}{} ] \invac G$ derives by \RuleName{nondetAssign} and \RuleName{comDual}.

			\item 
			\label{itm:compl-case-ode}
			$\vDash F \rightarrow \dbleft \evolution*{}{} \dbright \invac G$,
			then the evolution domain constraint~$\chi$ can be eliminated,
			since it is definable in FOD \cite[Lemma 5]{DBLP:journals/jar/Platzer08}.
			Hence, the remainder focuses on $\vDash F \rightarrow \dbleft \evolution*{}{non} \dbright \invac G$.
			Then obtain $\vDash F \rightarrow \fa{\tvar_0{=}\recvar} \fa{t{=}0} [ \evolution*{x' = \rp, t' = 1}{non} ] \tilde{G}$
			by~\RuleName{odeSync},
			where $\tilde{G} \equiv \fa{\recvar{=}\tvar_0} [ \send{\durchan}{t} ] \invac G$.
			Since $\tvar_0, t$ are fresh,
			$\vDash \tilde{F} \rightarrow [ \evolution*{x' = \rp, t' = 1}{non} ] \tilde{G}$,
			where $\tilde{F} \equiv F \wedge \tvar_0 = \recvar \wedge t = 0$.
			By \rref{lem:com-fod-expressiveness},
			there are $\toComFOD{\tilde{F}},\toComFOD{\tilde{G}} \in \comFOD$ such that $\vDash \tilde{F} \leftrightarrow \toComFOD{\tilde{F}}$ and $\vDash \tilde{G} \leftrightarrow \toComFOD{\tilde{G}}$.
			Since~$\tilde{F}$ and $\tilde{G}$ have less modalities than $\phi$
			when executing $\send{\durchan}{t}$ in $\tilde{G}$ as in case \rref{itm:relativeSafety-send},
			and $\toComFOD{\tilde{F}}, \toComFOD{\tilde{G}} \in \comFOD$,
			obtain $(\tilde{F} \rightarrow \toComFOD{\tilde{F}}) \indOrder \phi$ and $(\toComFOD{\tilde{G}} \rightarrow \tilde{G}) \indOrder \phi$.
			Hence, $\comFODderives \tilde{F} \rightarrow \toComFOD{\tilde{F}}$ and $\comFODderives \toComFOD{\tilde{G}} \rightarrow \tilde{G}$ by IH.
			Further, $\toComFOD{\tilde{F}} \rightarrow \dbleft \evolution*{}{non} \dbright \toComFOD{\tilde{G}}$ is a valid \comFOD formula,
			so $\comFODderives \toComFOD{\tilde{F}} \rightarrow \dbleft \evolution*{}{non} \dbright \toComFOD{\tilde{G}}$.
			This combines with $\comFODderives \toComFOD{\tilde{G}} \rightarrow \tilde{G}$ to $\comFODderives \toComFOD{\tilde{F}} \rightarrow \dbleft \evolution*{}{non} \dbright \tilde{G}$ by monotonicity \RuleName{acMono} and~\RuleName{acDiaMono},
			which combines with $\comFODderives \tilde{F} \rightarrow \toComFOD{\tilde{F}}$ to $\comFODderives \tilde{F} \rightarrow \dbleft \evolution*{}{non} \dbright \tilde{G}$ using \RuleName{MP}.
			Then $\comFODderives \fa{\tvar_0}\fa{t} (\tilde{F} \rightarrow \dbleft \evolution*{}{non} \dbright \tilde{G})$ by~\ForallGen.
			Hence, $\comFODderives \fa{\tvar_0}\fa{t} F \rightarrow \fa{\tvar_0{=}\recvar} \fa{t{=}0} \dbleft \evolution*{}{non} \dbright \tilde{G}$ 
			by~\RuleName{faDist},
			so $\comFODderives F \rightarrow \fa{\tvar_0{=}\recvar} \fa{t{=}0} \dbleft \evolution*{}{non} \dbright \tilde{G}$ by~\RuleName{vacuousQuanti},
			using that~$\tvar_0, t$ are fresh.
			Finally, $\comFODderives F \rightarrow \dbleft \evolution*{}{non} \dbright \invac G$ by \RuleName{odeSync}.

			\item $\vDash F \rightarrow [ \alpha \seq \beta ] \invac G$, then $\vDash F \rightarrow [ \alpha ] \invac [ \beta ] \invac G$ by \RuleName{acComposition}.
			Since $\alpha$ and~$\beta$ are simpler than $\alpha\seq\beta$,
			obtain $[ \alpha ] \invac [ \beta ] \invac G \indOrder [ \alpha\seq\beta ] \invac G$.
			Note
			that $[ \alpha ] \invac [ \beta ] \invac G$ is smaller in $\indOrder$, 
			even though the number of modalities increased,
			because the overall structural complexity of the programs got simpler by removing the sequential composition.
			Hence, $\comFODderives F \rightarrow [ \alpha ] \invac [ \beta ] \invac G$ by IH.
			Finally, $\comFODderives F \rightarrow [ \alpha\seq\beta ] \invac G$ by \RuleName{acComposition}.

			\item $\vDash F \rightarrow [ \alpha \cup \beta ] \invac G$,
			then $\vDash F \rightarrow [ \alpha ] \invac G \wedge [ \beta ] \invac G$
			using \RuleName{acChoice}.
			Since $\alpha$ and~$\beta$ are simpler than $\alpha\cup\beta$,
			obtain $([ \alpha ] \invac G \wedge [ \beta ] \invac G) \indOrder [ \alpha\cup\beta ] \invac G$.
			Hence, $\comFODderives F \rightarrow [ \alpha ] \invac G \wedge [ \beta ] \invac G$ by IH.
			Finally, $\comFODderives F \rightarrow [ \alpha \cup \beta ] \invac G$ by \RuleName{acChoice}.

			\item $\vDash F \rightarrow [ \repetition{\alpha} ] \invac G$,
			then by \rref{lem:com-fod-expressiveness},
			there is an \comFOD formula $\inv \equiv \toComFOD{([ \repetition{\alpha} ] \invac G)}$,
			which is equivalent to $[ \repetition{\alpha} ] \invac G$ .
			The formula $\inv$ is a sufficient invariant for~$\repetition{\alpha}$,
			because the following formulas derive in $\comFODderives$:
			\begin{enumerate}
				\item By \acCommit and totality 
				of the program semantics,
				\iest $(\pstate{v}, \epsilon, \bot) \in \sem{\repetition{\alpha}}{}$ for every state~$\pstate{v}$,
				obtain $\vDash [ \repetition{\alpha} ] \invac G \rightarrow \Commit$,
				so $F \rightarrow \Commit \wedge \inv$ is valid.
				Since $\inv\in\comFOD$ and $\Commit$ has less programs than $[ \repetition{\alpha} ] \invac G$,
				obtain $(\Commit \wedge \inv) \indOrder [ \repetition{\alpha} ] \invac G$,
				so $\comFODderives F \rightarrow \Commit \wedge \inv$ by IH.
				
				\item By~\RuleName{acIteration}, $\inv \rightarrow [ \alpha ] \invac \inv$ is valid.
				Since $(\inv \rightarrow [ \alpha ] \invac \inv) \indOrder [ \repetition{\alpha} ] \invac G$,
				because $\inv\in\comFOD$ and $\alpha$ is simpler than~$\repetition{\alpha}$,
				obtain $\comFODderives \inv \rightarrow [ \alpha ] \invac \inv$ by IH.
				
				\item By \RuleName{acIteration} again,
				$\vDash \inv \rightarrow [ \alpha^0 ] \invac G$,
				thus $\vDash \inv \rightarrow G$ 
				by~\RuleName{acNoCom} and \RuleName{test} as $\alpha^0 \equiv \test{\true}$.
				Since $\inv\in\comFOD$, and 
				$G$ has
				less programs than $[ \repetition{\alpha} ] \invac G$,
				obtain 
				$(\inv \rightarrow G) \indOrder [ \repetition{\alpha} ] \invac G$.
				Hence, $\comFODderives \inv \rightarrow G$ by IH.
			\end{enumerate}
			\vspace{.5em}	

			The following prooftree combines all observations 
			to a derivation of $F \rightarrow [ \repetition{\alpha} ] \invac G$ in $\comFODderives$,
			using the derivable induction rule \RuleName{acInvariant}:
			\begin{small}
				\begin{prooftree}
						\Axiom{$F \rightarrow \Commit \wedge \inv$}

							\Axiom{$\inv \rightarrow [ \alpha ] \invac \inv$}

							\RuleNameLeft{acInvariant}{}
							\UnaryInf{$\Commit \wedge \inv \rightarrow [ \repetition{\alpha} ] \invac \inv$}

							\Axiom{$\inv \rightarrow G$}

							\UnaryInf{$(\Commit \rightarrow \Commit) \wedge (\inv \rightarrow G)$}
	
							\RuleNameRight{acG}{}
							\UnaryInf{$[ \repetition{\alpha} ] \acpair{\true,\Commit\rightarrow\Commit} (\inv \rightarrow G)$} 

							\RuleNameRight{acModalMP}{}
							\UnaryInf{$[ \repetition{\alpha} ] \invac \inv \rightarrow [ \repetition{\alpha} ] \invac G$}
						
						\RuleNameRight{MP}{}
						\BinaryInf{$\Commit \wedge \inv \rightarrow [ \repetition{\alpha} ] \invac G$}
					
					\RuleNameRight{MP}{}
					\BinaryInf{$F \rightarrow [ \repetition{\alpha} ] \invac G$}
				\end{prooftree}
			\end{small}

			\item 
			\label{itm:relativeLiveness-iterationSpecial}
			In the special case $\vDash F \rightarrow \langle \repetition{\alpha} \rangle \acpair{\true, \false} G$,
			where the commitment $\false$ is unsatisfiable,  
			$\comFODderives F \rightarrow \langle \repetition{\alpha} \rangle \acpair{\true, \false} G$ derives by a generalization of an argument for \dL \cite{DBLP:journals/jar/Platzer08} to assumption-program pairs as modal actions.
			The general case is in case \ref{itm:compl-case-rep}.
			The variant $\varphi(v)$ for the convergence axiom \RuleName{acConvergence} is defined by combining the \comFOD representation $\toComFOD{(\langle \repetition{\alpha} \rangle \acpair{\true, \false} G)}$ (\rref{lem:com-fod-expressiveness}) and the %
			rendition (\rref{lem:rendition}) of the repetition $\repetition{\alpha}$.
			Since only runs to final states are relevant, 
			the predicate~$\finmarker{}$ is set to~$\true$ in the rendition of $\repetition{\alpha}$ and the formula is simplified accordingly.
			\begin{align*}
				\varphi(n-1) \equiv 
				\ex{\varvec[alt]} \Big( 
					& 
					(\toComFOD{G}) \subs{\varvec}{\varvec[alt]} 
					\wedge \isNat{n}
					\wedge
					\ex{\goedelat{\repstatevar}{n}{}}
					\big(
						\repstate{1} \!= \varvec
						\wedge \repstate{n} \!= \varvec[alt] \\
						&\qquad \wedge \fa{i{:}\naturals} (
							\orange{1}{i}{n} 
							\rightarrow 
							\reninter{i}{i+1}{\true} )
					\big)
				\Big)
			\end{align*}
			The variant $\varphi(v)$ expresses that if $\varphi(v)$ is satisfied in an initial state $\varvec$,
			where $\varvec$ are the variables of $\langle \repetition{\alpha} \rangle \acpair{\true, \false} G$,
			then a final state $\varvec[alt]$ satisfying $G$ is reachable by 
			an $\repetition{\alpha}$-run
			in~$v$ iterations.
			Moreover, observe that $\varphi(v) \indOrder \phi$ since $\varphi(v) \in \comFOD$.
			Then the following formulas derive in $\comFODderives$:
			\begin{enumerate}
				\item $\phi_0 \equiv \ex{v}\varphi(v) \rightarrow \langle \repetition{\alpha} \rangle \acpair{\true, \false} \ex{v{\le}0} \varphi(v)$:
				If $\varphi(v)$ is satisfied for some $v$,
				by definition of $\varphi(v)$,
				a final state $\varvec[alt]$ satisfying~$G$ is reachable by an 
				$\repetition{\alpha}$-run
				in~$v$ iterations.
				Hence, if $v > 0$,
				after one 
				$\alpha$-run,
				this final state is already reachable in $v-1$ iterations  
				such that $\chi \equiv v>0 \wedge \varphi(v) \rightarrow \langle \alpha \rangle \acpair{\true, \false} \varphi(v-1)$ is valid.
				Since $\chi \indOrder \phi$,
				because $\varphi(v) \in \comFOD$ and $\alpha$ is simpler than its repetition $\repetition{\alpha}$,
				obtain $\comFODderives \chi$ by IH.
				Then $\comFODderives [ \repetition{\alpha} ] \acpair{\true, \true} \fa{v{>}0} (\varphi(v) \rightarrow \langle \alpha \rangle \acpair{\true, \false} \varphi(v-1))$ by \ForallGen and Gödel generalization~\RuleName{acG}.

				Further, $\comFODderives \fa{v} (\varphi(v) \rightarrow \langle \repetition{\alpha} \rangle \acpair{\true, \false} \ex{v{\le}0} \varphi(v))$ by convergence \RuleName{acConvergence}.
				Hence, $\comFODderives \fa{v} \varphi(v) \rightarrow \langle \repetition{\alpha} \rangle \acpair{\true, \false} \ex{v{\le}0} \varphi(v)$ by \RuleName{faDist} and \RuleName{uniInstance} as $v$ is fresh.
				This yields $\comFODderives \phi_0$ using \RuleName{MP} because $\comFODderives \fa{v} \varphi(v) \rightarrow \ex{v} \varphi(v)$ by \RuleName{uniInstance}.

				\item $\phi_1 \equiv F \rightarrow \ex{v} \varphi(v)$ is valid by definition of $\varphi(v)$ because $F \rightarrow \langle \repetition{\alpha} \rangle \acpair{\true, \false} G$ is valid.
				Moreover, $\ex{v} \varphi(v) \indOrder \langle \repetition{\alpha} \rangle \acpair{\true, \false} G$ since $\ex{v} \varphi(v) \in \comFOD$.
				Hence, $\comFODderives F \rightarrow \ex{v} \varphi(v)$ by IH.
				
				\item $\phi_2 \equiv \langle \repetition{\alpha} \rangle \acpair{\true, \false} \ex{v{\le}0} \varphi(v) \rightarrow \langle \repetition{\alpha} \rangle \acpair{\true, \false} G$ derives in $\comFODderives$ from $\ex{v{\le}0} \varphi(v) \rightarrow G$ by monotonicity \RuleName{acDiaMono},
				and $\ex{v{\le}0} \varphi(v) \rightarrow G$ derives as follows:
				First, $(\ex{v{\le}0} \varphi(v) \rightarrow G) \indOrder \phi$ since $\ex{v{\le}0} \varphi(v) \in \comFOD$ and~$G$ has less programs than $\phi$.
				Moreover, $\ex{v{\le}0} \varphi(v) \rightarrow G$ is valid because if $\ex{v{\le}0} \varphi(v)$ holds,
				then $\varphi(v)$ is satisfied for some $v\le0$,
				and even $v=0$ as $\varphi(v)$ only holds for natural numbers.
				Then $\varphi(0)$ implies~$G$ by definition of $\varphi(v)$.
				Hence, $\comFODderives \ex{v{\le}0} \varphi(v) \rightarrow G$ by IH.
			\end{enumerate}
			Now, combine $\comFODderives \phi_0$ and $\comFODderives \phi_1$ by \RuleName{MP} and propositional reasoning into
			$\comFODderives F \rightarrow \langle \repetition{\alpha} \rangle \acpair{\true, \false} \ex{v{\le}0} \varphi(v)$.
			The latter and $\comFODderives \phi_2$ combine 
			into $\comFODderives F \rightarrow \langle \repetition{\alpha} \rangle \acpair{\true, \false} G$ 
			by \RuleName{MP} and propositional reasoning again.

			\item 
			\label{itm:compl-case-rep}
			In the general case $\vDash F \rightarrow \langle \repetition{\alpha} \rangle \invac G$,
			either $\vDash F \rightarrow \langle \alpha^0 \rangle \invac G$ or $\vDash F \rightarrow \langle \repetition{\alpha} \rangle \acpair{\true, \false} \tilde{G}$ by the derivable axiom \RuleName{acArrival},
			where $\tilde{G} \equiv \neg G \vee \langle \alpha \rangle \invac G$.
			Since $\langle \alpha^0 \rangle \ac G \indOrder \langle \repetition{\alpha} \rangle \invac G$,
			because $\alpha^0 \equiv \test{\true}$ is simpler than the repetition $\repetition{\alpha}$,
			obtain $\comFODderives F \rightarrow \langle \alpha^0 \rangle \invac G$ by IH if $\vDash F \rightarrow \langle \alpha^0 \rangle \ac G$.
			Otherwise, if $\vDash F \rightarrow \langle \repetition{\alpha} \rangle \acpair{\A, \false} \tilde{G}$,
			then $\comFODderives F \rightarrow \langle \repetition{\alpha} \rangle \acpair{\A, \false} \tilde{G}$ derives by \rref{itm:relativeLiveness-iterationSpecial}.
			In summary, $\comFODderives F \rightarrow \langle \alpha^0 \rangle \invac G \vee \langle \repetition{\alpha} \rangle \acpair{\A, \false} \tilde{G}$ by \RuleName{MP} and propositional reasoning,
			which yields $\comFODderives F \rightarrow \langle \repetition{\alpha} \rangle \invac G$ by axiom \RuleName{acArrival}.

			\item
			$\vDash F \rightarrow [ \alpha \parOp \beta ] \acpair{\true, \Commit} G$,
			then the proof crucially uses that strongest promises are expressible in \comFOD (\rref{sec:par-safety-lemmas})
			and that they can be decomposed for parallel composition
			into the strongest promises of the subprograms.
			For a detailed proof outline, refer to the beginning of \rref{sec:com-fod-completeness}.
			The proof uses the following abbreviations,
			where $\rvarvec[alt]_\gamma, \tvar_0$ are fresh,
			and throughout let $\otherprog{\alpha}\equiv\beta$ and $\otherprog{\beta}\equiv\alpha$:

			$\begin{aligned}
				& \rvarvec_\gamma = \SV(\gamma)\cap\RVar
				&& \cset_\gamma = 
				\SCN(\gamma) \cup \{\durchan\} \cup \SCN(\ogamma)^\complement \\
				& \rvarvec = \rvarvec_\alpha \cup \rvarvec_\beta 
				&& F_\gamma \equiv 
					((F {}_{\smash{\recvar}\vphantom{\rvarvec_\beta}}^{\tvar_0}) 
						{}_{\smash{\rvarvec_\alpha}\vphantom{\rvarvec_\beta}}^{\smash{\rvarvec[alt]_\alpha}\vphantom{\tvar_0}}) 
						{}_{\smash{\rvarvec_\beta}\vphantom{\rvarvec_\beta}}^{\smash{\rvarvec[alt]_\beta}\vphantom{\tvar_0}} 
					\wedge \rvarvec[alt]_\gamma = \rvarvec_\gamma \wedge \recvar \downarrow \cset_\gamma = \tvar_0 \downarrow \cset_\gamma \\
				& \rvarvec[alt] = \rvarvec[alt]_\alpha \cup \rvarvec[alt]_\beta
				&& F_0 \equiv 
					((F {}_{\smash{\recvar}\vphantom{\rvarvec_\beta}}^{\tvar_0}) 
						{}_{\smash{\rvarvec_\alpha}\vphantom{\rvarvec_\beta}}^{\smash{\rvarvec[alt]_\alpha}\vphantom{\tvar_0}}) 
						{}_{\smash{\rvarvec_\beta}\vphantom{\rvarvec_\beta}}^{\smash{\rvarvec[alt]_\beta}\vphantom{\tvar_0}}  
					\wedge \rvarvec[alt]= \rvarvec
					\wedge \tvar_0 = \recvar
			\end{aligned}$\\[.4em]
			
			\vspace{.5em}
			The precondition~$F_0$ freezes the initial state in $F$ 
			using the fresh variables $\tvar_0, \rvarvec[alt]_\gamma$
			such that~$F_0$ can be split into preconditions~$F_\alpha$ and $F_\beta$, 
			where the bound variables of the other subprogram are not free
			and the written channels are not accessed via the recorder~$\recvar$,
			except if they are shared.
			Hence, $\ogamma$ does not interfere (\rref{def:noninterference}) with $\nointfpair{\gamma}{F_\gamma}$.

			As the first step of the derivation,
			embed safety of each subprogram for its strongest promises into the parallel composition by parallel injection \RuleName{acDropComp}.
			For each $\gamma \in \{\alpha,\beta\}$, let $\strgCabbrev{\gamma} \equiv \fa{\rvarvec[alt]{=}\rvarvec} \strgCommit{\otherprog{\gamma}}{F_\gamma}{}{\gamma}$ and $\strgPabbrev{\gamma} \equiv \strgPost{\otherprog{\gamma}}{F_\gamma}{}{\gamma}$ be the strongest promises (\rref{lem:state_variations_comfod}) of~$\gamma$ \wrt the precondition~$F_\gamma$ and the environment $\ogamma$.
			Since~$\ogamma$ does not interfere (\rref{def:noninterference}) with $\nointfpair{\gamma}{F_\gamma}$,
			it neither interferes
			with $[ \gamma ] \acpair{\true, \strgCabbrev{\gamma}} \strgPabbrev{\gamma}$ by \rref{lem:state_variations_comfod}.
			Since $\strgCommit{\ogamma}{\varphi}{}{}$ and $\strgPost{\ogamma}{\varphi}{}{}$ are \comFOD formulas (\rref{lem:state_variations_comfod}),
			the premise $\phi_\gamma \equiv F_\gamma \rightarrow [ \gamma ] \acpair{\true, \strgCabbrev{\gamma}} \strgPabbrev{\gamma}$ 
			has less parallel compositions with a nesting depth greater or equal to $\alpha\parOp\beta$.
			This reduces the overall structural complexity of programs,
			so $\phi_\gamma \indOrder (F \rightarrow [ \alpha \parOp \beta ] \ac G)$.
			Hence, $\comFODderives \phi_\gamma$ by IH because~$\vDash \phi_\gamma$ by \rref{lem:reachable_states_correct}.
			The \comFOD formula $\tvar_0 = \recvar \rightarrow \recvar \downarrow \cset_\gamma = \tvar_0 \downarrow \cset_\gamma$ is valid,
			so derives in $\comFODderives$.
			Hence, the premise $\triangleleft_\gamma \equiv F_0 \rightarrow F_\gamma$ derives in $\comFODderives$ essentially by~\RuleName{MP}.
			\vspace{-.2cm}
			\begin{prooftree}		
						\Axiom{$\triangleleft_\alpha$}

						\Axiom{\rref{lem:reachable_states_correct} + IH}

						\UnaryInf{$F_\alpha \rightarrow [ \alpha ] \acpair{\true, \strgCabbrev{\alpha}} \strgPabbrev{\alpha}$}

						\RuleNameRight{acDropComp}{}
						\UnaryInf{$F_\alpha \rightarrow [ \alpha \parOp \beta ] \acpair{\true, \strgCabbrev{\alpha}} \strgPabbrev{\alpha}$}

					\RuleNameRight{MP}{}
					\BinaryInf{$F_0 \rightarrow [ \alpha \parOp \beta ] \acpair{\true, \strgCabbrev{\beta}} \strgPabbrev{\beta}$}

						\Axiom{$\triangleleft_\beta$}

						\Axiom{\rref{lem:reachable_states_correct} + IH}

						\UnaryInf{$F_\beta \rightarrow [ \beta ] \acpair{\true, \strgCabbrev{\beta}} \strgPabbrev{\beta}$}

						\RuleNameRight{acDropComp}{}
						\UnaryInf{$F_\beta \rightarrow [ \alpha \parOp \beta ] \acpair{\true, \strgCabbrev{\beta}} \strgPabbrev{\beta}$}

					\RuleNameRight{MP}{}
					\BinaryInf{$F_0 \rightarrow [ \alpha \parOp \beta ] \acpair{\true, \strgCabbrev{\beta}} \strgPabbrev{\beta}$}

				\RuleNameRight{MP}{}
				\BinaryInf{$F_0 \rightarrow [ \alpha \parOp \beta ] \acpair{\true, \strgCabbrev{\alpha}} \strgPabbrev{\alpha} \wedge [ \alpha \parOp \beta ] \acpair{\true, \strgCabbrev{\beta}} \strgPabbrev{\beta}$}

				\RuleNameRight{acBoxesDist}{}
				\UnaryInf{$F_0 \rightarrow [ \alpha \parOp \beta ] \acpair{\true, \strgCabbrev{\alpha} \wedge \strgCabbrev{\beta}} (\strgPabbrev{\alpha} \wedge \strgPabbrev{\beta})$}
			\end{prooftree}
			\vspace{.2cm}

			As the second step of the derivation,
			continuing the proof tree above,
			combine the strongest promises $\strgCabbrev{\alpha} \wedge \strgCabbrev{\beta}$ and $\strgPabbrev{\alpha} \wedge \strgPabbrev{\beta}$ for the subprograms to the strongest promises of 
			$\alpha\parOp\beta$.
			Let $\scSymb \equiv \fa{\rvarvec[alt]{=}\rvarvec} \strgCommit{\emptyset}{F_0}{}{\alpha\parOp\beta}$ and $\spSymb \equiv \strgPost{\emptyset}{F_0}{}{\alpha\parOp\beta}$ be the strongest promises 
			of \mbox{$\alpha\parOp\beta$} \wrt the precondition~$F_0$.
			Since $\vDash F_0 \rightarrow F$,
			obtain $\vDash F_0 \rightarrow [ \alpha\parOp\beta ] \acpair{\true, \Commit} G$.
			Hence, $\scSymb \rightarrow \Commit$ and $\spSymb \rightarrow G$ are valid by \rref{lem:reachable_states_correct}.
			By \rref{lem:ac_biggest_promise_split}, 
			the strongest promises~$\strgCabbrev{\gamma}$ and~$\strgPabbrev{\gamma}$ for the subprograms exactly demarcate the reachable states of $\alpha\parOp\beta$ 
			when combined with history invariance $H \equiv \recvar \succeq \tvar_0$ by~\RuleName{hExtension} to guarantee a linear history.
			In summary, 
			by \rreflems{lem:reachable_states_correct}{lem:ac_biggest_promise_split},
			the following formulas are valid: 
			\begin{align*}
				\qquad\qquad\triangleright_\scSymb &\equiv \strgCabbrev{\alpha} \wedge \strgCabbrev{\beta} \wedge H \rightarrow \Commit 
				&&&
				\triangleright_\spSymb &\equiv \strgPabbrev{\alpha} \wedge \strgPabbrev{\beta} \wedge H \rightarrow G
			\end{align*}
			Since $\strgCabbrev{\gamma}, \strgPabbrev{\gamma}$ only contain the program $\gamma$,
			the premises $\triangleright_\scSymb$ and~$\triangleright_\spSymb$
			have less parallel compositions with a nesting depth greater or equal to 
			\mbox{$\alpha\parOp\beta$ than $\phi \equiv F \rightarrow [ \alpha\parOp\beta ] \acpair{\true, \Commit} G$}.
			Hence, the overall structural complexity of programs 
			decreased,
			so $\triangleright_\scSymb \indOrder \phi$ and $\triangleright_\spSymb \indOrder \phi$.
			Since $\triangleright_\scSymb$ and $\triangleright_\spSymb$ are valid, $\comFODderives \triangleright_\scSymb$ and $\comFODderives \triangleright_\spSymb$ by IH.
			The premise $\triangleright_0 \equiv F_0 \rightarrow \tvar_0 = \recvar$ derives in $\comFODderives$ propositionally.
			\vspace{-.2cm}
			\begin{prooftree}
					\Axiom{see proof tree above}

					\UnaryInf{$F_0 \rightarrow [ \alpha \parOp \beta ] \acpair{\true, \strgCabbrev{\alpha} \wedge \strgCabbrev{\beta}} (\strgPabbrev{\alpha} \wedge \strgPabbrev{\beta})$}
				
					\Axiom{$*$}

					\RuleNameRight{hExtension}{}
					\UnaryInf{$\tvar_0 = \recvar \rightarrow [ \alpha \parOp \beta ] \acpair{\true, H} H$}

					\RuleNameRight{MP}{, $\triangleright_0$}
					\UnaryInf{$F_0 \rightarrow [ \alpha \parOp \beta ] \acpair{\true, H} H$}

				\RuleNameRight{MP}{}
				\BinaryInf{$F_0 \rightarrow [ \alpha \parOp \beta ] \acpair{\true, \strgCabbrev{\alpha} \wedge \strgCabbrev{\beta}} (\strgPabbrev{\alpha} \wedge \strgPabbrev{\beta}) \wedge [ \alpha \parOp \beta ] \acpair{\true, H} H$}	

				\RuleNameRight{acBoxesDist}{}
				\UnaryInf{$F_0 \rightarrow [ \alpha \parOp \beta ] \acpair{\true, \strgCabbrev{\alpha} \wedge \strgCabbrev{\beta} \wedge H}  (\strgPabbrev{\alpha} \wedge \strgPabbrev{\beta} \wedge H )$}

				\RuleNameRight{acMono}{, $\triangleright_\scSymb$, $\triangleright_\spSymb$}
				\UnaryInf{$F_0 \rightarrow [ \alpha \parOp \beta ] \acpair{\true,  \Commit} G$}

					\RuleNameRight{iG, subs}{}
					\UnaryInf{$F \rightarrow [ \alpha \parOp \beta ] \acpair{\true, \Commit} G$}
			\end{prooftree}
			\vspace{.2cm}

			This concludes the case $F \rightarrow [ \alpha\parOp\beta ] \invac G$,
			because the proof trees above combine to a derivation of $F \rightarrow [ \alpha\parOp\beta ] \invac G$
			in $\comFODderives$.

			\item 
			\label{itm:par-liveness-special}
			$\vDash F \rightarrow \langle \alpha \parOp \beta \rangle \invac G$,
			then $F \rightarrow \langle \alpha\parOp\beta \rangle \invac \false$ or $F \rightarrow \langle \alpha\parOp\beta \rangle G$
			is valid by~\RuleName{acSplitDia} and \RuleName{diasDual}.
			Hence, $\comFODderives F \rightarrow \langle \alpha \parOp \beta \rangle \invac G$
			derives by a case distinction whose branches use either \RuleName{acLiveParCommit} or \RuleName{acLivePar}.
			Recall the premises $\phi_\Commit$ and $\phi_G$ of \RuleName{acLiveParCommit} and \RuleName{acLivePar}, respectively,
			where the diamonds $\runcommit{\gamma}{\auxhis\downarrow\gamma}$ and $\runpost{\gamma}{}{\auxhis\downarrow\gamma}{M} \psi$
			with trace constraint $\auxhis\downarrow\gamma$
			are defined as in \rref{fig:calculus}.
			\begin{align*}
				\phi_\Commit & \equiv \acLiveParCommitRaw{\Commit\subs{\recvar}{\tvar_0\cdot\auxhis}}
				\label{eq:conv_commit} \\
				\phi_G &\equiv \acLiveParRaw{\psi\subs{\recvar}{\tvar_0\cdot\auxhis}}
			\end{align*}
			Then $F \rightarrow \phi_\Commit$ or $F \rightarrow \phi_G$ is valid,
			because the backward implication for \RuleName{acLiveParCommit} and \RuleName{acLivePar} holds as well (see their soundness proof in \rref{app:soundness}).

			Since $\alpha\parOp\beta$ is decomposed into $\alpha$ and $\beta$,
			the formulas $\phi_\Commit$ and $\phi_G$ have less parallel compositions with a nesting depth greater or equal to $\alpha\parOp\beta$.
			Hence, the overall structural complexity of the programs in $\phi_\Commit$ and $\phi_G$ is less than in $\langle \alpha\parOp\beta \rangle \invac G$,
			so $(F \rightarrow \phi_\Commit) \indOrder \phi$ and $(F \rightarrow \phi_G) \indOrder \phi$.
			Thus, $\vDash F \rightarrow \phi_\Commit$ implies $\comFODderives F \rightarrow \phi_\Commit$ by IH,
			and $\vDash F \rightarrow \phi_G$ implies $\comFODderives F \rightarrow \phi_G$ by IH.
			Finally, $\comFODderives F \rightarrow \langle \alpha\parOp\beta \rangle \invac G$
			derives by the left prooftree below if $\vDash F \rightarrow \langle \alpha\parOp\beta \rangle \invac \false$
			and by the right prooftree below if $\vDash F \rightarrow \langle \alpha\parOp\beta \rangle G$:
			
			\noindent\hspace*{\dimexpr\linewidth-\textwidth\relax}%
			\begin{center}
				\vspace*{-.3cm}
				\hspace{-.5cm}\begin{minipage}[t]{.48\textwidth}
					\begin{prooftree}[shape=justified, JustifiedScoresWidth=.85\textwidth]
						\Axiom{$F \rightarrow \phi_\Commit$}

						\RuleNameRight{acLiveParCommit}{}
						\UnaryInf{$F \rightarrow \langle \alpha\parOp\beta \rangle \acpair{\true, \Commit} \false$}

						\RuleNameRight{MP}{}
						\UnaryInf{$F \rightarrow \langle \alpha\parOp\beta \rangle \acpair{\true, \Commit} \false \vee \langle \alpha\parOp\beta \rangle \acpair{\true, \false} G$}
						
						\RuleNameRight{acSplitDia}{}
						\UnaryInf{$F \rightarrow \langle \alpha\parOp\beta \rangle \invac G$}
					\end{prooftree}
				\end{minipage}\hfill
				\begin{minipage}[t]{.48\textwidth}
					\begin{prooftree}[shape=justified, JustifiedScoresWidth=.85\textwidth]
						\Axiom{$F \rightarrow \phi_G$}

						\RuleNameRight{acLivePar}{}
						\UnaryInf{$F \rightarrow \langle \alpha\parOp\beta \rangle G$}

						\RuleNameRight{diasDual}{}
						\UnaryInf{$F \rightarrow \langle \alpha\parOp\beta \rangle \acpair{\true, \false} G$}

						\RuleNameRight{MP}{}
						\UnaryInf{$F \rightarrow \langle \alpha\parOp\beta \rangle \acpair{\true, \Commit} \false \vee \langle \alpha\parOp\beta \rangle \acpair{\true, \false} G$}
						
						\RuleNameRight{acSplitDia}{}
						\UnaryInf{$F \rightarrow \langle \alpha\parOp\beta \rangle \invac G$}
					\end{prooftree}
				\end{minipage}
			\end{center}
		\end{enumerate}
	\end{enumerate}
	This concludes the proof,
	since the cases above cover all \dLCHP formulas.
\end{proof}

Relative completeness (\rref{thm:com-fod-completeness}) confirms that \dLCHP provides a comprehensive characterization of all multi-dynamical aspects of parallel hybrid systems.
The proof itself further substantiates the careful axiom design:
Except for $\alpha\parOp\beta$,
the proof is reminiscent of completeness proofs for \dL \cite{DBLP:journals/jar/Platzer17,DBLP:conf/lics/Platzer12b} and \dGL \cite{DBLP:journals/tocl/Platzer15}.
Proof structures from~\dGL and \dL generalize to \dLCHP because \dLCHP stays close to the Pratt-Segerberg axioms \cite{Pratt1976,Segerberg1982} such that ac-reasoning causes minor overhead for previous arguments for $\alpha\seq\beta$, $\alpha\cup\beta$, and $\repetition{\alpha}$.
We expect that the convergence axiom \RuleName{acConvergence} is not necessary in a uniform substitution calculus for \dLCHP, as in the case of~\dL~\cite{DBLP:journals/tocl/Platzer15,DBLP:journals/jar/Platzer17}.
For the modal view onto ac-reasoning,
it is reassuring that convergence has a proper ac-generalization.

\rref{thm:com-fod-completeness} proves $[ \alpha\parOp\beta ]$ based on a conservative enumeration of all reachable states in the parallel product space by the strongest promises \cite{deRoever2001,AcSemantics_Zwiers}.
Unlike in Hybrid Hoare-logics \cite{Liu2010,Wang2012,Guelev2017},
this enumeration is not an inherent feature of the proof calculus
but expressible whenever necessary for completeness.
This is why \dLCHP proofs can use coarse mutual abstractions of the parallel dynamics that mitigate the state space explosion by compositional reduction.
For $\dbleft \alpha\parOp\beta \dbright$,
the assumption is applied to the parallel product using axiom~\RuleName{Atransfer}.
This addresses global assumptions,
which do not distribute to the subprograms, 
and avoids a fixed-point computation to find mutually sufficient assumptions and commitments for $[ \alpha\parOp\beta ]$.
Consequently, completeness does not need assumption weakening \RuleName{Aweak},
much as completeness for Hoare-style ac-reasoning \cite[Section 7.5.5]{deRoever2001} does not use the compositionality condition,
but \RuleName{Aweak} exactly identifies which underlying principle is unnecessary for completeness.
The \dLCHP calculus includes \RuleName{Aweak} because it guarantees schematic derivability of mutual abstractions,
which is imperative for the compositional state space reduction
by local reasoning about parallel program effects.

The proof of \rref{thm:com-fod-completeness} reduces $\dbleft \alpha\parOp\beta \dbright$ to \dLCHP formulas 
characterizing environmental interleaving locally from reachability $\langle\cdot\rangle$ for the subprograms
instead of a global encoding of their transition relation based on \rref{lem:com-fod-expressiveness}.
This novel local reduction 
is possible due to an induction order,
which gives precedence to program decomposition even when the logical complexity grows.
Globally, encoding is only required for $\dbleft \alpha\parOp\beta \dbright$ when the subprograms do.
This reflects that the state space explosion does not increase the proof-theoretical complexity of safety $[ \alpha\parOp\beta ]$ beyond the subprograms,
but liveness $\langle \alpha\parOp\beta \rangle$ follows the duality that $\exists$ is proof-theoretically harder than~$\forall$~\cite{DBLP:journals/tocl/Platzer15},
as apparent in the axioms~\RuleName{acLiveParCommit} and~\RuleName{acLivePar}.
In fact, parallel composition can increase the complexity of $\langle\cdot\rangle$ by modeling Turing-complete two-counter machines \cite{Minsky1961} from one-counter machines.
If the subprograms do not need encoding 
(no $\repetition{\alpha}$, no $\evolution*{}{non}$), 
\dLCHP reduces $\dbleft \alpha\parOp\beta \dbright$ 
to its first-order fragment,
by static evaluation of the trace terms,
even to decidable~\cite{Tarski1951} real arithmetic.
Assuming $\langle \repetition{\alpha} \rangle \ac G$ possesses an encoding-free reduction using uniform substitution as in case of \dL \cite{DBLP:journals/tocl/Platzer15},
\dLCHP only needs encoding for $\evolution*{}{non}$, and $[ \repetition{\alpha} ]$, and $\exists$ just like \dL does~\cite{DBLP:journals/tocl/Platzer15}.

\subsection{Completeness Relative to FOD}
\label{sec:fod-completeness}

\newcommand{\ofconcat}{\odot}

\newcommand{\maxdursuffixop}{\trianglerighteq}
\newcommand{\maxdurprefixop}{\trianglelefteq}
\newcommand{\maxdursuffix}[2]{#1 \trianglerighteq \comItem{\durchan,#2}}
\newcommand{\maxdurprefix}[2]{\comItem{\durchan,#2} \trianglelefteq #1}

\newcommand{\encvar}[1]{#1^\flat}
\newcommand{\enclen}[1]{\len{\encvar{#1}}}

\newcommand{\traceenc}{\mathbb{T}}

\newcommand{\goedelize}[1]{\mathop{\mathcal{G}}(#1)}

\newcommand{\seqat}[2]{\at{#1}{#2}}
\newcommand{\seqitm}[2]{\seqat{(\encvar{#1})}{#2}}

\newcommand{\seqchan}[2]{\chan{\seqitm{#1}{#2}}}
\newcommand{\seqval}[2]{\val{\seqitm{#1}{#2}}}

\newcommand{\eventencseq}[1]{#1:\traceenc}

\newcommand{\faencvar}[1]{\fa{\encvar{#1}{:}\traceenc}}

The previous section proved that the \dLCHP calculus (\rref{fig:calculus}) is complete relative to \comFOD,
the first-order logic of differential equation properties (FOD) augmented with communication traces.
Based on this result (\rref{thm:com-fod-completeness}),
this section shows that the \dLCHP calculus can be extended to a complete axiomatization of parallel hybrid systems relative to FOD.
The reduction from \comFOD to FOD follows the idea of a provably correct equitranslation~\cite{AbouElWafa2024}:
We define an effective semantic translation 
from \comFOD to FOD, 
which inductively simplifies the operators on traces
and uses $\reals$-Gödel encodings~\cite{DBLP:journals/jar/Platzer08} to represent the simplified traces in FOD.
In an extension (\rref{fig:plusClaculus}) of the \dLCHP calculus (\rref{fig:calculus}),
we then prove syntactically that this translation establishes an equivalence (\rref{prop:com-fod-to-fod}).
Since \dLCHP is relatively complete for \comFOD,
completeness of \dLCHP relative to FOD then follows by transitivity.

Communication traces are expressible in FOD
by compressing their finite sequence of events into a single real number
by $\reals$-Gödel encoding.
By \rref{lem:typecast},
the isomorphism $\goedelize{\cdot} : \traces \rightarrow \traceenc$\, translating between traces and their $\reals$-Gödel encodings $\traceenc \subseteq \reals$ is definable in \comFOD.
Completeness relative to FOD is thus based on an extension~(\rref{fig:plusClaculus}) of \dLCHP's proof calculus
internalizing bijectivity of $\goedelize{\cdot}$
together with a natural axiomatization of the operators on traces.
The latter is necessary to simplify traces so that $\reals$-Gödel encoding by $\goedelize{\cdot}$ can be applied.

\begin{lemma}
	[Trace encoding]
	\label{lem:typecast}
	There is a FOD formula $\eventencseq{x}$ characterizing a subset \mbox{$\traceenc \subseteq \reals$}
	that encodes the type of communication traces $\traces$,
	where the length~$\len{x}$, access $\seqat{x}{j}$, and selectors $\sel{\seqat{x}{j}}$ for $\selOp \in \{ \chanOp, \valOp \}$ 
	of the trace encoded in $x \in \RVar$ can be defined in FOD.
	That is, there is an isomorphism $\goedelize{\cdot} : \traces \rightarrow \traceenc$
	definable in \comFOD
	such that $\goedelize{\cdot}$ preserves lengths and entries.
	The proof is in \rref{app:com-fod-to-fod}.
\end{lemma}

The extension (\rref{fig:plusClaculus}) of the \dLCHP calculus is sound by \rref{thm:extended_soundness}.
We denote the extended calculus by $\dLCHPdash$.
The trace-encoding axioms~\RuleName{TraceToGoedel} and~\RuleName{GoedelToTrace} prove that every trace $\tvar$ has exactly one encoding in $\traceenc\subseteq\reals$ and vice versa (bijection),
where $\exone{x} \psi$ is unique $\exists$-quantification.
Axioms~\RuleName{accessDefault} and~\RuleName{opDefault} capture out-of-bounds access.
The Barcan axiom \RuleName{barcan} \cite{Barcan1946} and the vacuous axiom \RuleName{vacuous} enable to transfer trace terms 
over continuous evolutions.
Axiom \RuleName{extProj} characterizes projections inductively
using definable sequences $(\tvar_i)_{i=1}^{n}$ of traces (\rref{cor:defined_operators}).
The $\in$-relation in \RuleName{extProj} can be finitely unfolded since~$\cset$ is (co)finite,
\eg $\chan{\te} \in \{ \ch{}, \ch{dh} \}^\complement \equiv \neg (\chan{\te} = \ch{}) \wedge \neg (\chan{\te} = \ch{dh})$.
Axioms \RuleName{ext} and \RuleName{eqItem} are extensionality principles for traces and communication items.

The remaining axioms in \rref{fig:plusClaculus} 
decompose trace length and access.
Care is required 
for observable durations,
because the syntactic concatenation $\comItem{\durchan, \duration_1} \cdot \comItem{\durchan, \duration_2}$
semantically denotes a single event $\comItem{\durchan, \duration_1+\duration_2}$.
For this purpose, 
\emph{disjoint concatenation} $\te_1\ofconcat\te_2$ assumes that~$\te_1$ does not end,
\emph{or}~$\te_2$ does not begin,
with a nonzero duration $\comItem{\durchan, \duration}$ with $\duration\neq0$.
Axioms containing $\te_1\ofconcat\te_2$ are not applicable if $\te_1\cdot\te_2$ concatenates nonzero durations.
Otherwise, the axioms \RuleName{maxdursuffix} (and~\RuleName{maxdurprefix})
separate the maximal duration $d$ such that $\comItem{\durchan, d}$ is a suffix (or prefix)
of~$\tvar$,
and~\RuleName{durequi} and~\RuleName{durzero} internalize the equivalence of consecutive observable durations to their sum.

\renewcommand{\sidecondition}[1]{\quad\text{\color{darkgray}(#1)}}

\begin{figure}[t]
	\begin{minipage}{\textwidth}
		\begin{small}
			\begin{calculus}
				\startAxiom{TraceToGoedel}
					$\fa{\tvar}\exone{x{:}\traceenc} x {=} \goedelize{\tvar}$%
					\footnote{\label{ft:unique}Uniqueness quantification $\exone{x} \psi(x)$ is definable as usual by $\exone{x} \psi(x) \equiv \ex{x} (\psi(x) \wedge \fa{y} (\psi(y) \rightarrow y=x))$.}
				\stopAxiom
				\startAxiom{GoedelToTrace}
					$\fa{x{:}\traceenc} \exone{\tvar} x{=}\goedelize{\tvar}$%
					\footnoteref{ft:unique}
				\stopAxiom
			\end{calculus}\hspace{.2cm}%
			\begin{calculus}
				\startAxiom{ext}
					$\te_1 = \te_2 \leftrightarrow \len{\te_1} = \len{\te_2} \wedge \fa{k} \big( 
						0{\le}k{<}\len{\te_1} \rightarrow
						\at{\te_1}{k} = \at{\te_2}{k}	
					\big)$
				\stopAxiom
				\startAxiom{accessConcat}
					$\at{(\te_1 \ofconcat \te_2)}{k} = \ifelse{0{\le}k{<}\len{\te_1}}
						{\at{\te_1}{k}}{\at{\te_2}{k - \len{\te_1}}}$%
					\footnote{
						\label{ft:definable_operators}%
						Disjoint concatenation $\te_1\ofconcat\te_2$,
						the formulas $\maxdursuffix{\te}{d}$ and $\maxdurprefix{\te}{d}$
						characterizing suffixes and prefixes of maximal duration $d$,
						and sequences $(\tvar_i)_{i=1}^{n}$ of traces are definable in \comFOD (\rref{cor:defined_operators}).
					}
				\stopAxiom				
			\end{calculus}

			\begin{calculus}
				\startAxiom{lenEmpty}
					$\len{\epsilon} = 0$
				\stopAxiom
			\end{calculus}\hspace{.4cm}%
			\begin{calculus}
				\startAxiom{lenSumDist}
					$\len{\te_1 \ofconcat \te_2} = \len{\te_1} + \len{\te_2}$%
					\footnoteref{ft:definable_operators}
				\stopAxiom
			\end{calculus}\hspace{.4cm}%
			\begin{calculus}
				\startAxiom{lenAccess}
					$\len{\at{\te}{k}} = \ifelse{0{\le}k{<}\len{\te}}{1}{0}$
				\stopAxiom
			\end{calculus}

			\begin{calculus}
				\startAxiom{accessDefault}
					$\neg(0{\le}k{<}\len{\te}) \leftrightarrow \at{\te}{k} = \epsilon$
				\stopAxiom
				\startAxiom{prefix}
					$\te_1 \preceq \te_2 \leftrightarrow \ex{\tvar} \te_1 \ofconcat \tvar = \te_2$%
					\footnoteref{ft:definable_operators}
				\stopAxiom
				\startAxiom{barcan}
					$\fa{\avar} [ \alpha ] \psi \rightarrow [ \alpha ] \fa{\avar} \psi$
					\sidecondition{$\avar \not\in \alpha$}
				\stopAxiom
			\end{calculus}\hspace{.4cm}%
			\begin{calculus}		
				\startAxiom{accessBase}
					$\len{\te}=1 \rightarrow \at{\te}{k} = \ifelse{0{\le}k{<}1}{\te}{\epsilon}$
				\stopAxiom
				\startAxiom{opDefault}
					$\len{\te} \le 0 \rightarrow \selOp(\te) = 0$
					\sidecondition{$\selOp \in \{ \chanOp, \valOp \}$}
				\stopAxiom
				\startAxiom{vacuous}
					$\varphi \rightarrow [ \alpha ] \varphi$
					\sidecondition{$\SFV(\varphi) \cap \SBV(\alpha) = \emptyset$}
				\stopAxiom
			\end{calculus}

			\begin{calculus}
				\startAxiom{durzero}
					$\comItem{\durchan, 0} = \epsilon$
				\stopAxiom
			\end{calculus}\hspace{.4cm}%
			\begin{calculus}
				\startAxiom{durequi}
					$\comItem{\durchan, \duration_1} \cdot \comItem{\durchan, \duration_2} = \comItem{\durchan, \duration_1+\duration_2}$
				\stopAxiom
			\end{calculus}\hspace{.4cm}%
			\begin{calculus}
				\startAxiom{maxdursuffix}
					$\ex{d} \maxdursuffix{\tvar}{d}$%
					\footnoteref{ft:definable_operators}
				\stopAxiom
			\end{calculus}\hspace{.4cm}%
			\begin{calculus}
				\startAxiom{maxdurprefix}
					$\ex{d} \maxdurprefix{\tvar}{d}$%
					\footnoteref{ft:definable_operators}
				\stopAxiom
			\end{calculus}
			
			\begin{calculus}
				\startAxiom{eqItem}
					$\te\neq\comItem{\durchan,0}
					\rightarrow \big(
						\te = \comItem{\ch{}, \rp} \leftrightarrow \len{\te} = 1 \wedge \chan{\te} = \ch{} \wedge \val{\te} = \rp
					\big)$
				\stopAxiom
			\end{calculus}

			\begin{calculus}
				\startAxiom{extProj}
					$\phi(\te\downarrow\cset) \leftrightarrow \ex{(\tvar_i)_{i=1}^{\len{\te}}} \Big(
						\phi(\tvar_{\len{\te}})
						\wedge \fa{1{\le}k{\le}\len{\te}} \tvar_k = \tvar_{k-1}
						\cdot \big(\ifstat{\chan{\at{\te}{k}} \in \cset}{\at{\te}{k}}\big)
					\Big)$%
					\footnoteref{ft:definable_operators}%
					\footnote{The expression $\ifstat{\chi}{\te}$ 
						is short for $\ifelse{\chi}{\te}{\epsilon}$.	
					}
				\stopAxiom
			\end{calculus}
		\end{small}
	\end{minipage}
	\caption{Extension of the \dLCHP calculus}
	\label{fig:plusClaculus}
\end{figure}

\begin{theorem}
	[Soundness of $\dLCHPdash$]
	\label{thm:extended_soundness}
	The extension of the \dLCHP calculus in \rref{fig:plusClaculus} is sound,
	\iest all axioms in \rref{fig:plusClaculus} are valid formulas.
	Consequently, the extended \dLCHP calculus \dLCHPdash (\rreffigs{fig:calculus}{fig:plusClaculus}) is sound by \rref{thm:soundness}.
\end{theorem}
\vspace{-.7cm}
\begin{proof}
	Soundness of \RuleName{TraceToGoedel} and \RuleName{GoedelToTrace} follows from the fact that $\goedelize{\cdot}$ is a bijection $\traces \rightarrow \traceenc$ by \rref{lem:typecast}.
	Soundness proofs for \RuleName{vacuous} and \RuleName{barcan} are in the literature \cite{DBLP:conf/lics/Platzer12b}.
	Soundness of the remaining axioms in \rref{fig:plusClaculus} easily follows from the semantics of trace terms.
\end{proof}

The equitranslation by \rref{prop:com-fod-to-fod} effectively maps every \comFOD formula $\phi$ to a FOD formula $\toFOD{\phi}$
that is equivalent up to trace encoding by $\goedelize{\cdot}$.
The mapping~$\toFOD{(\cdot)}$ uniformly replaces every trace variable~$\tvar$ in $\phi$ with a fresh but fixed real variable~$\encvar{\tvar}$
and maps operators on traces to the corresponding operators on encodings (\rref{lem:typecast}).
Then $\phi \leftrightarrow \fa{\encvar{\hvarvec}{:}\traceenc{=}\goedelize{\hvarvec}} \toFOD{\phi}$
is provable in the extended \dLCHP calculus \dLCHPdash,
where $\fa{\encvar{\hvarvec}{:}\traceenc{=}\goedelize{\hvarvec}}$%
maps $\phi$'s free trace variables~$\hvarvec$ 
to their representation~$\encvar{\hvarvec}{:}\traceenc$\,in~$\toFOD{\phi}$.%
\footnote{
	The notation $\fa{x{:}\traceenc{=}\re} \psi$ is short for $\fa{x{:}\traceenc} (x{=}\re \rightarrow \psi)$,
	where $\fa{x{:}\traceenc} \chi \equiv \fa{x} ( x{:}\traceenc \rightarrow \chi)$.
}

\begin{proposition}
    [Equitranslation]
	\label{prop:com-fod-to-fod}
    For each \comFOD formula $\phi$
	over free trace variables~$\hvarvec$,
	there is effectively a FOD formula $\toFOD{\phi}$ such that $\phi \leftrightarrow \fa{\encvar{\hvarvec}{:}\traceenc{=}\goedelize{\hvarvec}} \toFOD{\phi}$ is provable in the extended \dLCHP calculus~\dLCHPdash,
	where~$\goedelize{\cdot}$ is applied point-wise.
	The proof is in \rref{app:com-fod-to-fod}.
\end{proposition}

By \rref{thm:com-fod-completeness},
every valid \dLCHP formula $\phi$ has a proof from \comFOD tautologies in the \dLCHP calculus (\rref{fig:calculus}).
By \rref{prop:com-fod-to-fod},
this proof can be extended in the extended \dLCHP calculus \dLCHPdash (\rreffigs{fig:calculus}{fig:plusClaculus}) to a proof of $\phi$ from FOD tautologies,
which proves \rref{thm:fod-completeness}.
A detailed proof of \rref{thm:fod-completeness} is in \rref{app:com-fod-to-fod}.
\begin{theorem}
    [Continuous completeness]
	\label{thm:fod-completeness}
    The extended \dLCHP calculus \dLCHPdash is complete relative to FOD,
    \iest every valid \dLCHP formula $\phi$
	can be proven in \dLCHPdash from FOD tautologies.
\end{theorem}

This concludes our completeness results.
Completeness relative to \comFOD (\rref{thm:com-fod-completeness}) 
shows that the \dLCHP calculus (\rref{fig:calculus})
covers all dynamical effects of parallel hybrid systems
because it can reduce all properties of CHPs to the assertion logic \comFOD.
Completeness relative to FOD (\rref{thm:fod-completeness}) 
proof-theoretically fully aligns parallel hybrid systems in \dLCHP
with reasoning about hybrid systems in \dL,
because provability reduces to FOD for \dL as well \cite{DBLP:journals/jar/Platzer17}.
In summary, in the extended \dLCHP calculus~\dLCHPdash,
properties of parallel hybrid systems can be proven 
whenever properties of hybrid or continuous systems can be proven.

By \rref{prop:com-fod-to-fod}, FOD is expressive for \dLCHP up to trace encoding,
in addition to \comFOD (\rref{lem:com-fod-expressiveness}).
This reduces the necessary assertion logic of \dLCHP to FOD,
and reveals that properties of parallel hybrid systems can be succinctly represented in FOD.
However, in practice, specific axioms for traces \cite{Brieger2023} 
are more useful than reasoning about encodings of traces as properties of differential equations.

\section{Related Work} \label{sec:related}

For clarity, the discussion is structured by group.

\paragraph{Models of Parallel Hybrid Systems}

Unlike CHPs, Hybrid CSP (HCSP) \cite{Jifeng1994} extends CSP \cite{Hoare1978} with \emph{lazily} terminating continuous evolution,
which ends deterministically only at the single point in time at which the evolution constraint is violated.
That is why parallel HCSP models only even have common runs and agree on a common duration if they all leave their domain constraints simultaneously.
Otherwise, HCSP has empty behavior resulting in vacuous proofs.
Instead of exploiting their compositional models as in \dLCHP,
hybrid process algebras are verified noncompositionally by combinatorial translation to model checking \cite{Man2005, Cong2013, Song2005}.
Unlike \dLCHP, which can model a variety of communication patterns by CHPs,
\eg loss and delay of communication, 
and reason about them thanks to completeness,
meta-level components \cite{Lunel2019,DBLP:journals/sttt/MullerMRSP18,Kamburjan2020,Kamburjan2022,Benvenuti2014,Frehse2004,Henzinger2001,Lynch2003} need to be designed from scratch in process algebras for different communication models such as lossy communication.

Quantified differential dynamic logic \QdL \cite{DBLP:conf/csl/Platzer10,DBLP:journals/lmcs/Platzer12b} can express parallel dynamics of an \emph{unbounded} number of %
hybrid systems but only if they all have 
\emph{the same}
structure.
By contrast, \dLCHP can model parallel interactions of entirely different subsystems. 
Fundamentally different from \dLCHP,
parallelism $\alpha\cap\beta$ in concurrent dynamic logic (CDL) \cite{Peleg1987} continues execution in all states reachable by $\alpha$ \emph{and} $\beta$ without ever merging again,
and the parallel programs cannot interact.
CDL with communication \cite{Peleg1987a} does neither support continuous dynamics nor a proof calculus for verification,
and even axioms self-evident in dynamic logic such as $[ \alpha \seq \beta ] \psi \leftrightarrow [ \alpha ] [ \beta ] \psi$ become unsound \cite[p.37]{Peleg1987a},
underlining the fundamentally different nature of their model of parallelism.

Unlike \dLCHP, which models the global flow of time in classical mechanics,
calculi for distributed real-time computation~\cite{Hooman1991,Hooman1987,Hooman1992} 
analyze the timing of discrete computation
or do not impose time synchronization upon parallel programs~\cite{Hooman1993}.
Further, these approaches \cite{Hooman1991,Hooman1987,Hooman1992,Hooman1993} cannot model continuous change by differential equations.

\paragraph{Hoare-logics}

Hybrid Hoare-logic (HHL) for HCSP \cite{Liu2010} 
features a proof calculus for HCSP
that is noncompositional \cite{Wang2012}.
Wang\etal \cite{Wang2012} extend HHL with assume-guarantee reasoning (AGR)%
\footnote{
	Assume-guarantee reasoning as a generic concept embraces a wide variety of techniques.
	It has also been applied to Hybrid Hoare-logic \cite{Guelev2017} but must not be confused with assumption-commitment reasoning, which is the specific proof technique for message-passing concurrency that we use in \dLCHP.
}
in a way that, unlike \dLCHP, becomes noncompositional again,
because their parallel composition rule explicitly unrolls all combinatorial interleavings of the communication traces.
Similarly, Guelev\etal~\cite{Guelev2017} encode the semantics of parallel composition 
by exhaustive unfolding
using the extended duration calculus~\cite{Chaochen1993} as assertion language.
Exposing 
all dynamics
of a subprogram to the other subprograms
in this way 
is said to devalue the whole point of 
compositionality \cite[Section 1.6.2]{deRoever2001} 
because it does not admit reduction of the state space by abstraction. 
Assumptions and guarantees in HHL \cite{Wang2012} cannot specify the communication history but fulfill the different purpose of reasoning about deadlock freedom.

HHL approaches lack completeness results \cite{Wang2012,Liu2010} or prove completeness noncompositionally \cite{Guelev2017} relative to the extended duration calculus \cite{Chaochen1993}.
This leaves open whether parallel hybrid systems proof-theoretically align with hybrid systems,
let alone how complete reasoning for parallel hybrid systems can be done compositionally%
---yet completeness is of limited value if a proof calculus exposes the whole semantics of parallelism \cite{Guelev2017}.
The actual challenge solved by \dLCHP is the development of minimal \emph{compositional} proof principles
that flexibly adapt to the simpler parallel interaction patterns in practice
but can capture all parallel behavior if needed.
Further, \dLCHP's completeness covers liveness modalities,
which are out-of-scope for Hoare-logics.

Completeness of calculi for distributed real-time computation either remains open~\cite{Hooman1993} or is relative to real-time versions of temporal logic \cite{Hooman1992, Hooman1989, Zhou1996} over $\rationals$ as time domain.
Such completeness relative to the data logic is not possible for hybrid systems~\cite{DBLP:conf/lics/Platzer12b} because their data logic is first-order real arithmetic,
which is decidable~\cite{Tarski1951}.
In contrast, \dLCHP is proven to be complete relative to continuous dynamics.

The \dLCHP calculus develops a new modularization of parallel systems safety reasoning,
based on the convincingly simple parallel injection axiom.
Only standard modal logic principles
are required to combine injections.
This reveals that parallel systems do not need complex and highly composite proof rules as in Hoare-style ac-reasoning~\cite{Misra1981,AcHoare_Zwiers,AcSemantics_Zwiers}.
The development of minimalistic proof calculi
further complements our work on uniform substitution \cite{Brieger2023,DBLP:conf/cade/Platzer19,DBLP:journals/jar/Platzer17},
which constitutes prover micro-kernels of small soundness-critical size.
Since ac-reasoning \cite{Misra1981,AcHoare_Zwiers} can be unified \cite{Xu1994} with rely-guarantee reasoning~\cite{Xu1997} for shared-variable parallelism,
modal logic foundations for shared-variable parallelism become an interesting research direction.

\paragraph{Differential Dynamic Logics}

Unlike \dL approaches with components \cite{Lunel2019,DBLP:journals/sttt/MullerMRSP18,Kamburjan2020,Kamburjan2022},
\dLCHP has a parallel operator as first-class citizen that can be arbitrarily nested with other hybrid programs,
rather than parallel composition of fixed-structure meta-level components.
Time-synchronization is built into the parallel operator,
in contrast to manual time modeling in component-based approaches \cite{Lunel2019,DBLP:journals/sttt/MullerMRSP18,Kamburjan2020,Kamburjan2022}.
Modeling of parallelism by nondeterministic choice additionally requires extra care to ensure execution periodicity~\cite{Lunel2019}.
In contrast to first-order constraints relating at most consecutive I/O events~\cite{Lunel2019,DBLP:journals/sttt/MullerMRSP18,Kamburjan2020,Kamburjan2022},
\dLCHP can reason about invariants of the whole communication history.
Orthogonally to our integrated reasoning about discrete, hybrid, and communication behavior,
Kamburjan\etal \cite{Kamburjan2020,Kamburjan2022} separates reasoning about communication from hybrid systems reasoning in entirely different programming languages.
Meta-level approaches do not study completeness \cite{Lunel2019,DBLP:journals/sttt/MullerMRSP18,Kamburjan2020,Kamburjan2022} but this may become possible via their encoding in \dLCHP with its completeness results.

In \QdL, structural and dimensional change of distributed networks of agents are an additional source of incompleteness \cite{DBLP:journals/lmcs/Platzer12b}
besides its discrete and continuous dynamics.
Unlike \dLCHP, which is complete relative to properties of continuous systems,
\QdL  is complete relative to properties of \emph{quantified} continuous systems \cite{DBLP:journals/lmcs/Platzer12b},
\iest systems with simultaneous change of unboundedly many continuous systems at once.
Unlike \dLCHP's uniform substitution calculus \cite{Brieger2023}, 
in this article,
\dLCHP's calculus relies on schematic axioms
to put the spotlight on the new completeness results. 
These results 
are the key for completeness
of the uniform substitution calculus but tackling both at once would make a comprehensible presentation of either result infeasible.

Temporal logic plays a central role in the verification of concurrency \cite{Manna1992}.
In \dLCHP, the ac-modalities are reminiscent of temporal logic by their quantification over communication traces.
Differential temporal logic \dTL extends \dL with temporal operators~\cite{DBLP:conf/cade/JeanninP14},
but does not support parallelism.
Unlike \dTL's temporal operators,
which quantify over continuous traces,
ac-modalities refer to discrete events.

\paragraph{Automata}

The parallel composition of hybrid automata \cite{Lynch2003, Frehse2004, Henzinger2001, Benvenuti2014}, just like 
HCSP \cite{Wang2012,Guelev2017}, always falls back to the combinatorial exponentiation of parallelism.
Consequently, even AGR approaches \cite{Lynch2003,Frehse2004,Henzinger1996,Benvenuti2014} for hybrid automata that mitigate the state space explosion for subautomata,
eventually resort to large product automata later.
In contrast, \dLCHP's parallel injection axiom exploits the built-in compositionality of the program semantics enabling verification of subprograms truly independently of their environment based on their shared communication interface.
Unlike ac-formulas in \dLCHP,
which can capture change, rate, delay, or noise for arbitrary pairings of communication channels,
overapproximation is limited to coarse abstractions by timed transition systems \cite{Frehse2004},
components completely dropping 
continuous behavior~\cite{Henzinger2001},
or static global contracts \cite{Benvenuti2014}.
Where \dLCHP inherits \dL's complete reasoning about differential equation invariants \cite{DBLP:journals/jacm/PlatzerT20},
automata approaches are often limited to linear continuous dynamics~\cite{Frehse2004,Henzinger2001}.

\section{Conclusion} \label{sec:conclusion}

This article shows completeness of the dynamic logic of communicating hybrid programs \dLCHP for modeling and verification of parallel interactions of hybrid systems.
These interactions go beyond the mere sum of hybrid and parallel systems 
because only their combination faces the challenge of true parallel synchronization in real time.
Despite this complexity, \dLCHP's compositional proof calculus disentangles the subtly intertwined dynamics of parallel hybrid systems into atomic pieces of discrete, continuous, and 
parallel behavior.
The calculus supports \emph{truly} compositional reasoning,
\iest the decomposition of parallel hybrid systems is along specifications of their observable behavior only,
which can always express enough to be complete
without being cluttered with exponential parallel product overhead where simple properties suffice.
Therefore, \dLCHP embeds assumption-commitment (ac) reasoning into dynamic logic,
and further replaces classical monolithic Hoare-style proof rules
with a far-reaching modularization of deductions about parallel systems that is driven by a stringent modal view onto ac-reasoning.
At the core of this development is the parallel injection axiom,
which proves properties of a parallel subprogram from its projection onto the subprogram alone.
Completeness shows that 
this convincingly simple axiom 
is the \emph{only} proof principle necessary for reasoning about safety of parallel hybrid systems.
Classical proof rules for parallel systems derive, 
but their soundness simply follows from the soundness of \dLCHP's small modular reasoning principles, simplifying side conditions that are otherwise notoriously subtle for parallel system verification.

The increased compositionality and modularity would be counterproductive if they were to miss phenomena in parallel hybrid systems.
The two effective completeness results show that this is not the case and prove adequacy of the calculus:
First, completeness is proven relative to the first-order logic of communication traces and differential equations \comFOD.
This shows that \dLCHP provides all axioms and proof rules necessary to reduce valid \dLCHP formulas to the assertion logic \comFOD,
and confirms that \dLCHP's calculus is a comprehensive characterization of all dynamical effects of parallel hybrid systems.
At the core of the proof is a complete reasoning pattern for safety of parallel hybrid systems.
This pattern points out all steps that can be necessary but in stark contrast with classical monolithic reasoning can be reduced whenever shortcuts are sufficient.
Further, completeness is proven relative to the first-order logic of differential equations FOD.
This result proof-theoretically aligns \dLCHP with reasoning about hybrid systems in \dL, which is complete relative to FOD as well.
Consequently, properties of parallel hybrid systems can be verified whenever properties of hybrid systems, continuous, and discrete systems can be verified.

Interesting directions for future work include
uniform substitution \cite{Brieger2023} that gets rid of subtle soundness-critical side conditions, 
which otherwise cause overwhelming implementations of theorem provers.
Uniform substitution is a subtle challenge on its own \cite{DBLP:journals/jar/Platzer17,DBLP:conf/cade/Platzer19},
so needs its own careful presentation,
but completeness of \dLCHP's schematic calculus proven in this article is a major step toward its completeness.

\appendix
\section*{Appendices}

\section{Computable Static Semantics}
\label{app:stat_sem}

This appendix provides sound (\rref{lem:sound_over}) computable supersets for the static semantics (\rref{def:static_semantics}).
The supersets for free variables of terms (\rref{def:free_term_vars}), programs (\rref{def:free_chp_vars}), and formulas (\rref{def:free_formula_vars}),
and bound variables of programs (\rref{def:bound_vars})
follow the standard definitions for \dL \cite{DBLP:journals/jar/Platzer17},
except for the new statements of \dLCHP.

\begin{definition}
	[Free names of terms]
	\label{def:free_term_vars}
	Free variables $\FV(\expr)$ and accessed channels per trace variable $\CN(\expr)$
	for a term~$\expr$
	are defined inductively as follows,
	where $\ratconst \in \rationals \cup \{\epsilon\}$ is any constant, 
	and $\bowtie \,\in \{+,\cdot\}$, 
	and $\sel{\te} \in \{ \chan{\te}, \val{\te}, \len{\te} \}$:

	\noindent\begin{minipage}{.5\textwidth}
		\begin{align*}
			\FV(\avar) & = \{ \avar \} 
			\sidecondition{$\avar \in \V$} \\
			\FV(\ratconst) & = \;\emptyset \\
			\FV(\expr_1 \bowtie \expr_2) & = \FV(\expr_1) \cup \FV(\expr_2) \\
			\FV(\sel{\te}) & = \FV(\te) \\
			\FV(\at{\te}{\re}) & = \FV(\te) \cup \FV(\re) \\
			\FV(\comItem{\ch{}, \rp}) & = \FV(\rp) \\
			\FV(\te \downarrow \cset) & = \FV(\te)
		\end{align*}
	\end{minipage}%
	\begin{minipage}{.5\textwidth}
		\begin{align*}
			\CN(\tvar) & = \{\tvar\} \times \Chan \sidecondition{$\tvar \in \TVar$}\\
			\CN(x) = \CN(\ratconst) & = \emptyset \sidecondition{$x \in \RVar$}\\
			\CN(\expr_1 \bowtie \expr_2) & = \CN(\expr_1) \cup \CN(\expr_2) \\
			\CN(\sel{\te}) & = \CN(\te) \\
			\CN(\at{\te}{\re}) & = \CN(\te) \cup \CN(\re) \\
			\CN(\comItem{\ch{}, \rp}) & = \emptyset \\
			\CN(\te \downarrow \cset) & = \CN(\te) \cap (\V_\traces \times \cset)
		\end{align*}
	\end{minipage}
\end{definition}

Every trace variable $\tvar$ adds all channel names $\Chan$ as accessed channels of that trace variable.
Conversely, the projection $\te \downarrow \cset$ 
restricts the accessed channels of all trace variables in its scope
to at most the channels $\cset$.
There are no channels accessed by the event $\comItem{\ch{},\rp}$,
because $\rp \in \polynoms{\rationals}{\RVar}$, so $\rp$ does not contain any trace variable.

For a program $\alpha$, \rref{def:bound_vars} defines its bound variables $\BV(\alpha)$ and must-bound variables $\MBV(\alpha) \subseteq \BV(\alpha)$.
A must-bound variable is bound and its initial value is overwritten on \emph{every} execution path of $\alpha$.
Consequently, these variables can be soundly removed from the free variables of the subsequent statements $\beta$ and $\psi$
in $\FV(\alpha\seq\beta)$ and $\FV([\alpha]\psi)$, respectively (\rref{def:free_formula_vars});
see previous work for a detailed discussion \cite{DBLP:journals/jar/Platzer17}.

\begin{definition}
	[Bound variables]
	\label{def:bound_vars}
	Bound variables $\BV(\alpha)$ and must-bound variables $\MBV(\alpha)$
	of a program $\alpha$ are defined inductively as follows:

	\noindent\begin{minipage}{.5\textwidth}
		\begin{align*}
			\BV(x \ceq \rp) & = \{ x \} \\
			\BV(x \ceq *) & = \{ x \} \\
			\BV(\evolution{}{}) & = \{ \recvar, x \} \\
			\BV(\test{}) & = \emptyset \\
			\BV(\send{}{}) & = \{ \recvar \} \\
			\BV(\receive{}{}) & = \{ \recvar, x \} \\
			\BV(\alpha \cup \beta) & = \BV(\alpha) \cup \BV(\beta) \\
			\BV(\alpha \seq \beta) & = \BV(\alpha) \cup \BV(\beta) \\
			\BV(\alpha \parOp \beta) & = \BV(\alpha) \cup \BV(\beta) \\
			\BV(\repetition{\alpha}) & = \BV(\alpha)
		\end{align*}%
	\end{minipage}%
	\begin{minipage}{.5\textwidth}
		\begin{align*}
			\MBV(x \ceq \rp) & = \{ x \} \\
			\MBV(x \ceq *) & = \{ x \} \\
			\MBV(\evolution{}{}) & = \{ x \} \\
			\MBV(\test{}) & = \emptyset \\
			\MBV(\send{}{}) & = \emptyset \\
			\MBV(\receive{}{}) & = \{ x \} \\
			\MBV(\alpha \cup \beta) & = \MBV(\alpha) \cap \MBV(\beta) \\
			\MBV(\alpha \seq \beta) & = \MBV(\alpha) \cup \MBV(\beta) \\
			\MBV(\alpha \parOp \beta) & = \MBV(\alpha) \cup \MBV(\beta) \\
			\MBV(\repetition{\alpha}) & = \emptyset
		\end{align*}
	\end{minipage}
\end{definition}

The recorder $\recvar$ is bound by every program that emits communication,
including the communication of durations along the channel $\durchan$
by continuous evolutions.
However, the recorder~$\recvar$ is not a must-bound variable,
because the state after emitting communication still also depends on the value of $\recvar$ before (\rref{sec:substitution}).
Hence, it would be unsound to remove $\recvar$ in \rref{def:free_formula_vars}  
from the free variables of $\psi$ in $\FV([\alpha] \psi)$.

\begin{definition}
	[Free variables of programs]
	\label{def:free_chp_vars}
	Free variables $\FV(\alpha)$ of a program $\alpha$
	are inductively defined as follows:

	\noindent\begin{minipage}{.5\textwidth}
		\begin{align*}
			\FV(x \ceq \rp) & = \FV(\rp) \\
			\FV(x \ceq *) & = \emptyset \\
			\FV(\test{}) & = \FV(\chi) \\
			\FV(\evolution{}{}) & = \{ x \} \cup \FV(\rp) \cup \FV(\chi) \\
			\FV(\send{}{}) & = \FV(\rp)
		\end{align*}
	\end{minipage}%
	\begin{minipage}{.5\textwidth}
		\begin{align*}
			\FV(\receive{}{}) & = \emptyset \\
			\FV(\alpha \seq \beta) & = \FV(\alpha) \cup (\FV(\beta) \setminus \MBV(\alpha)) \\
			\FV(\alpha \cup \beta) & = \FV(\alpha) \cup \FV(\beta) \\
			\FV(\alpha \parOp \beta) & = \FV(\alpha) \cup \FV(\beta) \\
			\FV(\repetition{\alpha}) & = \FV(\alpha)
		\end{align*}
	\end{minipage}
\end{definition}

\begin{definition}
	[Free variables of formulas]
	\label{def:free_formula_vars}
	Free variables $\FV(\varphi)$ and accessed channels per trace variable $\CN(\varphi)$
	for a formula $\varphi$
	are inductively defined as follows:

	\noindent\begin{minipage}{.5\textwidth}	
		\begin{align*}
			\FV(\expr_1 \sim \expr_2) & = \FV(\expr_1) \cup \FV(\expr_2) \\
			\FV(\neg \varphi) & = \FV(\varphi) \\
			\FV(\varphi \wedge \psi) & = \FV(\varphi) \cup \FV(\psi) \\
			\FV(\fa{\avar} \varphi) & = \FV(\varphi) \setminus \{ \avar \} \\
			\FV([ \alpha ] \psi) & = \FV(\alpha) \cup (\FV(\psi) \setminus \MBV(\alpha)) \\
			\FV([ \alpha ] \ac \psi) & = \FV([ \alpha ] \psi) \cup \FV(\A) \cup \FV(\Commit)
		\end{align*}
	\end{minipage}%
	\begin{minipage}{.5\textwidth}
		\begin{align*}
			\CN(\expr_1 \sim \expr_2) & = \CN(\expr_1) \cup \CN(\expr_2) \\
			\CN(\neg \varphi) & = \CN(\varphi) \\
			\CN(\varphi \wedge \psi) & = \CN(\varphi) \cup \CN(\psi) \\
			\CN(\fa{x} \varphi) & = \CN(\varphi) \\
			\CN(\fa{\tvar} \varphi) & = \CN(\varphi) \setminus (\{\tvar\} \times \Chan) \\
			\CN( [ \alpha ] \psi) & = \CN(\psi) \\
			\CN([ \alpha ] \ac \psi) & = \CN(\A) \cup \CN(\Commit) \cup \CN(\psi)
		\end{align*}
	\end{minipage}
\end{definition}

In case $\FV([ \alpha ] \ac \psi)$,
the must-bound variables $\MBV(\alpha)$
need not be removed from~$\A$ and~$\Commit$ 
because the ac-formulas are unaffected by state change in $\alpha$ (\rref{def:syntax_formulas}).
Since $\fa{\tvar}\varphi$ does not depend on $\tvar$ at all,
it neither depends on any channel via the trace variable $\tvar$,
which is why $\{\tvar\}\times\Chan$ can be soundly removed in case $\fa{\tvar}\varphi$.

Since CHPs do not depend on the history (\rref{cor:history_coincidence}),
\iest $\FV(\alpha) \subseteq \RVar$ for every program $\alpha$,
there are no accessed channels of programs.

\begin{definition}
	[Written channels of programs]
	\label{def:writen_chans}
	The written channels $\CN(\alpha)$ of a program $\alpha$
	are inductively defined as follows:

	\noindent\begin{minipage}{.5\textwidth}
		\begin{align*}
			\CN(x \ceq \rp) & = \emptyset \\
			\CN(x \ceq *) & = \emptyset \\
			\CN(\test{}) & = \emptyset \\
			\CN(\evolution{}{}) & = \{\durchan\} \\
			\CN(\send{}{}) & = \{\ch{}\}
		\end{align*}
	\end{minipage}%
	\begin{minipage}{.5\textwidth}
		\begin{align*}
			\CN(\receive{}{}) & = \{\ch{}\} \\
			\CN(\alpha \seq \beta) & = \CN(\alpha) \cup \CN(\beta) \\
			\CN(\alpha \cup \beta) & = \CN(\alpha) \cup \CN(\beta) \\
			\CN(\alpha \parOp \beta) & = \CN(\alpha) \cup \CN(\beta) \\
			\CN(\repetition{\alpha}) & = \CN(\alpha)
		\end{align*}
	\end{minipage}
\end{definition}

\section{Soundness of the Calculus}
\label{app:soundness}

\begin{proof}
	[Proof of \rref{lem:noninterference}.]
	First, show that
	the formula $\psi$ only depends on the communication $\trace \downarrow \alpha$ of $\alpha$,
	\iest $(\trace \downarrow \alpha) \downarrow \cset = \trace \downarrow \cset$,
	where $\cset = \relimage{\SCN(\psi)}{\{\recvar\}}$.
	This holds if only a communication event 
	$\comevent = \comItem{\ch{}, \semConst}$ in $\trace$,
	which is not removed by $\downarrow \cset$, is also not removed by $\downarrow \alpha$.
	Accordingly, let $\rawtrace \downarrow \cset = \rawtrace$.
	Then $\ch{} \in \cset$.
	If $\ch{} \not \in \SCN(\beta)$,
	then $\ch{} \in \SCN(\alpha)$ because~$\rawtrace$ is emitted by $\alpha \parOp \beta$.
	Otherwise, if $\ch{} \in \SCN(\beta)$,
	then $\ch{} \in \SCN(\alpha) \cup \{\durchan\}$ by noninterference (\rref{def:noninterference}).
	Hence, $\ch{} \in \CN(\alpha) \cup \{\durchan\}$,
	as $\CN(\alpha) \supseteq \SCN(\alpha)$ by \rref{lem:sound_over},
	such that $\rawtrace$ is not removed by~$\downarrow \alpha$,
	which is defined by $\downarrow (\CN(\alpha)\cup\{\durchan\})$.
	Since $(\trace \downarrow \alpha) \downarrow \cset = \trace \downarrow \cset$,
	for $\pstate{u} \in \{\pstate{v}, \pstate{w}_\alpha\}$, obtain
	\begin{equation} \label{eq:historyProjs}
		\big( \pstate{u} \cdot (\trace \downarrow \alpha) \big) \downarrow \cset
		= (\pstate{u} \downarrow \cset) \cdot \big( (\trace \downarrow \alpha) \downarrow \cset \big)
		= (\pstate{u} \downarrow \cset) \cdot (\trace \downarrow \cset)
		= (\pstate{u} \cdot \trace) \downarrow \cset
		\text{.}
	\end{equation}

	Using \rref{eq:historyProjs}, \rref{itm:invalidation1} holds by coincidence (\rref{lem:expr_coincidence}).
	For \rref{itm:invalidation2}, assume $\pstate{w} \neq \bot$. 
	Then $\pstate{w}_\alpha \neq \bot$ and $\pstate{w}_\beta \neq \bot$ by definition of $\merge$ in \rref{sec:semantics}. 
	First, observe that $\pstate{w}_\alpha = \pstate{w}$ on $\SBV(\alpha)$ by definition of $\merge$, so $\pstate{w}_\alpha = \pstate{w}$ on $\SBV(\alpha) \cap \SBV(\beta)^\complement$.
	Second, $\pstate{w}_\alpha = \pstate{v}$ on $\SBV(\alpha)^\complement$ and $\pstate{v} = \pstate{w}_\beta$ on $\SBV(\beta)^\complement$ by the bound effect property (\rref{lem:bound_effect}),
	and $\pstate{w}_\beta = \pstate{w}$ on $\SBV(\alpha)^\complement$ by definition of $\merge$.
	In summary, $\pstate{w}_\alpha = \pstate{w}$ on $\SBV(\alpha)^\complement \cap \SBV(\beta)^\complement$.
	Third, since $\pstate{w}_\alpha = \pstate{v} = \pstate{w}_\beta$ on~$\TVar$ by \rref{lem:bound_effect},
	obtain $\pstate{w}_\alpha = \pstate{w}_\alpha \merge \pstate{w}_\beta = \pstate{w}$ on $\TVar$.
	In summary, $\pstate{w}_\alpha = \pstate{w}$ on
	\begin{equation}
		\label{eq:vars_noninterference}
		(\SBV(\alpha) \cap \SBV(\beta)^\complement) \cup (\SBV(\alpha)^\complement \cap \SBV(\beta)^\complement) \cup \TVar 
		= \SBV(\beta)^\complement \cup \TVar \text{.}
	\end{equation} 
	By noninterference (\rref{def:noninterference}),
	obtain $\SFV(\psi) \subseteq \SBV(\beta)^\complement \cup \gvars$.
	Hence, $\pstate{w}_\alpha = \pstate{w}$ on $\SFV(\psi)$ by \rref{eq:vars_noninterference}.
	Therefore, by \rref{eq:historyProjs},
	obtain $(\pstate[ind=\alpha]{w} \cdot (\trace \downarrow \alpha)) \downarrow \cset
	\overset{(\ref{eq:historyProjs})}{=} 
		(\pstate[ind=\alpha]{w} \cdot \trace) \downarrow \cset
	= (\pstate{w} \cdot \trace) \downarrow \cset$ on $\SFV(\psi)$.
	Finally, \mbox{\rref{itm:invalidation2}} holds by \rref{lem:expr_coincidence}.
	\qedhere
\end{proof}

\begin{proof}[Proof of \rref{thm:soundness}.]
	We prove soundness of the novel ac-axioms and rules. 
	Since \dLCHP is a conservative extension of \dL \cite[Proposition 1]{BriegerMP2023}, 
	we can soundly use the \dL proof calculus for reasoning about \dL formulas in \dLCHP.
	Hence, we point to the literature for soundness of the axioms and rules adopted from~$\dL$ \cite{DBLP:journals/jar/Platzer08, DBLP:journals/jar/Platzer17,Platzer18}. 

	\begin{itemize}[leftmargin=1em, itemindent=-1em, align=left]
		\itemsep.5em
		\item[] \RuleName{boxesDual}:
		The implication $\rightarrow$ uses that the commitment holds trivially
		and $\leftarrow$ uses that the assumption holds trivially.

		\item[] \RuleName{acdbDual}:
		The axiom is a simple consequence of the semantics of ac-box and ac-diamond.

		\item[] \RuleName{Atransfer}:
		Let $\pstate{v} \vDash \tvar_0 = \recvar$,
		and recall that $\Atrace{\A} \equiv \AclDef{}{\recvar}$
		for \mbox{$\sim \,\in \{ \prec, \preceq \}$}.
		Then observe that $\assCP{\sim}{\pstate{v}}{\trace} \vDash \A$ iff
		$\pstate{v} \cdot \trace \vDash \Atrace{\A}$
		(even $\pstate{w} \cdot \trace \vDash \Atrace{\A}$ if $\pstate{w} \neq \bot$ by \rref{cor:communicative_coincidence})
		for every $\run \in \sem{\alpha}{}$.
		The axiom follows by \acCommit and \acPost.		

		\item[] \RuleName{acNoCom}:
		Let $\pstate{v} \vDash [ \alpha ] \ac \psi$. 
		Then $\pstate{v} \vDash \Commit$ by \acCommit since $(v, \epsilon, \bot) \in \sem{\alpha}{}$ by totality and prefix-closedness.
		Now, assume $\pstate{v} \vDash \A$
		and let $\computation \in \sem{\alpha}{}$ with $\pstate{w} \neq \bot$.
		Then $\trace = \epsilon$ because $\SCN(\alpha) = \emptyset$. 
		Hence, $\assPost{\pstate{v}}{\trace} = \{\pstate{v}\}$ in \acPost,
		which implies $\pstate{w} \vDash \psi$. 
		Conversely, let $\pstate{v} \vDash \Commit \wedge (\A \rightarrow [ \alpha ] \psi)$ and $\computation \in \sem{\alpha}{}$.
		Then \acCommit holds since $\trace = \epsilon$ and $\pstate{v} \vDash \Commit$. 
		For \acPost, assume $\pstate{w} \neq \bot$ and $\assPost{\pstate{v}}{\trace} \vDash \A$, so $\pstate{v} \vDash \A$.
		Hence, $\pstate{v} \vDash \A \rightarrow [ \alpha ] \psi$ implies $\pstate{w} \vDash \psi$ as $\trace = \epsilon$.
		
		\item[] \RuleName{Aweak}:
		Let $\pstate{v} \vDash [ \alpha ] \acpair{\true, \Commit \wedge \weakA \rightarrow \A} \true$ and $\pstate{v} \vDash [ \alpha ] \acpair{\A, \Commit} \psi$.
		First, observe that for every $\run \in \sem{\alpha}{}$,
		the stronger assumption $\assCommit{\pstate{v}}{\trace} \vDash \weakA$ implies $\assCommit{\pstate{v}}{\trace} \vDash \A$.
		This is proven by induction on the structure of $\trace$:
		\begin{enumerate}[leftmargin=*]
			\item $\trace = \epsilon$, then  $\assCommit{\pstate{v}}{\trace} \vDash \A$ holds trivially since $\assCommit{\pstate{v}}{\trace} = \emptyset$.
			
			\item $\trace = \trace_0 \cdot \comevent$ with $\semLen{\comevent} = 1$, then assume $\assCommit{\pstate{v}}{\trace} \vDash \weakA$.
			Hence, $\pstate{v} \cdot \trace_0 \vDash \weakA$ and $\assCommit{\pstate{v}}{\trace_0} \vDash \weakA$,
			where the latter implies $\assCommit{\pstate{v}}{\trace_0} \vDash \A$ by the induction hypothesis.
			Since $(\pstate{v}, \trace_0, \bot) \in \sem{\alpha}{}$ by prefix-closedness and $\pstate{v} \vDash [ \alpha ] \acpair{\A, \Commit} \psi$, obtain $\pstate{v} \cdot \trace_0 \vDash \Commit$.
			The latter, and $\pstate{v} \cdot \trace_0 \vDash \weakA$, and $\pstate{v} \vDash [ \alpha ] \acpair{\true, \Commit \wedge \weakA \rightarrow \A} \true$ together imply $\pstate{v} \cdot \trace_0 \vDash \A$.
			Thus, $\assCommit{\pstate{v}}{\trace_0} \vDash \A$ extends  to $\assCommit{\pstate{v}}{\trace} \vDash \A$.
		\end{enumerate}
	
		To prove $\pstate{v} \vDash [ \alpha ] \acpair{\weakA, \Commit} \psi$,
		let $\run \in \sem{\alpha}{}$.
		For \acCommit, assume $\assCommit{\pstate{v}}{\trace} \vDash \weakA$,
		which implies $\assCommit{\pstate{v}}{\trace} \vDash \A$.
		Hence, $\pstate{v} \cdot \trace \vDash \Commit$ by $\pstate{v} \vDash [ \alpha ] \ac \psi$.
		For \acPost, assume $\pstate{w} \neq \bot$ and $\assPost{\pstate{v}}{\trace} \vDash \weakA$,
		which implies $\assCommit{\pstate{v}}{\trace} \vDash \A$.
		Then $\pstate{v} \cdot \trace \vDash \Commit$ by $\pstate{v} \vDash [ \alpha ] \ac \psi$ again.
		Since $\assPost{\pstate{v}}{\trace} \vDash \weakA$ contains $\pstate{v} \cdot \trace \vDash \weakA$,
		obtain $\pstate{v} \cdot \trace \vDash \A$ by $\pstate{v} \vDash [ \alpha ] \acpair{\true, \Commit \wedge \weakA \rightarrow \A} \true$.
		In summary, $\assPost{\pstate{v}}{\trace} \vDash \A$.
		Finally, $\pstate{v} \vDash [ \alpha ] \acpair{\A, \Commit} \psi$ implies $\pstate{w} \cdot \trace \vDash \psi$.
		
		\item[] \RuleName{acComposition}
		Let $\pstate{v} \vDash [ \alpha \seq \beta ] \ac \psi$. 
		To show $\pstate{v} \vDash [ \alpha ] \ac [\beta] \ac \psi$, 
		let $(v, \trace_1, \pstate{u}) \in \sem{\alpha}{}$.
		For \acCommit, assume $\assCommit{\pstate{v}}{\trace_1} \vDash \A$.
		By prefix-closedness, $(v, \trace_1, \bot) \in \botop{\sem{\alpha}{}} \subseteq \sem{\alpha \seq \beta}{}$.
		Then $\pstate{v} \cdot \trace_1 \vDash \Commit$ by $\pstate{v} \vDash [ \alpha \seq \beta ] \ac \psi$. 
		For \acPost, assume $\pstate{u} \neq \bot$ and $\assPost{\pstate{v}}{\trace_1} \vDash \A$.
		To prove $\pstate{u} \cdot \trace_1 \vDash [ \beta ] \ac \psi$, 
		let $(\pstate{u} \cdot \trace_1, \trace_2, \pstate{w} \cdot \trace_1) \in \sem{\beta}{}$
		(\wlossg by coincidence (\rref{cor:history_coincidence})).	
		\begin{enumerate}[leftmargin=*]
			\item For \acCommit, assume $\assCommit{\pstate{u} \cdot \trace_1}{\trace_2} \vDash \A$.
			Since $\SFV(\A) \cap \SBV(\alpha) \subseteq \TVar$ by well-formedness (\rref{def:syntax_formulas}),
			obtain $\assCommit{\pstate{v} \cdot \trace_1}{\trace_2} \vDash \A$ by coincidence (\rref{cor:communicative_coincidence}).
			Hence, $\assCommit{\pstate{v}}{\trace_1\cdot\trace_2} \vDash \A$,
			because $\assPost{\pstate{v}}{\trace_1} \vDash \A$.
			Since $(v, \trace_1, \pstate{u}) \in \sem{\alpha}{}$ and $(\pstate{u} \cdot \trace_1, \trace_2, \pstate{w} \cdot \trace_1) \in \sem{\beta}{}$,
			obtain $(\pstate{v}, \trace_1 \cdot \trace_2, \pstate{w}) \in \sem{\alpha}{} \continuation \sem{\beta}{} \subseteq \sem{\alpha\seq\beta}{}$
			using \rref{cor:history_coincidence}.
			Hence, $\pstate{v} \cdot \trace_1 \cdot \trace_2 \vDash \Commit$ by $\pstate{v} \vDash [ \alpha \seq \beta ] \ac \psi$.
			Finally, $\pstate{u} \cdot \trace_1 \cdot \trace_2 \vDash \Commit$ by \rref{cor:communicative_coincidence}.
			
			\item For \acPost, assume $\pstate{w} \neq \bot$ and $\assPost{\pstate{u} \cdot \trace_1}{\trace_2} \vDash \A$.
			Then $\assPost{\pstate{v}}{\trace_1\cdot\trace_2} \vDash \A$ 
			and $(\pstate{v}, \trace_1 \cdot \trace_2, \pstate{w}) \in \sem{\alpha\seq\beta}{}$
			as above.
			Finally, $\pstate{w} \cdot \trace_1 \cdot \trace_2 \vDash \psi$ since $\pstate{v} \vDash [ \alpha\seq\beta ] \ac \psi$.
		\end{enumerate}
		
		Conversely, let $\pstate{v} \vDash [ \alpha ] \ac [ \beta ] \ac \psi$. 
		To prove $\pstate{v} \vDash [ \alpha\seq\beta ] \ac \psi$, let $\computation \in \sem{\alpha \seq \beta}{}$.
		If $\computation \in \botop{\sem{\alpha}{}}$, \acCommit holds by the assumption,
		and \acPost holds trivially as $\pstate{w} = \bot$.
		Otherwise, if $\computation \in \sem{\alpha}{} \continuation \sem{\beta}{}$,
		there are $(\pstate{v}, \trace_1, \pstate{u}) \in \sem{\alpha}{}$ and $(\pstate{u}, \trace_2, \pstate{w}) \in \sem{\beta}{}$ with $\trace = \trace_1 \cdot \trace_2$.
		By \rref{cor:history_coincidence},
		obtain $(\pstate{u} \cdot \trace_1, \trace_2, \pstate{w} \cdot \trace_1) \in \sem{\beta}{}$.
		\begin{enumerate}[leftmargin=*]
			\item 
			For \acCommit, assume $\assCommit{\pstate{v}}{\trace} \vDash \A$. 
			If $\trace_2 = \epsilon$, \acCommit holds because $(\pstate{v}, \trace, \bot) \in \botop{\sem{\alpha}{}}$ by prefix-closedness.
			If $\trace_2 \neq \epsilon$, then $\assPost{\pstate{v}}{\trace_1} \vDash \A$ and $\assCommit{\pstate{v} \cdot \trace_1}{\trace_2} \vDash \A$.
			Hence, $\pstate{u} \cdot \trace_1 \vDash [ \beta ] \ac \psi$ by \acPost,
			and $\assCommit{\pstate{u} \cdot \trace_1}{\trace_2} \vDash \A$ by \rref{cor:communicative_coincidence}.
			Since $\pstate{u} \cdot \trace_1 \vDash [ \beta ] \ac \psi$,
			obtain $\pstate{u} \cdot \trace_1 \cdot \trace_2 \vDash \Commit$ by \acCommit.
			By \rref{cor:communicative_coincidence},
			obtain $\pstate{v} \cdot \trace \vDash \Commit$.
			
			\item 
			For \acPost, assume $\pstate{w} \neq \bot$ and $\assPost{\pstate{v}}{\trace} \vDash \A$. 
			Then $\pstate{u} \cdot \trace_1 \vDash [ \beta ] \ac \psi$ as above.
			Since $\assPost{\pstate{v}}{\trace} \vDash \A$,
			obtain $\assPost{\pstate{v} \cdot \trace_1}{\trace_2} \vDash \A$,
			so $\assPost{\pstate{u} \cdot \trace_1}{\trace_2} \vDash \A$ by \rref{cor:communicative_coincidence}.
			Finally, $\pstate{w} \cdot \trace \vDash \psi$ 
			by $\pstate{u} \cdot \trace_1 \vDash [ \beta ] \ac \psi$.
		\end{enumerate}
		\item[] \RuleName{acChoice}:
		The axiom follows directly from the semantics $\sem{\alpha \cup \beta}{} = \sem{\alpha}{} \cup \sem{\beta}{}$ of choice. 
		\item[] \RuleName{acIteration}
		Since $\sem{\repetition{\alpha}}{} = \bigcup_{n \in \naturals} \sem{\alpha^\semvar}{}$, the formula $[ \repetition{\alpha} ] \ac \psi \leftrightarrow [ \alpha^0 ] \ac \psi \wedge [ \alpha \seq \alpha^* ] \ac \psi$ is valid.
		Then \RuleName{acIteration} follows by axiom \RuleName{acComposition}.    
		\item[] \RuleName{acInduction}:
		Let $\pstate{v} \vDash [ \repetition{\alpha} ] \ac \psi$. 
		Then $\pstate{v} \vDash [ \alpha^0 ] \ac \psi \wedge [ \alpha \seq \repetition{\alpha} ] \ac \psi$ by axioms \RuleName{acIteration} and \RuleName{acComposition}.
		Since $\sem{\alpha \seq \repetition{\alpha}}{} = \sem{\repetition{\alpha} \seq \alpha}{}$,%
		\footnote{This fact can be proven by an induction using the fact that sequential composition is associative.} 
		obtain $\pstate{v} \vDash [ \repetition{\alpha} \seq \alpha ] \ac \psi$.
		Hence, $\pstate{v} \vDash [ \repetition{\alpha} \seq \alpha ] \ac \psi$ by \acPost,
		which implies $\pstate{v} \vDash [ \repetition{\alpha} ] \ac [ \alpha ] \ac \psi$ by axiom \RuleName{acComposition}. 
		Finally, $\pstate{v} \vDash [ \repetition{\alpha} ] \acpair{\A, \true} (\psi \rightarrow [ \alpha ] \ac \psi)$ by rule \RuleName{acMono}.
		
		Conversely, let $\pstate{v} \vDash [ \alpha^0 ] \ac \psi \wedge [ \repetition{\alpha} ] \acpair{\A, \true} ( \psi \rightarrow [ \alpha ] \ac \psi )$.
		For proving $\pstate{v} \vDash [ \repetition{\alpha} ] \ac \psi$, let $\computation \in \sem{\repetition{\alpha}}{}$.
		Then $\computation \in \sem{\alpha^\semvar}{}$ for some $\semvar \in \naturals$.
		Now, prove \acCommit and \acPost by induction on $\semvar$:
		If $\semvar = 0$,
		then \acCommit and \acPost hold by $\pstate{v} \vDash [ \alpha^0 ] \ac \psi$.
		If $\semvar > 0$,
		then $\run \in \sem{\alpha^\semvar}{} = \sem{\alpha \seq \alpha^{\semvar-1}}{}$.
		By associativity of sequential composition,
		$\run \in \sem{\alpha^{\semvar-1} \seq \alpha}{} = \botop{\sem{\alpha^{\semvar-1}}{}} \cup \sem{\alpha^{\semvar-1}}{} \continuation \sem{\alpha}{}$.
		In case $\run \in \botop{\sem{\alpha^{\semvar-1}}{}} \subseteq \sem{\alpha^{\semvar-1}}{}$, 
		\acCommit and \acPost hold by the induction hypothesis.
		Otherwise, if $\computation \in \sem{\alpha^{\semvar-1}}{} \continuation \sem{\alpha}{}$, there are $(\pstate{v}, \trace_1, \pstate{u}) \in \sem{\alpha^{\semvar-1}}{}$ and $(\pstate{u}, \trace_2, \pstate{w}) \in \sem{\alpha}{}$ with $\trace = \trace_1 \cdot \trace_2$.
		Further, $(\pstate{u} \cdot \trace_1, \trace_2, \pstate{w} \cdot \trace_1) \in \sem{\alpha}{}$ by coincidence (\rref{cor:history_coincidence}).
		\begin{enumerate}[leftmargin=*]
			\item If $\trace_2 = \epsilon$, \acCommit holds by the induction hypothesis (IH) since $(\pstate{v}, \trace, \bot) \in \sem{\alpha^{\semvar-1}}{}$ by prefix-closedness.
			If $\trace_2 \neq \epsilon$, assume $\assCommit{\pstate{v}}{\trace} \vDash \A$, so $\assPost{\pstate{v}}{\trace_1} \vDash \A$.
			Hence, $\pstate{u} \cdot \trace_1 \vDash \psi$ by the IH,
			and $\pstate{u} \cdot \trace_1 \vDash \psi \rightarrow [ \alpha ] \ac \psi$ by $\pstate{v} \vDash [ \repetition{\alpha} ] \acpair{\A, \true} (\psi \rightarrow [ \alpha ] \ac \psi)$ since $\sem{\alpha^{\semvar-1}}{} \subseteq \sem{\repetition{\alpha}}{}$.
			Therefore, $\pstate{u} \cdot \trace_1 \vDash [ \alpha ] \ac \psi$.
			Since $\assCommit{\pstate{v}}{\trace} \vDash \A$,
			obtain $\assCommit{\pstate{v} \cdot \trace_1}{\trace_2} \vDash \A$,
			so $\assCommit{\pstate{u} \cdot \trace_1}{\trace_2} \vDash \A$ by coincidence (\rref{cor:communicative_coincidence}).
			Thus, $\pstate{u} \cdot \trace_1 \cdot \trace_2 \vDash \Commit$, so $\pstate{v} \cdot \trace \vDash \Commit$ by \rref{cor:communicative_coincidence} again and $\trace = \trace_1 \cdot \trace_2$.

			\item For \acPost, assume $\pstate{w} \neq \bot$ and $\assPost{\pstate{v}}{\trace} \vDash \A$.
			Hence, $\assPost{\pstate{v}}{\trace_1} \vDash \A$ and $\assPost{\pstate{v} \cdot \trace_1}{\trace_2} \vDash \A$. 
			By \rref{cor:communicative_coincidence},
			obtain $\assPost{\pstate{u} \cdot \trace_1}{\trace_2} \vDash \A$.
			As in case \acCommit, obtain $\pstate{u} \cdot \trace_1 \vDash [ \alpha ] \ac \psi$ (using the IH).
			Therefore, $\pstate{w} \cdot \trace_1 \cdot \trace_2 \vDash \psi$,
			so $\pstate{w} \cdot \trace \vDash \psi$ because $\trace = \trace_1 \cdot \trace_2$.
		\end{enumerate}
		
		\item[] \RuleName{acDropComp}:
		Let $\pstate{v} \vDash [ \alpha ] \ac \psi$ and $\computation \in \sem{\alpha \parOp \beta}{}$. 
		Then $(\pstate{v}, \trace \downarrow \alpha, \pstate{w}_\alpha) \in \sem{\alpha}{}$ with $\pstate{w} = \pstate{w}_\alpha \merge \pstate{w}_\beta$ for some $\pstate{w}_\beta \in \botop{\states}$.
		For \acCommit, assume $\assCommit{\pstate{v}}{\trace} \vDash \A$.
		Observe that if $\trace[pre]_\alpha \prec \trace \downarrow \alpha$, then $\trace[pre] \prec \trace$ exists such that $\trace[pre]_\alpha = \trace[pre] \downarrow \alpha$. 
		Thus, $(\pstate{v}, \trace[pre], \bot) \in \sem{\alpha \parOp \beta}{}$ and $(\pstate{v}, \trace[pre]_\alpha, \bot) \in \sem{\alpha}{}$ by prefix-closedness.
		Since $\beta$ does not interfere with $[ \alpha ] \ac \psi$ (\rref{def:noninterference}), 
		obtain $\assCommit{\pstate{v}}{\trace \downarrow \alpha} \vDash \A$,
		using \rref{lem:noninterference} for each $\trace[pre] \prec \trace \downarrow \alpha$.
		Hence, $\pstate{v} \cdot (\trace \downarrow \alpha) \vDash \Commit$ by $\pstate{v} \vDash [ \alpha ] \ac \psi$,
		which implies $\pstate{v} \cdot \trace \vDash \Commit$ using \rref{lem:noninterference} again.
		For \acPost, assume $\pstate{w} \neq \bot$ and $\assPost{\pstate{v}}{\trace} \vDash \A$. 
		Then $\pstate[ind=\alpha]{w} \neq \bot$ by definition of $\merge$ in \rref{sec:semantics}.
		Moreover, $\assPost{\pstate{v}}{\trace \downarrow \alpha} \vDash \A$ by \rref{lem:noninterference} as above. 
		Hence, $\pstate{w}_\alpha \cdot (\trace \downarrow \alpha) \vDash \psi$ by $\pstate{v} \vDash [ \alpha ] \ac \psi$.
		Finally, $\pstate{w} \cdot \trace \vDash \psi$ by \rref{lem:noninterference} again.
		
		\item[] \RuleName{acCom}:
		For all $(\pstate{v}, \trace, \pstate{w}) \in \sem{\send{}{}}{}$,
		obtain $(\trace, \pstate{w}) \preceq (\comItem{\ch{}, \semConst}, \pstate{v})$
		with $\semConst = \sem{\rp}{\pstate{v}}$.
		Hence, for any $\pstate{v}$, 
		satisfaction $\pstate{v} \vDash [ \send{}{} ] \ac \psi$ 
		can be divided into the following cases:
		\begin{enumerate*}[label=(\roman*)]
			\item For $(\pstate{v}, \epsilon, \bot) \in \sem{\send{}{}}{}$,
			\acCommit holds iff $\pstate{v} \vDash \Commit$.
			\item For $(\pstate{v}, \comItem{\ch{}, \semConst}, \bot) \in \sem{\send{}{}}{}$,
			the run $(\pstate{v}, \comItem{\ch{}, \semConst}, \pstate{w}) \in \sem{\alpha}{}$ 
			with $\pstate{w}=\pstate{v}$ exists, 
			because sending $\send{}{}$ never aborts execution.
			Thus, \acCommit holds for $(\pstate{v}, \comItem{\ch{}, \semConst}, \bot) \in \sem{\alpha}{}$
			iff $\pstate{v} \vDash \A$ implies $\pstate{w} \cdot \trace \vDash \Commit$.
			\item For $(\pstate{v}, \comItem{\ch{}, \semConst}, \pstate{w}) \in \sem{\send{}{}}{}$
			with $\pstate{w}=\pstate{v}$,
			\acPost holds
			iff $\pstate{v} \vDash \A$ and $\pstate{w} \cdot \trace \vDash \A$ imply $\pstate{w} \cdot \trace \vDash \psi$.
		\end{enumerate*}
		Consequently, $[ \send{}{} ] \ac \psi \leftrightarrow [ \test{\true} ] \ac [ \send{}{} ] [ \test{\true} ] \ac \psi$ is valid,
		since $[ \test{\true} ] \ac \lambda \leftrightarrow \Commit \wedge (\A \rightarrow \lambda)$ is valid for every~$\lambda$ by the axioms \RuleName{acNoCom} and~\RuleName{test}.

		\item[] \RuleName{send}:
		Let $\pstate{v} \vDash [ \send{}{} ] \psi(\recvar)$.
		For proving $\fa{\tvar_0}$, assume $\pstate{v} \subs{\tvar_0}{\trace_0} \vDash \tvar_0 = \recvar \cdot \comItem{\ch{}, \rp}$
		for some trace~$\trace_0$.
		Hence, $\sem{\recvar \cdot \comItem{\ch{}, \rp}}{\pstate{v}} = \sem{\tvar_0}{\pstate{v}} = \trace_0$.
		Since $(\pstate{v}, \trace, \pstate{v}) \in \sem{\send{}{}}{}$ with $\trace = \comItem{\ch{}, \sem{\rp}{\pstate{v}}}$,
		obtain $\pstate{v} \cdot \trace \vDash \psi(\recvar)$ by $\pstate{v} \vDash [ \send{}{} ] \psi(\recvar)$.
		Then $\pstate{v} \subs{\recvar}{\trace_0} \vDash \psi(\recvar)$ because $\pstate{v} \cdot \trace = \pstate{v} \subs{\recvar}{\sem{\recvar \cdot \comItem{\ch{}, \rp}}{\pstate{v}}} = \pstate{v} \subs{\recvar}{\trace_0}$.
		Finally, $\pstate{v} \subs{\tvar_0}{\trace_0} \vDash \psi(\tvar_0)$ by substitution (\rref{lem:rec_substitution}).
		Conversely, let $\pstate{v} \vDash \fa{\tvar_0} (\tvar_0 = \recvar \cdot \comItem{\ch{}, \rp} \rightarrow \psi(\tvar_0))$ and $\run \in \sem{\send{}{}}{}$ with $\pstate{w} \neq \bot$.
		Hence, $\pstate{u} \vDash \psi(\tvar_0)$
		for $\pstate{u} = \pstate{v} \subs{\tvar_0}{\trace_0}$ with $\trace_0 = \sem{\recvar \cdot \comItem{\ch{}, \rp}}{\pstate{v}}$.
		Then $\pstate{u} \subs{\recvar}{\sem{\tvar_0}{\pstate{u}}} \vDash \psi(\recvar)$ by \rref{lem:rec_substitution}.
		Since
		$\sem{\tvar_0}{\pstate{u}} = \sem{\recvar \cdot \comItem{\ch{}, \rp}}{\pstate{u}} = \sem{\recvar \cdot \comItem{\ch{}, \rp}}{\pstate{v}} = \trace_0$ by coincidence (\rref{lem:expr_coincidence}),
		as $\tvar_0$ is fresh,
		obtain $\pstate{u} \subs{\recvar}{\trace_0} \vDash \psi(\recvar)$.
		Hence, $\pstate{v} \subs{\recvar}{\trace_0} \vDash \psi(\recvar)$ by \rref{lem:expr_coincidence}, as $\tvar_0$ is fresh.
		Since $\trace = \comItem{\ch{}, \sem{\rp}{\pstate{v}}}$, 
		obtain $\pstate{v} \cdot \trace = \pstate{v} \subs{\recvar}{\trace_0}$.
		Finally, $\pstate{w} \cdot \trace \vDash \psi(\recvar)$ because $\pstate{w} = \pstate{v}$.

		\item[] \RuleName{comDual}:
		Let $\pstate{v} \vDash [ \receive{}{} ] \ac \psi$.
		Then $\pstate{v} \vDash [ x \ceq * ] [ \send{}{x} ] \ac \psi$,
		because $\pstate{v} \subs{x}{\semConst} \vDash [ \send{}{x} ] \ac \psi$
		for any $d\in\reals$ as follows:
		Let $(\pstate{v} \subs{x}{\semConst}, \trace, \pstate{w}) \in \sem{\send{}{x}}{}$.
		Hence, $\run \in \sem{\receive{}{}}{}$.
		For \acCommit, assume $\assCommit{\pstate{v} \subs{x}{\semConst}}{\trace} \vDash \A$.
		Hence, $\assCommit{\pstate{v}}{\trace} \vDash \A$ by coincidence (\rref{cor:communicative_coincidence}),
		so $\pstate{v} \cdot \trace \vDash \Commit$ by $\pstate{v} \vDash [ \receive{}{} ] \ac \psi$.
		By \rref{cor:communicative_coincidence}, $\pstate{v} \subs{x}{\semConst} \cdot \trace \vDash \Commit$.
		For \acPost, assume $\pstate{w} \neq \bot$ and $\assPost{\pstate{v} \subs{x}{\semConst}}{\trace} \vDash \A$.
		By \rref{cor:communicative_coincidence}, $\assPost{\pstate{v}}{\trace} \vDash \A$.
		Hence, $\pstate{w} \cdot \trace \vDash \psi$ by $\pstate{v} \vDash [ \receive{}{} ] \ac \psi$.
		Conversely, let $\pstate{v} \vDash [ x \ceq * ] [ \send{}{x} ] \ac \psi$.
		To prove $\pstate{v} \vDash [ \receive{}{} ] \ac \psi$,
		let $\run \in \sem{\receive{}{}}{}$.
		Then $(\pstate{v}, \epsilon, \pstate{v} \subs{x}{\semConst}) \in \sem{x \ceq *}{}$ and $(\pstate{v} \subs{x}{\semConst}, \trace, \pstate{w}) \in \sem{\send{}{x}}{}$
		with $d\in\reals$ exist.
		Hence, $\pstate{v} \subs{x}{\semConst} \vDash [ \send{}{x} ] \ac \psi$ 
		by $\pstate{v} \vDash [ x \ceq * ] [ \send{}{x} ] \ac \psi$.
		For \acCommit, assume $\assCommit{\pstate{v}}{\trace} \vDash \A$.
		By \rref{cor:communicative_coincidence}, 
		$\assCommit{\pstate{v} \subs{x}{\semConst}}{\trace} \vDash \A$.
		Hence, $\pstate{v} \subs{x}{\semConst} \cdot \trace \vDash \Commit$ by $\pstate{v} \subs{x}{\semConst} \vDash [ \send{}{x} ] \ac \psi$, 
		so $\pstate{v} \cdot \trace \vDash \Commit$ by \rref{cor:communicative_coincidence}.
		For \acPost, assume $\pstate{w} \neq \bot$ and $\assPost{\pstate{v}}{\trace} \vDash \A$.
		By \rref{cor:communicative_coincidence}, $\assPost{\pstate{v} \subs{x}{\semConst}}{\trace} \vDash \A$.
		Hence, $\pstate{w} \cdot \trace \vDash \psi$ by $\pstate{v} \subs{x}{\semConst} \vDash [ \send{}{x} ] \ac \psi$.

		\newcommand{\odeone}{\mathtt{ode}}
		\newcommand{\odetwo}{\evolution*{t'=1,\mathtt{ode}}{non}}

		\item[] \RuleName{odeSync}:
		For readability, 
		let $\odeone \equiv \evolution*{x'=\rp}{}$.
		Let $\pstate{v} \vDash [ \odeone ] \ac \psi$,
		and assume $\pstate{v}(\tvar_0) = \pstate{v}(\recvar)$ and $\pstate{v}(t) = 0$.
		To prove $\pstate{v} \vDash [ \odetwo ] \fa{\recvar{=}\tvar_0} [ \send{\durchan}{t} ] \ac \psi$,
		let $\run \in \sem{\odetwo}{}$ with $\pstate{w}\neq\bot$.
		Hence, $(\trace, \pstate{w}) = (\comItem{\durchan, \duration}, \odeSolution(\duration))$
		for a solution $\odeSolution : [ 0, \duration ] \rightarrow \states$
		with $\odeSolution(\zeta)(\avar') = \solutionDerivative{\odeSolution}{\avar}(\zeta)$
		for $\avar\in\{x,t\}$
		and $\odeSolution(\zeta) = \odeSolution(0)$ on $\{x,t\}^\complement$
		such that $\odeSolution(\zeta) \vDash \odetwo$ 
		for all $\zeta\in[0,\duration]$.
		Then $\odeSolution[alt] : [ 0, \duration ] \rightarrow \states$ 
		with $\odeSolution[alt](\zeta) = \odeSolution(\zeta) \subs{t}{0}$ for $\zeta\in[0,\duration]$ 
		is a solution for $\odeone$
		such that $(\pstate{v}, \trace, \odeSolution[alt](\duration)) \in \sem{\odeone}{}$.
		For $\rawtrace \preceq \trace$,
		obtain $\pstate{v} \cdot \rawtrace \vDash \Commit$
		if $\assCommit{\pstate{v}}{\rawtrace} \vDash \A$
		from $\pstate{v} \vDash [ \odeone ] \ac \psi$ 
		by \acCommit.
		Further, 
		$\odeSolution[alt](\duration) \cdot \trace \vDash \psi$  
		if $\assPost{\pstate{v}}{\trace} \vDash \A$
		by $\pstate{v} \vDash [ \odeone ] \ac \psi$ and \acPost.
		By \rref{cor:communicative_coincidence}, 
		$\pstate{v} = \odeSolution[alt](\duration)$ on $\SFV(\A, \Commit)$.
		Further, $\odeSolution[alt](\duration) = \odeSolution(\duration)$ on $\SFV(\A, \Commit, \psi)$
		since $t$ is fresh.
		Therefore, $\odeSolution(\duration) \vDash [ \send{\durchan}{t} ] \ac \psi$
		by \acCommit, and \acPost, and coincidence (\rref{lem:expr_coincidence}),
		because $\trace = \comItem{\durchan, \duration} = \comItem{\durchan, \odeSolution(\duration)(t)}$.
		Since $\odeSolution(\duration)(\recvar) = \pstate{v}(\recvar) = \pstate{v}(\tvar_0) = \odeSolution(\duration)(\tvar_0)$
		by the bound effect property (\rref{lem:bound_effect}), 
		obtain $\odeSolution(\duration) \vDash \fa{\recvar{=}\tvar_0} [ \send{\durchan}{t} ] \ac \psi$.
		Hence, $\odeSolution(\duration) \cdot \trace \vDash \fa{\recvar{=}\tvar_0} [ \send{\durchan}{t} ] \ac \psi$,
		by \rref{lem:expr_coincidence},
		since $\odeSolution(\duration) \cdot \trace = \odeSolution(\duration)$ on $\{\recvar\}^\complement$.
		This proves $\pstate{v} \vDash \fa{\tvar_0{=}\recvar} \fa{t{=}0} [ \odetwo ] \fa{\recvar{=}\tvar_0} [ \send{\durchan}{t} ] \ac \psi$.
		The converse can be obtained 
		analogously,
		defining a solution $\odeSolution$ for $\odetwo$ from a solution $\odeSolution[alt]$ for $\odeone$
		by $\odeSolution(\zeta) = \odeSolution[alt](\zeta) \subs{t}{\zeta}$.

		\item[] \RuleName{acG}:
		If $\Commit \wedge \psi$ is valid,
		\acCommit and \acPost for $\pstate{v} \vDash [ \alpha ] \ac \psi$ hold in any state $\pstate{v}$.

		\item[] \RuleName{hExtension}:
		Let $\pstate{v} \vDash \tvar_0 = \recvar$ and $\run \in \sem{\alpha}{}$.
		Since $\recvar$ is $\alpha$'s recorder,
		$\trace(\tvar_0) = \epsilon$.
		For \acCommit, $\pstate{v} \cdot \trace \vDash \recvar \succeq \tvar_0$
		because $(\pstate{v} \cdot \trace)(\recvar) 
			\succeq \pstate{v}(\recvar)
			= \pstate{v}(\tvar_0)
			= (\pstate{v} \cdot \trace)(\tvar_0)$.
		For \acPost, assume $\pstate{w} \neq \bot$.
		By the bound effect property (\rref{lem:bound_effect}), 
		$\pstate{v} \vDash \tvar_0 = \recvar$ implies $\pstate{w} \vDash \tvar_0 = \recvar$.
		Then $(\pstate{w} \cdot \trace)(\recvar) 
			\succeq \pstate{w}(\recvar) 
			= \pstate{w}(\tvar_0)
			= (\pstate{w} \cdot \trace)(\tvar_0)$
		such that $\pstate{w} \cdot \trace \vDash \recvar \succeq \tvar_0$.

		\item[] \RuleName{acLiveParCommit}:
		Intuitively,
		there is a run for $\alpha\parOp\beta$ 
		satisfying $\Commit(\recvar)$
		by the premise,
		because there is a communication history $\auxhis$ for $\alpha\parOp\beta$
		satisfying $\Commit(\tvar_0 \cdot \auxhis)$,
		which is observable from the subprograms by $\runcommit{\gamma}{\auxhis\downarrow\gamma}$
		for $\gamma\in\{\alpha,\beta\}$.
		Formally, let $\pstate[alt]{v} \vDash \ex{\auxhis{=}\auxhis\downarrow(\alpha{\parOp}\beta)} \fa{\tvar_0{=}\recvar} \phi$,
		where $\phi \equiv
			\runcommit{\alpha}{\auxhis\downarrow\alpha}
			\wedge \mbox{$\runcommit{\beta}{\auxhis\downarrow\beta}$}
			\wedge \Commit(\tvar_0\cdot\auxhis)$.
		Then $\pstate{v} \vDash \phi$ for some $\pstate{v}$ with $\pstate{v} = \pstate[alt]{v}$ on $\{\auxhis, \tvar_0\}^\complement$,
		and $\pstate{v}(\auxhis) \downarrow (\alpha\parOp\beta) = \pstate{v}(\auxhis)$, and $\pstate{v}(\tvar_0) = \pstate [alt]{v}(\recvar)$.
		Observe that the semantics of the diamond $\runcommit{\gamma}{\auxhis\downarrow\gamma}$
		with trace constraint $\te$
		is as follows by its definition in \rref{fig:calculus} and \acCommit:
		\begin{equation*}
			\sem{\runcommit{\gamma}{\auxhis\downarrow\gamma}}{} = \{ \pstate{v} \mid \bfexists \trace : (\pstate{v} \subs{\recvar}{\epsilon}, \trace, \bot) \in \sem{\gamma}{} \text{ and } \trace = \pstate{v}(\auxhis) \downarrow \gamma \}
		\end{equation*}

		Hence, by $\pstate{v} \vDash \phi$, 
		there is a run $(\pstate{v} \subs{\recvar}{\epsilon}, \trace_\gamma, \bot) \in \sem{\gamma}{}$ for each $\gamma\in\{\alpha,\beta\}$
		with $\trace_\gamma = \pstate{v}(\auxhis) \downarrow \gamma$.
		This implies $(\pstate{v} \subs{\recvar}{\epsilon}, \pstate{v}(\auxhis), \bot) \in \sem{\alpha\parOp\beta}{}$,
		as $\pstate{v}(\auxhis) \downarrow (\alpha\parOp\beta) = \pstate{v}(\auxhis)$.
		Since $\pstate{v} \subs{\recvar}{\epsilon} = \pstate[alt]{v}$ on $\{\recvar, \auxhis, \tvar_0\}^\complement \supseteq \SFV(\alpha\parOp\beta)$,
		obtain \mbox{$(\pstate[alt]{v}, \pstate{v}(\auxhis), \bot) \in \sem{\alpha\parOp\beta}{}$} by coincidence (\rref{lem:program_coincidence}).
		Further, $\pstate{v} \subs{\recvar}{\trace} \vDash \Commit(\recvar)$ by substitution (\rref{lem:rec_substitution}),
		where \mbox{$\trace = \pstate{v}(\tvar_0) \cdot \pstate{v}(\auxhis)$},
		because $\pstate{v} \vDash \Commit(\tvar_0\cdot\auxhis)$ by $\pstate{v} \vDash \phi$.
		Then $\pstate[alt]{v} \subs{\recvar}{\trace} \vDash \Commit(\recvar)$
		by coincidence (\rref{lem:expr_coincidence}),
		as $\tvar_0, \auxhis \not\in \SFV(\Commit(\recvar))$.
		Since $\trace = \pstate[alt]{v}(\recvar) \cdot \pstate{v}(\auxhis)$,
		because $\pstate{v}(\tvar_0) = \pstate[alt]{v}(\recvar)$,
		obtain $\pstate[alt]{v} \cdot \pstate{v}(\auxhis) \vDash \Commit(\recvar)$.
		Finally, $\pstate[alt]{v} \vDash \langle \alpha\parOp\beta \rangle \acpair{\true, \Commit(\recvar)} \false$
		by \acCommit.

		Observe that all steps can be reversed,
		so that the axiom becomes an equivalence.
		This is not necessary for completeness of the proof calculus,
		but used in the completeness proof to transfer validity.

		\item[] \RuleName{acLivePar}:
		Recall the axiom and introduce abbreviations $\mathcal{Q}$ and $\phi$ for the proof:
		\begin{equation*}
			\underbrace{
				\ex{\auxhis{=}\auxhis\downarrow(\alpha{\parOp}\beta)}
				\fa{\tvar_0{=}\recvar}
			}_{\mathcal{Q}}
			\underbrace{
				\runpost{\alpha}{}{\auxhis\downarrow\alpha}{M}
				\runpost{\beta}{}{\auxhis\downarrow\beta}{M}
				\psi(\tvar_0 \cdot \auxhis)
			}_{\phi}
			\rightarrow \axkey{\langle \alpha\parOp\beta \rangle \psi(\recvar)}
		\end{equation*}
		
		As for \RuleName{acLiveParCommit},
		there is a communication history $\auxhis$ for $\alpha\parOp\beta$ by the premise.
		Further, sequential reachability of a state satisfying $\psi$ 
		guarantees parallel reachability of this state,
		because the subprograms do not share state (\rref{def:syntax_chps}),
		\iest they do not write each other's free and bound variables.
		Consequently, the $\beta$-run can be aligned with the $\alpha$-run by coincidence (\rref{lem:program_coincidence}),
		so that both runs combine to a run of $\alpha\parOp\beta$.
		See \rref{fig:livePar} for an illustration.

		Let $\pstate[alt]{v} \vDash \mathcal{Q}\, \phi$.
		Hence, $\pstate{v} \vDash \phi$ 
		for some~$\pstate{v}$ with $\pstate{v} = \pstate[alt]{v}$ on $\{\auxhis, \tvar_0 \}^\complement$,
		and $\pstate{v}(\auxhis) \downarrow (\alpha\parOp\beta) = \pstate{v}(\auxhis)$, and $\pstate{v}(\tvar_0) = \pstate[alt]{v}(\recvar)$.
		Observe that the semantics of the diamond $\runpost{\gamma}{}{\auxhis\downarrow\gamma}{M} \psi$
		with trace constraint $\te$
		is as follows by its definition in \rref{fig:calculus} and \acPost,
		using the bound effect property (\rref{lem:bound_effect}) to observe that $\pstate{v} \subs{\recvar}{\epsilon} = \pstate{w}$
		on $\SBV(\gamma)^\complement \supseteq \{\auxhis \}$:
		\begin{equation*}
			\sem{\runpost{\gamma}{}{\auxhis\downarrow\gamma}{M} \psi}{} = 
			\Bigg\{ \pstate{v} \;\bigg\vert\;
				\begin{aligned}
					\bfexists \trace, \omega :\;
					& (\pstate{v} \subs{\recvar}{\epsilon}, \trace, \pstate{w}) \in \sem{\gamma}{}
					\text{, and } \trace = \pstate{v}(\auxhis) \downarrow \gamma \text{,} \\
					&\text{and } \pstate{w}\neq\bot \text{, and } \pstate{w} \cdot \trace \vDash \psi
				\end{aligned} 
			\Bigg\}
		\end{equation*}

		Hence, by $\pstate{v} \vDash \phi$,
		there is a run $(\pstate{v}_\alpha, \trace_\alpha, \pstate{w}_\alpha) \in \sem{\alpha}{}$ 
		with $\pstate{v}_\alpha = \pstate{v} \subs{\recvar}{\epsilon}$,
		and $\trace_\alpha = \pstate{v}(\auxhis) \downarrow \alpha$,
		and $\pstate{w}_\alpha \neq \bot$,
		and $\pstate{w}_\alpha \cdot \trace_\alpha \vDash \runpost{\beta}{}{\auxhis\downarrow\beta}{M}
				\psi(\tvar_0 \cdot \auxhis)$.
		Consequently, there is also a run $(\pstate{v}_\beta, \trace_\beta, \pstate{w}_\beta) \in \sem{\beta}{}$
		with $\pstate{v}_\beta = (\pstate{w}_\alpha \cdot \trace_\alpha) \subs{\recvar}{\epsilon}$,
		and $\trace_\beta = (\pstate{w}_\alpha \cdot \trace_\alpha)(\auxhis) \downarrow \beta$,
		and $\pstate{w}_\beta \neq \bot$,
		and $\pstate{w}_\beta \cdot \trace_\beta \vDash \psi(\tvar_0 \cdot \auxhis)$.
		Using \rref{lem:bound_effect}, obtain $\trace_\beta = \pstate{v}(\auxhis) \downarrow \beta$.

		Parallel programs do not share state (\rref{def:syntax_chps}),
		\iest $\SV(\beta) \setminus \gvars \subseteq \SBV(\alpha)^\complement$,
		where $\SV(\cdot) = \SBV(\cdot) \cup \SFV(\cdot)$.
		Hence, $\pstate[alt]{v} = \pstate{v}_\gamma$ on $\SV(\gamma)\setminus\{\recvar\}\supseteq\SFV(\gamma)$
		for $\gamma\in\{\alpha,\beta\}$,
		using \rref{lem:bound_effect} if $\gamma\equiv\beta$,
		because $\auxhis, \tvar_0$ are fresh
		and $\recvar\not\in\SFV(\gamma)$ by \rref{cor:history_coincidence}.
		By coincidence (\rref{lem:program_coincidence}),
		obtain $(\pstate[alt]{v}, \trace_\gamma, \pstate[alt]{w}_\gamma) \in \sem{\gamma}{}$ such that $\pstate[alt]{w}_\gamma = \pstate{w}_\gamma$ on $\SV(\gamma) \setminus \{\recvar\}$.
		Hence, $(\pstate[alt]{v}, \pstate{v}(\auxhis), \pstate[alt]{w}_\alpha \merge \pstate[alt]{w}_\beta) \in \sem{\alpha\parOp\beta}{}$.
		
		Since $\pstate{v}_\beta = (\pstate{w}_\alpha \cdot \trace_\alpha) \subs{\recvar}{\epsilon}$
		and $\SV(\beta) \setminus \gvars \subseteq \SBV(\alpha)^\complement$,
		obtain $\pstate{w}_\beta = \pstate{w}_\alpha$ on $\SBV(\alpha)\setminus\{\recvar\}$ by \rref{lem:bound_effect}.
		Hence, $\pstate{w}_\beta = \pstate[alt]{w}_\alpha \merge \pstate[alt]{w}_\beta$ on $\varset = (\SBV(\alpha) \cup \SBV(\beta)) \setminus \{\recvar\}$,
		because $\pstate[alt]{w}_\gamma = \pstate{w}_\gamma$ on $\SV(\gamma) \setminus \{\recvar\}$.
		Further, $\pstate{w}_\beta = \pstate[alt]{v} = \pstate[alt]{w}_\alpha \merge \pstate[alt]{w}_\beta$ on 
		$\varset^\complement \setminus \{ \recvar, \auxhis, \tvar_0 \}$
		by \rref{lem:bound_effect}.
		Hence, $\pstate{w}_\beta = \pstate[alt]{w}_\alpha \merge \pstate[alt]{w}_\beta$ on $\SFV(\psi(\recvar))\setminus \{\recvar\}$, 
		as $\auxhis, \tvar_0$ are fresh.
		Since $\pstate{w}_\beta \cdot \trace_\beta \vDash  \psi(\tvar_0 \cdot \auxhis)$,
		obtain $\pstate{w}_\beta \vDash \psi(\tvar_0 \cdot \auxhis)$,
		as~$\trace_\beta$ only changes $\recvar$.
		By substitution (\rref{lem:rec_substitution}),
		$(\pstate{w}_\beta) \subs{\recvar}{\trace} \vDash \psi(\recvar)$,
		where $\trace = \pstate{w}_\beta(\tvar_0) \cdot \pstate{w}_\beta(\auxhis)$.
		Hence, $(\pstate[alt]{w}_\alpha \merge \pstate[alt]{w}_\beta) \subs{\recvar}{\trace} \vDash \psi(\recvar)$
		by coincidence (\rref{lem:expr_coincidence}).
		Since
		\begin{equation*}
			((\pstate[alt]{w}_\alpha \merge \pstate[alt]{w}_\beta) \cdot \pstate{v}(\auxhis))(\recvar)
			= \pstate[alt]{v}(\recvar) \cdot \pstate{v}(\auxhis)
			= \pstate{v}(\tvar_0) \cdot \pstate{v}(\auxhis)
			= \pstate{w}_\beta(\tvar_0) \cdot \pstate{w}_\beta(\recvar)
			= \trace
			\text{,}
		\end{equation*}
		using \rref{lem:bound_effect},
		obtain $(\pstate[alt]{w}_\alpha \merge \pstate[alt]{w}_\beta) \cdot \pstate{v}(\auxhis) \vDash \psi(\recvar)$.
		In summary, $\pstate[alt]{v} \vDash \langle \alpha\parOp\beta \rangle \psi(\recvar)$
		by \acPost.

		Since all steps can be reversed,
		the axiom becomes an equivalence,
		which enables to transfer validity in the completeness proof.

		\item[] \RuleName{acConvergence}:	
		Let $\pstate{v} \vDash \A \wedge [ \repetition{\alpha} ] \acpair{\A, \true} \fa{v{>}0} (
			\varphi(v) \rightarrow \langle \alpha \rangle \acpair{\A, \false} \varphi(v-1)
		)$.
		Then prove $\pstate{v} \vDash \varphi(v) \rightarrow \langle \repetition{\alpha} \rangle \acpair{\A, \false} \ex{v{\le}0} \varphi(0)$ for all $\semConst = \pstate{v}(v)$
		by a well-founded induction on $\semConst$ for all states $\pstate{v}$.
		This proves $\pstate{v} \vDash \fa{v} (\varphi(v) \rightarrow \langle \repetition{\alpha} \rangle \acpair{\A, \false} \ex{v{\le}0} \varphi(0))$
		by coincidence (\rref{lem:expr_coincidence}),
		because $v$ is fresh,
		so neither free nor bound in $\A$ and $\alpha$.
		The proof uses \acPost without further notice.
		\begin{enumerate}[leftmargin=*]
			\item If $\semConst \le 0$,
			let $\pstate{v} \vDash  \varphi(v)$.
			Since $\pstate{v} \vDash v \le 0$, obtain $\pstate{v} \cdot \epsilon \vDash \ex{v{\le}0} \varphi(v)$.
			Hence, $\pstate{v} \vDash \langle \repetition{\alpha} \rangle \acpair{\A, \false} \ex{v{\le}0} \varphi(v)$,
			because $\assPost{\pstate{v}}{\epsilon} \vDash \A$ by $\pstate{v} \vDash \A$
			and $(\pstate{v}, \epsilon, \pstate{v}) \in \sem{\alpha^0}{} \subseteq \sem{\repetition{\alpha}}{}$.
		
			\item If $\semConst > 0$,
			let $\pstate{v} \vDash \varphi(v)$.
			Since $(\pstate{v}, \epsilon, \pstate{v}) \in \sem{\repetition{\alpha}}{}$ and $\assPost{\pstate{v}}{\epsilon} \vDash \A$,
			obtain $\pstate{v} \vDash v > 0 \wedge \varphi(v) \rightarrow \langle \alpha \rangle \acpair{\A, \false} \varphi(v-1)$ by the premise.
			Hence, there is a run $(\pstate{v}, \trace_1, \pstate{u}) \in \sem{\alpha}{}$ 
			with $\assPost{\pstate{v}}{\trace_1} \vDash \A$ 
			such that $\pstate{u} \cdot \trace_1 \vDash \varphi(v-1)$,
			because \mbox{$\pstate{v} \vDash v > 0$}.
			Since~$v$ is not bound by $\alpha$,
			\iest $v \not\in \SBV(\alpha)$,
			obtain $\pstate{v}(v) = (\pstate{u} \cdot \trace_1)(v)$ by the bound effect property (\rref{lem:bound_effect}),
			so that $(\pstate{u} \cdot \trace_1) \subs{v}{d-1} \vDash \varphi(v)$
			by substitution (\rref{lem:rec_substitution}).
			By the induction hypothesis,
			there is $(\pstate{u} \cdot \trace_1, \trace_2, \pstate{w}) \in \sem{\repetition{\alpha}}{}$ 
			with $\assPost{\pstate{u}\cdot\trace_1}{\trace_2} \vDash \A$
			such that $\pstate{w} \cdot \trace_2 \vDash \ex{v{\le}0} \varphi(v)$.
			Hence, $(\pstate{u}, \trace_2, \pstate[alt]{w}) \in \sem{\repetition{\alpha}}{}$ 
			with $\pstate{w} = \pstate[alt]{w} \cdot \trace_1$
			by \rref{cor:history_coincidence},
			so that $(\pstate{v}, \trace_1\cdot\trace_2, \pstate[alt]{w}) \in \sem{\repetition{\alpha}}{}$.
			Further, $\assPost{\pstate{v}\cdot\trace_1}{\trace_2} \vDash \A$ by \rref{cor:communicative_coincidence},
			using that $\SFV(\A) \cap \SBV(\alpha) = \emptyset$ by well-formedness (\rref{def:syntax_formulas}).
			Hence, $\assPost{\pstate{v}}{\trace_1\cdot\trace_2} \vDash \A$,
			because $\assPost{\pstate{v}}{\trace_1} \vDash \A$.
			Finally, $\pstate{v} \vDash \langle \repetition{\alpha} \rangle \acpair{\A, \false} \ex{v{\le}0} \varphi(v)$,
			because $\pstate[alt]{w} \cdot \trace_1 \cdot \trace_2 \vDash \ex{v{\le}0} \varphi(v)$
			by $\pstate{w} \cdot \trace_2 \vDash \ex{v{\le}0} \varphi(v)$.
			The induction is well-founded as $v$ decreased at least by one.
			\qedhere
		\end{enumerate}
	\end{itemize}
\end{proof}

\begin{proof}
	[Proof of \rref{cor:derived}.]
	The axioms and proof rules in \rref{fig:derived} are sound,
	because they derive in \dLCHP's calculus (\rref{fig:calculus}) 
	as follows:
	\begin{enumerate}
		\item Rule \RuleName{acMono} obtains monotonicity of the promises from \RuleName{acModalMP} using \RuleName{acG} following standard arguments,
		and likewise uses \RuleName{Aweak} and \RuleName{acG} for antitonicity of the assumption:%
		\footnote{%
			An alternative proof \cite{Brieger2023} additionally requires a contextual equivalence rule.
		}

		\begin{small}
			\begin{prooftree}
					\Axiom{$\A_2 \rightarrow \A_1$}

					\UnaryInf{$\Commit_2 \wedge \A_2 \rightarrow \A_1$}

					\RuleNameLeft{acG}{}
					\UnaryInf{$[ \alpha ] \acpair{\true, \Commit_2 \wedge \A_2 \rightarrow \A_1} \true$}
					
					\UnaryInf{$[ \alpha ] \acpair{\A_1, \Commit_1} \psi_1 \rightarrow [ \alpha ] \acpair{\true, \Commit_2 \wedge \A_2 \rightarrow \A_1} \true$}
		
						\Axiom{$\Commit_1 \rightarrow \Commit_2$}
		
						\Axiom{$\psi_1 \rightarrow \psi_2$}
		
					\RuleNameRight{MP}{}
					\BinaryInf{$(\Commit_1 \rightarrow \Commit_2) \wedge (\psi_1 \rightarrow \psi_2)$}
		
					\RuleNameRight{acG}{}
					\UnaryInf{$[ \alpha ] \acpair{\A_1, \Commit_1 \rightarrow \Commit_2} (\psi_1 \rightarrow \psi_2)$}
		
					\RuleNameRight{acModalMP}{}
					\UnaryInf{$[ \alpha ] \acpair{\A_1, \Commit_1} \psi_1 \rightarrow [ \alpha ] \acpair{\A_1, \Commit_2} \psi_2$}

				\RuleNameRight{Aweak}{}
				\BinaryInf{$[ \alpha ] \acpair{\A_1, \Commit_1} \psi_1 \rightarrow [ \alpha ] \acpair{\A_2, \Commit_2} \psi_2$}
			\end{prooftree}
		\end{small}
	
		\item Axiom \RuleName{acBoxesDist} derives from ac-modal modus ponens \RuleName{acModalMP} and ac-monotonicity~\RuleName{acMono} applying a standard modal logic argument to all promises~\cite{Brieger2023}.
	
		\item For \RuleName{diasDual}, chain the axioms \RuleName{dbDual}, \RuleName{boxesDual}, and \RuleName{acdbDual} (\rref{fig:axioms_modal_dualities}).
		The negation $\neg\false$ required by \RuleName{acdbDual} as commitment instead of $\true$ obtained by~\RuleName{boxesDual} can be introduced by~\RuleName{acMono}.
	
		\item The axioms \RuleName{acDiaNoCom} and \RuleName{acArrival}, and \RuleName{acSplitDia} 
		and the rule \RuleName{acDiaMono}
		derive from their ac-box counterpart \RuleName{acNoCom}, \RuleName{acInduction}, \RuleName{acBoxesDist}, and \RuleName{acMono}, respectively, by ac-duality \RuleName{acdbDual}. 

		\item The rule \RuleName{acInvariant} derives from the induction axiom \RuleName{acInduction} as follows:
		\begin{small}
			\begin{prooftree}
					\Axiom{$*$}
				
					\UnaryInf{$\Commit \wedge \psi \rightarrow (\Commit \wedge (\A \rightarrow (\true \rightarrow \psi)))$}
			
					\RuleNameLeft{acNoCom, test}{}
					\UnaryInf{$\Commit \wedge \psi \rightarrow [ \alpha^0 ] \ac \psi$}
				
					\Axiom{$\psi \rightarrow [ \alpha ] \ac \psi$}

					\UnaryInf{$\true \wedge (\psi \rightarrow [ \alpha ] \ac \psi)$}
			
					\RuleNameRight{acG}{}
					\UnaryInf{$\Commit \wedge \psi \rightarrow [ \repetition{\alpha} ] \acpair{\A, \true} (\psi \rightarrow [ \alpha ] \ac \psi)$}
			
				\RuleNameRight{MP}{}
				\BinaryInf{$\Commit \wedge \psi \rightarrow [ \alpha^0 ] \ac \psi \wedge [ \repetition{\alpha} ] \acpair{\A, \true} (\psi \rightarrow [ \alpha ] \ac \psi)$}
			
				\RuleNameRight{acInduction}{}
				\UnaryInf{$\Commit \wedge \psi \rightarrow [ \repetition{\alpha} ] \ac \psi$}
			\end{prooftree}
		\end{small}

		\item The rule \RuleName{acLoop} derives from \RuleName{acInvariant} following a standard argument \cite{Platzer18},
		which uses a cut with $\inv \rightarrow [ \repetition{\alpha} ] \ac \inv$ to introduce the invariant $\inv$.
		\qedhere
	\end{enumerate}
\end{proof}

\section{Definability of $\reals$-Gödel Encodings}
\label{app:real_goedel}

Completeness relies on $\reals$-Gödel encodings \cite{DBLP:journals/jar/Platzer08} being definable in FOD (\rref{lem:real_goedel}),
the FOD encoding of natural numbers (\rref{lem:isnat}),
and definability of rounding and some operations on traces (\rref{cor:defined_operators}).

\begin{lemma}
	[Definability of $\naturals$ {\cite[Theorem 2]{DBLP:journals/jar/Platzer08}}]
	\label{lem:isnat}
	The formula $\isNat{n}$, 
	which holds iff~the value of the variable $n \in \RVar$ is a natural number, 
	is definable in FOD.
	For a formula $\varphi$, 
	define $\fa{n{:}\naturals} \varphi \equiv \fa{n} (\isNat{n} \rightarrow \varphi)$
	and $\ex{n{:}\naturals} \varphi \equiv \ex{n} (\isNat{n} \wedge \varphi)$.
\end{lemma}

\begin{corollary}
	[Defined operators]
	\label{cor:defined_operators}

	The following operators can be defined in \comFOD:
	\begin{enumerate}
		\item Rounding $\lfloor \re \rfloor$ of a real number $\re$ (even definable in FOD).
		
		\item The slice $\at{\te}{0, y}$ of $\te$ from the $0$-th (inclusive) up to the $\lfloor y \rfloor$-th item (exclusive).

		\item The formulas $\maxdursuffix{\te}{d}$ and $\maxdurprefix{\te}{d}$ that hold if the absolute value $\abs{d}$
		is the maximal duration
		such that $\comItem{\durchan,d}$ is a suffix or prefix of $\te$, respectively.

		\item Disjoint concatenation $\te_1\ofconcat\te_2$,
		which assumes that $\te_1\cdot\te_2$ does not append nonzero duration events $\comItem{\durchan, \duration_i}$
		with $\duration_i\neq0$ for $i=1,2$.

		\item Sequences $(\tvar_i)_{i=1}^{n}$ of traces for $n \in \naturals$. 
	\end{enumerate}
\end{corollary}
\vspace{-.7cm}
\begin{proof}
	For a formula $\phi(\avar)$,
	define
	\begin{enumerate}
		\item $\phi(\lfloor \re \rfloor) \equiv \ex{n{:}\naturals} \big(
			\re{-}1 {<} n {\le} \re \wedge \phi(n) 
		\big)$,
		where $\ex{n{:}\naturals}$ is definable in FOD (\rref{lem:isnat})

		\item $\phi(\at{\te}{x, y}) 
			\equiv \ex{\tvar} \big(
				\len{\tvar} = \lfloor y \rfloor - \lfloor x \rfloor \wedge
				\fa{0 {\le} k {<} \len{\tvar}}
					\at{\tvar}{k} = \at{\te}{\lfloor x \rfloor + k}
				\wedge \phi(\tvar)
			\big)$

		\item $\maxdursuffix{\te}{d} \equiv \ex{\tvar} \big( \te = \tvar \cdot \comItem{\durchan, d} \wedge \neg\ex{\tvar',d{\neq}0} \tvar = \tvar' \cdot \comItem{\durchan, d} \big)$
		
		$\maxdurprefix{\te}{d} \equiv \ex{\tvar} \big( \te = \comItem{\durchan, d} \cdot \tvar \wedge \neg\ex{\tvar',d{\neq}0} \tvar = \comItem{\durchan, d} \cdot \tvar' \big)$

		\item $\phi(\te_1\ofconcat\te_2) \equiv (\maxdursuffix{\te_1}{0} \vee \maxdurprefix{\te_2}{0}) \rightarrow \phi(\te_1\cdot\te_2)$
			
		\item $\ex{(\tvar_i)_{i=1}^{n}} \phi(\tvar_i) \equiv
			\ex{\tvar_0}\ex{\iarr{:}\traces} \big(
				\len{\iarr} = n+1 \wedge \fa{1{\le}k{\le}n} \phi(\at{\tvar_0}{\at{\iarr}{k-1}, \at{\iarr}{k}})
			\big)$
		\qedhere
	\end{enumerate}
\end{proof}

\begin{lemma}
	[$\reals$-Gödel encoding {\cite[Lemma 4]{DBLP:journals/jar/Platzer08}}]
	\label{lem:real_goedel}
	Let $Z$, $n$ $j$, and $x$ be real variables.
	Then the formula $\rawgoedelat{Z, n, j, x}$,
	which holds iff~$Z$ represents a Gödel encoding of a sequence of $n$ real numbers
	such that $x$ is at position $j$ is definable in FOD.
	For a formula $\phi(x)$,
	write $\phi(\goedelat{Z}{n}{j})$ to abbreviate $\ex{x} ( \rawgoedelat{Z, n, j, x} \wedge \phi(x))$.
\end{lemma}

\section{Completeness Relative to \texorpdfstring{\comFOD}{Omega}}
\label{app:com-fod-completeness}

\begin{proof}
	[Proof of \rref{lem:rendition}.]
	The proof generalizes the rendition for \dL \cite[Lemma 5]{DBLP:journals/jar/Platzer08}.
	Figure \ref{fig:homogeneous_rendition} defines the formula $\rendition{\alpha}{}{}$ in \comFOD
	inductively along the structure of~$\alpha$.
	Separate cases $\rendition{\alpha}{}{\renparams{}{}{\false}}$ and $\rendition{\alpha}{}{\renparams{}{}{\true}}$
	define $\rendition{\alpha}{}{\renparams{}{}{}}$ 
	by $(\finmarker{} \rightarrow \rendition{\alpha}{}{\renparams{}{}{\true}}) \wedge (\neg\finmarker{} \rightarrow \rendition{\alpha}{}{\renparams{}{}{\false}})$.
	
	Communication $\send{}{}$ and $\receive{}{}$ essentially appends $\comItem{\ch{}, \rp}$ and $\comItem{\ch{}, y}$ for some~$y$ to the history $\recvar$.
	For $\receive{}{}$, 
	the $\exists$-quantification $\ex{y}\!$ expresses that the environment controls the received value.
	In case \mbox{$\evolution*{}{}$},
	the domain constraint $\chi$ is eliminated by reversing the flow and checking $\chi$ backwards along the differential equation.
	Nested modalities can be avoided with appropriate care \mbox{\cite[Lemma 5]{DBLP:journals/jar/Platzer08}}.
	Further, the observable duration $\comItem{\durchan, t}$ is appended to the recorder $\recvar$.
	Runs to intermediate states of the atomic statements except for tests
	are encoded modularly as prefixes of runs to final states.
	Tests cannot be covered this way, because if the test fails, there is no run at all.
	A run of~$\alpha\seq\beta$,
	is either an unfinished run of $\alpha$,
	so $\neg\finmarker{}$ holds and $\alpha$ runs to an intermediate state $\varvec[alt]$ by $\rendition{\alpha}{}{\renparams{}{}{\false}}$,
	or~$\alpha$ reaches a final state $\varvec[opt]$ by $\rendition{\alpha}{}{\renparams{}{\varvec[opt]}{\true}}$ from which $\beta$ continues by $\rendition{\beta}{}{\renparams{\varvec[opt]}{}{}}$.

	In case $\repetition{\alpha}$,
	a finite formula must capture unboundedly many 
	intermediate states.
	For each real variable $\avar_j\in\varvec$,
	the sequence of intermediate states is stored in a trace variable~$\tvar_j$.
	For $\recvar\in\varvec$,
	the reached history $\tvarof{\avar[alt]}\in\varvec[alt]$ already contains the intermediate states as prefixes. 
	To identify the endpoints %
	of these prefixes in $\tvarof{\avar[alt]}$,
	the trace variable $\iarr$ serves as an index,
	so that $\at{\tvarof{\avar[alt]}}{0, \at{\iarr}{i-1}}$ is the history after $i{-}1$ iterations.
	The slice $\at{\te}{0, y}$ of $\te$ from the $0$-th (inclusive) up to the $\lfloor y \rfloor$-th item (exclusive) is definable in \comFOD (\rref{cor:defined_operators}).
	Taken together, the vector 
	$\repstate{i} \equiv \big( \at{\tvarof{\avar[alt]}}{0, \at{\iarr}{i{-}1}}, \val{\at{(\tvar_1,\ldots,\tvar_\semvar)}{i-1}} \big)$ 
	stores the $i$-th intermediate state
	of a repetition with $n{-}1$ loop passes,
	and $\ex{\goedelat{\repstatevar}{n}{}} \psi$ is an abbreviation as follows:
	\begin{equation*}
		\ex{\goedelat{\repstatevar}{n}{}} \psi \equiv \ex{\tvar_1 \ldots, \tvar_\semvar} \ex{\iarr{:}\traces}  (n = \len{\tvar_1} \wedge \ldots \wedge n = \len{\tvar_\semvar} \wedge n = \len{\iarr} \wedge \psi)
	\end{equation*}
	By $\orange{1}{i}{n-1} \vee \finmarker{}$,
	all but the last iteration must reach a final state.
	Quantification $\fa{n{:}\naturals} \phi$ and $\ex{n{:}\naturals} \phi$
	over natural numbers
	is definable in \comFOD by \rref{lem:isnat}.
	
	In case $\alpha\parOp\beta$,
	each subprogram $\gamma$ has a run from~$\varvec$ to an individual state $\varvec[alt]_\gamma$,
	where $\varvec[alt]_\gamma \equiv (\tvar_\gamma, \avar[alt]_{\gamma1}, ..., \avar[alt]_{\gamma\semvar})$ is fresh.
	By~$\tvar_\gamma = \recvar \cdot (\auxhis\downarrow\gamma)$, 
	the subtrace $\auxhis\downarrow\gamma$ of the overall history~$\auxhis$
	matches the history of the $\gamma$-run.
	By $\auxhis=\auxhis\downarrow(\alpha{\parOp}\beta)$, there is no noncausal communication.
	Like merging~$\merge$ on states,
	the real-valued part~$(\avar[alt]_1, \ldots, \avar[alt]_\semvar)$ of the reached state $\varvec[alt]$ results from merging~$\varvec[alt]_\alpha$ and~$\varvec[alt]_\beta$ by $\mergeren{j}$,
	where $\mergeren{j} \equiv \avar[alt]_{\alpha j}$ if $\avar_j\in\SBV(\alpha)$
	and $\mergeren{j} \equiv \avar[alt]_{\beta j}$ otherwise.
	In summary, $\rendition{\alpha\parOp\beta}{}{}$ exactly captures the parallel semantics.
\end{proof}

\begin{figure}
	\newcommand{\syncom}[1]{\recvar \cdot \comItem{\ch{}, #1}}
	\small
	\begin{align*}
		\rendition{x \ceq \rp}{}{\renparams{}{}{\true}}
			& \equiv \varvec[alt] = \varvec \subs{x}{\rp}
			& \rendition{\send{}{}}{}{\renparams{}{}{\true}}
			& \equiv \varvec[alt] = \varvec \subs{\recvar}{\syncom{\rp}} \\
		\rendition{x \ceq *}{}{\renparams{}{}{\true}}
			& \equiv \ex{y} \varvec[alt] = \varvec \subs{x}{y}
		& \rendition{\receive{}{}}{}{\renparams{}{}{\true}}
			& \equiv \ex{y} \varvec[alt] = (\varvec \subs{x}{y}) \subs{\recvar}{\syncom{y}} \\
		\rendition{\test{\chi}}{}{\renparams{}{}{\true}}
			& \equiv \chi \wedge \varvec[alt] = \varvec 
		& \rendition{\test{\chi}}{}{\renparams{}{}{\false}}
			& \equiv \varvec[alt] = \varvec
	\end{align*}
	\begin{align*}
		\rendition{\evolution*{}{}}{}{\renparams{}{}{\true}}
			& \equiv\ex{t{=}0} \langle \evolution*{x'=\rp, t' = 1}{non} \rangle \big(
				\varvec[alt] = \varvec \subs{\recvar}{\recvar\cdot\comItem{\durchan, t}} \wedge
				[ \evolution*{x' = -\rp, t' = -1}{non} ] (t{\ge}0 \rightarrow \chi)
			\big) \\
		\rendition{\alpha}{}{\renparams{}{}{\false}} 
			& \equiv \ex{\varvec[opt]} \big(
				\rendition{\alpha}{}{\renparams{}{\varvec[opt]}{\true}} \wedge \ex{\tvar'} (\recvar \preceq \tvar' \preceq \tvarof{\avar[opt]} \wedge \varvec[alt] = \varvec \subs{\recvar}{\recvar\cdot\tvar'})
			\big)
			\sidecond{\alpha \text{ atomic} \not\equiv\test{\chi}} \\[.1em]
		\rendition{\alpha\cup\beta}{}{}
			& \equiv \rendition{\alpha}{}{} \vee \rendition{\beta}{}{} \\
		\rendition{\alpha\seq\beta}{}{}
			& \equiv 
			\big( \neg\finmarker{} \wedge \rendition{\alpha}{}{\renparams{}{}{\false}} \big)
			\vee
			\ex{\varvec[opt]} \big( 
				\rendition{\alpha}{}{\renparams{}{\varvec[opt]}{\true}} 
				\wedge 
				\rendition{\beta}{}{\renparams{\varvec[opt]}{}{}}
			\big) \\
		\rendition{\repetition{\alpha}}{}{}
			& \equiv \ex{n{:}\naturals}
				\ex{\repstate{}}
				\Big(
					\repstate{1} \!= \varvec
					\wedge \repstate{n} \!= \varvec[alt] \\
					&\hspace{.7cm} \wedge \fa{i{:}\naturals} \big(
						\orange{1}{i}{n} 
						\rightarrow 
						\reninter{i}{i+1}{\orange{1}{i}{n-1} \vee \finmarker{}} 
					\big)
				\Big)
			\\
		\rendition{\alpha \parOp \beta}{}{}
			& \equiv \ex{\auxhis{=}\auxhis\downarrow(\alpha\parOp\beta)}\ex{\varvec[alt]_\alpha, \varvec[alt]_\beta} \Big(
				\tvarof{\avar[alt]} = \recvar \cdot \auxhis 
				\wedge \big(
					\textstyle\bigwedge_{j\in\{1,\ldots,\semvar\}} \avar[alt]_j = \mergeren{j}
				\big) \\
				&\hspace{.7cm}
				\wedge \rendition{\alpha}{}{\renparams{}{\varvec[alt]_\alpha}{\finmarker{}}}
						\wedge
						\tvar_\alpha = \recvar \cdot (\auxhis \downarrow \alpha)
						\wedge \rendition{\beta}{}{\renparams{}{\varvec[alt]_\beta}{\finmarker{}}}
						\wedge
						\tvar_\beta = \recvar \cdot (\auxhis \downarrow \beta)
				\big)
			\Big)
	\end{align*}
	\caption{
		Encoding of the transition semantics of CHPs in \comFOD (\rref{lem:rendition})
	}
	\label{fig:homogeneous_rendition}
\end{figure}

\begin{proof}
	[Proof of \rref{lem:com-fod-expressiveness}.]
	The proof is by induction on the structure of $\phi$,
	generalizing a result for \dL \cite[Lemma 6]{DBLP:journals/jar/Platzer08} to ac-modalities.
	The formula $\phi$ can be assumed to contain no dynamic modalities rewriting them by the equivalences $[ \alpha ] \psi \leftrightarrow [ \alpha ] \acpair{\true, \true} \psi$ and $\langle \alpha \rangle \psi \leftrightarrow \langle \alpha \rangle \acpair{\true, \false} \psi$.
	Throughout the proof, IH abbreviates usage of the induction hypothesis.
	\begin{enumerate}
		\item If $\phi$ is an \comFOD formula, 
		then define $\toComFOD{\phi} \equiv \phi$.

		\item 
		\label{itm:expressibility-and}
		If $\phi \equiv \varphi \wedge \psi$,
	 	by IH, $\toComFOD{\varphi}$, $\toComFOD{\psi}$ exist such that $\vDash \varphi \leftrightarrow \toComFOD{\varphi}$ and $\vDash \psi \leftrightarrow \toComFOD{\psi}$.
		Now, define $\toComFOD{\phi} \equiv \toComFOD{\varphi} \wedge \toComFOD{\psi}$.
		Then $\vDash \phi \leftrightarrow \toComFOD{\phi}$.

		\item Other propositional connectives and quantifiers ($\neg, \forall$) are handled analogous to \rref{itm:expressibility-and}.

		\item If $\phi \equiv [ \alpha ] \ac \psi$,
		then by IH, $\toComFOD{\A}$, $\toComFOD{\Commit}$, and $\toComFOD{\psi}$ exist
		such that $\vDash \chi \leftrightarrow \toComFOD{\chi}$ for each $\chi \in \{\A, \Commit, \psi\}$.
		To express~$\alpha$'s transition semantics in \comFOD, 
		the rendition $\rendition{\alpha}{}{}$ from \rref{lem:rendition} is used,
		where $\varvec = (\recvar, \avar_1, \ldots, \avar_\semvar)$ are the variables of $\alpha$ including the recorder~$\recvar$
		and $\varvec[alt] = (\tvar, \avar[alt]_1, \ldots, \avar[alt]_\semvar)$ is fresh and compatible with $\varvec$.
		The formula $\rendition{\alpha}{}{}$ holds if there is an $\alpha$-run from $\varvec$ to $\varvec[alt]$,
		where $\finmarker{}$ tells whether $\varvec[alt]$ characterizes an intermediate or final state.
		To capture the assumption~$\A$ between $\alpha$'s initial history $\recvar$ and its reached history $\tvar$,
		let $\Atrace{\A} \equiv \fa{\recvar{\preceq}\tvar'{\sim}\tvar} \A\subs{\recvar}{\tvar'}$,
		where~$\tvar'$ is fresh
		and $\sim\,\in \{\prec, \preceq\}$.
		Then $\toComFOD{\phi}$ is defined as follows,
		where $\fa{\finmarker{}} \phi(\finmarker{})$ is short for $\phi(\false) \wedge \phi(\true)$:
		\begin{equation}
			\label{eq:rendition-acbox}
			\toComFOD{\phi} \equiv \fa{\varvec[alt]} \fa{\finmarker{}} \Big( 
				\rendition{\alpha}{\toComFOD{\A}}{} 
				\rightarrow 
				\big(
					\Abutlast{\toComFOD{\A}} \rightarrow
					(\toComFOD{\Commit}) \subs{\recvar}{\tvar}
				\big) 
				\wedge 
				\big( 
					\finmarker{} \wedge \Aglobally{\toComFOD{\A}} \rightarrow (\toComFOD{\psi}) \subs{\varvec}{\varvec[alt]}
				\big) 
			\Big)
		\end{equation}
		The conjuncts in \rref{eq:rendition-acbox} straightforwardly reflect \acCommit and \acPost.
		Hence, for the ac-contract $\Atrace{\toComFOD{\A}}$ and $\toComFOD{\Commit}$,
		only the recorder~$\recvar$ is updated
		while for the postcondition $\toComFOD{\psi}$ the overall state $\varvec$ is updated.

		\item If $\phi \equiv \langle \alpha \rangle \ac \psi$, then $\toComFOD{\A}$, $\toComFOD{\Commit}$, and $\toComFOD{\psi}$ exist by IH such that $\vDash \chi \leftrightarrow \toComFOD{\chi}$ for each $\chi \in \{\A, \Commit, \psi\}$.
		Moreover, let $\rendition{\alpha}{\toComFOD{\A}}{}$ and $\Atrace{\A}$ be as in case $[ \alpha ] \ac \psi$.
		Then $\toComFOD{\phi}$ is defined as follows,
		where $\ex{\finmarker{}} \phi(\finmarker{}) \equiv \phi(\false) \vee \phi(\true)$:
		\begin{equation*}
			\toComFOD{\phi} \equiv \ex{\varvec[alt]} \ex{\finmarker{}} \Big( 
				\rendition{\alpha}{\toComFOD{\A}}{}
				\wedge
				\Big(
					\big(
						\Abutlast{\toComFOD{\A}} \wedge (\toComFOD{\Commit}) \subs{\recvar}{\tvar} 
					\big)
					\vee 
					\big(
						\finmarker{} \wedge \Aglobally{\toComFOD{\A}} \wedge (\toComFOD{\psi}) \subs{\varvec}{\varvec[alt]}
					\big)
				\Big)
			\Big)
		\end{equation*}
		\qedhere
	\end{enumerate}
\end{proof}

\begin{proof}
	[Proof of \rref{lem:state_variations_comfod}.]
	Let $\varvec = (\recvar, \avar_1, \ldots, \avar_\semvar)$
	be all variables in $\{\recvar\} \cup \SFV(\varphi) \cup \SV(\alpha)$
	and let~$\varvec[alt] = (\tvar, \avar[alt]_1, \ldots \avar[alt]_\semvar)$ be fresh.
	Further, let $\cset_0 = \relimage{\SCN(\varphi)}{\{\recvar\}} \cup \SCN(\alpha) \cup \{\durchan\} \cup \cset^\complement$.
	Then the formulas $\strgCommit{\cset}{\varphi}{}{\alpha}$ and $\strgPost{\cset}{\varphi}{}{\alpha}$ below characterize the sets of intermediate state variations $\reachableInter{\cset, \varphi}{\alpha}$ and final state variations $\reachableFin{\cset, \varphi}{\alpha}$, respectively,
	where $\tvar_\alpha, \tvar_0$ are fresh:
	\begin{align*}
		\strgCommit{\cset}{\varphi}{}{\alpha} 
		& \equiv
			\fa{\tvar{=}(\recvar{\downarrow}\cset_0)} 
			\ex{\recvar} 
		 	\big( 
				\varphi \wedge \ex{\tvar_\alpha} \big(
					\tvar = \recvar \cdot \tvar_\alpha 
					\wedge \ex{\tvar_0{=}\recvar}  \langle \alpha \rangle \acpair{\true, \recvar = \tvar_0 \cdot \tvar_\alpha \downarrow (\alpha \cup \cset^\complement)} \false 
				\big)
			\big) \\
		\strgPost{\cset}{\varphi}{}{\alpha} 
		& \equiv
			\fa{\varvec[alt]{=}(\varvec\subs{\recvar}{\recvar{\downarrow}\cset_0})}
			\ex{\varvec} \big( 
				\varphi \wedge \ex{\tvar_\alpha} \big(
					\tvar = \recvar \cdot \tvar_\alpha 
					\wedge \ex{\tvar_0{=}\recvar} \langle \alpha \rangle \varvec = \varvec[alt] \subs{\tvar}{\tvar_0 \cdot \tvar_\alpha \downarrow (\alpha \cup \cset^\complement)}
				\big)
			\big)
	\end{align*}
	The formula $\fa{\varvec[alt]{=}\varvec} \ex{\varvec} ( \varphi \wedge \langle \alpha \rangle \varvec[alt] = \varvec )$,
	where $\varvec$ are all variables of $\varphi$ and $\alpha$,
	is satisfied in exactly those final states reachable by~$\alpha$ from some state satisfying $\varphi$.
	As $\alpha$ potentially writes~$\varvec$,
	the fresh variables $\varvec[alt]$ relate the initial and final state of $\alpha$.
	Closely mirroring \rref{def:state_variations},
	$\strgCommit{\cset}{\varphi}{}{\alpha}$ and $\strgPost{\cset}{\varphi}{}{\alpha}$ generalize this observation to environmental state variations.
	
	For $\scspSymb \in \{ 
		\strgCommit{\beta}{\varphi}{}{\alpha},
		\strgPost{\beta}{\varphi}{}{\alpha}
	\}$,
	the projection $\downarrow\cset_0$ ensures
	$\relimage{\SCN(\scspSymb)}{\{\recvar\}} \subseteq \relimage{\SCN(\varphi)}{\{\recvar\}} \cup \SCN(\alpha) \cup \{\durchan\} \cup \SCN(\beta)^\complement$.
	Since $\beta$ does not interfere with $\nointfpair{\alpha}{\varphi}$,
	obtain $\relimage{\SCN(\scspSymb)}{\{\recvar\}} \cap \SCN(\beta) \subseteq \SCN(\alpha) \cup \{\durchan\}$,
	\iest via the recorder~$\recvar$,
	the strongest promise $\scspSymb$ only depends on channels of the environment $\beta$
	if these channels are shared with $\alpha$.
	Further, $\SFV(\scspSymb) \subseteq \varvec = \{\recvar\} \cup \SFV(\varphi) \cup \SV(\alpha)$.
	Since~$\SBV(\beta) \cap \SV(\alpha) \subseteq \{\recvar\}$,
	as $\alpha\parOp\beta$ is well-formed (\rref{def:syntax_chps}),
	and $\beta$ does not interfere with $\nointfpair{\alpha}{\varphi}$,
	obtain $\SFV(\scspSymb) \cap \SBV(\beta)
	\subseteq \{\recvar\}$.
	In summary, $\beta$ does not interfere with $\nointfpair{\alpha}{\scspSymb}$.
\end{proof}

\begin{proof}
	[Proof of \rref{lem:reachable_states_correct}.]
	The items are proven separately:
	\begin{enumerate}
		\item
		Let $\vDash \varphi \rightarrow [ \alpha ] \acpair{\true, \Commit} \psi$.
		For \ref{itm:strongest_com_implies_commit}, assume $\pstate{varfin} \vDash \strgCommit{\emptyset}{\varphi}{}{\alpha}$.
		By  \rref{lem:state_variations_comfod} and \rref{def:state_variations}, 
		a run $\comvarrun{}{\emptyset^\complement} \in \sem{\alpha}{}$ exists
		such that
		$\pstate{v} \vDash \varphi$ and $\pstate{varfin} = \pstate{v} \cdot \trace$.
		Since $\comvariation{}{\emptyset^\complement} = \Chan$, 
		obtain $\run \in \sem{\alpha}{}$.
		By $\vDash \varphi \rightarrow [ \alpha ] \acpair{\true, \Commit} \psi$,
		obtain $\pstate{v} \cdot \trace \vDash \Commit$,
		so $\pstate{varfin} \vDash \Commit$. 
		For \ref{itm:strongest_reach_implies_post}, assume $\pstate{varfin} \vDash \strgPost{\emptyset}{\varphi}{}{\alpha}$.
		By \rref{lem:state_variations_comfod} and \rref{def:state_variations},
		a run $\comvarrun{}{\emptyset^\complement} \in \sem{\alpha}{}$ exists
		such that $\pstate{v} \vDash \varphi$ and $\pstate{varfin} = \pstate{w} \cdot \trace$,
		so $\run \in \sem{\alpha}{}$.
		Finally, $\pstate{varfin} \vDash \psi$ by $\vDash \varphi \rightarrow [ \alpha ] \acpair{\true, \Commit} \psi$ as $\pstate{varfin} = \pstate{w} \cdot \trace$.

		\item 
		Let $\pstate{v} \vDash \varphi_0$,
		where $\varphi_0 \equiv \rvarvec[alt] = \rvarvec \wedge \varphi$, and $\run \in \sem{\alpha}{}$.
		For \acCommit,
		observe $\pstate{v} \cdot \trace \in \reachableInter{\cset, \varphi_0}{\alpha}$ by \rref{def:state_variations}
		as $\trace \downarrow (\comvariation{}{}) = \trace$ because $\trace$ is $\alpha$-communication.
		Hence, $\pstate{v} \cdot \trace \vDash \strgCommit{\cset}{\varphi_0}{}{\alpha}$ by \rref{lem:state_variations_comfod}.
		Since $\pstate{v} \vDash \rvarvec[alt] = \rvarvec$ and $\pstate{v} = \pstate{v} \cdot \trace$ on $\rvarvec, \rvarvec[alt]$,
		obtain $\pstate{v} \cdot \trace \vDash \rvarvec[alt] = \rvarvec$,
		thus $\pstate{v} \cdot \trace \vDash \fa{\rvarvec{=}\rvarvec[alt]} \strgCommit{\cset}{\varphi_0}{}{\alpha}$.
		For \acPost, assume $\pstate{w} \neq \bot$.
		Hence, $\pstate{w} \cdot \trace \in \reachableFin{\cset, \varphi_0}{\alpha}$ by \rref{def:state_variations},
		which implies $\pstate{w} \cdot \trace \vDash \strgPost{\cset}{\varphi_0}{}{\alpha}$ by \rref{lem:state_variations_comfod}.
		\qedhere
	\end{enumerate}
\end{proof}

\begin{proof}
	[Proof of \rref{lem:ac_biggest_promise_split}.]
	The proof focuses on the case of strongest postconditions,
	because the proof transfers to strongest commitments with minor adjustments.
	By \rref{lem:state_variations_comfod},
	the proof can argue semantically about state variations (\rref{def:state_variations}),
	\iest 
	it suffices to prove:
	\begin{equation*}
		\reachableFin{\beta, F_\alpha}{\alpha}
		\cap
		\reachableFin{\alpha, F_\beta}{\beta}
		\cap
		\sem{\recvar \succeq \tvar_0}{} 
		\subseteq 
		\reachableFin{\emptyset,F}{\alpha\parOp\beta}
	\end{equation*}
	
	Let $\pstate{varfin} \in \reachableFin{\otherprog{\gamma}, F_\gamma}{\gamma}$
	for $(\gamma,\ogamma) \in \{(\alpha,\beta), (\beta,\alpha)\}$.
	Further, let $\pstate{varfin} \in \sem{\recvar \succeq \tvar_0}{}$.
	Then, by \rref{def:state_variations},
	there is a trace $\trace[alt]_\gamma$ and a run $(\pstate{v}_\gamma, \trace_\gamma, \pstate{w}_\gamma) \in \sem{\gamma}{}$ with $\trace_\gamma = \trace[alt]_\gamma \downarrow (\comvariation{\gamma}{(\ogamma)^\complement})$
	such that $\pstate{v}_\gamma \vDash F_\gamma$ and $\pstate{varfin} = \pstate{w}_\gamma \cdot \trace[alt]_\gamma$.
	That is, the state $\pstate{varfin}$ can be reached via both runs 
	by mutually interleaving some $\ogamma$-communication into~$\trace_\gamma$.
	In fact, $\trace[alt]_\alpha = \trace[alt]_\beta$
	as follows:

	By \mbox{$\pstate{varfin} \in \sem{\recvar\succeq\tvar_0}{}$},
	there is a $\trace$ such that $\pstate{varfin}(\recvar) = \pstate{varfin}(\tvar_0) \cdot \trace$.
	Since $\pstate{w}_\gamma = \pstate{v}_\gamma$ on $\TVar$ by the bound effect property (\rref{lem:bound_effect}),
	obtain
	\begin{equation}
		\label{eq:hext}
		\pstate{v}_\gamma(\recvar) \cdot \trace[alt]_\gamma
		= \pstate{w}_\gamma(\recvar) \cdot \trace[alt]_\gamma
		= \pstate{varfin}(\recvar) 
		= \pstate{varfin}(\tvar_0) \cdot \trace
		= \pstate{w}_\gamma(\tvar_0) \cdot \trace 
		= \pstate{v}_\gamma(\tvar_0) \cdot \trace
		\text{.}
	\end{equation}
	Hence, 
	$\pstate{v}_\gamma(\recvar) \downarrow\gamma \cdot \trace_\gamma
	= \pstate{v}_\gamma(\tvar_0) \downarrow\gamma \cdot \trace \downarrow\gamma$,
	because $\trace[alt]_\gamma \downarrow \gamma = \trace_\gamma$
	by $\trace_\gamma = \trace[alt]_\gamma \downarrow (\gamma\cup(\ogamma)^\complement)$.
	Since $\pstate{v}_\gamma(\recvar) \downarrow \gamma = \pstate{v}_\gamma(\tvar_0) \downarrow \gamma$ by $\pstate{v}_\gamma \vDash F_\gamma$,
	obtain $\trace \downarrow \gamma = \trace_\gamma$.
	By $\trace_\gamma = \trace[alt]_\gamma \downarrow (\comvariation{\gamma}{(\ogamma)^\complement})$,
	only $\ogamma$-communication distinguishes $\trace_\gamma$ and $\trace[alt]_\gamma$,
	so $\trace[alt]_\gamma \downarrow (\alpha\parOp\beta) = \trace[alt]_\gamma$.
	Hence, $\trace[alt]_\gamma$ is a suffix of~$\trace$,
	since~$\trace$ and $\trace[alt]_\gamma$ are suffixes of~$\pstate{varfin}(\recvar)$
	with $\trace[alt]_\gamma \downarrow \gamma = \trace_\gamma$ and $\trace \downarrow \gamma = \trace_\gamma$.
	Further, any communication of~$\trace$ that lies before $\trace[alt]_\gamma$
	must not be communication of~$\alpha$ and~$\beta$.
	However, this is impossible by \rref{eq:hext},
	because $\pstate{v}_\gamma(\tvar_0) \downarrow \cset_\gamma = \pstate{v}(\tvar_0) \downarrow \cset_\gamma$ by $\pstate{v}_\gamma \vDash F_\gamma$
	and $\cset_\gamma \supseteq \SCN(\ogamma)^\complement$,
	so that $\trace = \trace[alt]_\gamma$.

	Since $\trace = \trace[alt]_\alpha = \trace[alt]_\beta$,
	obtain $\trace \downarrow (\alpha\parOp\beta) = \trace$.
	Further, $\pstate{w} \equalsdef \pstate{w}_\alpha = \pstate{w}_\beta$,
	as $\pstate{varfin} = \pstate{w}_\gamma \cdot \trace[alt]_\gamma$.
	The remaining proof merges the individual runs for $\alpha$ and $\beta$ into a run for $\alpha\parOp\beta$ as follows:

	\begin{center}
		\begin{tikzpicture}[node distance=1.2cm]
			\tikzstyle{similar state}=[double, double distance=1.2pt, gray]

			\node (valpha) {$\pstate{v}_\alpha$};
				
			\node [below right=of valpha] (finstate) {$\pstate{w}$};

			\node [above right=of finstate] (vbeta) {$\pstate{v}_\beta$};

			\draw[-stealth] (valpha) -- node[fill=white, yshift=.5mm, inner sep=.2em] {$\trace_\alpha$} (finstate);

			\draw[-stealth] (vbeta) -- node[fill=white, yshift=.5mm, inner sep=.2em] {$\trace_\beta$} (finstate);

			\draw[dashed] (valpha) edge[bend left=50] (finstate.north);

			\draw[dashed] (vbeta) edge[bend right=50] 
			node[left=1mm, fill=white, inner sep=.3em] {$\trace$}
			(finstate.north);

			\draw[-stealth] (finstate.north)--++(0,-0.0001);

			\node [right=3cm of finstate] (walpha) {$\pstate[alt]{w}_\alpha$};

			\node[right=.8cm of walpha] (merge) {$\pstate[alt]{w}_\alpha\merge\pstate[alt]{w}_\beta$};

			\node at (merge |- vbeta) (v) {$\pstate{v}$};

			\node[right=.8cm of merge] (wbeta) {$\smash{\pstate[alt]{w}_\beta}\phantom{\pstate[alt]{w}_\alpha}$};

			\draw[-stealth] (v) -- node[fill=white, yshift=.5mm, inner sep=.2em] {$\trace$} (merge);

			\draw[-stealth] (v) -- node[fill=white, yshift=.5mm, inner sep=.2em] {$\trace_\alpha$} (walpha);

			\draw[-stealth] (v) -- node[fill=white, yshift=.5mm, inner sep=.2em] {$\trace_\beta$} (wbeta);

			\draw[similar state, gray] (walpha) -- node[xshift=1mm, yshift=-.3mm, above] {\footnotesize $\SBV(\alpha)$} (merge);
			\draw[similar state, gray] (wbeta) -- node[xshift=-1mm, yshift=-.3mm, above] {\footnotesize $\SBV(\alpha)^\complement$} (merge);

			\node [above right=.1cm and 1.6cm of finstate] (arrow) {$\underset{\text{\rref{lem:program_coincidence}}}{\overset{\text{by coincidence}}{\Longrightarrow}}$};

			\draw (vbeta) edge[similar state, gray] node[above, xshift=-7mm] {\footnotesize $\varvec_\gamma \supseteq \SFV(F_\gamma) \cup \SFV(\gamma)$} (v);
			\draw (valpha) edge[similar state, gray, bend left=15] 
			(v);

			\draw (finstate) edge[similar state, gray] node[below,xshift=5mm] {\footnotesize $\varvec_\gamma \supseteq \SV(\gamma)$} (walpha);
			\draw (finstate) edge[similar state, gray, bend right=15] (wbeta);

			\draw[similar state] ([xshift=-2mm]valpha.west) -- node[left] {\shortstack[l]{
				\footnotesize on $\{\tvar_0\}$\\
				\footnotesize $\cup\SBV(\alpha)^\complement \cup \SBV(\beta)^\complement$
				\\ by \rref{lem:bound_effect}
			}} ([xshift=-2mm]valpha.west|-finstate.west);
		\end{tikzpicture}
	\end{center}

	Let $\varvec_\gamma = \SFV(F_\gamma) \cup \SV(\gamma)$ be the local variables of $\gamma$.
	Note that $\{\recvar, \tvar_0\} \subseteq \SFV(F_\gamma) \subseteq \varvec_\gamma$.
	By merging $\pstate{v}_\alpha$, $\pstate{v}_\beta$, and $\pstate{varfin}$,
	define a state $\pstate{v}$ as follows:
	\begin{equation*}
		\pstate{v}(\avar) = 
			\begin{cases}
				\pstate{varfin}(\tvar_0) & \text{if } \avar=\recvar \\
				\pstate{v}_\alpha(\avar) & \text{else, if $\avar\in\varvec_\alpha$} \\
				\pstate{v}_\beta(\avar) & \text{else}
			\end{cases}
	\end{equation*}
	
	Since $\SFV(\varphi_\gamma) \subseteq \SBV(\otherprog{\gamma})^\complement \cup \gvars$,
	as $\ogamma$ does not interfere (\rref{def:noninterference}) with $(\gamma, \varphi_\gamma)$,
	and $\SV(\gamma) \subseteq \SBV(\ogamma)^\complement \cup \gvars$ by well-formedness (\rref{def:syntax_chps}),
	and $\{\tvar_0\} \subseteq \SBV(\ogamma)^\complement$,
	obtain
	\begin{equation*}
		\varvec_\gamma 
		= \SFV(F_\gamma) \cup \SV(\gamma)
		\subseteq \SFV(\varphi_\gamma) \cup \{\recvar, \tvar_0\} \cup \SV(\gamma)
		\subseteq \SBV(\ogamma)^\complement \cup \gvars
	\end{equation*}

	Since $\varvec_\alpha \cap \varvec_\beta \subseteq \SBV(\alpha)^\complement \cap \SBV(\beta)^\complement \cup \gvars$,
	obtain $\pstate{v} = \pstate{v}_\gamma$ on $\varvec_\gamma \cap \{ \recvar \}^\complement$
	by definition of~$\pstate{v}$,
	using that $\pstate{v}_\alpha = \pstate{w} = \pstate{v}_\beta$ on $\{\tvar_0\} \cup \SBV(\alpha)^\complement \cup \SBV(\beta)^\complement$ by \rref{lem:bound_effect}.
	Further, $\pstate{v}(\recvar) = \pstate{v}_\gamma(\recvar)$,
	using $\pstate{v}(\recvar) = \pstate{varfin}(\tvar_0)$, \rref{eq:hext}, and $\trace[alt]_\gamma = \trace$.
	In summary, $\pstate{v} = \pstate{v}_\gamma$ on $\varvec_\gamma$.
	
	Since $\pstate{v} = \pstate{v}_\gamma$ on $\varvec_\gamma \supseteq \SFV(F_\gamma)$, 
	obtain $\pstate{v} \vDash F_\gamma$
	by $\pstate{v}_\gamma \vDash F_\gamma$ and coincidence (\rref{lem:expr_coincidence}),
	so $\pstate{v} \vDash \varphi_\alpha \wedge \varphi_\beta$.
	Further,
	$\pstate{v}(\recvar)
	= \pstate{varfin}(\tvar_0)
	= \pstate{v}_\alpha(\tvar_0)
	= \pstate{v}(\tvar_0)$
	such that $\pstate{v} \vDash \recvar = \tvar_0$.
	In summary, $\pstate{v} \vDash F$.
	Since $\pstate{v} = \pstate{v}_\gamma$ on $\varvec_\gamma \supseteq \SFV(\gamma)$,
	there are $\gamma$-runs $(\pstate{v}, \trace_\gamma, \pstate[alt]{w}_\gamma) \in \sem{\gamma}{}$
	with $\pstate[alt]{w}_\gamma = \pstate{w}$ on $\varvec_\gamma$
	by coincidence (\rref{lem:program_coincidence}).
	Hence, $(\pstate{v}, \trace, \pstate[alt]{w}_\alpha \merge \pstate[alt]{w}_\beta) \in \sem{\alpha\parOp\beta}{}$.
	Since $\pstate{v} \vDash F$, and $(\pstate{v}, \trace, \pstate{w}) \in \sem{\alpha\parOp\beta}{}$, 
	and $\trace \downarrow ((\alpha\parOp\beta) \cup \emptyset^\complement)
	= \trace$, 
	obtain $\pstate{varfin} = \pstate{w} \cdot \trace \in \reachableFin{\emptyset, F}{\alpha\parOp\beta}$,
	using that $\pstate[alt]{w}_\alpha\merge\pstate[alt]{w}_\beta = \pstate{w}$ as follows:
	
	On~$\{\recvar\}^\complement \cap \SBV(\gamma)$,
	obtain 
	$\pstate[alt]{w}_\alpha\merge\pstate[alt]{w}_\beta
	= \pstate[alt]{w}_\gamma
	= \pstate{w}$,
	using $\SBV(\beta)\subseteq\SBV(\alpha)^\complement\cup\gvars$ by \rref{def:syntax_chps}.
	By \rref{lem:bound_effect},
	$\pstate{v}_\alpha = \pstate{w} = \pstate{v}_\beta$
	on $\varset = \{\recvar\}^\complement \cap (\SBV(\alpha)^\complement \cap \SBV(\beta)^\complement)$,
	so $\pstate[alt]{w}_\alpha\merge\pstate[alt]{w}_\beta = \pstate{v} = \pstate{w}$ on $\varset$.
	Further,
	$((\pstate[alt]{w}_\alpha\merge\pstate[alt]{w}_\beta) \cdot \trace)(\recvar)
	=
		\pstate{v}(\recvar) \cdot \trace(\recvar)
	= \pstate{varfin}(\tvar_0) \cdot \trace
	= \pstate{varfin}(\recvar)
	= (\pstate{w} \cdot \trace)(\recvar)$.
	In summary, $\pstate[alt]{w}_\alpha\merge\pstate[alt]{w}_\beta = \pstate{w}$.
	\qedhere
\end{proof}

\section{Completeness Relative to FOD}
\label{app:com-fod-to-fod}

\begin{proof}
	[Proof of \rref{lem:typecast}.]
	Communication traces
	can be represented in FOD by a nested $\reals$-Gödel encoding (\rref{lem:real_goedel}) 
	that first compresses every event (channel and value)
	and then compresses the resulting finite sequence into a single real number.
	For disambiguation,
	the encoding of the trace is further paired with its length.
	Recall that if $Z$ is a $\reals$-Gödel encoding of a sequence of length $n$,
	then $\goedelat{Z}{n}{j}$ denotes the $j$-th encoded real number (\rref{lem:real_goedel}).
	For real variables $x$ and $k$, 
	define $\len{x} \equiv \goedelat{x}{2}{1}$
	and $\seqat{x}{k} \equiv \goedelat{(\goedelat{x}{2}{2})}{\len{x}}{\lfloor k \rfloor}$,
	where rounding $\lfloor \cdot \rfloor$ is definable in FOD (\rref{cor:defined_operators}).
	Further, define
	$\chan{\seqat{x}{k}} \equiv \goedelat{(\seqat{x}{k})}{2}{1}$
	and $\val{\seqat{x}{k}} \equiv \goedelat{(\seqat{x}{k})}{2}{2}$.
	Then define $\eventencseq{x}$\,as follows,
	where $\isNat{\cdot}$ is definable in FOD by \rref{lem:isnat}:
	\begin{equation*}
		\eventencseq{x} \equiv \isNat{\len{x}} 
		\wedge \big(
			\len{x}=0 \rightarrow \goedelat{x}{2}{2}=0
		\big) 
		\wedge \fa{0{\le}k{<}\len{x}} \isNat{\chan{\seqat{x}{k}}}
	\end{equation*}
	Further, define an \comFOD formula $\goedelize{x, \tvar}$,
	where $\langle \_, \_ \rangle$ constructs a communication item:
	\begin{align}
		\label{eq:goedelize}
		\goedelize{x, \tvar} & \equiv 
			\len{x}=\len{\tvar}
			\wedge \fa{k} \Big(
				0{\le}k{<}\len{\tvar} \rightarrow \at{\tvar}{k} = \big\langle
					\chan{\seqat{x}{k}},
					\val{\seqat{x}{k}}
				\big\rangle
			\Big)
	\end{align}
	Observe that $\eventencseq{x}$ restricts the encoding of channel names to $\Chan = \naturals$
	and disambiguates the encoding of the empty trace.
	Hence, $\goedelize{x, \tvar}$ characterizes a bijection $\goedelize{\cdot} : \traces \rightarrow \traceenc$,
	\iest for every trace $\tvar : \traces$,
	there is exactly one encoding $\eventencseq{x}$\,such that $\goedelize{x, \tvar}$ holds and vice versa,
	as $\reals$-Gödel encodings are unique for a specific length \cite[Lemma 4]{DBLP:journals/jar/Platzer08}.
	This justifies to write $x = \goedelize{\tvar}$ instead of $\goedelize{x, \tvar}$.
	Finally, observe that $\goedelize{\cdot}$ preserves lengths and entries,
	\iest $\len{\tvar} = \len{\goedelize{\tvar}}$,
	and $\sel{\at{\tvar}{k}} = \sel{\seqat{(\goedelize{\tvar})}{k}}$ for all $0{\le}k{<}\len{\tvar}$ and $\selOp \in \{ \chanOp, \valOp \}$.
\end{proof}

In preparation for the proof of \rref{prop:com-fod-to-fod},
\rref{lem:com-fod-to-extensional-com-fod} simplifies \comFOD formulas
to a provably equivalent form.
An \comFOD formula $\phi$ is called \emph{extensional} if every relation $\te_1 \sim \te_2$
in~$\phi$ has the form $\at{\tvar_1}{k} = \at{\tvar_2}{j}$, 
where $\tvar_1, \tvar_2 \in \TVar$ and $k, j \in \RVar$,
and if the operators $\chan{\cdot}$, $\val{\cdot}$, and~$\len{\cdot}$ are only applied to variables. 
In particular, extensional \comFOD formulas do not contain the prefix relation $\preceq$.

\begin{lemma}
	[Extensional \comFOD]
	\label{lem:com-fod-to-extensional-com-fod}
	For every \comFOD formula~$\phi$,
	there is effectively an equivalent extensional \comFOD formula $\toExtComFOD{\phi}$
	over the same free variables
	such that $\phi \leftrightarrow \toExtComFOD{\phi}$ is provable in the extended \dLCHP calculus~\dLCHPdash.
\end{lemma}
\vspace{-.7cm}
\begin{proof}
	The formula $\toExtComFOD{\phi}$ is inductively defined in \rref{fig:com-fod-to-extensional-com-fod}.
	Figure~\ref{fig:com-fod-to-extensional-com-fod-simplifications} eliminates prefixing~$\preceq$, 
	and normalizes trace equality to the form $\tvar = \te$,
	where $\tvar \in \TVar$ and $\te$ is restricted to 
	$\tvar_0 \mid \epsilon \mid \comItem{\ch{}, \rp} \mid \tvar_1 \cdot \tvar_2 \mid \at{\tvar_0}{k} \mid \tvar_0 \downarrow \cset$.
	Figure~\ref{fig:com-fod-to-extensional-com-fod-base} expresses every $\tvar = \te$
	extensionally based on the length $\len{\te}$ and the positions $\at{\te}{k}$,
	where maximal suffixes $\maxdursuffixop$ and prefixes $\maxdurprefixop$ of durations
	and sequences $(\tvar_i)_{i=1}^n$ of traces (\rref{cor:defined_operators}) 
	can be defined as extensional \comFOD formulas with appropriate care.

	Derivability of $\phi \leftrightarrow \toExtComFOD{\phi}$ in \dLCHPdash
	is proven by an induction on $\phi$ along the recursion of \rref{fig:com-fod-to-extensional-com-fod}.
	For \rref{fig:com-fod-to-extensional-com-fod-simplifications},
	$\dLCHPdash \phi \leftrightarrow \toExtComFOD{\phi}$ follows by first-order reasoning,
	using \RuleName{prefix} for $\phi\equiv\te_1\preceq\te_2$,
	and \RuleName{acMono} and \RuleName{acDiaMono} for $\phi \equiv \dbleft \evolution*{}{non} \dbright \psi$.
	In \rref{fig:com-fod-to-extensional-com-fod-base}, 
	$\dLCHPdash \phi \leftrightarrow \toExtComFOD{\phi}$
	is by \RuleName{ext} for $\phi \equiv \tvar = \tvar_0$,
	by~\RuleName{lenEmpty} for $\phi \equiv \tvar = \epsilon$,
	and by \RuleName{extProj} for $\phi \equiv \tvar = \tvar_0 \downarrow \cset$.
	For $\phi \equiv \tvar = \comItem{\ch{}, \rp}$,
	combine~\RuleName{durzero} and~\RuleName{eqItem} if $\comItem{\ch{}, \rp}\neq\comItem{\durchan,0}$, 
	and~\RuleName{durzero} and \RuleName{lenEmpty} if $\comItem{\ch{}, \rp}=\comItem{\durchan,0}$.
	
	For $\phi \equiv \tvar = \at{\tvar_0}{k}$,
	use~\RuleName{lenEmpty} and~\RuleName{accessDefault} if $\len{\tvar}=0$.
	Otherwise, if $\len{\tvar}=1$,
	use \RuleName{ext}, and \RuleName{lenAccess}, and \RuleName{accessConcat}.
	For $\rightarrow$, all cases have been considered,
	because $\len{\tvar} = \len{\at{\tvar_0}{k}} \in \{0,1\}$ by \RuleName{lenAccess}.
	
	For $\phi \equiv \tvar = \tvar_1 \cdot \tvar_2$,
	the axioms \RuleName{ext}, and \RuleName{lenSumDist}, and~\RuleName{accessConcat} are combined, 
	using the axioms~\RuleName{maxdursuffix}, and \RuleName{maxdurprefix}
	to obtain suffixes and prefixes with maximal duration if necessary
	when $\tvar_1\cdot\tvar_2$ appends consecutive durations,
	and \RuleName{durequi} and \RuleName{durzero} combine these durations.
\end{proof}

\begin{figure}
	\begin{subfigure}{\textwidth}
		\begin{small}
			\centering
			$\begin{aligned}
				& \toExtComFOD{(\neg \varphi)}
					\equiv \neg \toExtComFOD{\varphi} \\
				& \toExtComFOD{(\varphi \wedge \psi)}
					\equiv \toExtComFOD{\varphi} \wedge \toExtComFOD{\psi} \\
				& \toExtComFOD{(\fa{\avar} \varphi)}
					\equiv \fa{\avar} \toExtComFOD{\varphi} \\
				& \toExtComFOD{(\te_1 \preceq \te_2)}
					\equiv \ex{\tvar} \toExtComFOD{(\te_1 \ofconcat \tvar = \te_2)} \\
				& \toExtComFOD{(\dbleft \evolution*{}{non} \dbright \psi)}
					\equiv \dbleft \evolution*{}{non} \dbright \toExtComFOD{\psi}
			\end{aligned}$\hspace{.4cm}%
			$\begin{aligned}
				& \toExtComFOD{(\te_0 = \te)} 
					\equiv \ex{\tvar} \toExtComFOD{\big( 
						\tvar = \te_0 \wedge \tvar = \te_2  
					\big)} \\
				& \toExtComFOD{\big( \tvar = \te_0 \cdot \te \big)}
					\equiv \ex{\tvar_0} \toExtComFOD{\big(
						\tvar_0 = \te_0 \wedge \tvar = \tvar_0 \cdot \te
					\big)} \\
				& \toExtComFOD{\big(\tvar = \te \cdot \te_0 \big)}
					\equiv \ex{\tvar_0} \toExtComFOD{\big(
						\tvar_0 = \te_0 \wedge 
						\tvar = \te \cdot \tvar_0
					\big)} \\
				& \toExtComFOD{\big( \varphi(\fsymb[builtin](\te_0)) \big)}
					\equiv \ex{\tvar} \toExtComFOD{\big(
						\tvar = \te_0 
						\wedge
						\varphi(\fsymb[builtin](\tvar)) 
					\big)}
					\label{eq:simplify-selectors} \\
				& \toExtComFOD{\big( \tvar = \at{\te}{\re_0} \big)}
					\equiv \ex{k} \toExtComFOD{\big(
						k = \re_0 
						\wedge
						\tvar = \at{\te}{k} 
					\big)}	
			\end{aligned}$
		\end{small}
		\caption{}
		\label{fig:com-fod-to-extensional-com-fod-simplifications}
	\end{subfigure}

	\begin{subfigure}{\textwidth}
		\begin{footnotesize}
			$\begin{aligned}
				& \toExtComFOD{(\tvar = \tvar_0)}
					\equiv \len{\tvar} = \len{\tvar_0} \wedge \fa{0{\le}k{<}\len{\tvar}} \at{\tvar}{k} = \at{\tvar_0}{k}
				\qquad\toExtComFOD{(\tvar = \epsilon)}
					\equiv \len{\tvar} = 0
				\qquad\toExtComFOD{\phi}
					\equiv \phi \sidecondition{$\phi \in \FolRA$}
				\\[.5em]
				& \toExtComFOD{(\tvar = \comItem{\ch{}, \rp})}
					\equiv (\len{\tvar} = 0 \wedge \ch{} = \durchan \wedge \rp = 0) 
					\vee (\len{\tvar} = 1 \wedge \chan{\tvar} = \ch{} \wedge \val{\tvar} = \rp) \\
				& \toExtComFOD{(\tvar = \tvar_1 \cdot \tvar_2)}
					\equiv 
					\ex{d_1,d_2} \Big( 
						\maxdursuffix{\tvar_1}{d_1} \wedge \maxdurprefix{\tvar_2}{d_2}
						\wedge
						\ex{\tvar_1',\tvar_2'} \Big(
							\tvar_1 = \tvar_1' \otimes \comItem{\durchan, d_1} \wedge \tvar_2 = \comItem{\durchan, d_2} \otimes \tvar_2' \\
							&\qquad \wedge 
							\big(
								d_1=-d_2 \rightarrow \tvar = \tvar_1' \otimes \tvar_2'
							\big) 
							\wedge 
							\big(
								d_1\neq-d_2 \rightarrow
								\ex{\tvar_0} (
									\tvar_0 = \tvar_1' \otimes \comItem{\durchan, d_1+2_2} \wedge \tvar = \tvar_0 \otimes \tvar_2'
								)
							\big)
						\Big)
					\Big) \\
				& \tvar = \tvar_1 \otimes \tvar_2
					\equiv \len{\tvar} = \len{\tvar_1} + \len{\tvar_2} 
					\wedge \fa{0{\le}j{<}\len{\tvar_1}} \at{\tvar}{j} = \at{\tvar_1}{j} 
					\wedge \fa{0{\le}j{<}\len{\tvar_2}}
					\toExtComFOD{\big( 
						\at{\tvar}{j + \len{\tvar_1}} = \at{\tvar_2}{j}	
					\big)} \\
				& \toExtComFOD{(\tvar = \at{\tvar_0}{k})}
					\equiv \big(
						\len{\tvar} = 1 \wedge 0{\le}k{<}\len{\tvar_0} \wedge \at{\tvar}{0} = \at{\tvar_0}{k} 
					\big)
					\vee 
					\big( 
						\len{\tvar} = 0 \wedge \neg(0{\le}k{<}\len{\tvar_0})
					\big) \\
				& \toExtComFOD{(\tvar = \tvar_0\downarrow\cset)}
					\equiv \ex{(\tvar_i)_{i=1}^{\len{\tvar_0}}} \Big(
							\tvar_0 = \tvar_{\len{\te}}
							\wedge \fa{1{\le}k{\le}\len{\te}} \tvar_k = \tvar_{k-1}
							\cdot \big(\ifstat{\chan{\at{\te}{k}} \in \cset}{\at{\te}{k}}\big)
						\Big)
			\end{aligned}$
		\end{footnotesize}
		\caption{}
		\label{fig:com-fod-to-extensional-com-fod-base}
	\end{subfigure}
	\caption{
		Recursive definition of an extensional \comFOD formula~$\toExtComFOD{\phi}$
		equivalent to $\phi$,
		where $\fsymb[builtin](\te) \in \{ \chan{\te},\val{\te},\len{\te},\te\downarrow\cset, \at{\te}{\re} \}$,
		and $\te_0 \not\in \TVar$, and $\re_0 \not\in\RVar$
	}
	\label{fig:com-fod-to-extensional-com-fod}
\end{figure}

\begin{figure}[t]
	\begin{subfigure}{.5\textwidth}%
		\begin{align*}
			\toFOD{(\re_1 \sim \re_2)} & \equiv \toFOD{\re_1} \sim \toFOD{\re_2} \\
			\toFOD{(\at{\tvar_1}{k} = \at{\tvar_2}{j})} & \equiv \seqat{(\encvar{\tvar_1})}{k} = \seqat{(\encvar{\tvar_2})}{j} \\
			\toFOD{(\neg\varphi)} & \equiv \neg \toFOD{\varphi} \\
			\toFOD{(\varphi\wedge\psi)} & \equiv \toFOD{\varphi} \wedge \toFOD{\psi} \\
			\toFOD{(\fa{x} \varphi)} & \equiv \fa{x} \toFOD{\varphi} \\
			\toFOD{(\fa{\tvar} \varphi)} & \equiv \fa{\encvar{\tvar}{:}\traceenc} \toFOD{\varphi} \\
			\toFOD{(\dbleft \evolution*{}{non} \dbright \psi)} & \equiv \dbleft \evolution*{}{non} \dbright \toFOD{\psi}
		\end{align*}%
		\caption{
			Cases for formulas, where $\sim\, \in \{=,\le\}$
		}
		\label{fig:extensional-com-FOD-to-FOD-formulas}
	\end{subfigure}%
	\begin{subfigure}{.5\textwidth}
		\begin{align*}
			\toFOD{x} & \equiv x \\
			\toFOD{\ratconst} & \equiv \ratconst \\
			\toFOD{(\re_1 \bowtie \re_2)} & \equiv \toFOD{\re_1} \bowtie \toFOD{\re_2} \\
			\toFOD{(\chan{\tvar})} 
				& \equiv 
				\seqchan{\tvar}{\enclen{\tvar}-1} \\
			\toFOD{(\val{\tvar})} 
				& \equiv 
				\seqval{\tvar}{\enclen{\tvar}-1} \\
			\toFOD{(\len{\tvar})} & \equiv \enclen{\tvar}
		\end{align*}
		\caption{Cases for real terms, where $\bowtie\, \in \{+, \cdot\}$}
		\label{fig:extensional-com-FOD-to-FOD-terms}
	\end{subfigure}%
	\caption{Inductive definition of a FOD formula $\toFOD{\phi}$ that is equivalent to the extensional \comFOD formula $\phi$ up to type-casting}
	\label{fig:extensional-com-FOD-to-FOD}
\end{figure}

\CreateRuleRef{fol}

\begin{proof}
	[Proof of \rref{prop:com-fod-to-fod}.]

	By \rref{lem:com-fod-to-extensional-com-fod},
	assume $\phi$ is an extensional \comFOD formula.
	Moreover, by first-order reasoning,
	assume $\phi$ is of the form $\fa{\hvarvec_0{=}\hvarvec} (\phi_0)\subs{\hvarvec}{\hvarvec_0}$, 
	where~$\hvarvec$ are the free trace variables of $\phi_0$
	and $\hvarvec_0$ is fresh.
	Further, for every $\selOp \in \{ \chanOp, \valOp \}$ in $\phi$, 
	replace $\phi(\selOp(\tvar))$ once by $(\len{\tvar} {>} 0 \rightarrow \phi(\selOp(\tvar))) \vee (\len{\tvar} {\le} 0 \wedge \phi(0))$
	using axiom \RuleName{opDefault},
	and by the axiom \RuleName{accessDefault},
	replace every $\at{\tvar_1}{k} = \at{\tvar_2}{j}$ in $\phi$ once with the formula
	\begin{equation*}
		\big(
			0 {\le} k {<} \len{\tvar_1} \wedge 0 {\le} j {<} \len{\tvar_2} \rightarrow \at{\tvar_1}{k} = \at{\tvar_2}{j}	
		\big) \vee 
		\neg\big(
			0 {\le} k {<} \len{\tvar_1} \wedge 0 {\le} j {<} \len{\tvar_2}
		\big)
		\text{,}
	\end{equation*}
	In the resulting formula, every selector and access in $\phi$ is guarded by a range check.
	
	Then the formula $\toFOD{\phi}$ is inductively defined in \rref{fig:extensional-com-FOD-to-FOD}.
	The mapping $\toFOD{(\cdot)}$ uniformly replaces every trace variable $\tvar$ in $\phi$ with a fresh but fixed real variable $\encvar{\tvar} \in \RVar$
	and every operator on traces with the corresponding operator on encodings (\rref{lem:typecast}).
	Since $\phi \equiv \fa{\hvarvec_0{=}\hvarvec} (\phi_0)\subs{\hvarvec}{\hvarvec_0}$,
	every~$\encvar{\tvar}$ contains a trace encoding by $\eventencseq{\encvar{\tvar}}$
	in case $\toFOD{(\fa{\tvar} \varphi)}$,
	where $\traceenc$ is definable in FOD (\rref{lem:typecast}).
	For ease of presentation,	
	the type $\traceenc$ is omitted in the following
	if not required.

	\newcommand{\fagoedelize}[2]{\fa{\encvar{#1}{=}\goedelize{#2}}}
	\newcommand{\exgoedelize}[2]{\ex{\encvar{#1}{=}\goedelize{#2}}}

	Since there is a unique encoding $\eventencseq{\encvar{\tvar}}$ for every $\tvar$ and vice versa
	by \RuleName{TraceToGoedel} and \RuleName{GoedelToTrace},
	respectively,
	existential and universal quantification collapse by first-order reasoning:
	\begin{equation}
		\label{eq:goedel_dual}
		\exgoedelize{\hvarvec}{\hvarvec} \toFOD{\varphi} \leftrightarrow \fagoedelize{\hvarvec}{\hvarvec} \toFOD{\varphi}
	\end{equation}

	Then for $\hvarvec \supseteq \SFV(\phi) \cap \V_\traces$,
	show $\dLCHPdash \phi \leftrightarrow \fa{\encvar{\hvarvec}{=}\goedelize{\hvarvec}} \toFOD{\phi}$
	by induction on the structure of~$\phi$,
	where IH is short for induction hypothesis.
	Propositional reasoning (including \RuleName{MP}) is used without further notice.

	\begin{enumerate}[leftmargin=*]
		\item $\phi \equiv \re_1 \sim \re_2$ or $\phi \equiv \at{\tvar_1}{k} = \at{\tvar_2}{j}$,
		then $\dLCHPdash \encvar{\hvarvec}=\goedelize{\hvarvec} \rightarrow \len{\tvar} = \len{\encvar{\tvar}}$
		and  
		\begin{equation*}
			\dLCHPdash \encvar{\hvarvec}=\goedelize{\hvarvec} \rightarrow (0{\le}l{<}\len{\tvar} 
			\rightarrow \at{\tvar}{l} {=} 
			\big\langle
				\chan{\seqitm{\tvar}{l}},
				\val{\seqitm{\tvar}{l}}
			\big\rangle
		\end{equation*}
		for every $\tvar$ in $\phi$
		by definition of $\goedelize{\cdot}$ in \rref{eq:goedelize}.
		By the range checks above,
		$0{\le}\len{\tvar}{-}1{<}\len{\tvar}$ for every $\tvar$ in $\re_1 \sim \re_2$,
		and $0{\le}k{<}\len{\tvar_1}$ and $0{\le}j{<}\len{\tvar_2}$ for $\at{\tvar_1}{k} = \at{\tvar_2}{j}$
		are provable from~$\phi$.
		Hence, for $\re_1 \sim \re_2$,
		obtain $\dLCHPdash \encvar{\hvarvec}=\goedelize{\hvarvec} \rightarrow \selOp(\tvar) = \sel{\seqitm{\tvar}{\len{\tvar}-1}}$
		for every $\selOp(\tvar)$ in $\re_1 \sim \re_2$ 
		by \RuleName{eqItem},
		as $\selOp(\tvar)$ can be considered a shorthand for $\selOp(\at{\tvar}{\len{\tvar}-1})$.
		Then $\dLCHPdash \phi \rightarrow (\encvar{\hvarvec}=\goedelize{\hvarvec} \rightarrow \toFOD{\phi})$
		by equality,
		so $\dLCHPdash \phi \rightarrow \fagoedelize{\hvarvec}{\hvarvec} \toFOD{\phi}$ by first-order reasoning (\RuleName{fol}).
		Further, $\dLCHPdash \encvar{\hvarvec}=\goedelize{\hvarvec} \rightarrow (\toFOD{\phi} \rightarrow \phi)$
		by equality.
		Finally, $\dLCHPdash \fagoedelize{\hvarvec}{\hvarvec} \toFOD{\phi} \rightarrow \phi$ 
		by \RuleName{fol},
		as there is an encoding $\encvar{\hvarvec}$
		such that $\encvar{\hvarvec} = \goedelize{\hvarvec}$ holds
		by \RuleName{TraceToGoedel}.

		\item $\phi \equiv \neg\varphi$,
		then $\dLCHPdash \varphi \leftrightarrow \fagoedelize{\hvarvec}{\hvarvec} \toFOD{\varphi}$ by IH.
		Hence, $\dLCHPdash \varphi \leftrightarrow \exgoedelize{\hvarvec}{\hvarvec} \toFOD{\varphi}$ by \rref{eq:goedel_dual}.
		Finally, $\dLCHPdash \neg\varphi \leftrightarrow \fagoedelize{\hvarvec}{\hvarvec} \toFOD{(\neg\varphi)}$
		as
		$\exgoedelize{\hvarvec}{\hvarvec} \toFOD{\varphi} \equiv \neg \fagoedelize{\hvarvec}{\hvarvec} \neg \toFOD{\varphi}$
		and $\toFOD{(\neg\varphi)} \equiv \neg\toFOD{\varphi}$.

		\item $\phi \equiv \fa{x} \varphi$, then $\dLCHPdash \varphi \leftrightarrow \fagoedelize{\hvarvec}{\hvarvec}\toFOD{\varphi}$ by IH.
		Hence, $\dLCHPdash \fa{x} \varphi \leftrightarrow \fa{x} \fagoedelize{\hvarvec}{\hvarvec}\toFOD{\varphi}$ by \RuleName{fol},
		so $\dLCHPdash \fa{x} \varphi \leftrightarrow \fagoedelize{\hvarvec}{\hvarvec} \toFOD{(\fa{x} \varphi)}$
		by $\forall$-reordering
		as $\toFOD{(\fa{x} \varphi)} \equiv \fa{x} \toFOD{\varphi}$.

		\item $\phi \equiv \fa{\tvar} \varphi$, 
		then let 
		$\hvarvec = (\SFV(\varphi)\cap \V_\traces) \setminus \{\tvar\}$.
		Then $\dLCHPdash \varphi \leftrightarrow \fagoedelize{\tvar}{\tvar} \toFOD{\chi}$ by IH,
		where~$\toFOD{\chi} \equiv \fagoedelize{\hvarvec}{\hvarvec} \toFOD{\varphi}$.
		By \RuleName{fol},
		obtain $\dLCHPdash \varphi \leftrightarrow (\encvar{\tvar}{=}\goedelize{\tvar} \rightarrow \toFOD{\chi})$,
		using $\encvar{\tvar} \not\in \SFV(\varphi)$ for $\leftarrow$.
		Then \mbox{$\dLCHPdash \fa{\tvar} \varphi \rightarrow \faencvar{\tvar} \toFOD{\chi}$} 
		by \RuleName{fol}
		because, by~\RuleName{GoedelToTrace}, every $\eventencseq{\encvar{\tvar}}$
		encodes a trace $\tvar$ such that $\encvar{\tvar}=\goedelize{\tvar}$
		and $\tvar$ is not free in $\faencvar{\tvar} \toFOD{\chi}$.
		Further, $\dLCHPdash \faencvar{\tvar} \toFOD{\chi} \rightarrow \fa{\tvar} \varphi$
		by \RuleName{fol}
		because by~\RuleName{TraceToGoedel} an encoding $\eventencseq{\tvar}$ exists for every trace~$\tvar$.
		Finally, $\dLCHPdash \fa{\tvar} \varphi \leftrightarrow \fa{\encvar{\hvarvec}{=}\goedelize{\hvarvec}} \toFOD{(\fa{\tvar} \varphi)}$ by $\forall$-reordering
		since $\toFOD{(\fa{\tvar}\varphi)} \equiv \faencvar{\tvar} \toFOD{\varphi}$.

		\newcommand{\anymodal}{\scalebox{0.8}{$\bigcirc$}}

		\item $\phi \equiv \dbleft  \evolution*{}{non} \dbright \psi$,
		then let $\square \equiv [ \evolution*{}{non} ]$ and $\Diamond \equiv \langle \evolution*{}{non} \rangle$.
		For all $\varphi,\alpha,\psi$,
		the formula $\varphi \wedge \langle \alpha \rangle \psi \leftrightarrow \langle \alpha \rangle (\varphi \wedge \psi)$
		derives if $\SFV(\varphi) \cap \SBV(\alpha) = \emptyset$ using \RuleName{vacuous}.%
		\footnote{
			Essentially by \RuleName{acMono},
			$\dLCHPdash \langle \alpha \rangle (\varphi \wedge \psi) \rightarrow \langle \alpha \rangle \varphi \wedge \langle \alpha \rangle \psi$.
			Then $\dLCHPdash \langle \alpha \rangle (\varphi \wedge \psi) \rightarrow \varphi \wedge \langle \alpha \rangle \psi$
			by the derivable dual of \RuleName{vacuous}.
			By \RuleName{acModalMP} and duality \RuleName{dbDual},
			obtain $\dLCHPdash [ \alpha ] \varphi \rightarrow (\langle \alpha \rangle \psi \rightarrow \langle \alpha \rangle (\varphi \wedge \psi))$.
			Then $\varphi \rightarrow (\langle \alpha \rangle \psi \rightarrow \langle \alpha \rangle (\varphi \wedge \psi))$ by \RuleName{vacuous}.
			Finally, $\varphi \wedge \langle \alpha \rangle \psi \rightarrow \langle \alpha \rangle (\varphi \wedge \psi)$ propositionally.
		}
		Further, assume $\recvar \not\in \hvarvec = \SFV(\phi) \cap \TVar = \SFV(\psi) \cap \TVar$
		using axiom \RuleName{odeSync}.
		Hence, $\dLCHPdash (\encvar{\hvarvec}{=}\goedelize{\hvarvec} \wedge \Diamond \toFOD{\psi}) \leftrightarrow \Diamond (\encvar{\hvarvec}{=}\goedelize{\hvarvec} \wedge \toFOD{\psi})$ 
		since $\encvar{\hvarvec}, \hvarvec \not\in \SBV(\evolution*{}{non})$.
		Then $\dLCHPdash \exgoedelize{\hvarvec}{\hvarvec} \Diamond \toFOD{\psi} \leftrightarrow \ex{\encvar{\hvarvec}} \Diamond (\encvar{\hvarvec}{=}\goedelize{\hvarvec} \wedge \toFOD{\psi})$ by \RuleName{fol}.
		Hence,
		$\dLCHPdash \exgoedelize{\hvarvec}{\hvarvec} \Diamond \toFOD{\psi} \leftrightarrow \Diamond \exgoedelize{\hvarvec}{\hvarvec} \toFOD{\psi}$
		by~\RuleName{barcan} as $\encvar{\hvarvec} \not\in \evolution*{}{non}$.
		This implies $\dLCHPdash \fagoedelize{\hvarvec}{\hvarvec} \anymodal \toFOD{\psi} \leftrightarrow \anymodal \fagoedelize{\hvarvec}{\hvarvec} \toFOD{\psi}$
		by duality \RuleName{dbDual} for $\anymodal \equiv \square$
		and \rref{eq:goedel_dual} for $\anymodal \equiv \lozenge$.
		Since $\dLCHPdash \psi \leftrightarrow \fagoedelize{\hvarvec}{\hvarvec} \toFOD{\psi}$ by IH,
		obtain $\dLCHPdash \anymodal \psi \leftrightarrow \anymodal \fagoedelize{\hvarvec}{\hvarvec} \toFOD{\psi}$  by \RuleName{acMono} or \RuleName{acDiaMono}.
		The latter combines with $\dLCHPdash \fagoedelize{\hvarvec}{\hvarvec} \anymodal \toFOD{\psi} \leftrightarrow \anymodal \fagoedelize{\hvarvec}{\hvarvec} \toFOD{\psi}$ 
		to $\dLCHPdash \phi \leftrightarrow \fagoedelize{\hvarvec}{\hvarvec} \anymodal \toFOD{\psi}$.
		\qedhere
	\end{enumerate}
\end{proof}

\begin{proof}
	[Proof of \rref{thm:fod-completeness}.]
	Let $\phi$ be a valid \dLCHP formula.
	By \rref{thm:com-fod-completeness},
	there are \comFOD tautologies $\phi_1, \ldots, \phi_\semvar$ from which $\phi$ derives in \dLCHP's calculus (\rref{fig:calculus}).
	Since $\phi_k$ is a tautology for each $1\le k\le\semvar$,
	assume that $\phi_k$ contains no free trace variables by axiom \RuleName{uniInstance}.
	Then there is a FOD formula~$\toFOD{\phi_k}$ by \rref{prop:com-fod-to-fod} 
	such that $\phi_k \leftrightarrow \toFOD{\phi_k}$ derives in the extended \dLCHP calculus \dLCHPdash.
	Note that there is no quantifier around $\toFOD{\phi_k}$ since~$\phi_k$ has no free trace variables.
	By soundness (\rref{thm:extended_soundness}),
	$\toFOD{\phi_k}$ is a tautology because $\phi_k$ is.
	Thus, the derivation of $\phi$
	can be extended to a derivation of $\phi$ in $\dLCHPdash$ from the FOD tautologies $\toFOD{\phi_1}, \ldots, \toFOD{\phi_\semvar}$.
\end{proof}

\section*{Declarations}

Marvin Brieger wrote the main manuscript text. All authors contributed to the introduction, related work, and conclusion. All authors reviewed the manuscript.

\section*{Acknowledgements}

We thank the anonymous reviewers for their helpful feedback.

\section*{Funding}

This project was funded in part by the Deutsche For\-schungs-gemeinschaft (DFG) -- \href{https://gepris.dfg.de/gepris/projekt/378803395?context=projekt&task=showDetail&id=378803395&}{378803395} (ConVeY), 
an Alexander von Humboldt Professorship,
by the AFOSR under grant no.\,FA9550-16-1-0288,
and by the NSF under grant no.\,CCF2427581.

\newcommand{\doi}[1]{doi: \href{https://doi.org/#1}{\nolinkurl{#1}}}
\bibliography{platzer, literature}


\begin{thebibliography}{77}
\ifx \bisbn   \undefined \def \bisbn  #1{ISBN #1}\fi
\ifx \binits  \undefined \def \binits#1{#1}\fi
\ifx \bauthor  \undefined \def \bauthor#1{#1}\fi
\ifx \batitle  \undefined \def \batitle#1{#1}\fi
\ifx \bjtitle  \undefined \def \bjtitle#1{#1}\fi
\ifx \bvolume  \undefined \def \bvolume#1{\textbf{#1}}\fi
\ifx \byear  \undefined \def \byear#1{#1}\fi
\ifx \bissue  \undefined \def \bissue#1{#1}\fi
\ifx \bfpage  \undefined \def \bfpage#1{#1}\fi
\ifx \blpage  \undefined \def \blpage #1{#1}\fi
\ifx \burl  \undefined \def \burl#1{\textsf{#1}}\fi
\ifx \doiurl  \undefined \def \doiurl#1{\url{https://doi.org/#1}}\fi
\ifx \betal  \undefined \def \betal{\textit{et al.}}\fi
\ifx \binstitute  \undefined \def \binstitute#1{#1}\fi
\ifx \binstitutionaled  \undefined \def \binstitutionaled#1{#1}\fi
\ifx \bctitle  \undefined \def \bctitle#1{#1}\fi
\ifx \beditor  \undefined \def \beditor#1{#1}\fi
\ifx \bpublisher  \undefined \def \bpublisher#1{#1}\fi
\ifx \bbtitle  \undefined \def \bbtitle#1{#1}\fi
\ifx \bedition  \undefined \def \bedition#1{#1}\fi
\ifx \bseriesno  \undefined \def \bseriesno#1{#1}\fi
\ifx \blocation  \undefined \def \blocation#1{#1}\fi
\ifx \bsertitle  \undefined \def \bsertitle#1{#1}\fi
\ifx \bsnm \undefined \def \bsnm#1{#1}\fi
\ifx \bsuffix \undefined \def \bsuffix#1{#1}\fi
\ifx \bparticle \undefined \def \bparticle#1{#1}\fi
\ifx \barticle \undefined \def \barticle#1{#1}\fi
\bibcommenthead
\ifx \bconfdate \undefined \def \bconfdate #1{#1}\fi
\ifx \botherref \undefined \def \botherref #1{#1}\fi
\ifx \url \undefined \def \url#1{\textsf{#1}}\fi
\ifx \bchapter \undefined \def \bchapter#1{#1}\fi
\ifx \bbook \undefined \def \bbook#1{#1}\fi
\ifx \bcomment \undefined \def \bcomment#1{#1}\fi
\ifx \oauthor \undefined \def \oauthor#1{#1}\fi
\ifx \citeauthoryear \undefined \def \citeauthoryear#1{#1}\fi
\ifx \endbibitem  \undefined \def \endbibitem {}\fi
\ifx \bconflocation  \undefined \def \bconflocation#1{#1}\fi
\ifx \arxivurl  \undefined \def \arxivurl#1{\textsf{#1}}\fi
\csname PreBibitemsHook\endcsname

\bibitem[\protect\citeauthoryear{Alur et~al.}{1993}]{Alur1993}
\begin{bchapter}
\bauthor{\bsnm{Alur}, \binits{R.}},
\bauthor{\bsnm{Courcoubetis}, \binits{C.}},
\bauthor{\bsnm{Henzinger}, \binits{T.A.}},
\bauthor{\bsnm{Ho}, \binits{P.}}:
\bctitle{Hybrid automata: {A}n algorithmic approach to the specification and verification of hybrid systems}.
In: \beditor{\bsnm{Grossman}, \binits{R.L.}},
\beditor{\bsnm{Nerode}, \binits{A.}},
\beditor{\bsnm{Ravn}, \binits{A.P.}},
\beditor{\bsnm{Rischel}, \binits{H.}} (eds.)
\bbtitle{Proc. 1th and 2nd Intl. Workshop Hybrid Systems {(HS)}}.
\bsertitle{LNCS},
vol. \bseriesno{736},
pp. \bfpage{209}--\blpage{229}.
\bpublisher{Springer}
(\byear{1993}).
\doiurl{10.1007/3-540-57318-6\_30}
\end{bchapter}
\endbibitem

\bibitem[\protect\citeauthoryear{Apt et~al.}{2010}]{AptdeBoerOlderog10}
\begin{bbook}
\bauthor{\bsnm{Apt}, \binits{K.R.}},
\bauthor{\bsnm{Boer}, \binits{F.S.}},
\bauthor{\bsnm{Olderog}, \binits{E.-R.}}:
\bbtitle{Verification of Sequential and Concurrent Programs},
\bedition{3rd} edn.
\bpublisher{Springer}
(\byear{2010}).
\doiurl{10.1007/978-1-84882-745-5}
\end{bbook}
\endbibitem

\bibitem[\protect\citeauthoryear{Alur}{2011}]{Alur2011}
\begin{bchapter}
\bauthor{\bsnm{Alur}, \binits{R.}}:
\bctitle{Formal verification of hybrid systems}.
In: \beditor{\bsnm{Chakraborty}, \binits{S.}},
\beditor{\bsnm{Jerraya}, \binits{A.}},
\beditor{\bsnm{Baruah}, \binits{S.K.}},
\beditor{\bsnm{Fischmeister}, \binits{S.}} (eds.)
\bbtitle{Proc. 11th Intl. Conf. Embedded Software {(EMSOFT)}},
pp. \bfpage{273}--\blpage{278}.
\bpublisher{{ACM} Press}
(\byear{2011}).
\doiurl{10.1145/2038642.2038685}
\end{bchapter}
\endbibitem

\bibitem[\protect\citeauthoryear{{Abou El Wafa} and Platzer}{2024}]{AbouElWafa2024}
\begin{bchapter}
\bauthor{\bsnm{{Abou El Wafa}}, \binits{N.}},
\bauthor{\bsnm{Platzer}, \binits{A.}}:
\bctitle{Complete game logic with sabotage}.
In: \beditor{\bsnm{Sobocinski}, \binits{P.}},
\beditor{\bsnm{Lago}, \binits{U.D.}},
\beditor{\bsnm{Esparza}, \binits{J.}} (eds.)
\bbtitle{Proc. 39th {ACM/IEEE} Symp. Logic in Computer Science {(LICS)}},
pp. \bfpage{1}--\blpage{15}.
\bpublisher{{ACM}}
(\byear{2024}).
\doiurl{10.1145/3661814.3662121}
\end{bchapter}
\endbibitem

\bibitem[\protect\citeauthoryear{Barcan}{1946}]{Barcan1946}
\begin{barticle}
\bauthor{\bsnm{Barcan}, \binits{R.C.}}:
\batitle{A functional calculus of first order based on strict implication}.
\bjtitle{J. Symb. Log.}
\bvolume{11}(\bissue{1}),
\bfpage{1}--\blpage{16}
(\byear{1946})
\doiurl{10.2307/2269159}
\end{barticle}
\endbibitem

\bibitem[\protect\citeauthoryear{Benvenuti et~al.}{2014}]{Benvenuti2014}
\begin{barticle}
\bauthor{\bsnm{Benvenuti}, \binits{L.}},
\bauthor{\bsnm{Bresolin}, \binits{D.}},
\bauthor{\bsnm{Collins}, \binits{P.}},
\bauthor{\bsnm{Ferrari}, \binits{A.}},
\bauthor{\bsnm{Geretti}, \binits{L.}},
\bauthor{\bsnm{Villa}, \binits{T.}}:
\batitle{Assume–guarantee verification of nonlinear hybrid systems with {ARIADNE}}.
\bjtitle{Intl. J. Robust Nonlinear Control}
\bvolume{24},
\bfpage{699}--\blpage{724}
(\byear{2014})
\doiurl{10.1002/rnc.2914}
\end{barticle}
\endbibitem

\bibitem[\protect\citeauthoryear{Brieger et~al.}{2023a}]{BriegerMP2023}
\begin{botherref}
\oauthor{\bsnm{Brieger}, \binits{M.}},
\oauthor{\bsnm{Mitsch}, \binits{S.}},
\oauthor{\bsnm{Platzer}, \binits{A.}}:
Dynamic logic of communicating hybrid programs.
CoRR
\textbf{abs/2302.14546}
(2023)
\doiurl{10.48550/arXiv.2302.14546}
\end{botherref}
\endbibitem

\bibitem[\protect\citeauthoryear{Brieger et~al.}{2023b}]{Brieger2023}
\begin{bchapter}
\bauthor{\bsnm{Brieger}, \binits{M.}},
\bauthor{\bsnm{Mitsch}, \binits{S.}},
\bauthor{\bsnm{Platzer}, \binits{A.}}:
\bctitle{Uniform substitution for dynamic logic with communicating hybrid programs}.
In: \beditor{\bsnm{Pientka}, \binits{B.}},
\beditor{\bsnm{Tinelli}, \binits{C.}} (eds.)
\bbtitle{Proc. 29th Intl. Conf. on Automated Deduction {(CADE)}}.
\bsertitle{LNCS},
vol. \bseriesno{14132},
pp. \bfpage{96}--\blpage{115}.
\bpublisher{Springer}
(\byear{2023}).
\doiurl{10.1007/978-3-031-38499-8\_6}
\end{bchapter}
\endbibitem

\bibitem[\protect\citeauthoryear{Brieger et~al.}{2024}]{Brieger2024}
\begin{botherref}
\oauthor{\bsnm{Brieger}, \binits{M.}},
\oauthor{\bsnm{Mitsch}, \binits{S.}},
\oauthor{\bsnm{Platzer}, \binits{A.}}:
Complete dynamic logic of communicating hybrid programs.
CoRR
\textbf{abs/2408.05012}
(2024)
\doiurl{10.48550/ARXIV.2408.05012}
\end{botherref}
\endbibitem

\bibitem[\protect\citeauthoryear{Brookes}{1996}]{Brookes1996}
\begin{barticle}
\bauthor{\bsnm{Brookes}, \binits{S.D.}}:
\batitle{Full abstraction for a shared-variable parallel language}.
\bjtitle{Inf. Comput.}
\bvolume{127}(\bissue{2}),
\bfpage{145}--\blpage{163}
(\byear{1996})
\doiurl{10.1006/inco.1996.0056}
\end{barticle}
\endbibitem

\bibitem[\protect\citeauthoryear{Brookes}{2002}]{Brookes2002}
\begin{bchapter}
\bauthor{\bsnm{Brookes}, \binits{S.D.}}:
\bctitle{Traces, pomsets, fairness and full abstraction for communicating processes}.
In: \beditor{\bsnm{Brim}, \binits{L.}},
\beditor{\bsnm{Jancar}, \binits{P.}},
\beditor{\bsnm{Kret{\'{\i}}nsk{\'{y}}}, \binits{M.}},
\beditor{\bsnm{Kucera}, \binits{A.}} (eds.)
\bbtitle{Proc. 13th Intl. Conf. Concurrency Theory {(CONCUR)}}.
\bsertitle{LNCS},
vol. \bseriesno{2421},
pp. \bfpage{466}--\blpage{482}.
\bpublisher{Springer}
(\byear{2002}).
\doiurl{10.1007/3-540-45694-5\_31}
\end{bchapter}
\endbibitem

\bibitem[\protect\citeauthoryear{Clarke et~al.}{2001}]{Clarke2001}
\begin{bbook}
\bauthor{\bsnm{Clarke}, \binits{E.M.}},
\bauthor{\bsnm{Grumberg}, \binits{O.}},
\bauthor{\bsnm{Peled}, \binits{D.A.}}:
\bbtitle{Model Checking, 1st Edition}.
\bpublisher{{MIT} Press}
(\byear{2001})
\end{bbook}
\endbibitem

\bibitem[\protect\citeauthoryear{Cook}{1978}]{Cook1978}
\begin{barticle}
\bauthor{\bsnm{Cook}, \binits{S.A.}}:
\batitle{Soundness and completeness of an axiom system for program verification}.
\bjtitle{{SIAM} J. Comput.}
\bvolume{7}(\bissue{1}),
\bfpage{70}--\blpage{90}
(\byear{1978})
\doiurl{10.1137/0207005}
\end{barticle}
\endbibitem

\bibitem[\protect\citeauthoryear{Chaochen et~al.}{1993}]{Chaochen1993}
\begin{bchapter}
\bauthor{\bsnm{Chaochen}, \binits{Z.}},
\bauthor{\bsnm{Ravn}, \binits{A.P.}},
\bauthor{\bsnm{Hansen}, \binits{M.R.}}:
\bctitle{An extended duration calculus for hybrid real-time systems}.
In: \beditor{\bsnm{Grossman}, \binits{R.L.}},
\beditor{\bsnm{Nerode}, \binits{A.}},
\beditor{\bsnm{Ravn}, \binits{A.P.}},
\beditor{\bsnm{Rischel}, \binits{H.}} (eds.)
\bbtitle{Proc. 1th and 2nd Intl. Workshop Hybrid Systems {(HS)}}.
\bsertitle{LNCS},
vol. \bseriesno{736},
pp. \bfpage{36}--\blpage{59}.
\bpublisher{Springer}
(\byear{1993}).
\doiurl{10.1007/3-540-57318-6\_23}
\end{bchapter}
\endbibitem

\bibitem[\protect\citeauthoryear{Cong et~al.}{2013}]{Cong2013}
\begin{bchapter}
\bauthor{\bsnm{Cong}, \binits{X.}},
\bauthor{\bsnm{Yu}, \binits{H.}},
\bauthor{\bsnm{Xu}, \binits{X.}}:
\bctitle{Verification of hybrid chi model for cyber-physical systems using {PHAVer}}.
In: \beditor{\bsnm{Barolli}, \binits{L.}},
\beditor{\bsnm{You}, \binits{I.}},
\beditor{\bsnm{Xhafa}, \binits{F.}},
\beditor{\bsnm{Leu}, \binits{F.}},
\beditor{\bsnm{Chen}, \binits{H.}} (eds.)
\bbtitle{Proc. 7th Intl. Conf. Innovative Mobile and Internet Services in Ubiquitous Computing, {(IMIS)}},
pp. \bfpage{122}--\blpage{128}.
\bpublisher{{IEEE} Computer Society}
(\byear{2013}).
\doiurl{10.1109/IMIS.2013.29}
\end{bchapter}
\endbibitem

\bibitem[\protect\citeauthoryear{de~Roever}{1997}]{Roever1997}
\begin{bchapter}
\bauthor{\bsnm{Roever}, \binits{W.P.}}:
\bctitle{The need for compositional proof systems: {A} survey}.
In: \beditor{\bsnm{Roever}, \binits{W.P.}},
\beditor{\bsnm{Langmaack}, \binits{H.}},
\beditor{\bsnm{Pnueli}, \binits{A.}} (eds.)
\bbtitle{Intl. Symp. Compositionality: The Significant Difference {(COMPOS)}}.
\bsertitle{LNCS},
vol. \bseriesno{1536},
pp. \bfpage{1}--\blpage{22}.
\bpublisher{Springer}
(\byear{1997}).
\doiurl{10.1007/3-540-49213-5\_1}
\end{bchapter}
\endbibitem

\bibitem[\protect\citeauthoryear{de~Roever et~al.}{2001}]{deRoever2001}
\begin{bbook}
\bauthor{\bsnm{Roever}, \binits{W.P.}},
\bauthor{\bsnm{Boer}, \binits{F.S.}},
\bauthor{\bsnm{Hannemann}, \binits{U.}},
\bauthor{\bsnm{Hooman}, \binits{J.J.M.}},
\bauthor{\bsnm{Lakhnech}, \binits{Y.}},
\bauthor{\bsnm{Poel}, \binits{M.}},
\bauthor{\bsnm{Zwiers}, \binits{J.}}:
\bbtitle{Concurrency Verification: Introduction to Compositional and Noncompositional Methods}.
\bsertitle{Cambridge Tracts in Theoretical Computer Science},
vol. \bseriesno{54}.
\bpublisher{Cambridge University Press}
(\byear{2001})
\end{bbook}
\endbibitem

\bibitem[\protect\citeauthoryear{Frehse et~al.}{2004}]{Frehse2004}
\begin{bchapter}
\bauthor{\bsnm{Frehse}, \binits{G.}},
\bauthor{\bsnm{Han}, \binits{Z.}},
\bauthor{\bsnm{Krogh}, \binits{B.H.}}:
\bctitle{Assume-guarantee reasoning for hybrid {I/O}-automata by over-approximation of continuous interaction}.
In: \bbtitle{Proc. 43rd IEEE Conf. Decision and Control {(CDC)}},
pp. \bfpage{479}--\blpage{484}.
\bpublisher{{IEEE}}
(\byear{2004}).
\doiurl{10.1109/CDC.2004.1428676}
\end{bchapter}
\endbibitem

\bibitem[\protect\citeauthoryear{Fitting and Mendelsohn}{1999}]{Fitting1999}
\begin{bbook}
\bauthor{\bsnm{Fitting}, \binits{M.}},
\bauthor{\bsnm{Mendelsohn}, \binits{R.L.}}:
\bbtitle{First-Order Modal Logic}.
\bpublisher{Kluwer}
(\byear{1999}).
\doiurl{10.1007/978-94-011-5292-1}
\end{bbook}
\endbibitem

\bibitem[\protect\citeauthoryear{Guelev et~al.}{2017}]{Guelev2017}
\begin{bchapter}
\bauthor{\bsnm{Guelev}, \binits{D.P.}},
\bauthor{\bsnm{Wang}, \binits{S.}},
\bauthor{\bsnm{Zhan}, \binits{N.}}:
\bctitle{Compositional {H}oare-style reasoning about hybrid {CSP} in the duration calculus}.
In: \beditor{\bsnm{Larsen}, \binits{K.G.}},
\beditor{\bsnm{Sokolsky}, \binits{O.}},
\beditor{\bsnm{Wang}, \binits{J.}} (eds.)
\bbtitle{Proc. 3rd Intl. Symp. Dependable Software Engineering. Theories, Tools, and Applications {(SETTA)}}.
\bsertitle{LNCS},
vol. \bseriesno{10606},
pp. \bfpage{110}--\blpage{127}.
\bpublisher{Springer}
(\byear{2017}).
\doiurl{10.1007/978-3-319-69483-2\_7}
\end{bchapter}
\endbibitem

\bibitem[\protect\citeauthoryear{Gödel}{1931}]{Goedel1931}
\begin{barticle}
\bauthor{\bsnm{Gödel}, \binits{K.}}:
\batitle{{\"U}ber formal unentscheidbare {S}ätze der {P}rincipia {M}athematica und verwandter {S}ysteme {I}}.
\bjtitle{Monatshefte für Mathematik und Physik}
\bvolume{38},
\bfpage{173}--\blpage{198}
(\byear{1931})
\doiurl{10.1007/BF01700692}
\end{barticle}
\endbibitem

\bibitem[\protect\citeauthoryear{Harel}{1979}]{Harel1979}
\begin{bbook}
\bauthor{\bsnm{Harel}, \binits{D.}}:
\bbtitle{First-Order Dynamic Logic}.
\bsertitle{LNCS},
vol. \bseriesno{68}.
\bpublisher{Springer}
(\byear{1979}).
\doiurl{10.1007/3-540-09237-4}
\end{bbook}
\endbibitem

\bibitem[\protect\citeauthoryear{Hooman and de~Roever}{1992}]{Hooman1992}
\begin{barticle}
\bauthor{\bsnm{Hooman}, \binits{J.J.M.}},
\bauthor{\bsnm{Roever}, \binits{W.P.}}:
\batitle{An introduction to compositional methods for concurrency and their application to real-time}.
\bjtitle{Sādhanā}
\bvolume{17}(\bissue{1}),
\bfpage{29}--\blpage{73}
(\byear{1992})
\doiurl{10.1007/BF02811338}
\end{barticle}
\endbibitem

\bibitem[\protect\citeauthoryear{Henzinger}{1996}]{Henzinger1996}
\begin{bchapter}
\bauthor{\bsnm{Henzinger}, \binits{T.A.}}:
\bctitle{The theory of hybrid automata}.
In: \bbtitle{Proc. 11th {IEEE} Symp. Logic in Computer Science {(LICS)}},
pp. \bfpage{278}--\blpage{292}.
\bpublisher{{IEEE}}
(\byear{1996}).
\doiurl{10.1109/LICS.1996.561342}
\end{bchapter}
\endbibitem

\bibitem[\protect\citeauthoryear{Henzinger et~al.}{1995}]{Henzinger1995}
\begin{bchapter}
\bauthor{\bsnm{Henzinger}, \binits{T.A.}},
\bauthor{\bsnm{Kopke}, \binits{P.W.}},
\bauthor{\bsnm{Puri}, \binits{A.}},
\bauthor{\bsnm{Varaiya}, \binits{P.}}:
\bctitle{What's decidable about hybrid automata?}
In: \beditor{\bsnm{Leighton}, \binits{F.T.}},
\beditor{\bsnm{Borodin}, \binits{A.}} (eds.)
\bbtitle{Proc. 27th Annual {ACM} Symp. on Theory of Computing},
pp. \bfpage{373}--\blpage{382}.
\bpublisher{{ACM}}
(\byear{1995}).
\doiurl{10.1145/225058.225162}
\end{bchapter}
\endbibitem

\bibitem[\protect\citeauthoryear{Harel et~al.}{2000}]{Harel_et_al_2000}
\begin{bbook}
\bauthor{\bsnm{Harel}, \binits{D.}},
\bauthor{\bsnm{Kozen}, \binits{D.}},
\bauthor{\bsnm{Tiuryn}, \binits{J.}}:
\bbtitle{Dynamic Logic}.
\bpublisher{MIT Press},
(\byear{2000}).
\doiurl{10.7551/mitpress/2516.001.0001}
\end{bbook}
\endbibitem

\bibitem[\protect\citeauthoryear{Harel et~al.}{1977}]{Harel1977}
\begin{bchapter}
\bauthor{\bsnm{Harel}, \binits{D.}},
\bauthor{\bsnm{Meyer}, \binits{A.R.}},
\bauthor{\bsnm{Pratt}, \binits{V.R.}}:
\bctitle{Computability and completeness in logics of programs (preliminary report)}.
In: \beditor{\bsnm{Hopcroft}, \binits{J.E.}},
\beditor{\bsnm{Friedman}, \binits{E.P.}},
\beditor{\bsnm{Harrison}, \binits{M.A.}} (eds.)
\bbtitle{Proc. 9th Annual {ACM} Symp. Theory of Computing},
pp. \bfpage{261}--\blpage{268}.
\bpublisher{{ACM}}
(\byear{1977}).
\doiurl{10.1145/800105.803416}
\end{bchapter}
\endbibitem

\bibitem[\protect\citeauthoryear{Henzinger et~al.}{2001}]{Henzinger2001}
\begin{bchapter}
\bauthor{\bsnm{Henzinger}, \binits{T.A.}},
\bauthor{\bsnm{Minea}, \binits{M.}},
\bauthor{\bsnm{Prabhu}, \binits{V.S.}}:
\bctitle{Assume-guarantee reasoning for hierarchical hybrid systems}.
In: \beditor{\bsnm{Benedetto}, \binits{M.D.D.}},
\beditor{\bsnm{Sangiovanni{-}Vincentelli}, \binits{A.L.}} (eds.)
\bbtitle{Proc. 4th Intl. Workshop Hybrid Systems: Computation and Control {(HSCC)}}.
\bsertitle{LNCS},
vol. \bseriesno{2034},
pp. \bfpage{275}--\blpage{290}.
\bpublisher{Springer}
(\byear{2001}).
\doiurl{10.1007/3-540-45351-2\_24}
\end{bchapter}
\endbibitem

\bibitem[\protect\citeauthoryear{Hoare}{1978}]{Hoare1978}
\begin{barticle}
\bauthor{\bsnm{Hoare}, \binits{C.A.R.}}:
\batitle{Communicating sequential processes}.
\bjtitle{Communications of the {ACM}}
\bvolume{21}(\bissue{8}),
\bfpage{666}--\blpage{677}
(\byear{1978})
\doiurl{10.1145/359576.359585}
\end{barticle}
\endbibitem

\bibitem[\protect\citeauthoryear{Hooman}{1987}]{Hooman1987}
\begin{bchapter}
\bauthor{\bsnm{Hooman}, \binits{J.}}:
\bctitle{A compositional proof theory for real-time distributed message passing}.
In: \beditor{\bsnm{Bakker}, \binits{J.W.}},
\beditor{\bsnm{Nijman}, \binits{A.J.}},
\beditor{\bsnm{Treleaven}, \binits{P.C.}} (eds.)
\bbtitle{PARLE, Parallel Architectures and Languages Europe, Volume {II:} Parallel Languages, Eindhoven, The Netherlands, June 15-19, 1987, Proceedings}.
\bsertitle{LNCS},
vol. \bseriesno{259},
pp. \bfpage{315}--\blpage{332}.
\bpublisher{Springer}
(\byear{1987}).
\doiurl{10.1007/3-540-17945-3\_18}
\end{bchapter}
\endbibitem

\bibitem[\protect\citeauthoryear{Hooman}{1991}]{Hooman1991}
\begin{bbook}
\bauthor{\bsnm{Hooman}, \binits{J.}}:
\bbtitle{Specification and Compositional Verification of Real-Time Systems}.
\bsertitle{LNCS},
vol. \bseriesno{558}.
\bpublisher{Springer}
(\byear{1991}).
\doiurl{10.1007/3-540-54947-1}
\end{bbook}
\endbibitem

\bibitem[\protect\citeauthoryear{Hooman}{1993}]{Hooman1993}
\begin{bchapter}
\bauthor{\bsnm{Hooman}, \binits{J.}}:
\bctitle{A compositional approach to the design of hybrid systems}.
In: \beditor{\bsnm{Grossman}, \binits{R.L.}},
\beditor{\bsnm{Nerode}, \binits{A.}},
\beditor{\bsnm{Ravn}, \binits{A.P.}},
\beditor{\bsnm{Rischel}, \binits{H.}} (eds.)
\bbtitle{Proc. 1th and 2nd Intl. Workshop Hybrid Systems {(HS)}}.
\bsertitle{LNCS},
vol. \bseriesno{736},
pp. \bfpage{121}--\blpage{148}.
\bpublisher{Springer}
(\byear{1993}).
\doiurl{10.1007/3-540-57318-6\_27}
\end{bchapter}
\endbibitem

\bibitem[\protect\citeauthoryear{Hooman and Widom}{1989}]{Hooman1989}
\begin{bchapter}
\bauthor{\bsnm{Hooman}, \binits{J.}},
\bauthor{\bsnm{Widom}, \binits{J.}}:
\bctitle{A temporal-logic based compositional proof system for real-time message passing}.
In: \beditor{\bsnm{Odijk}, \binits{E.}},
\beditor{\bsnm{Rem}, \binits{M.}},
\beditor{\bsnm{Syre}, \binits{J.}} (eds.)
\bbtitle{Proc. Parallel Architectures and Languages Europe ({PARLE}), Volume {II:} Parallel Languages}.
\bsertitle{LNCS},
vol. \bseriesno{366},
pp. \bfpage{424}--\blpage{441}.
\bpublisher{Springer}
(\byear{1989}).
\doiurl{10.1007/3-540-51285-3\_56}
\end{bchapter}
\endbibitem

\bibitem[\protect\citeauthoryear{Jifeng}{1994}]{Jifeng1994}
\begin{bchapter}
\bauthor{\bsnm{Jifeng}, \binits{H.}}:
\bctitle{From {CSP} to Hybrid Systems}.
\bbtitle{A classical mind: essays in honour of C. A. R. Hoare},
pp. \bfpage{171}--\blpage{189}.
\bpublisher{Prentice Hall International}
(\byear{1994})
\end{bchapter}
\endbibitem

\bibitem[\protect\citeauthoryear{Jeannin and Platzer}{2014}]{DBLP:conf/cade/JeanninP14}
\begin{bchapter}
\bauthor{\bsnm{Jeannin}, \binits{J.}},
\bauthor{\bsnm{Platzer}, \binits{A.}}:
\bctitle{{dTL$^2$}: Differential temporal dynamic logic with nested temporalities for hybrid systems}.
In: \beditor{\bsnm{Demri}, \binits{S.}},
\beditor{\bsnm{Kapur}, \binits{D.}},
\beditor{\bsnm{Weidenbach}, \binits{C.}} (eds.)
\bbtitle{IJCAR}.
\bsertitle{LNCS},
vol. \bseriesno{8562},
pp. \bfpage{292}--\blpage{306}.
\bpublisher{Springer},
(\byear{2014}).
\doiurl{10.1007/978-3-319-08587-6_22}
\end{bchapter}
\endbibitem

\bibitem[\protect\citeauthoryear{King}{1976}]{King1976}
\begin{barticle}
\bauthor{\bsnm{King}, \binits{J.C.}}:
\batitle{Symbolic execution and program testing}.
\bjtitle{Commun. {ACM}}
\bvolume{19}(\bissue{7}),
\bfpage{385}--\blpage{394}
(\byear{1976})
\doiurl{10.1145/360248.360252}
\end{barticle}
\endbibitem

\bibitem[\protect\citeauthoryear{Kamburjan et~al.}{2022}]{Kamburjan2022}
\begin{botherref}
\oauthor{\bsnm{Kamburjan}, \binits{E.}},
\oauthor{\bsnm{Mitsch}, \binits{S.}},
\oauthor{\bsnm{H{\"{a}}hnle}, \binits{R.}}:
A hybrid programming language for formal modeling and verification of hybrid systems.
Leibniz Trans. Embed. Syst.
\textbf{8}(2)
(2022)
\doiurl{10.4230/LITES.8.2.4}
\end{botherref}
\endbibitem

\bibitem[\protect\citeauthoryear{Kamburjan et~al.}{2020}]{Kamburjan2020}
\begin{bchapter}
\bauthor{\bsnm{Kamburjan}, \binits{E.}},
\bauthor{\bsnm{Schlatte}, \binits{R.}},
\bauthor{\bsnm{Johnsen}, \binits{E.B.}},
\bauthor{\bsnm{Tarifa}, \binits{S.L.T.}}:
\bctitle{Designing distributed control with hybrid active objects}.
In: \beditor{\bsnm{Margaria}, \binits{T.}},
\beditor{\bsnm{Steffen}, \binits{B.}} (eds.)
\bbtitle{Proc. 9th Intl. Symp. Leveraging Applications of Formal Methods : Tools and Trends {(ISoLA)}}.
\bsertitle{LNCS},
vol. \bseriesno{12479},
pp. \bfpage{88}--\blpage{108}.
\bpublisher{Springer}
(\byear{2020}).
\doiurl{10.1007/978-3-030-83723-5\_7}
\end{bchapter}
\endbibitem

\bibitem[\protect\citeauthoryear{Levin and Gries}{1981}]{LevinGries1981}
\begin{barticle}
\bauthor{\bsnm{Levin}, \binits{G.}},
\bauthor{\bsnm{Gries}, \binits{D.}}:
\batitle{A proof technique for communicating sequential processes}.
\bjtitle{Acta Informatica}
\bvolume{15}(\bissue{3}),
\bfpage{281}--\blpage{302}
(\byear{1981})
\doiurl{10.1007/BF00289266}
\end{barticle}
\endbibitem

\bibitem[\protect\citeauthoryear{Liu et~al.}{2010}]{Liu2010}
\begin{bchapter}
\bauthor{\bsnm{Liu}, \binits{J.}},
\bauthor{\bsnm{Lv}, \binits{J.}},
\bauthor{\bsnm{Quan}, \binits{Z.}},
\bauthor{\bsnm{Zhan}, \binits{N.}},
\bauthor{\bsnm{Zhao}, \binits{H.}},
\bauthor{\bsnm{Zhou}, \binits{C.}},
\bauthor{\bsnm{Zou}, \binits{L.}}:
\bctitle{A calculus for hybrid {CSP}}.
In: \beditor{\bsnm{Ueda}, \binits{K.}} (ed.)
\bbtitle{Proc. 8th Asian Symp. Programming Languages and Systems {(APLAS)}}.
\bsertitle{LNCS},
vol. \bseriesno{6461},
pp. \bfpage{1}--\blpage{15}.
\bpublisher{Springer}
(\byear{2010}).
\doiurl{10.1007/978-3-642-17164-2\_1}
\end{bchapter}
\endbibitem

\bibitem[\protect\citeauthoryear{Lunel et~al.}{2019}]{Lunel2019}
\begin{bchapter}
\bauthor{\bsnm{Lunel}, \binits{S.}},
\bauthor{\bsnm{Mitsch}, \binits{S.}},
\bauthor{\bsnm{Boyer}, \binits{B.}},
\bauthor{\bsnm{Talpin}, \binits{J.}}:
\bctitle{Parallel composition and modular verification of computer controlled systems in differential dynamic logic}.
In: \beditor{\bsnm{Beek}, \binits{M.H.}},
\beditor{\bsnm{McIver}, \binits{A.}},
\beditor{\bsnm{Oliveira}, \binits{J.N.}} (eds.)
\bbtitle{Proc. 3rd World Congr. Formal Methods - The Next 30 Years {(FM)}}.
\bsertitle{LNCS},
vol. \bseriesno{11800},
pp. \bfpage{354}--\blpage{370}.
\bpublisher{Springer}
(\byear{2019}).
\doiurl{10.1007/978-3-030-30942-8\_22}
\end{bchapter}
\endbibitem

\bibitem[\protect\citeauthoryear{Loos and Platzer}{2016}]{DBLP:conf/lics/LoosP16}
\begin{bchapter}
\bauthor{\bsnm{Loos}, \binits{S.M.}},
\bauthor{\bsnm{Platzer}, \binits{A.}}:
\bctitle{Differential refinement logic}.
In: \beditor{\bsnm{Grohe}, \binits{M.}},
\beditor{\bsnm{Koskinen}, \binits{E.}},
\beditor{\bsnm{Shankar}, \binits{N.}} (eds.)
\bbtitle{LICS},
pp. \bfpage{505}--\blpage{514}.
\bpublisher{ACM},
(\byear{2016}).
\doiurl{10.1145/2933575.2934555}
\end{bchapter}
\endbibitem

\bibitem[\protect\citeauthoryear{Loos et~al.}{2011}]{DBLP:conf/fm/LoosPN11}
\begin{bchapter}
\bauthor{\bsnm{Loos}, \binits{S.M.}},
\bauthor{\bsnm{Platzer}, \binits{A.}},
\bauthor{\bsnm{Nistor}, \binits{L.}}:
\bctitle{Adaptive cruise control: Hybrid, distributed, and now formally verified}.
In: \beditor{\bsnm{Butler}, \binits{M.}},
\beditor{\bsnm{Schulte}, \binits{W.}} (eds.)
\bbtitle{FM}.
\bsertitle{LNCS},
vol. \bseriesno{6664},
pp. \bfpage{42}--\blpage{56}.
\bpublisher{Springer},
(\byear{2011}).
\doiurl{10.1007/978-3-642-21437-0_6}
\end{bchapter}
\endbibitem

\bibitem[\protect\citeauthoryear{Lynch et~al.}{2003}]{Lynch2003}
\begin{barticle}
\bauthor{\bsnm{Lynch}, \binits{N.A.}},
\bauthor{\bsnm{Segala}, \binits{R.}},
\bauthor{\bsnm{Vaandrager}, \binits{F.W.}}:
\batitle{Hybrid {I/O} automata}.
\bjtitle{Information and Computation}
\bvolume{185}(\bissue{1}),
\bfpage{105}--\blpage{157}
(\byear{2003})
\doiurl{10.1016/S0890-5401(03)00067-1}
\end{barticle}
\endbibitem

\bibitem[\protect\citeauthoryear{Misra and Chandy}{1981}]{Misra1981}
\begin{barticle}
\bauthor{\bsnm{Misra}, \binits{J.}},
\bauthor{\bsnm{Chandy}, \binits{K.M.}}:
\batitle{Proofs of networks of processes}.
\bjtitle{{IEEE} Transactions on Software Engineering}
\bvolume{7}(\bissue{4}),
\bfpage{417}--\blpage{426}
(\byear{1981})
\doiurl{10.1109/TSE.1981.230844}
\end{barticle}
\endbibitem

\bibitem[\protect\citeauthoryear{Minsky}{1961}]{Minsky1961}
\begin{barticle}
\bauthor{\bsnm{Minsky}, \binits{M.L.}}:
\batitle{Recursive unsolvability of post's problem of "tag" and other topics in theory of turing machines}.
\bjtitle{Annals of Mathematics}
\bvolume{74}(\bissue{3}),
\bfpage{437}--\blpage{455}
(\byear{1961})
\end{barticle}
\endbibitem

\bibitem[\protect\citeauthoryear{M{\"{u}}ller et~al.}{2018}]{DBLP:journals/sttt/MullerMRSP18}
\begin{barticle}
\bauthor{\bsnm{M{\"{u}}ller}, \binits{A.}},
\bauthor{\bsnm{Mitsch}, \binits{S.}},
\bauthor{\bsnm{Retschitzegger}, \binits{W.}},
\bauthor{\bsnm{Schwinger}, \binits{W.}},
\bauthor{\bsnm{Platzer}, \binits{A.}}:
\batitle{Tactical contract composition for hybrid system component verification}.
\bjtitle{STTT}
\bvolume{20}(\bissue{6}),
\bfpage{615}--\blpage{643}
(\byear{2018})
\doiurl{10.1007/s10009-018-0502-9} .
\bcomment{Special issue for selected papers from FASE'17}
\end{barticle}
\endbibitem

\bibitem[\protect\citeauthoryear{Manna and Pnueli}{1992}]{Manna1992}
\begin{bbook}
\bauthor{\bsnm{Manna}, \binits{Z.}},
\bauthor{\bsnm{Pnueli}, \binits{A.}}:
\bbtitle{The Temporal Logic of Reactive and Concurrent Systems - Specification}.
\bpublisher{Springer}
(\byear{1992}).
\doiurl{10.1007/978-1-4612-0931-7}
\end{bbook}
\endbibitem

\bibitem[\protect\citeauthoryear{Man et~al.}{2005}]{Man2005}
\begin{barticle}
\bauthor{\bsnm{Man}, \binits{K.L.}},
\bauthor{\bsnm{Reniers}, \binits{M.A.}},
\bauthor{\bsnm{Cuijpers}, \binits{P.J.L.}}:
\batitle{Case studies in the hybrid process algebra {HyPA}}.
\bjtitle{Int. J. Softw. Eng. Knowl. Eng.}
\bvolume{15}(\bissue{2}),
\bfpage{299}--\blpage{306}
(\byear{2005})
\doiurl{10.1142/S0218194005002385}
\end{barticle}
\endbibitem

\bibitem[\protect\citeauthoryear{Owicki and Gries}{1976}]{OwickiGries1976}
\begin{barticle}
\bauthor{\bsnm{Owicki}, \binits{S.S.}},
\bauthor{\bsnm{Gries}, \binits{D.}}:
\batitle{An axiomatic proof technique for parallel programs {I}}.
\bjtitle{Acta Informatica}
\bvolume{6},
\bfpage{319}--\blpage{340}
(\byear{1976})
\doiurl{10.1007/BF00268134}
\end{barticle}
\endbibitem

\bibitem[\protect\citeauthoryear{Platzer and Clarke}{2009}]{DBLP:conf/fm/PlatzerC09}
\begin{bchapter}
\bauthor{\bsnm{Platzer}, \binits{A.}},
\bauthor{\bsnm{Clarke}, \binits{E.M.}}:
\bctitle{Formal verification of curved flight collision avoidance maneuvers: A case study}.
In: \beditor{\bsnm{Cavalcanti}, \binits{A.}},
\beditor{\bsnm{Dams}, \binits{D.}} (eds.)
\bbtitle{FM}.
\bsertitle{LNCS},
vol. \bseriesno{5850},
pp. \bfpage{547}--\blpage{562}.
\bpublisher{Springer},
(\byear{2009}).
\doiurl{10.1007/978-3-642-05089-3_35}
\end{bchapter}
\endbibitem

\bibitem[\protect\citeauthoryear{Peleg}{1987a}]{Peleg1987a}
\begin{barticle}
\bauthor{\bsnm{Peleg}, \binits{D.}}:
\batitle{Communication in concurrent dynamic logic}.
\bjtitle{J. Comput. Syst. Sci.}
\bvolume{35}(\bissue{1}),
\bfpage{23}--\blpage{58}
(\byear{1987})
\doiurl{10.1016/0022-0000(87)90035-3}
\end{barticle}
\endbibitem

\bibitem[\protect\citeauthoryear{Peleg}{1987b}]{Peleg1987}
\begin{barticle}
\bauthor{\bsnm{Peleg}, \binits{D.}}:
\batitle{Concurrent dynamic logic}.
\bjtitle{J. {ACM}}
\bvolume{34}(\bissue{2}),
\bfpage{450}--\blpage{479}
(\byear{1987})
\doiurl{10.1145/23005.23008}
\end{barticle}
\endbibitem

\bibitem[\protect\citeauthoryear{Pandya and Joseph}{1991}]{Pandya1991}
\begin{barticle}
\bauthor{\bsnm{Pandya}, \binits{P.K.}},
\bauthor{\bsnm{Joseph}, \binits{M.}}:
\batitle{{P -- A} logic - {A} compositional proof system for distributed programs}.
\bjtitle{Distributed Computing}
\bvolume{5}(\bissue{1}),
\bfpage{37}--\blpage{54}
(\byear{1991})
\doiurl{10.1007/BF02311231}
\end{barticle}
\endbibitem

\bibitem[\protect\citeauthoryear{Platzer}{2008}]{DBLP:journals/jar/Platzer08}
\begin{barticle}
\bauthor{\bsnm{Platzer}, \binits{A.}}:
\batitle{Differential dynamic logic for hybrid systems.}
\bjtitle{J. Autom. Reas.}
\bvolume{41}(\bissue{2}),
\bfpage{143}--\blpage{189}
(\byear{2008})
\doiurl{10.1007/s10817-008-9103-8}
\end{barticle}
\endbibitem

\bibitem[\protect\citeauthoryear{Platzer}{2010a}]{Platzer10}
\begin{bbook}
\bauthor{\bsnm{Platzer}, \binits{A.}}:
\bbtitle{Logical Analysis of Hybrid Systems: Proving Theorems for Complex Dynamics}.
\bpublisher{Springer},
(\byear{2010}).
\doiurl{10.1007/978-3-642-14509-4}
\end{bbook}
\endbibitem

\bibitem[\protect\citeauthoryear{Platzer}{2010b}]{DBLP:conf/csl/Platzer10}
\begin{bchapter}
\bauthor{\bsnm{Platzer}, \binits{A.}}:
\bctitle{Quantified differential dynamic logic for distributed hybrid systems}.
In: \beditor{\bsnm{Dawar}, \binits{A.}},
\beditor{\bsnm{Veith}, \binits{H.}} (eds.)
\bbtitle{CSL}.
\bsertitle{LNCS},
vol. \bseriesno{6247},
pp. \bfpage{469}--\blpage{483}.
\bpublisher{Springer}
(\byear{2010}).
\doiurl{10.1007/978-3-642-15205-4_36}
\end{bchapter}
\endbibitem

\bibitem[\protect\citeauthoryear{Platzer}{2012a}]{DBLP:journals/lmcs/Platzer12b}
\begin{barticle}
\bauthor{\bsnm{Platzer}, \binits{A.}}:
\batitle{A complete axiomatization of quantified differential dynamic logic for distributed hybrid systems}.
\bjtitle{Log. Meth. Comput. Sci.}
\bvolume{8}(\bissue{4:17}),
\bfpage{1}--\blpage{44}
(\byear{2012})
\doiurl{10.2168/LMCS-8(4:17)2012} .
\bcomment{Special issue for selected papers from CSL'10}
\end{barticle}
\endbibitem

\bibitem[\protect\citeauthoryear{Platzer}{2012b}]{DBLP:conf/lics/Platzer12b}
\begin{bchapter}
\bauthor{\bsnm{Platzer}, \binits{A.}}:
\bctitle{The complete proof theory of hybrid systems}.
In: \bbtitle{LICS},
pp. \bfpage{541}--\blpage{550}.
\bpublisher{IEEE},
(\byear{2012}).
\doiurl{10.1109/LICS.2012.64}
\end{bchapter}
\endbibitem

\bibitem[\protect\citeauthoryear{Platzer}{2015}]{DBLP:journals/tocl/Platzer15}
\begin{barticle}
\bauthor{\bsnm{Platzer}, \binits{A.}}:
\batitle{Differential game logic}.
\bjtitle{{ACM} Trans. Comput. Log.}
\bvolume{17}(\bissue{1}),
\bfpage{1}--\blpage{1151}
(\byear{2015})
\doiurl{10.1145/2817824}
\end{barticle}
\endbibitem

\bibitem[\protect\citeauthoryear{Platzer}{2016}]{DBLP:conf/cade/Platzer16}
\begin{bchapter}
\bauthor{\bsnm{Platzer}, \binits{A.}}:
\bctitle{Logic {\&} proofs for cyber-physical systems}.
In: \beditor{\bsnm{Olivetti}, \binits{N.}},
\beditor{\bsnm{Tiwari}, \binits{A.}} (eds.)
\bbtitle{IJCAR}.
\bsertitle{LNCS},
vol. \bseriesno{9706},
pp. \bfpage{15}--\blpage{21}.
\bpublisher{Springer},
(\byear{2016}).
\doiurl{10.1007/978-3-319-40229-1_3}
\end{bchapter}
\endbibitem

\bibitem[\protect\citeauthoryear{Platzer}{2017}]{DBLP:journals/jar/Platzer17}
\begin{barticle}
\bauthor{\bsnm{Platzer}, \binits{A.}}:
\batitle{A complete uniform substitution calculus for differential dynamic logic}.
\bjtitle{J. Autom. Reas.}
\bvolume{59}(\bissue{2}),
\bfpage{219}--\blpage{265}
(\byear{2017})
\doiurl{10.1007/s10817-016-9385-1}
\end{barticle}
\endbibitem

\bibitem[\protect\citeauthoryear{Platzer}{2018}]{Platzer18}
\begin{bbook}
\bauthor{\bsnm{Platzer}, \binits{A.}}:
\bbtitle{Logical Foundations of Cyber-Physical Systems}.
\bpublisher{Springer},
(\byear{2018}).
\doiurl{10.1007/978-3-319-63588-0}
\end{bbook}
\endbibitem

\bibitem[\protect\citeauthoryear{Platzer}{2019}]{DBLP:conf/cade/Platzer19}
\begin{bchapter}
\bauthor{\bsnm{Platzer}, \binits{A.}}:
\bctitle{Uniform substitution at one fell swoop}.
In: \beditor{\bsnm{Fontaine}, \binits{P.}} (ed.)
\bbtitle{CADE}.
\bsertitle{LNCS},
vol. \bseriesno{11716},
pp. \bfpage{425}--\blpage{441}.
\bpublisher{Springer}
(\byear{2019}).
\doiurl{10.1007/978-3-030-29436-6_25}
\end{bchapter}
\endbibitem

\bibitem[\protect\citeauthoryear{Platzer and Quesel}{2009}]{DBLP:conf/icfem/PlatzerQ09}
\begin{bchapter}
\bauthor{\bsnm{Platzer}, \binits{A.}},
\bauthor{\bsnm{Quesel}, \binits{J.-D.}}:
\bctitle{{European Train Control System}: A case study in formal verification}.
In: \beditor{\bsnm{Breitman}, \binits{K.}},
\beditor{\bsnm{Cavalcanti}, \binits{A.}} (eds.)
\bbtitle{ICFEM}.
\bsertitle{LNCS},
vol. \bseriesno{5885},
pp. \bfpage{246}--\blpage{265}.
\bpublisher{Springer},
(\byear{2009}).
\doiurl{10.1007/978-3-642-10373-5_13}
\end{bchapter}
\endbibitem

\bibitem[\protect\citeauthoryear{Pratt}{1976}]{Pratt1976}
\begin{bchapter}
\bauthor{\bsnm{Pratt}, \binits{V.R.}}:
\bctitle{Semantical considerations on floyd-hoare logic}.
In: \bbtitle{17th Annual Symposium on Foundations of Computer Science, Houston, Texas, USA, 25-27 October 1976},
pp. \bfpage{109}--\blpage{121}.
\bpublisher{{IEEE} Computer Society}
(\byear{1976}).
\doiurl{10.1109/SFCS.1976.27}
\end{bchapter}
\endbibitem

\bibitem[\protect\citeauthoryear{Platzer and Tan}{2020}]{DBLP:journals/jacm/PlatzerT20}
\begin{barticle}
\bauthor{\bsnm{Platzer}, \binits{A.}},
\bauthor{\bsnm{Tan}, \binits{Y.K.}}:
\batitle{Differential equation invariance axiomatization}.
\bjtitle{J. ACM}
\bvolume{67}(\bissue{1}),
\bfpage{6}--\blpage{1666}
(\byear{2020})
\doiurl{10.1145/3380825}
\end{barticle}
\endbibitem

\bibitem[\protect\citeauthoryear{Song et~al.}{2005}]{Song2005}
\begin{bchapter}
\bauthor{\bsnm{Song}, \binits{H.}},
\bauthor{\bsnm{Compton}, \binits{K.J.}},
\bauthor{\bsnm{Rounds}, \binits{W.C.}}:
\bctitle{{SPHIN:} {A} model checker for reconfigurable hybrid systems based on {SPIN}}.
In: \beditor{\bsnm{Lazic}, \binits{R.}},
\beditor{\bsnm{Nagarajan}, \binits{R.}} (eds.)
\bbtitle{Proc. 5th Intl. Workshop Automated Verification of Critical Systems {(AVoCS)}}.
\bsertitle{ENTCS},
vol. \bseriesno{145},
pp. \bfpage{167}--\blpage{183}.
\bpublisher{Elsevier}
(\byear{2005}).
\doiurl{10.1016/j.entcs.2005.10.011}
\end{bchapter}
\endbibitem

\bibitem[\protect\citeauthoryear{Segerberg}{1982}]{Segerberg1982}
\begin{barticle}
\bauthor{\bsnm{Segerberg}, \binits{K.}}:
\batitle{A completeness theorem in the modal logic of programs}.
\bjtitle{Banach Center Publications}
\bvolume{9},
\bfpage{31}--\blpage{46}
(\byear{1982})
\end{barticle}
\endbibitem

\bibitem[\protect\citeauthoryear{Tarski}{1951}]{Tarski1951}
\begin{bbook}
\bauthor{\bsnm{Tarski}, \binits{A.}}:
\bbtitle{A Decision Method for Elementary Algebra and Geometry},
\bedition{2nd} edn.
\bpublisher{University of California Press},
(\byear{1951}).
\doiurl{10.1525/9780520348097}
\end{bbook}
\endbibitem

\bibitem[\protect\citeauthoryear{Wang et~al.}{2012}]{Wang2012}
\begin{bchapter}
\bauthor{\bsnm{Wang}, \binits{S.}},
\bauthor{\bsnm{Zhan}, \binits{N.}},
\bauthor{\bsnm{Guelev}, \binits{D.P.}}:
\bctitle{An assume/guarantee based compositional calculus for hybrid {CSP}}.
In: \beditor{\bsnm{Agrawal}, \binits{M.}},
\beditor{\bsnm{Cooper}, \binits{S.B.}},
\beditor{\bsnm{Li}, \binits{A.}} (eds.)
\bbtitle{Proc. 9th Conf. Theory and Applications of Models of Computation {(TAMC)}}.
\bsertitle{LNCS},
vol. \bseriesno{7287},
pp. \bfpage{72}--\blpage{83}.
\bpublisher{Springer}
(\byear{2012}).
\doiurl{10.1007/978-3-642-29952-0_13}
\end{bchapter}
\endbibitem

\bibitem[\protect\citeauthoryear{Xu et~al.}{1994}]{Xu1994}
\begin{bchapter}
\bauthor{\bsnm{Xu}, \binits{Q.}},
\bauthor{\bsnm{Cau}, \binits{A.}},
\bauthor{\bsnm{Collette}, \binits{P.}}:
\bctitle{On unifying assumption-commitment style proof rules for concurrency}.
In: \beditor{\bsnm{Jonsson}, \binits{B.}},
\beditor{\bsnm{Parrow}, \binits{J.}} (eds.)
\bbtitle{Proc. 5th Intl. Conf. Concurrency Theory ({CONCUR})}.
\bsertitle{LNCS},
vol. \bseriesno{836},
pp. \bfpage{267}--\blpage{282}.
\bpublisher{Springer}
(\byear{1994}).
\doiurl{10.1007/978-3-540-48654-1\_22}
\end{bchapter}
\endbibitem

\bibitem[\protect\citeauthoryear{Xu et~al.}{1997}]{Xu1997}
\begin{barticle}
\bauthor{\bsnm{Xu}, \binits{Q.}},
\bauthor{\bsnm{Roever}, \binits{W.P.}},
\bauthor{\bsnm{He}, \binits{J.}}:
\batitle{The rely-guarantee method for verifying shared variable concurrent programs}.
\bjtitle{Formal Aspects of Comput.}
\bvolume{9}(\bissue{2}),
\bfpage{149}--\blpage{174}
(\byear{1997})
\doiurl{10.1007/BF01211617}
\end{barticle}
\endbibitem

\bibitem[\protect\citeauthoryear{Zwiers et~al.}{1983}]{AcHoare_Zwiers}
\begin{bchapter}
\bauthor{\bsnm{Zwiers}, \binits{J.}},
\bauthor{\bsnm{Bruin}, \binits{A.}},
\bauthor{\bsnm{Roever}, \binits{W.P.}}:
\bctitle{A proof system for partial correctness of dynamic networks of processes (extended abstract)}.
In: \beditor{\bsnm{Clarke}, \binits{E.M.}},
\beditor{\bsnm{Kozen}, \binits{D.}} (eds.)
\bbtitle{Proc. Carnegie Mellon Workshop Logics of Programs 1983}.
\bsertitle{LNCS},
vol. \bseriesno{164},
pp. \bfpage{513}--\blpage{527}.
\bpublisher{Springer}
(\byear{1983}).
\doiurl{10.1007/3-540-12896-4_384}
\end{bchapter}
\endbibitem

\bibitem[\protect\citeauthoryear{Zwiers et~al.}{1985}]{AcSemantics_Zwiers}
\begin{bchapter}
\bauthor{\bsnm{Zwiers}, \binits{J.}},
\bauthor{\bsnm{Roever}, \binits{W.P.}},
\bauthor{\bsnm{Emde~Boas}, \binits{P.}}:
\bctitle{Compositionality and concurrent networks: {S}oundness and completeness of a proofsystem}.
In: \beditor{\bsnm{Brauer}, \binits{W.}} (ed.)
\bbtitle{Proc. 12th Intl. Coll. Automata, Languages and Programming {(ICALP)}}.
\bsertitle{LNCS},
vol. \bseriesno{194},
pp. \bfpage{509}--\blpage{519}.
\bpublisher{Springer}
(\byear{1985}).
\doiurl{10.1007/BFb0015776}
\end{bchapter}
\endbibitem

\bibitem[\protect\citeauthoryear{Zhou et~al.}{1996}]{Zhou1996}
\begin{barticle}
\bauthor{\bsnm{Zhou}, \binits{P.}},
\bauthor{\bsnm{Hooman}, \binits{J.}},
\bauthor{\bsnm{Kuiper}, \binits{R.}}:
\batitle{Compositional verification of real-time systems with explicit clock temporal logic}.
\bjtitle{Formal Aspects Comput.}
\bvolume{8}(\bissue{3}),
\bfpage{294}--\blpage{323}
(\byear{1996})
\doiurl{10.1007/BF01214917}
\end{barticle}
\endbibitem

\bibitem[\protect\citeauthoryear{Zwiers et~al.}{1996}]{Zwiers1996}
\begin{bchapter}
\bauthor{\bsnm{Zwiers}, \binits{J.}},
\bauthor{\bsnm{Hannemann}, \binits{U.}},
\bauthor{\bsnm{Lakhnech}, \binits{Y.}},
\bauthor{\bsnm{Roever}, \binits{W.P.}},
\bauthor{\bsnm{Stomp}, \binits{F.A.}}:
\bctitle{Modular completeness: Integrating the reuse of specified software in top-down program development}.
In: \beditor{\bsnm{Gaudel}, \binits{M.}},
\beditor{\bsnm{Woodcock}, \binits{J.}} (eds.)
\bbtitle{Proc. 3rd Intl. Symp. Formal Methods Europe {(FME)}: Industrial Benefit and Advances in Formal Methods}.
\bsertitle{LNCS},
vol. \bseriesno{1051},
pp. \bfpage{595}--\blpage{608}.
\bpublisher{Springer}
(\byear{1996}).
\doiurl{10.1007/3-540-60973-3\_109}
\end{bchapter}
\endbibitem

\end{thebibliography}

\iflongversion

\newpage
\section*{Additional Appendices}

\section{Substitution and Recorder Renaming}
\label{app:substitution}

This appendix reports details for \rref{sec:substitution}.
Definition \ref{def:rec_substitution} inductively defines substitution for trace variables.
Further, a detailed proof for \rref{lem:rec_substitution} is given.

\begin{definition}
	[Substitution for trace variables]
	\label{def:rec_substitution}
	For a formula $\phi$,
	a trace variable $\tvar$,
	and a trace term $\te$,
	the substitution $\phi \subs{\tvar}{\te}$ of $\te$ for $\tvar$ in $\phi$ is inductively defined
	as follows:
	\begin{align*}
		(\expr_1 \sim \expr_2) \subs{\tvar}{\te} &\equiv (\expr_1) \subs{\tvar}{\te} \sim (\expr_2) \subs{\tvar}{\te} \\
		(\neg \varphi) \subs{\tvar}{\te} &\equiv \neg \varphi \subs{\tvar}{\te} \\
		(\varphi \wedge \psi) \subs{\tvar}{\te} &\equiv \varphi \subs{\tvar}{\te} \wedge \psi \subs{\tvar}{\te} \\
		(\fa{\avar} \varphi) \subs{\tvar}{\te} &\equiv 
		\begin{cases}
			\fa{\avar} \varphi 
				&\text{if } \avar \equiv \tvar \\
			\fa{\avar_0} (\varphi \subs{\avar}{\avar_0}) \subs{\tvar}{\te} 
				&\text{if } \avar \not\equiv \tvar \text{ and } \avar \in \SFV(\te) \text{, and } \avar_0 \text{ is fresh} \\
			\fa{\avar} \varphi \subs{\tvar}{\te} 
				&\text{if } \avar \not\equiv \tvar \text{ and } \avar \not\in \SFV(\te)
		\end{cases} \\
		([ \recchp{\tilde{\tvar}}{\alpha} ] \psi) \subs{\tvar}{\te} &\equiv
		\begin{cases}
			[ \recchp{\tilde{\tvar}}{\alpha} ] \psi \subs{\tvar}{\te}
			& \text{if } \tvar \not\equiv \tilde{\tvar} \text{ and } \SFV(\te)\cap\SBV(\recchp{\tilde{\tvar}}{\alpha})=\emptyset \\
			\fa{\tvar_0{=}\te} [ \recchp{{\tilde{\tvar}}\subs{\tvar}{{\tvar_0}}}{\alpha} ] \psi \subs{\tvar}{{\tvar_0}} 
			& \text{else, where } \tvar_0 \text{ is fresh}
		\end{cases} \\
		(\langle \recchp{\tilde{\tvar}}{\alpha} \rangle \psi) \subs{\tvar}{\te} &\equiv
		\begin{cases}
			\langle \recchp{\tilde{\tvar}}{\alpha} \rangle \psi \subs{\tvar}{\te}
			& \text{if } \tvar \not\equiv \tilde{\tvar} \text{ and } \SFV(\te)\cap\SBV(\recchp{\tilde{\tvar}}{\alpha})=\emptyset \\
			\ex{\tvar_0{=}\te} \langle \recchp{{\tilde{\tvar}}\subs{\tvar}{{\tvar_0}}}{\alpha} \rangle \psi \subs{\tvar}{{\tvar_0}} 
			& \text{else, where } \tvar_0 \text{ is fresh}
		\end{cases} \\
		([ \recchp{\tilde{\tvar}}{\alpha} ] \ac \psi) \subs{\tvar}{\te} &\equiv
		\begin{cases}
			[ \recchp{\tilde{\tvar}}{\alpha} ] \acpair{\A \subs{\tvar}{\te}, \Commit \subs{\tvar}{\te}} \psi \subs{\tvar}{\te}
			& \text{if } \tvar \not\equiv \tilde{\tvar} \text{ and } \SFV(\te)\cap\SBV(\recchp{\tilde{\tvar}}{\alpha}) = \emptyset \\
			\fa{\tvar_0{=}\te} [ \recchp{{\tilde{\tvar}}\subs{\tvar}{{\tvar_0}}}{\alpha} ] \acpair{\A \subs{\tvar}{{\tvar_0}}, \Commit \subs{\tvar}{{\tvar_0}}} \psi \subs{\tvar}{{\tvar_0}}
			& \text{else, where } \tvar_0 \text{ is fresh}
		\end{cases} \\
		(\langle \recchp{\tilde{\tvar}}{\alpha} \rangle \ac \psi) \subs{\tvar}{\te} &\equiv
		\begin{cases}
			\langle \recchp{\tilde{\tvar}}{\alpha} \rangle \acpair{\A \subs{\tvar}{\te}, \Commit \subs{\tvar}{\te}} \psi \subs{\tvar}{\te}
			& \text{if } \tvar \not\equiv \tilde{\tvar} \text{ and } \SFV(\te)\cap\SBV(\recchp{\tilde{\tvar}}{\alpha}) = \emptyset \\
			\ex{\tvar_0{=}\te} \langle \recchp{{\tilde{\tvar}}\subs{\tvar}{{\tvar_0}}}{\alpha} \rangle \acpair{\A \subs{\tvar}{{\tvar_0}}, \Commit \subs{\tvar}{{\tvar_0}}} \psi \subs{\tvar}{{\tvar_0}}
			& \text{else, where } \tvar_0 \text{ is fresh}
		\end{cases}
	\end{align*}
\end{definition}	

\begin{proof}
	[Proof of \rref{lem:rec_substitution}.]
	The proof is by induction on the structure of $\phi$.
	The only nonstandard cases are $\phi \equiv \dbleft \recchp{\tilde{\tvar}}{\alpha} \dbright \psi$ and $\phi \equiv \dbleft \recchp{\tilde{\tvar}}{\alpha} \dbright \ac \psi$.
	In the following, the case $\phi \equiv [ \recchp{\tilde{\tvar}}{\alpha} ] \psi$ is proven.
	The remaining cases $\langle \recchp{\tilde{\tvar}}{\alpha} \rangle \psi$ and $\dbleft \recchp{\tilde{\tvar}}{\alpha} \dbright \ac \psi$ are analogous.
	In this proof, let state-trace concatenation $\pstate{v} \cdot (\tvar, \trace)$ on $\tvar$ be defined by $\pstate{v} \subs{\tvar}{\pstate{v}(\tvar)\cdot\trace}$
	for all $\pstate{v}, \tvar, \trace$.
	\begin{enumerate}
		\item 
		\label{itm:rec_subs_one}
		If $\tvar\not\equiv\tilde{\tvar}$ and $\SFV(\te)\cap\SBV(\recchp{\tilde{\tvar}}{\alpha})=\emptyset$,
		then $\phi \subs{\tvar}{\te} \equiv [ \recchp{\tilde{\tvar}}{\alpha}  ] \psi \subs{\tvar}{\te}$.
		Then let $\pstate{v} \vDash \phi \subs{\tvar}{\te}$.
		To prove $\pstate{v} \subs{\tvar}{\sem{\te}{\pstate{v}}} \vDash [ \recchp{\tilde{\tvar}}{\alpha}  ] \psi$,
		let $(\pstate{v} \subs{\tvar}{\sem{\te}{\pstate{v}}}, \trace, \pstate[alt]{w}) \in \sem{\alpha}{}$ with $\pstate[alt]{w} \neq \bot$.
		By coincidence (\rref{cor:history_coincidence}),
		there is a run $\run \in \sem{\alpha}{}$ with $\pstate[alt]{w} = \pstate{w} \subs{\tvar}{\sem{\te}{\pstate{v}}}$.
		Hence, $\pstate{w} \cdot (\tilde{\tvar}, \trace) \vDash \psi \subs{\tvar}{\te}$ 
		by \acPost
		since $\pstate{v} \vDash \phi \subs{\tvar}{\te}$.
		By IH, $(\pstate{w} \cdot (\tilde{\tvar}, \trace)) \subs{\tvar}{\sem{\te}{(\pstate{w} \cdot (\tilde{\tvar}, \trace))}} \vDash \psi$.
		Further, $\sem{\te}{(\pstate{w} \cdot (\tilde{\tvar}, \trace))} = \sem{\te}{\pstate{v}}$
		by coincidence (\rref{lem:expr_coincidence}),
		because $\SFV(\te)\cap\SBV(\recchp{\tilde{\tvar}}{\alpha})=\emptyset$ and $\pstate{v} = \pstate{w} \cdot (\tilde{\tvar}, \trace)$ on $\SBV(\recchp{\tilde{\tvar}}{\alpha})^\complement$
		by the bound effect property (\rref{lem:bound_effect}).
		Hence, $(\pstate{w} \cdot (\tilde{\tvar}, \trace)) \subs{\tvar}{\sem{\te}{\pstate{v}}} \vDash \psi$,
		so that $(\pstate{w} \subs{\tvar}{\sem{\te}{\pstate{v}}}) \cdot (\tilde{\tvar}, \trace) \vDash \psi$
		as $\tvar\not\equiv\tilde{\tvar}$.
		Since $\pstate[alt]{w} = \pstate{w} \subs{\tvar}{\sem{\te}{\pstate{v}}}$,
		obtain $\pstate[alt]{w} \cdot (\tilde{\tvar}, \trace) \vDash \psi$,
		which implies $\pstate{v} \subs{\tvar}{\sem{\te}{\pstate{v}}} \vDash [ \alpha ] \psi$ by \acPost.
		This concludes this case, because all steps can be reversed.

		\item If $\tvar \not\equiv \tilde{\tvar}$
		and $\SFV(\te) \cap \SBV(\recchp{\tilde{\tvar}}{\alpha}) \neq \emptyset$,
		then $\phi \subs{\tvar}{\te} \equiv \fa{\tvar_0{=}\te} [ \recchp{{\tilde{\tvar}}}{\alpha} ] \psi \subs{\tvar}{{\tvar_0}}$.
		Then $\pstate{v} \vDash \phi\subs{\tvar}{\te}$
		iff, by \rref{lem:expr_coincidence} using that $\tvar_0$ is fresh,
		$\pstate{v} \subs{\tvar_0}{\sem{\te}{\pstate{v}}} \vDash [ \recchp{{\tilde{\tvar}}}{\alpha} ] \psi \subs{\tvar}{{\tvar_0}}$
		iff, by the definition of substitution,
		$\pstate{v} \subs{\tvar_0}{\sem{\te}{\pstate{v}}} \vDash ([ \recchp{{\tilde{\tvar}}}{\alpha} ] \psi) \subs{\tvar}{{\tvar_0}}$
		iff, by \rref{itm:rec_subs_one},
		$(\pstate{v} \subs{\tvar_0}{\sem{\te}{\pstate{v}}}) \subs{\tvar}{\sem{\te}{\pstate{v}}} \vDash [ \recchp{{\tilde{\tvar}}}{\alpha} ] \psi$
		iff, by \rref{lem:expr_coincidence} using that~$\tvar_0$ is fresh,
		$\pstate{v} \subs{\tvar}{\sem{\te}{\pstate{v}}} \vDash [ \recchp{{\tilde{\tvar}}}{\alpha} ] \psi$.

		\item If $\tvar \equiv \tilde{\tvar}$,
		then $\phi \subs{\tvar}{\te} \equiv \fa{\tvar_0{=}\te} [ \recchp{\tvar_0}{\alpha} ] \psi \subs{\tvar}{{\tvar_0}}$.
		Then let $\pstate{v} \vDash \phi\subs{\tvar}{\te}$.
		Hence, $\pstate{v}\subs{\tvar_0}{\sem{\te}{\pstate{v}}} \vDash [ \recchp{\tvar_0}{\alpha} ] \psi \subs{\tvar}{\tvar_0}$.
		To prove $\pstate{v} \subs{\tvar}{\sem{\te}{\pstate{v}}} \vDash [\recchp{\tilde{\tvar}}{\alpha}] \psi$,
		let $(\pstate{v} \subs{\tvar}{\sem{\te}{\pstate{v}}}, \trace, \pstate{w}) \in \sem{\alpha}{}$.
		By \rref{cor:history_coincidence},
		$(\pstate{v} \subs{\tvar_0}{\sem{\te}{\pstate{v}}}, \trace, \pstate[alt]{w}) \in \sem{\alpha}{}$
		with $\pstate[alt]{w} = (\pstate{w} \subs{\tvar}{\pstate{v}(\tvar)}) \subs{\tvar_0}{\sem{\te}{\pstate{v}}}$.
		Hence, $\pstate[alt]{w} \cdot (\tvar_0, \trace) \vDash \psi \subs{\tvar}{\tvar_0}$
		by \acPost
		since $\pstate{v}\subs{\tvar_0}{\sem{\te}{\pstate{v}}} \vDash [ \recchp{\tvar_0}{\alpha} ] \psi \subs{\tvar}{\tvar_0}$.
		Then $(\pstate[alt]{w} \cdot (\tvar_0, \trace)) \subs{\tvar}{\pstate[alt]{w}(\tvar_0)\cdot\trace} \vDash \psi$
		by IH,
		which implies $\pstate{w} \subs{\tvar}{\sem{\te}{\pstate{v}}\cdot\trace} \vDash \psi$
		by \rref{lem:expr_coincidence} since $\tvar_0\not\in\SFV(\psi)$ as~$\tvar_0$ is fresh.
		Since $\pstate{w}(\tvar) = \sem{\te}{\pstate{v}}$ by \rref{lem:bound_effect},
		obtain $\pstate{w} \cdot (\tilde{\tvar},\trace) \vDash \psi$,
		using that $\tilde{\tvar} \equiv \tvar$.
		Finally, $\pstate{v} \subs{\tvar}{\sem{\te}{\pstate{v}}} \vDash [\recchp{\tilde{\tvar}}{\alpha}] \psi$
		by \acPost.
		This concludes this case, as all steps can be reversed.
		\qedhere
	\end{enumerate}	
\end{proof}

\section{Induction Order}
\label{app:order}

\rref{thm:com-fod-completeness}
is proven by well-founded induction along an order $\indOrder$ on \dLCHP formulas (\rref{def:ind_order}),
which lexicographically combines orders measuring different aspects of structural complexity.
Definition~\ref{def:rank} formally defines these orders from rank functions,
which justifies their well-foundedness and makes the 
complexity measures
explicit.

\begin{definition}
	[Rank of a formula]
	\label{def:rank}
	For every formula $\phi \not\in\comFOD$, \rref{fig:rank} defines the rank functions $\progrank{\phi}$ and $\fmlrank{\phi}$ by recursion on the structure of $\phi$.
	For $\phi \in \comFOD$, define $\progrank{\phi} = \fmlrank{\phi} = 0$.
	\begin{enumerate}
		
		\item The \emph{rank by program complexity} $\progrank{\phi}$ measures the overall structural complexity of programs in $\phi$.
		The rank induces a well-founded order~$\progorder$ on formulas by $\varphi \progorder \psi$ if $\progrank{\varphi} < \progrank{\psi}$.
		
		\item The \emph{rank by logical complexity} $\fmlrank{\phi}$ measures the structural complexity of the formula $\phi$ itself.
		The rank induces a well-founded order~$\fmlorder$ on formulas by $\varphi \fmlorder \psi$ if $\fmlrank{\varphi} < \fmlrank{\psi}$.
	\end{enumerate}
\end{definition}

The decisive characteristic of the rank by program complexity $\progrank{\phi}$ (\rref{def:rank}) is that compound programs have a higher rank than the sum of their pieces.
Hence, a formula becomes smaller in the order if a program gets removed, \eg $\fa{x} (x^2 \ge 0) \progorder [ x \ceq \rp ] x = y$,
or if a program gets decomposed, \eg $[ \alpha ] [ \beta ] \psi \progorder [ \alpha \seq \beta ] \psi$.

Parallel composition even receives a rank that is more than twice the rank of its subprograms by \rref{def:rank}.
A formula that contains two copies of both subprograms of a parallel composition is thus still simpler than a formula that contains the parallel composition itself,
if no other program got worse,
\eg $[ \alpha ] \langle \alpha \rangle \psi \wedge [ \beta ] \langle \beta \rangle \psi \progorder [ \alpha\parOp\beta ] \psi$.

The rank of a formula by logical complexity (\rref{def:rank}) ranks 
every formula
higher than the sum of its subformulas.
This induces the standard structural complexity order on formulas,
\iest subformulas are smaller in the order, \eg $\varphi \fmlorder \varphi \wedge \psi$,
and a formula becomes smaller if some subformula does,
\eg $[ \alpha ] \lambda \fmlorder [ \alpha ] \psi$ if $\lambda \fmlorder \psi$.

\begin{definition}
	[Induction order]
	\label{def:ind_order}
	The partial order $\indOrder$ on \dLCHP formulas is the lexicographic combination of the orders $\progorder$ and $\fmlorder$ (\rref{def:rank}),
	\iest for \dLCHP formulas $\varphi,\psi$
	define $\varphi \indOrder \psi$
	if $\varphi \progorder \psi$ \emph{or} $\varphi =_\alpha \psi$ and $\varphi \fmlorder \psi$,
	where $\varphi =_\alpha \psi$ if neither $\varphi \progorder \psi$ nor $\psi \progorder \varphi$.
	The order $\indOrder$ is well-founded as lexicographic combination of well-founded orders. 
\end{definition}

\begin{figure}
	\begin{small}
		$\begin{aligned}
			& \progrank{\alpha} = 1 \qquad\sidecondition{$\alpha$ atomic} \\
			& \progrank{\alpha \seq \beta} = 1 + \progrank{\alpha} + \progrank{\beta} \\
			& \progrank{\alpha \cup \beta} = 1 + \progrank{\alpha} + \progrank{\beta} \\
			& \progrank{\alpha \parOp \beta} = 1 + 2 \cdot (\progrank{\alpha} + \progrank{\beta}) \\
			& \progrank{\repetition{\alpha}} = 1 + \progrank{\alpha}
		\end{aligned}$\hspace{.4cm}%
		$\begin{aligned}
			& \progrank{\expr_1 \sim \expr_2} = 0 \\
			& \progrank{\neg \varphi} = \progrank{\varphi} \\
			& \progrank{\varphi \wedge \psi} = \progrank{\varphi} + \progrank{\psi} \\
			& \progrank{\fa{\avar} \varphi} = \progrank{\varphi} \\
			& \progrank{\dbleft \alpha \dbright \psi} = \progrank{\alpha} + \progrank{\psi}
		\end{aligned}$
		\vspace{.25cm}

		$\progrank{\dbleft \alpha \dbright \ac \psi} = \progrank{\dbleft \alpha \dbright \psi} + \progrank{\A} + \progrank{\Commit}$
		\vspace{.5cm}

		$\begin{aligned}
			& \fmlrank{\expr_1 \sim \expr_2} = 1 \\
			& \fmlrank{\neg \varphi} = 1 + \fmlrank{\varphi} \\
			& \fmlrank{\dbleft \alpha \dbright \psi} = 1 + \fmlrank{\psi} \\
		\end{aligned}$\hspace{.4cm}%
		$\begin{aligned}
			& \fmlrank{\fa{\avar} \varphi} = 1 + \fmlrank{\varphi} \\
			& \fmlrank{\varphi \wedge \psi} = 1 + \fmlrank{\varphi} + \fmlrank{\psi} \\
			& \fmlrank{\dbleft \alpha \dbright \ac \psi} = \fmlrank{\dbleft \alpha \dbright \psi} + \fmlrank{\A} + \fmlrank{\Commit}
		\end{aligned}$
		\vspace{.25cm}
	\end{small}

	\caption{$\progrank{\phi}$ and $\fmlrank{\phi}$ for a formula $\phi \not\in \comFOD$ (\rref{def:rank})}
	\label{fig:rank}
\end{figure}

\section{Details of the Example}
\label{app:example}

This appendix reports details for \rref{sec:examples}.
Figures~\ref{fig:follower_dist} and~\ref{fig:leader} prove the open premises of \rrefexamples{ex:derivation}{ex:follower_proof}.

\begin{figure}[h]
	\small
	\begin{subfigure}{\textwidth}
		\small
		\begin{prooftree}[shape=justified, JustifiedScoresWidth=.9\textwidth]
			\Axiom{$*$}
			
			\UnaryInf{$\ldots, \tvar_1 = \tvar_0 \cdot \comItem{\durchan, t} \vdash F(\tvar_1, x_f + t \cdot v_f, \waitvar + t)$}

			\RuleNameRight{forallR, implR}{}
			\UnaryInf{$\ldots \vdash \fa{\tvar_1} \big( \tvar_0 = \tvar_0 \cdot \comItem{\durchan, t} \rightarrow F(\tvar_1, x_f + t \cdot v_f, \waitvar + t) \big)$}

			\RuleNameRight{send}{}
			\UnaryInf{$\ldots \vdash [ \recchp{\tvar_0}{\send{\durchan}{t}} ] F(\tvar_0, x_f + t \cdot v_f, \waitvar + t)$}

			\RuleNameRight{composition, assign}{}
			\UnaryInf{$\begin{aligned}
				&F, \tvar_0 = \recvar \cdot \comItem{\ch{vel}, \tarvelo}, \A(\tvar_0), d \le \periodicity \maxvelo, t \ge 0, \solutionDomain{\waitvar + s \le \periodicity} \\
				&\qquad\qquad\qquad \vdash [ x_f \ceq x_f + t \cdot v_f \seq \waitvar \ceq \waitvar + t ] [ \recchp{\tvar_0}{\send{\durchan}{t}} ] F(\tvar_0, x_f, \waitvar)
			\end{aligned}$}

			\RuleNameRight{acSolutionR}{}
			\UnaryInf{$F, \tvar_0 = \recvar \cdot \comItem{\ch{vel}, \tarvelo}, \A(\tvar_0), d \le \periodicity \maxvelo \vdash [ \recchp{\tvar_0}{\Plant_f} ] F(\tvar_0)$}
		\end{prooftree}
		\caption{}
		\label{fig:follower_velo_d_unsafe}
	\end{subfigure}

	\begin{subfigure}{\textwidth}
		\begin{prooftree}[shape=justified, JustifiedScoresWidth=.9\textwidth]
			\Axiom{$*$}
			
			\UnaryInf{$\ldots, \tvar_1 = \tvar_0 \cdot \comItem{\durchan, t}
				\vdash F(\tvar_1, v_0, x_f + t \cdot v_0, \waitvar + t))$}
			
			\RuleNameRight{forallR, implR}{}
			\UnaryInf{$\ldots
				\vdash \fa{\tvar_1} (\tvar_1 = \tvar_0 \cdot \comItem{\durchan, t} \rightarrow F(\tvar_1, v_0, x_f + t \cdot v_0, \waitvar + t))$}
			
			\RuleNameRight{send}{}
			\UnaryInf{$\ldots, t\ge0, \solutionDomain{\waitvar + s \le \periodicity} 
				\vdash [ \recchp{\tvar_0}{\send{\durchan}{t}} ] F(\tvar_0, v_0, x_f + t \cdot v_0, \waitvar + t)$}

			\RuleNameRight{composition, assign}{}
			\UnaryInf{$\begin{aligned}
				& \ldots, t\ge0, \solutionDomain{\waitvar + s \le \periodicity} \\
				& \qquad\qquad \vdash [ x_f \ceq x_f + t \cdot v_0 \seq \waitvar \ceq \waitvar + t ] [ \recchp{\tvar_0}{\send{\durchan}{t}} ] F(\tvar_0, v_0, x_f, \waitvar)
			\end{aligned}$}
	
			\RuleNameRight{acSolutionR}{}
			\UnaryInf{$F_0, d = \mespos - x_f, d \le \periodicity \maxvelo, \orange{0}{v_0}{\safevelo{d}}, \waitvar = 0 \vdash [ \recchp{\tvar_0}{\Plant_f(v_0)} ] F(\tvar_0, v_0)$}
	
			\RuleNameRight{assign}{}
			\UnaryInf{$F_0, d = \mespos - x_f, d \le \periodicity \maxvelo, \orange{0}{v_0}{\safevelo{d}} \vdash [ \waitvar \ceq 0 ] [ \recchp{\tvar_0}{\Plant_f(v_0)} ] F(\tvar_0, v_0)$}

			\RuleNameRight{test, implR}{}
			\UnaryInf{$F_0, d = \mespos - x_f, d \le \periodicity \maxvelo \vdash [ 
				\test{\orange{0}{v_0}{\safevelo{d}}}
			] [ \waitvar \ceq 0 ] [ \recchp{\tvar_0}{\Plant_f(v_0)} ] F(\tvar_0, v_0)$}
	
			\SideAx{$\triangleright$ \rref{fig:follower_dist_d_safe}}

			\RuleNameRight{nondetAssign, forallR}{}
			\UnaryInf{$\ldots \vdash [ 
					v_f \ceq * ] [ \test{\orange{0}{v_f}{\safevelo{d}}}
				] [ \waitvar \ceq 0 ] [ \recchp{\tvar_0}{\Plant_f(v_f)} ] F(\tvar_0, v_f)$}
			
			\RuleNameRight{composition}{}
			\UnaryInf{$F_0, d = \mespos - x_f, d \le \periodicity \maxvelo \vdash [ 
				v_f \ceq * \seq \test{\orange{0}{v_f}{\safevelo{d}}}
			\seq \waitvar \ceq 0 ] [ \recchp{\tvar_0}{\Plant_f} ] F(\tvar_0)$}

			\RuleNameRight{conditional}{}
			\UnaryInf{$F_0, d = \mespos - x_f \vdash [ 
			\ifstat{d \le \periodicity \maxvelo}{%
				\{v_f \ceq * \seq \test{\orange{0}{v_f}{\safevelo{d}}}\}
			} \seq \waitvar \ceq 0 ] [ \recchp{\tvar_0}{\Plant_f} ] F(\tvar_0)$}
	
			\RuleNameRight{composition, assign}{}
			\UnaryInf{$F_0 \vdash [ d \ceq \mespos - x_f \seq 
			\ifstat{d \le \periodicity \maxvelo}{%
				\{v_f \ceq * \seq \test{\orange{0}{v_f}{\safevelo{d}}}\}
			} \seq \waitvar \ceq 0 ] [ \recchp{\tvar_0}{\Plant_f} ] F(\tvar_0)$}
	
			\RuleNameRight{acReceiveRight}{}
			\UnaryInf{$F \vdash [ \receive{\ch{pos}}{\mespos} ] \acpair{\A, \true} 
			[ d \ceq \mespos - x_f \seq \ifstat{d \le \periodicity \maxvelo}{ \{ \alpha \} } \seq \waitvar \ceq 0 ] [ \Plant_f ] F(\recvar)$}
	
			\RuleNameRight{acMono, boxesDual}{}
			\UnaryInf{$F \vdash [ \receive{\ch{pos}}{\mespos} ] \acpair{\A, \true} 
			[ d \ceq \mespos - x_f \seq \ifstat{d \le \periodicity \maxvelo}{ \{ \alpha \} } \seq \waitvar \ceq 0 ] \acpair{\A, \true} [ \Plant_f ] \acpair{\A, \true} F$}
	
			\RuleNameRight{acComposition}{}
			\UnaryInf{$F \vdash [ d \ceq \mespos - x_f \seq 
				\ifstat{d \le \periodicity \maxvelo}{%
					\{v_f \ceq * \seq \test{\orange{0}{v_f}{\safevelo{d}}}\}
				} \seq \waitvar \ceq 0 ] \acpair{\A, \true} [ \Plant_f ] \acpair{\A, \true} F$}
		\end{prooftree}
		\caption{
			$F_0 \equiv F, \tvar_0 = \recvar \cdot \comItem{\ch{pos}, m}, \A(\tvar_0)$
			and
			$\alpha \equiv v_f \ceq * \seq \test{\orange{0}{v_f}{\safevelo{d}}}$
		}
		\label{fig:follower_dist_d_unsafe}
	\end{subfigure}

	\begin{subfigure}{\textwidth}
		\begin{prooftree}[shape=justified, JustifiedScoresWidth=.9\textwidth]
			\Axiom{$*$}

			\UnaryInf{$\ldots, t\ge0, \solutionDomain{\waitvar + s \ge \periodicity}, \tvar_1 = \tvar_0 \cdot \comItem{\durchan, t} \vdash F(\tvar_1, x_f + t \cdot v_f, \waitvar + t)$}

			\RuleNameRight{forallR, implR}{}
			\UnaryInf{$\ldots, t\ge0, \solutionDomain{\waitvar + s \ge \periodicity} \vdash \fa{\tvar_1} ( \tvar_1 = \tvar_0 \cdot \comItem{\durchan, t} \rightarrow F(\tvar_1, x_f + t \cdot v_f, \waitvar + t))$}

			\RuleNameRight{send}{}
			\UnaryInf{$\ldots, t\ge0, \solutionDomain{\waitvar + s \ge \periodicity} \vdash [ \recchp{\tvar_0}{\send{\durchan}{t}} ] F(\tvar_0, x_f + t \cdot v_f, \waitvar + t)$}

			\RuleNameRight{composition, assign}{}
			\UnaryInf{$\ldots, t\ge0, \solutionDomain{\waitvar + s \ge \periodicity} \vdash [ x_f \ceq x_f + t \cdot v_f \seq \waitvar \ceq \waitvar + t ] [ \recchp{\tvar_0}{\send{\durchan}{t}} ] F(\tvar_0, x_f, \waitvar)$}
	
			\RuleNameRight{acSolutionR}{}
			\UnaryInf{$F_0, d = \mespos - x_f, d > \periodicity \maxvelo, \waitvar = 0 \vdash [ \recchp{\tvar_0}{\Plant_f} ] F(\tvar_0)$}

			\RuleNameRight{assign}{}
			\UnaryInf{$F_0, d = \mespos - x_f, d > \periodicity \maxvelo \vdash
			 [ \waitvar \ceq 0 ] [ \recchp{\tvar_0}{\Plant_f} ] F(\tvar_0)$}
		\end{prooftree}
		\caption{}
		\label{fig:follower_dist_d_safe}
	\end{subfigure}

	\caption{Derivation of the subproof about the $\progtt{follower}$ in \rref{ex:derivation}}
	\label{fig:follower_dist}
\end{figure}

\begin{figure}[h]
	\begin{subfigure}{\textwidth}
		\begin{prooftree}[shape=justified, JustifiedScoresWidth=.87\textwidth]
			\Axiom{$*$}
			
			\UnaryInf{$L, 0 {\le} v_0 {\le} \maxvelo, \tvar_0 = \tvar \cdot \comItem{\ch{vel}, v_0}, t\ge0, \tvar_1 = \tvar_0 \cdot \comItem{\durchan, t} \vdash L(\tvar_1, x_l + t \cdot v_0)$}

			\RuleNameRight{forallR, implR}{}
			\UnaryInf{$\ldots, t\ge0 \vdash \fa{\tvar_1} ( \tvar_1 = \tvar_0 \cdot \comItem{\durchan, t} \rightarrow L(\tvar_1, x_l + t \cdot v_0) )$}

			\RuleNameRight{acSendRight}{}
			\UnaryInf{$\ldots, t\ge0 \vdash [ \recchp{\tvar_0}{\send{\durchan}{t}} ] \acpair{\true, \Commit(\tvar_0)} L(\tvar_0, x_l + t \cdot v_0)$}

			\RuleNameRight{assign}{}
			\UnaryInf{$\ldots, t\ge0 \vdash [ x_l \ceq x_l + t \cdot v_0 ] [ \recchp{\tvar_0}{\send{\durchan}{t}} ] \acpair{\true, \Commit(\tvar_0)} L(\tvar_0, x_l)$}

			\RuleNameRight{acSolutionR}{}
			\UnaryInf{$L, 0 {\le} v_0 {\le} \maxvelo, \tvar_0 = \recvar \cdot \comItem{\ch{vel}, v_0} \vdash [ \recchp{\tvar_0}{\Plant_l(v_0)} ] \acpair{\true, \Commit(\tvar_0)} L(\tvar_0)$}
			
			\SideAx{$\triangleright$ \rref{fig:leader_velo_skip}}

			\RuleNameRight{acSendRight}{}
			\UnaryInf{$L, 0 {\le} v_0 {\le} \maxvelo \vdash [ \send{\ch{vel}}{v_0} ] \acpair{\true, \Commit(\recvar)} [ \Plant_l(v_0) ] \acpair{\true, \Commit(\recvar)} L(\recvar)$}

			\RuleNameRight{acChoice, andR}{}
			\UnaryInf{$L, 0 {\le} v_0 {\le} \maxvelo \vdash [ \send{\ch{vel}}{v_0} \cup \skipProg ] \acpair{\true, \Commit} [ \Plant_l(v_0) ] \acpair{\true, \Commit} L$}

			\RuleNameRight{test, implR}{}
			\UnaryInf{$L \vdash [ \test{0 {\le} v_0 {\le} \maxvelo} ] [ \send{\ch{vel}}{v_0} \cup \skipProg ] \acpair{\true, \Commit} [ \Plant_l(v_0) ] \acpair{\true, \Commit} L$}
			
			\SideAx{$\triangleright$ by \RuleName{fol}}

			\RuleNameRight{composition, nondetAssign, forallR}{}
			\UnaryInf{$L \vdash [ v_l \ceq * \seq \test{0 {\le} v_l {\le} \maxvelo} ] [ \send{\ch{vel}}{v_l} \cup \skipProg ] \acpair{\true, \Commit} [ \Plant_l(v_l) ] \acpair{\true, \Commit} L$}

			\RuleNameRight{acNoCom, andR}{}
			\UnaryInf{$L \vdash [ v_l \ceq * \seq \test{0 {\le} v_l {\le} \maxvelo} ] \acpair{\true, \Commit} [ \send{\ch{vel}}{v_l} \cup \skipProg ] \acpair{\true, \Commit} [ \Plant_l ] \acpair{\true, \Commit} L$}

			\SideAx{$\triangleright$ \rref{fig:leader_pos}}

			\RuleNameRight{acComposition}{}
			\UnaryInf{$L \vdash [ \ctrlNotify ] \acpair{\true, \Commit} [ \Plant_l ] \acpair{\true, \Commit} L$}

			\RuleNameRight{andR}{}
			\UnaryInf{$L \vdash [ \ctrlNotify ] \acpair{\true, \Commit} [ \Plant_l ] \acpair{\true, \Commit} L \wedge [ \ctrlUpdate ] \acpair{\true, \Commit} [ \Plant_l ] \acpair{\true, \Commit} L$}

			\RuleNameRight{acComposition, acChoice}{}
			\UnaryInf{$L \vdash [ (\ctrlNotify \cup \ctrlUpdate) \seq \Plant_l ] \acpair{\true, \Commit} L$}
	
			\RuleNameRight{acLoop}{}
			\UnaryInf{$\Gamma \vdash [ \progtt{leader}^* ] \acpair{\true, \Commit} \val[x_0]{\ch{pos}} \le x_l$}
		\end{prooftree}
		\caption{
			Note that $[ \alpha ] \ac [ \beta ] \ac \psi \leftrightarrow [ \alpha ] \ac [ \beta ] \psi$ derives if $\beta$ is sufficiently simple
		}
		\label{fig:leader_velo_send}
	\end{subfigure}
	
	\begin{subfigure}{\textwidth}
		\begin{prooftree}[shape=justified, JustifiedScoresWidth=.87\textwidth]
			\Axiom{$*$}

			\UnaryInf{$L, 0 {\le} v_0 {\le} \maxvelo, t\ge0, \tvar_0 = \recvar \cdot \comItem{\durchan, t} \vdash L(x_l + t \cdot v_0, \tvar_0)$}
			
			\RuleNameRight{acSendRight}{}
			\UnaryInf{$L, 0 {\le} v_0 {\le} \maxvelo, t\ge0 \vdash [ \send{\durchan}{t} ] \acpair{\true, \Commit} L(x_l + t \cdot v_0, \recvar)$}

			\RuleNameRight{assign}{}
			\UnaryInf{$L, 0 {\le} v_0 {\le} \maxvelo, t\ge0 \vdash [ x_l \ceq x_l + t \cdot v_0 ] [ \send{\durchan}{t} ] \acpair{\true, \Commit} L(x_l)$}

			\RuleNameRight{acSolutionR}{}
			\UnaryInf{$L, 0 {\le} v_0 {\le} \maxvelo \vdash [ \Plant_l(v_0) ] \acpair{\true, \Commit} L$}

			\SideAx{$\triangleright$ by \RuleName{fol}}

			\UnaryInf{$L, 0 {\le} v_0 {\le} \maxvelo \vdash [ \skipProg ] [ \Plant_l(v_0) ] \acpair{\true, \Commit} L$}
			
			\RuleNameRight{acNoCom, andR}{}
			\UnaryInf{$L, 0 {\le} v_0 {\le} \maxvelo \vdash [ \skipProg ] \acpair{\true, \Commit} [ \Plant_l(v_0) ] \acpair{\true, \Commit} L$}
		\end{prooftree}
		\caption{}
		\label{fig:leader_velo_skip}
	\end{subfigure}

	\begin{subfigure}{\textwidth}
		\begin{prooftree}[shape=justified, JustifiedScoresWidth=.87\textwidth]
			\Axiom{$*$}

			\UnaryInf{$L, \tvar_0 = \tvar \cdot \comItem{\ch{pos}, x_l}, t\ge0, \tvar_1 = \tvar_0 \cdot \comItem{\durchan, t} \vdash L(\tvar_1, x_l + t \cdot v_l)$}

			\RuleNameRight{forallR, implR}{}
			\UnaryInf{$L, \tvar_0 = \tvar \cdot \comItem{\ch{pos}, x_l}, t\ge0 \vdash \fa{\tvar_1} ( \tvar_1 = \tvar_0 \cdot \comItem{\durchan, t} \rightarrow L(\tvar_1, x_l + t \cdot v_l))$}
			
			\RuleNameRight{acSendRight}{}
			\UnaryInf{$L, \tvar_0 = \tvar \cdot \comItem{\ch{pos}, x_l}, t\ge0 \vdash [ \recchp{\tvar_0}{\send{\durchan}{t}} ] \acpair{\true, \Commit(\tvar_0)} L(\tvar_0, x_l + t \cdot v_l)$}

			\RuleNameRight{assign}{}
			\UnaryInf{$L, \tvar_0 = \tvar \cdot \comItem{\ch{pos}, x_l}, t\ge0 \vdash [ x_l \ceq x_l + t \cdot v_l ] [ \recchp{\tvar_0}{\send{\durchan}{t}} ] \acpair{\true, \Commit(\tvar_0)} L(\tvar_0, x_l)$}

			\RuleNameRight{acSolutionR}{}
			\UnaryInf{$L, \tvar_0 = \tvar \cdot \comItem{\ch{pos}, x_l} \vdash [ \recchp{\tvar_0}{\Plant_l} ] \acpair{\true, \Commit(\tvar_0)} L(\tvar_0)$}

			\RuleNameRight{acSendRight}{}
			\UnaryInf{$L \vdash [ \send{\ch{pos}}{x_l} ] \acpair{\true, \Commit(\recvar)} [ \Plant_l ] \acpair{\true, \Commit(\recvar)} L(\recvar)$}
		\end{prooftree}
		\caption{}
		\label{fig:leader_pos}
	\end{subfigure}
	\caption{
		Derivation of the subproof about the $\progtt{leader}$ in \rref{ex:derivation},
		where $\Gamma \equiv \tvar_0 = \tvar, x_0 = x_l, \varphi$
		and $L \equiv \val[x_0]{\tvar\downarrow\ch{pos}} \le x_l \wedge \Commit$
	}
	\label{fig:leader}
\end{figure}

\else\fi

\end{document}